\documentclass{article}
\usepackage{mystyle}

\begin{document}
\title{Maximum Likelihood Estimation for Entity Ranking under Iterative Synthetic Data Augmentation}
\author{
    Yiqiao Jin\thanks{Department of Statistics and Data Science, Washington University in St. Louis. Email: \texttt{j.yiqiao@wustl.edu}} \qquad
    Mengxin Yu\thanks{Department of Statistics and Data Science, Washington University in St. Louis. Email: \texttt{myu@wustl.edu}}
}
\maketitle
\begin{abstract}
   We study entity ranking and inference under the Bradley--Terry--Luce (BTL) model using pairwise comparisons collected over sparse comparison graphs. In practice, collecting high-quality human judgments can be expensive and time-consuming, motivating the use of synthetic data augmentation in model training. We analyze an iterative synthetic \textit{augmentation} workflow in which synthetic comparisons generated from fitted models are successively added to the original dataset. In this process, the proportion of real data may vanish as the number of iterations grows. However, emerging literature has shown that recursive training on synthetic data can lead to \textit{model collapse}, raising concerns about the statistical reliability of such augmentation procedures. To this end, we systematically analyze the resulting MLE in finite-sample, high-dimensional regimes.  For the resulting iterative maximum likelihood estimator (MLE), we derive its optimal finite-sample  $\ell_2$ and $\ell_{\infty}$ statistical rates and establish its asymptotic normality under natural identifiability conditions. We further characterize regimes in which model collapse is avoided despite the diminishing fraction of real data. We validate our theoretical findings through large-scale numerical experiments and an application to the Arena Human Preference 140k dataset.
\end{abstract}
{\small
}
\clearpage
\section{Introduction}
Ranking arises in various real-world applications, include web search \citep{dwork2001rank, wang2016learning}, recommendation systems \citep{baltrunas2010group, he2018adversarial}, scientific journals \citep{stigler1994citation, ji2022co}, election candidates \citep{plackett1975analysis, mattei2013preflib}, assortment optimization \citep{rusmevichientong2010dynamic, aouad2018approximability}, sports competition \citep{massey1997statistical, turner2012bradley}, and even gene expression \citep{boulesteix2009stability, plaisier2010rank}. Therefore, ranking problems have attracted significant attention across many fields such as psychology \citep{thurstone1927method, thurstone2017law}, healthcare \citep{wang2018association, adelman2020efficient}, operations research  \citep{McFadden1974conditional, mohammadi2020ensemble}, statistics \citep{hunter2004mm, chen2019bspectral,fan2025ranking}, and artificial intelligence \citep{ouyang2022training, chiang2024chatbot}.

One of the most widely used parametric models for pairwise comparisons in ranking problems is the Bradley--Terry--Luce (BTL) model \citep{bradley1952rank, luce2012individual}. Under the BTL model, a collection of $n$ items is ranked according to their latent preference scores $\theta_i^{\ast}$, for $i=1,\ldots,n$. Specifically, for any pair of items $i$ and $j$, the probability that item $i$ is preferred over item $j$ is given by
\[
\mathbb{P}\left(\text{item $i$ is preferred over item $j$}\right)
=
\frac{e^{\theta_i^{\ast}}}{e^{\theta_i^{\ast}}+e^{\theta_j^{\ast}}}.
\]
To model which pairs of items are compared, following the literature \citep{chen2015spectral,chen2019bspectral,chen2022partial,fan2024uncertainty}, we assume that the comparisons among the $n$ items are generated according to an Erd\H{o}s--R\'enyi random comparison graph \citep{erdHos1960evolution}. In particular, each unordered pair $(i,j)$ is independently included in the comparison graph with probability $p$, and for each observed pair, $L$ independent comparisons are collected. Given the resulting binary comparison outcomes, the goal is typically to recover the underlying preference scores, or to identify the $K$ items with the largest preference scores. Beyond pairwise comparisons, the Plackett-Luce (PL) model \citep{plackett1975analysis} extends this framework to $M$-way comparisons by modeling a full ranking over multiple items.

It is worth noting that, under both the Bradley--Terry--Luce (BTL) and Plackett--Luce (PL) models, ranking recovery relies heavily on the quality, authenticity, and reliability of human preference labels. However, collecting a sufficiently large volume of human preference data can be expensive and time-consuming \citep{ji2024reinforcement}. To alleviate this burden, recent work has explored various data synthesis approaches that leverage generative models to supplement human annotations and facilitate downstream inference, particularly in the reinforcement learning from human feedback (RLHF) and reinforcement learning from AI feedback (RLAIF) literature \citep{ouyang2022training, bai2022constitutional, lee2023rlaif}.

While synthetic data offer a promising way to alleviate the scarcity and high cost of human annotations, the phenomenon of \textit{model collapse} \citep{shumailov2024ai} has raised concerns that recursively training  models on model-generated data can gradually degrade both model performance and distributional fidelity, eventually rendering the trained models ineffective. This issue is becoming increasingly relevant as model-generated content becomes more prevalent on the Internet, while the amount of training data used by large language models continues to grow rapidly. For example, Llama~1 was trained on 1.4 trillion tokens \citep{touvron2023llama}, compared with 2 trillion for Llama~2 \citep{touvron2023llama2}, 15 trillion for Llama~3 \citep{grattafiori2024llama3}, and 30 trillion for Llama~4 \citep{adcock2026llama}.

Since then, a rapidly growing literature has studied the empirical mechanisms, theoretical properties, and potential mitigation strategies for model collapse \citep{alemohammad2024self, bertrand2024stability, dohmatob2024model, gerstgrasser2024model, dey2024universality,dohmatob2025strong, barzilai2026models}. Among these works, \cite{gerstgrasser2024model} considered a more realistic \textit{augmentation workflow} in linear models, where the original real data remain available at each generation of model fitting, while synthetic data are gradually added to the training set. Similar results have been further explored in exponential families \citep{dey2024universality,barzilai2026models}. However, these works either rely on estimators with closed-form expressions or require restrictive assumptions on the loss function, making their analysis difficult to extend to ranking problems and other application settings, particularly high-dimensional ones. In addition, they mainly focus on the $\ell_2$ statistical error and do not provide the entrywise error control needed for accurate ranking recovery. Moreover, most existing results are asymptotic and developed under low-dimensional regimes, while the high-dimensional regime has received little attention.

Motivated by the phenomenon of model collapse and the limitations discussed above, this paper studies entity ranking under iterative synthetic data augmentation and asks whether such a framework can avoid progressive degradation in ranking accuracy.  Specifically, we study the iterative MLE for each round $t$,  under the BTL model and derive its \emph{non-asymptotic} statistical rates in both $\ell_{2}$ and $\ell_{\infty}$ norms for estimating the latent preference scores $\theta^{\ast}.$ We show that both rates are optimal under a sparse comparison graph, provided that $\log t \left[ \frac{\sqrt{\log n}}{\sqrt{nL}p} + \frac{1}{\sqrt{np}} \left( 1 + \frac{\log n}{\sqrt{np}} \right) \right] = \mathcal{O}(1)$ and $\frac{1}{t+1}  \left( \frac{\log n}{np} + \frac{1}{\sqrt{np}} \right) \leq c$ for some fixed constant $c > 0$. These non-asymptotic guarantees remain to hold even as the fraction of real data vanishes under the accumulating setting. 

Our analysis relies  on a induction argument together with several leave-one-out constructions \cite{chen2019bspectral,chen2022partial}, which allow us to handle the temporal dependence across iterations and control how the estimation errors propagate over time. As a consequence, when the number of retraining iterations grows at most polynomially with $n$, the iterative MLE remains consistent and enjoys optimal $\ell_2,$ $\ell_{\infty}$ statistical rates, showing that model collapse can be avoided in this setting. We further establish the asymptotic normality of the iterative MLE, which provides a theoretical basis for uncertainty quantification under iterative retraining. Finally, we conduct extensive numerical experiments to support our theoretical results and demonstrate the proposed framework on the Arena Human Preference 140K dataset.

To summarize, this work makes several contributions. We study the MLE under the BTL model with iterative synthetic data augmentation and develop a new induction framework to obtain non-asymptotic error bounds for the iterative MLE in both $\ell_2$ and $\ell_\infty$ norms. We further establish optimal statistical rates under a sparse comparison graph, showing that model collapse is not inevitable in ranking problems when the number of retraining iterations grows at most polynomially in $n$. In addition, we derive the asymptotic normality of the iterative MLE and validate our theoretical findings through large-scale numerical experiments on both synthetic and real data.

\subsection{Related Works}

In this section, we review the literature on ranking, as well as related work on iterative model retraining with synthetic data augmentation and its connection to model collapse.
\paragraph{Ranking.}
Recently, a growing body of literature has studied ranking problems under the Bradley--Terry--Luce (BTL) model \citep{bradley1952rank, luce2012individual}. Among these works, \cite{negahban2012iterative} introduced \emph{Rank Centrality}, a spectral ranking method, and established optimal $\ell_2$ statistical rates for estimating the latent preference scores. Later, \cite{chen2015spectral} proposed a two-stage procedure combining a spectral method with maximum likelihood estimation (MLE) under an Erd\H{o}s--R\'enyi comparison graph, and established  top-$K$ items recovery with optimal sample complexity guarantees in the sparsest possible regime of the comparison graph. Under a similar setting, \cite{chen2019bspectral} established optimal statistical rates in both the $\ell_2$ and $\ell_\infty$ norms for the latent preference scores using the maximum likelihood estimator (MLE) or spectral methods only, showing that both achieve exact top-$K$ recovery optimality when the condition number is of constant order. In addition, \cite{chen2022partial} further studied partial recovery and proved that both estimators from MLE and spectral methods attain optimal statistical rates in the $\ell_2$ and $\ell_\infty$ norms, whereas spectral methods can be suboptimal when the condition number is allowed to vary.

Beyond estimation and ranking recovery, several recent works have developed asymptotic distribution theory and uncertainty quantification for ranking scores. Early results by \cite{simons1999asymptotics} and \cite{han2020asymptotic} established the asymptotic normality of the MLE for fully connected graphs ($p=1$) or relatively dense comparison graphs satisfying $p \gtrsim (\log n)^{1/5}/n^{1/10}$, where $p$ denotes the edge probability of the Erd\H{o}s--R'enyi comparison graph and $n$ is the number of items being compared.  More recently, \cite{liu2023lagrangian} developed a Lagrangian debiasing approach to establish asymptotic distributional results in the sparsest graph regime, where $np \gtrsim \log n$. Their results, however, require a relatively large number of repeated comparisons, with $L \gtrsim n^{2}\log^{2} n$. This requirement was substantially relaxed by \cite{gao2023uncertainty}, who introduced a leave-two-out technique and established asymptotic normality under sparse comparison graphs with $p\asymp 1/n$ up to logarithmic factors and only $L=O(1)$ repeated comparisons. They further constructed simultaneous confidence intervals using a Bonferroni-type procedure. Subsequently, \cite{fan2025ranking} developed sharper simultaneous inference procedures for ranks under the Plackett--Luce (PL) model \citep{plackett1975analysis}. Their approach reduces rank uncertainty to inference on pairwise score differences and approximates the distribution of the resulting maximum statistic using a valid Gaussian multiplier bootstrap.

Recent studies have also extended the classical framework for ranking estimation and inference under the BTL model in several directions, including nonparametric comparison models \citep{wang2024confidence, fan2025uncertainty, li2025efficient}, multiway comparisons \citep{han2025unified, fan2025ranking, fan2026spectral}, and models incorporating individual-level covariate information \citep{li2022bayesian, fan2024covariate, fan2024uncertainty, li2025efficient, dong2026statistical}.

Ranking methods have also become increasingly important in a range of AI applications. For example, in reinforcement learning from human feedback (RLHF) \citep{ouyang2022training}, human annotators compare pairs of responses to the same prompt, and the resulting preference data are used to fine-tune language models so that their outputs better align with human preferences. Ranking methods also play a central role in language model evaluation \citep{chiang2024chatbot, zhang2025fisher, cen2026perturbation, liu2025uncertainty}, where they are used to compare and rank the performance of different language models across a wide range of tasks based on human preference data.

For a comprehensive review of recent developments in the BTL model and its extensions, we refer readers to \cite{fang2026recent}.

\paragraph{Model collapse.}

With the rapid development of artificial intelligence, publicly available text, image, and video corpora are increasingly mixed with AI-generated content. Repeatedly training future models on such data may lead to \textit{model collapse}, a degenerative phenomenon in which the learned distribution progressively drifts away from the original data-generating distribution through repeated retraining on synthetic data, leading to degraded performance, reduced diversity, and compromised reliability \citep{briesch2023large, shumailov2024ai, xie2026reviewing}. A growing body of empirical and theoretical work has investigated model collapse under various forms of synthetic data contamination.

A first line of work considers a fully synthetic setting, often referred to as the \textit{discard workflow}, in which the model is retrained exclusively on synthetic data generated by the model from the previous iteration, with no real data retained. In this setting, \cite{shumailov2024ai} showed that even for Gaussian distribution estimation, the estimated covariance can progressively collapse toward zero while the estimated mean drifts away from the true mean. Similar degradation was later established for high-dimensional linear regression by \cite{dohmatob2024model}.

A complementary line of work asks whether model collapse can be mitigated when real data are retained during retraining. \cite{gerstgrasser2024model} and \cite{kazdan2024collapse} considered an accumulating-data setting, often referred to as the \textit{augmentation workflow}, in which a fixed real dataset is retained throughout training and synthetic data generated in previous iterations are progressively added. Under this regime, they showed that model collapse can be avoided for regression and Gaussian distribution estimation under suitable conditions, and provided supporting empirical evidence across a range of LLMs. Extending these results beyond linear regression, \cite{dey2024universality} showed that the augmentation workflow remains effective for exponential-family models, establishing a universal $\pi^2/6$-type risk bound. More recently, \cite{barzilai2026models} studied iterative maximum likelihood estimation in the low-dimensional setting and established non-asymptotic error bounds under mild regularity conditions. Their results show that model collapse can be avoided even when the fraction of real data vanishes over time.

More recent work has further refined this picture by studying settings in which real and synthetic data are mixed at each retraining round. \cite{bertrand2024stability} considered an iterative retraining procedure in which each generation combines a fixed real dataset with synthetic data generated by the most recent model, and established stability when the fraction of synthetic data is sufficiently small. In contrast, \cite{alemohammad2024self} studied a fresh-data augmentation setting, in which new real data are collected at each generation, and empirically identified a trade-off between exploiting synthetic data and preserving generative model performance. Under this setting, \cite{he2025golden} derived an optimal strategy for combining real and synthetic data in linear regression and Gaussian estimation, while \cite{garg2026preventing} characterized high-dimensional asymptotic prediction risks and studied optimal mixing weights in overparameterized linear regression. More broadly, \cite{wu2026does} investigated model collapse in systems of interacting models, establishing finite-sample guarantees for linear regression and asymptotic results for general M-estimators.

Existing studies have also examined model collapse across a broad range of settings, including regression models \citep{dohmatob2024model, dohmatob2025strong, dey2024universality, vu2025happens, garg2026preventing, he2025golden}, maximum likelihood estimation \citep{bertrand2024stability, suresh2024rate, kanabar2025model, barzilai2026models}, variational autoencoders \citep{shumailov2023curse, shumailov2024ai, gerstgrasser2024model}, diffusion models \citep{alemohammad2024self, yoon2024model, Shi2025a, khelifa2026error}, and flow-based models \citep{zhu2024analyzing, kim2025simple, kwon2026balanced}. For a recent review of model collapse across different application domains, together with strategies for mitigating it, we refer readers to \cite{xie2026reviewing}.

While synthetic data can be harmful when used indiscriminately, recent advances in generative models have also opened up new ways to use synthetic data to improve statistical inference. Examples include the bootstrap \citep{tran2026generative}, imbalanced classification \citep{lyu2025bias, ma2026synthetic}, and prediction-powered inference \citep{angelopoulos2023prediction, van2026calibeating}. We refer readers to \cite{abdel2026harnessing} for a recent overview of this literature, which organizes existing approaches into four paradigms and discusses how synthetic data can be used to support trustworthy scientific discovery.

\subsection{Notation}
Unless otherwise stated, we let $\mathbb{N} := \{0,1, 2,\ldots\}$ denote the set of non-negative integers and write $[n] := \{1,2,\dots,n\}$ for any positive integer $n$. For a vector $\mathbf{x} \in \mathbb{R}^{d}$ and $q \in (0,\infty]$, we denote its $\ell_q$-norm by $\|\mathbf{x}\|_{q} := \|\mathbf{x}\|_{\ell_q}$. The inner product $\langle \mathbf{x}, \mathbf{y} \rangle$ between $\mathbf{x}, \mathbf{y} \in \mathbb{R}^{d}$ is the usual Euclidean inner product $\mathbf{x}^{\top}\mathbf{y}$. We use $\boldsymbol{1}_{d}$ to denote the $d$-dimensional vector of all ones. For a matrix $\mathbf{X} \in \mathbb{R}^{d_{1} \times d_{2}}$, we write $\mathbf{X} \succeq 0$ (respectively $\mathbf{X} \preceq 0$) if $\mathbf{X}$ is positive semidefinite (respectively negative semidefinite), and we denote by $\mathbf{X}^{\dagger}$ its Moore–Penrose pseudoinverse. Given two nonnegative sequences $\{a_{n}\}$ and $\{b_{n}\}$, the notation $a_{n} \lesssim b_{n}$ or $a_{n} = \mathcal{O}(b_{n})$ means that there exists a constant $C > 0$ such that $a_{n} \le C b_{n}$. Similarly, $a_{n} \gtrsim b_{n}$ or $a_{n} = \Omega(b_{n})$ means that there exists $C > 0$ such that $a_{n} \ge C b_{n}$. We write $a_{n} \asymp b_{n}$ or $a_{n} = \Theta(b_{n})$ if both $a_{n} \lesssim b_{n}$ and $a_{n} \gtrsim b_{n}$ hold. Finally, $a_{n} \ll b_{n}$ or $a_{n} = o(b_{n})$ means $\lim_{n \to \infty} a_{n}/b_{n} = 0$. Throughout, $\mathbb{P}$ and $\mathbb{E}$ denote probability measure and expectation whose distribution is determined from the context. We use $\xrightarrow{p}$ to denote convergence in probability and $\rightsquigarrow$ to denote convergence in distribution.

\subsection{Roadmap}
The remainder of the paper is organized as follows. We introduce the BTL model under the augment workflow, formalize the estimation problem, and derive the corresponding iterative maximum likelihood estimator in \S\ref{sec:Model_Formulation}. In \S\ref{sec:Main_Results}, we establish our main theoretical results, including non-asymptotic estimation error bounds and asymptotic normality for the iterative MLE. In \S\ref{sec:Numerical_Results}, we conduct large-scale simulation evidence on synthetic and real data that supports the theoretical results. In \S\ref{sec:Proof_Sketch_of_Theorem_1}, we outline the main ideas of the proofs and provide a flowchart to illustrate the logical structure. Supplementary technical results and proofs are deferred to Appendix~\ref{sec:Preliminaries}--\ref{appendix:last}.


\section{Model Formulation}
\label{sec:Model_Formulation}

In this section, we formulate the problem and introduce the iterative synthetic augmentation framework. The procedure begins by fitting the model to the original real data. At each subsequent iteration, synthetic comparisons are generated from the current fitted model and incorporated into the training data, after which the model is retrained on the augmented dataset. We next formalize this iterative procedure and introduce the corresponding model and data-generating process.

 We assume there are $n$ items and each item is associated with a latent score $\theta_i^{\star}$, and we let $\theta^{(0)} \equiv \theta^{\ast} := \left( \theta_{1}^{\ast}, \cdots, \theta_{n}^{\ast} \right)^{\top} \in \Theta$ denote as the latent score vector. For identifiability, we impose the standard identifiablity condition and define the parameter space as $\Theta := \left\{ \theta \in \mathbb{R}^{n}: \boldsymbol{1}^{\top} \theta = 0 \right\}$, as widely adopted in the existing ranking literature \citep{chen2019bspectral,chen2022partial, fan2024covariate}. We further assume that the latent scores lie in a bounded range $[\theta_{\min}, \theta_{\max}]$ for some fixed lower and upper bounds $\theta_{\min}$ and $\theta_{\max}$.

In practice, collecting high-quality human-labeled  data is often expensive. A natural alternative is to iteratively augment the training set with synthetic data generated by the current model and update the model parameters accordingly. Our goal is to formalize the statistical behavior of this iterative workflow and provide rigorous theoretical guarantees for recovering the latent score vector $\theta^\star$ using a mix of human-labeled and synthetic data.


Next, we discuss our data generation process. At any iteration $T \in \mathbb{N}$, let $\mathcal{G}^{(T)} := \left( [n], \mathcal{E}^{(T)} \right)$ denote the comparison graph, with adjacency matrix $G^{(T)}$, where $(i,j) \in \mathcal{E}^{(T)}$, or equivalently $G_{ij}^{(T)} = 1$, if and only if items $i$ and $j$ are compared. Throughout the paper, we assume that the comparison graphs in all rounds are independently generated according to the Erd\H{o}s--R\'enyi model,
$\mathcal{G}^{(T)} \overset{\mathrm{i.i.d.}}{\sim} \mathcal{G}(n,p).$
Equivalently, the edge indicators satisfy
$G_{ij}^{(T)}\overset{\mathrm{i.i.d.}}{\sim} \operatorname{Bernoulli}(p), 1 \le i < j \le n,$
with $G_{ji}^{(T)} = G_{ij}^{(T)}$ and $G_{ii}^{(T)} = 0$.

In addition, at each iteration $T=0,1,\ldots$, for every pair $(i,j)\in\mathcal{E}^{(T)}$, we observe $L$ independent and identically distributed pairwise comparison outcomes

$$
y_{i,j}^{(T,l)}
\overset{\mathrm{i.i.d.}}{\sim}
\operatorname{Bernoulli}\left(
\sigma\left(\theta_j^{(T)}-\theta_i^{(T)}\right)
\right),
\qquad l=1,\ldots,L,
$$
where $\sigma(\cdot)$ denotes the sigmoid function, defined as $\sigma(a)=1/(1+e^{-a})$, and $y_{i,j}^{(T,l)}=1-y_{j,i}^{(T,l)}$. Here, $\theta^{(0)} = \theta^\ast$ denotes the true latent score vector. For each iteration $T \ge 1$, we let $\theta^{(T)}$ be the MLE estimator computed from all data accumulated up to round $T-1$ (formally defined in \eqref{eq:argminLossT}), which is then used to generate synthetic pairwise comparisons for iteration $T$. Consequently, the initial preference observations $y_{i,j}^{(0,l)}$ for $l \in [L]$ represent real human feedback drawn from the BTL model with ground-truth score $\theta^\star$, whereas all subsequent preference feedbacks ($\{y_{i,j}^{(T,l)}\}$ with $T \ge 1$) are synthetically generated from the fitted model of the preceding round under the BTL model.

Next, we rigorously define the iterative MLE obtained at each round of iteration. For ease of presentation, we represent the collection of sufficient statistics as
\begin{align*}
    Y^{(\leq T)} := \bigcup_{t=0}^{T} Y^{(t)}, \quad Y^{(T)} := \left\{ y^{(T)}_{j,i} : (i,j) \in \mathcal{E}^{(T)},\, i>j \right\}, \quad y^{(T)}_{j,i} = \frac{1}{L}\sum_{l=1}^{L} y_{j,i}^{(T,l)}.
\end{align*}

Under the identifiability constraint imposed by $\Theta$ \citep{chen2019bspectral,chen2022partial,fan2024covariate}, we define the iterative MLE at each round $T\geq 0$ as:
\begin{align} \label{eq:argminLossT}
    \theta^{(T+1)} = \arg\min_{\theta \in \Theta} \ \mathcal{L}^{(\leq T)}(\theta),
\end{align} 
where 
\begin{align*}
     \mathcal{L}^{(\leq T)}(\theta) 
     := 
     \sum_{t=0}^{T} \ell_{t} (\theta)
     =
     \sum_{t=0}^{T} \sum_{(i,j) \in \mathcal{E}^{(t)},\, i>j} \left\{ - y_{j,i}^{(t)} (\theta_i - \theta_j) + \log\left(1 + e^{\theta_i - \theta_j} \right) \right\}.
\end{align*}
This means that, at each round $T$, the updated parameter $\theta^{(T+1)}$ is estimated using all data accumulated up to that round, including the original human comparison data and the synthetic comparisons generated in the subsequent rounds. As $T$ increases, the amount of synthetic data continues to grow, while the amount of original human data remains fixed. Consequently, the proportion of real data in the augmented dataset gradually decreases. We refer the reader to Algorithm~\ref{algorithm:unregularized_IMLE} for a summary of the iterative MLE under the synthetic augmentation workflow.


The cumulative augmentation workflow described above has been studied in linear models, exponential families, and general maximum likelihood estimation \citep{alemohammad2024self, bertrand2024stability, dohmatob2024model, gerstgrasser2024model, dey2024universality,dohmatob2025strong, barzilai2026models}. However, it remains unclear how such synthetic data augmentation behaves in high-dimensional ranking problems, particularly from a non-asymptotic perspective. We study this workflow under the Bradley--Terry--Luce (BTL) model over sparse comparison graphs, deriving optimal finite-sample $\ell_2$ and $\ell_\infty$ statistical rates and establishing asymptotic normality for the iterative MLE. Our results characterize high-dimensional, non-asymptotic regimes in which model collapse can be avoided even as the proportion of real data vanishes.

\begin{algorithm}[h] 
\caption{Iterative MLE for BTL Model Under the Augment Workflow}
\label{algorithm:unregularized_IMLE}
\begin{algorithmic}[1] 
\Require Parameter space $\Theta \subseteq \mathbb{R}^n$; number of samples  per iteration; 
\State Set $\theta^{(0)} := \theta^\ast$. 
\For{$T = 0, 1, 2, \dots$}
    \State Construct pairwise comparison set $\mathcal{E}_{T}$ from $\mathcal{G}^{(T)} \sim \mathcal{G}(n,p)$.
    \State Collect comparison samples $Y^{(T)} := \{y^{(T)}_{j,i} \mid (i,j) \in \mathcal{E}^{(T)},\, i>j \}$: 
    \begin{align*}
        y_{j,i}^{(T)} = \frac 1L\sum_{l=1}^{L}y_{j,i}^{(T,l)}, \text{where } y_{j,i}^{(T,l)} \overset{\mathrm{i.i.d.}}{\sim} \operatorname{Bernoulli}\!\left( \sigma(\theta_i^{(T)}-\theta_j^{(T)})\right) \text{ with } \ell=1,\ldots,L.
    \end{align*}
    \State Define cumulative dataset: $Y^{(\leq T)} := \bigcup_{t=0}^{T} Y^{(t)}$.
    \State Update model parameters using $Y^{(\leq T)}$:
    \begin{align*}
        \theta^{(T+1)}
        &:= 
        \arg\min_{\theta \in \Theta} \ \mathcal{L}^{(\leq T)}(\theta),
        \\ 
        \mathcal{L}^{(\leq T)}(\theta) 
        &:= 
        \sum_{t=0}^{T} \ell_{t} (\theta),
        \\ 
        \textrm{where~} \ell_{t}(\theta) 
        &:= 
        \sum_{(i,j) \in \mathcal{E}^{(t)},\, i>j} \left\{ - y_{j,i}^{(t)} (\theta_i - \theta_j) + \log\left(1 + e^{\theta_i - \theta_j} \right) \right\}.
    \end{align*}
\EndFor
\end{algorithmic}
\end{algorithm}
\section{Main Results}
\label{sec:Main_Results}
In this section, we formally establish that the iterative MLE obtained at each round of Algorithm~\ref{algorithm:unregularized_IMLE} achieves optimal statistical rates under mild conditions. In particular, we derive non-asymptotic $\ell_2$ and $\ell_{\infty}$ bounds on the estimation error relative to the ground-truth parameter $\theta^\star$. Crucially, our results demonstrate that model collapse is not inevitable, even as the number of iterations $T$ grows polynomially with $n$ and the fraction of real human data vanishes to zero.
\subsection
{Rate of Convergence of Iterative MLE} \label{subsec:Rate_of_Convergence}

We summarize our theoretical guarantees for the estimators $\theta^{(t)}$, $t\geq 1$, in the following Theorem~\ref{thm:iterative_loss}. 

\begin{theorem}[$\ell_2$- and $\ell_{\infty}$- statistical rates]
\label{thm:iterative_loss}
Consider the sequence of MLEs $\{\theta^{(t)}\}_{t=1}^{\infty}$ estimated from BTL model and the iterative synthetic data augment workflow described in Algorithm~\ref{algorithm:unregularized_IMLE}. Suppose that $L \leq d_{0}\cdot n^{d_{1}}$ for any fixed constants $d_{0}, d_{1} > 0$. Consider iterations $t$ satisfying
 $t \lesssim \frac{n^{3}}{p\log n}$, $\log t \left[ \frac{\sqrt{\log n}}{\sqrt{nL}p} + \frac{1}{\sqrt{np}} \left( 1 + \frac{\log n}{\sqrt{np}} \right) \right] = \mathcal{O}(1)$ and $\frac{1}{t+1}  \left( \frac{\log n}{np} + \frac{1}{\sqrt{np}} \right) \leq c$ for some fixed constant $c > 0$. Then, with probability at least $1-\mathcal{O}(n^{-5})$, the following bounds hold uniformly over all iterations $t$ satisfying the above conditions:
\begin{align*}
    \norm{\theta^{(t)} - \theta^{\ast}}_{2} 
    &\lesssim 
    \sqrt{\frac{1}{pL}}, 
    \\
    \norm{\theta^{(t)} - \theta^{\ast}}_{\infty} 
    &\lesssim 
    \sqrt{\frac{\log n}{npL}}.
\end{align*}
\end{theorem}

\begin{proof}
For interested readers, we refer to Section~\ref{appendix:proof_outline} for the complete proof.
\end{proof}

Theorem~\ref{thm:iterative_loss} establishes that the iterative MLE achieves error rates of $\mathcal{O}\big(\sqrt{1/pL}\big)$ and $\mathcal{O}\big(\sqrt{\log n/npL}\big)$ under the $\ell_2$ and $\ell_{\infty}$ losses, respectively. Both rates are known to be optimal in the existing ranking literature \citep{chen2019bspectral,chen2022partial,Negahban2017rank,fan2024covariate,fan2025ranking}. Crucially, these results demonstrate that model collapse is not inevitable: even when synthetic data are repeatedly generated and incorporated into training—causing the proportion of real data to vanish—the iterative estimator still achieves optimal statistical rates.

We further highlight that much of the existing literature on model collapse focuses on asymptotic analyses or low-dimensional settings \citep{alemohammad2024self, bertrand2024stability, dohmatob2024model, gerstgrasser2024model, dey2024universality,dohmatob2025strong, barzilai2026models, garg2026preventing}, where the behavior of recursively trained models can often be characterized through asymptotic distributional approximations. In contrast, we consider a high-dimensional regime in which the number of compared items $n$ is allowed to grow, and establish non-asymptotic error bounds, including coordinatewise control of the latent score estimates through the $\ell_{\infty}$ loss. Establishing such guarantees uniformly across iterations requires a more delicate analysis due to the dependence introduced by recursively generated synthetic data. Our proof proceeds through an induction argument that controls the estimation error across successive iterations. We provide a detailed proof outline in \S\ref{appendix:proof_outline} and summarize the main proof structure in Figure~\ref{fig:proof_structure} and Section~\ref{sec:Proof_Sketch_of_Theorem_1}. To the best of our knowledge, Theorem~\ref{thm:iterative_loss} provides the first theoretical characterization of the statistical rates of ranking estimators under iterative synthetic data augmentation.

Next, we discuss the conditions required for Theorem~\ref{thm:iterative_loss}. The assumptions
$$
\log t \left[
\frac{\sqrt{\log n}}{\sqrt{nL}p}
+
\frac{1}{\sqrt{np}}
\left(
1+\frac{\log n}{\sqrt{np}}
\right)
\right]
=\mathcal{O}(1), \quad \frac{1}{t+1}
\left(
\frac{\log n}{np}
+
\frac{1}{\sqrt{np}}
\right)
\leq c
$$
for some fixed constant $c>0$ are relatively mild. In particular, these conditions are satisfied when
$
\sqrt{nL}p \gtrsim (\log n)^{3/2}
\text{ and } 
np \gtrsim (\log n)^2.
$
Although this condition is stricter than the minimal sampling requirement $np \gtrsim \log n$ needed to guarantee the connectivity of an Erd\H{o}s--R\'enyi comparison graph, this gap stems from a technical challenge in controlling the $\ell_{\infty}$ statistical error throughout the iterative augmentation procedure. We elaborate on this point in Section~\ref{subsec:asymptotic} after presenting the non-asymptotic expansion of the estimator.

Regarding the constraint $t \lesssim n^{3}/(p\log n)$, this bound ensures that the desired events continue to hold with high probability even after taking a union bound uniformly across all $t$ iterations. This condition can be relaxed to $t \lesssim n^{\alpha}/(p\log n)$ for any constant $\alpha > 0$, provided that the per-iteration high-probability guarantees for the $\ell_2$ and $\ell_{\infty}$ losses hold with probability at least $1 - \mathcal{O}(n^{-q})$ for a sufficiently large constant $q > 0$. Strengthening the probability bound in this way only changes the constants in the $\ell_2$ and $\ell_{\infty}$ error bounds, while leaving the statistical rates unchanged. For simplicity, we set $\alpha=3$ throughout the rest of the proof.

The next corollary provides the conditions on the recovery of the top-$K$ items when there exists a gap between the true scores of the $K$-th and $(K + 1)$-th items.

\begin{corollary}[Exact top-$K$ recovery]
 Under the conditions of Theorem~\ref{thm:iterative_loss}, suppose that $\theta_{(K)}^{\ast} - \theta_{(K+1)}^{\ast} \ge \Delta_{K} > 0$, where $\theta_{(i)}^{\ast}$ denotes the true score of the item ranked $i$-th for $i \in [n]$. If the sample complexity satisfies
\begin{align*}
    npL \gtrsim \frac{\log n}{\Delta_{K}^2},
\end{align*}
then the estimator $\theta^{(T+1)}$ achieves exact top-$K$ identification with probability at least $1 - \mathcal{O}(n^{-5})$. That is, the estimated top-$K$ set coincides with the true top-$K$ set:
\begin{align*}
    \Big\{ i \in [n] : \theta_{i}^{\ast} \ge \theta_{(K)}^{\ast} \Big\} = \Big\{ i \in [n] : \theta_{i}^{(T+1)} \ge \theta_{(K)}^{(T+1)} \Big\}.
\end{align*}
\end{corollary}
\begin{proof}
The proof of this corollary follows directly from Theorem~\ref{thm:iterative_loss} and is therefore omitted.
\end{proof}
\subsection
{Asymptotic Normality of Iterative MLE} \label{subsec:asymptotic}
In this section, we first derive a non-asymptotic expansion of $\theta^{(T+1)}$ for any $T \ge 0$ by decomposing $\theta^{(T+1)} - \theta^{\ast}$ into a leading term that captures the variance of the estimation error and a higher-order residual term. Leveraging this expansion, we then apply the martingale central limit theorem \citep{brown1971martingale,mourrat2013rate,bolthausen1982exact,haeusler1988rate} to establish the asymptotic normality of linear projections of $\theta^{(T+1)}$.

Directly studying the statistical properties of $\theta^{(T+1)}$ is challenging. To address this, following the approach in \citet{fan2024uncertainty}, we approximate $\left\{ \theta^{(t+1)} \right\}_{t = 0}^{T}$ via its quadratic surrogate:
\begin{align} \label{eq:argminLosst_quadratic}
    \overline{\theta}^{(t+1)} := \underset{\theta \in \Theta}{\arg\min} \;\overline{\mathcal{L}}^{(\leq t)}(\theta),
\end{align}
where $\overline{\mathcal{L}}^{(\leq t)}(\theta)$ denotes the quadratic expansion of the loss function $\mathcal{L}^{(\leq t)}(\theta)$ at $\theta^{(t)}$ with
\begin{align*}
    \overline{\mathcal{L}}^{(\leq t)}(\theta) = \mathcal{L}^{(\leq t)}(\theta^{(t)}) + \nabla \mathcal{L}^{(\leq t)}(\theta^{(t)})^{\top} \left( \theta - \theta^{(t)} \right) + \frac{1}{2} \left( \theta - \theta^{(t)} \right)^{\top} \nabla^{2} \mathcal{L}^{(\leq t)}(\theta^{(t)}) \left( \theta - \theta^{(t)} \right).
\end{align*}

According to this definition, $\overline{\theta}^{(t+1)}$ can be equivalently characterized as the solution to the following system of linear equations:
\begin{align*}
    \begin{cases}
        \mathcal{P} \nabla \mathcal{L}^{(\leq t)}(\theta^{(t)}) + \mathcal{P} \nabla^{2} \mathcal{L}^{(\leq t)}(\theta^{(t)}) ( \overline{\theta}^{(t+1)} - \theta^{(t)} ) = \bm{0},
        \\
        \mathcal{P} \overline{\theta}^{(t+1)} = \overline{\theta}^{(t+1)},
    \end{cases}
\end{align*}
where $\mathcal{P} = \boldsymbol{I}_{n} - \frac{1}{n} \boldsymbol{1}_{n} \boldsymbol{1}_{n}^{\top}$ denotes the orthogonal projection matrix onto the parameter space $\Theta$. Moreover, the Hessian satisfies $\mathcal{P} \nabla^{2} \mathcal{L}^{(\leq t)}(\theta^{(t)}) = \nabla^{2} \mathcal{L}^{(\leq t)}(\theta^{(t)}) \mathcal{P} = \nabla^{2} \mathcal{L}^{(\leq t)}(\theta^{(t)})$ and the gradient also enjoys $\mathcal{P} \nabla \mathcal{L}^{(\leq t)}(\theta^{(t)})=\nabla \mathcal{L}^{(\leq t)}(\theta^{(t)})$ according to the definitions. Therefore, the above system can be equivalently written as
\begin{align*}
    \overline{\theta}^{(t+1)} - \theta^{(t)} 
    =
    -\left[ \nabla^{2} \mathcal{L}^{(\leq t)}(\theta^{(t)}) \right]^{\dagger} \nabla \ell_{t}(\theta^{(t)}).
\end{align*} 
Here, $ \nabla^2\mathcal{L}(\cdot) ^{\dagger}$ denotes the Moore--Penrose pseudoinverse of the matrix $\nabla^2\mathcal{L}(\cdot)$. Moreover, by the definition of $\theta^{(t)}$, we have $\nabla\mathcal{L}^{(\le t-1)}(\theta^{(t)}) = \mathbf{0}$, which implies
    $\nabla\mathcal{L}^{(\le t)}(\theta^{(t)}) = \nabla\ell_t(\theta^{(t)}).$
This characterization further confirms that $\overline{\theta}^{(t+1)} \in \Theta$.

Next, we present a non-asymptotic expansion of $\mathbf{c}^{\top} (\theta^{(T+1)} - \theta^{\ast})$ for any direction vector $\mathbf{c} \in \mathbb{R}^n$ by decomposing it into a leading term and an approximation error. Under mild conditions, this approximation error is of a smaller order than the leading term. Leveraging this expansion, we invoke a martingale central limit theorem to derive the asymptotic distribution of linear combinations of the estimation error $\theta^{(T+1)} - \theta^{\ast}$.

\begin{theorem}[Asymptotic normality of MLE] \label{thm:asy_norm}
    Given any $\bm{c} \in \mathbb{R}^{n}$, under the assumptions of Theorem~\ref{thm:iterative_loss}, for all $T \lesssim \frac{n^{3}}{p\log n}$ we have the following decomposition
    \begin{align*}
        \sqrt{L} (\bm{c}^{\top} \theta^{(T+1)} - \bm{c}^{\top} \theta^{\ast}) = \underbrace{\sum_{t=0}^{T}\sqrt{L} (\bm{c}^{\top} \theta^{(t+1)} - \bm{c}^{\top} \overline{\theta}^{(t+1)})}_{\text{approximation error}} + \underbrace{\sum_{t=0}^{T}\sqrt{L} (\bm{c}^{\top} \overline{\theta}^{(t+1)} - \bm{c}^{\top} \theta^{(t)})}_{\text{leading term}},
    \end{align*}
    where
    \begin{align*}
        \left| \frac{\sqrt{L} \sum_{t=0}^{T}\left(\bm{c}^{\top} \theta^{(t+1)} - \bm{c}^{\top} \overline{\theta}^{(t+1)}\right)}{\sqrt{\sum_{t=0}^{T}\bm{c}^{\top} \left[ \nabla^{2} \mathcal{L}^{(\leq t)}(\theta^{(t)}) \right]^{\dagger} \nabla^{2} \ell_{t}(\theta^{(t)}) \left[ \nabla^{2} \mathcal{L}^{(\leq t)}(\theta^{(t)}) \right]^{\dagger} \bm{c}}} \right| 
        \lesssim
        \log T \left[ \frac{\log n}{\sqrt{nL}p} + \sqrt{\frac{\log n}{np}} \left( 1 + \frac{\log n}{\sqrt{np}} \right) \right] \frac{\norm{\bm{c}}_{1}}{\norm{\bm{c}}_{2}}
    \end{align*}
    with probability at least $1 - \mathcal{O}(n^{-5})$ and 
    \begin{align*}
        \sup_{x \in \mathbb{R}} \left| \mathbb{P} \left( \frac{\sqrt{L} \sum_{t=0}^{T} \left( \bm{c}^{\top} \overline{\theta}^{(t+1)} - \bm{c}^{\top} \theta^{(t)} \right)}{\sqrt{\sum_{t=0}^{T}\bm{c}^{\top} \left[ \nabla^{2} \mathcal{L}^{(\leq t)}(\theta^{(t)}) \right]^{\dagger}   \nabla^{2} \ell_{t}(\theta^{(t)})  \left[ \nabla^{2} \mathcal{L}^{(\leq t)}(\theta^{(t)}) \right]^{\dagger} \bm{c}}} \leq x \right) - \mathbb{P} \left( \mathcal{N}(0,1) \leq x \right) \right| = r_{n}
    \end{align*}
    with $r_{n} = o(1)$ as long as $np \gg \log n$. Combining the approximation error and asymptotic distribution together, we further obtain
    \begin{align*}
        & \sup_{x \in \mathbb{R}} \left| \mathbb{P} \left( \frac{\sqrt{L} \left( \bm{c}^{\top} \theta^{(T+1)} - \bm{c}^{\top} \theta^{\ast} \right)}{\sqrt{\sum_{t=0}^{T} \bm{c}^{\top} \left[ \nabla^{2} \mathcal{L}^{(\leq t)}(\theta^{(t)}) \right]^{\dagger} \nabla^{2} \ell_{t}(\theta^{(t)}) \left[ \nabla^{2} \mathcal{L}^{(\leq t)}(\theta^{(t)}) \right]^{\dagger} \bm{c}}} \leq x \right) - \mathbb{P} \left( \mathcal{N}(0,1) \leq x \right) \right|
        \\ \lesssim &
        r_{n} + \log T \left[ \frac{\log n}{\sqrt{nL}p} + \sqrt{\frac{\log n}{np}} \left( 1 + \frac{\log n}{\sqrt{np}} \right) \right] \frac{\norm{\bm{c}}_{1}}{\norm{\bm{c}}_{2}} + n^{-5}.
    \end{align*}
\end{theorem}

\begin{proof}
    See \S\ref{appendix:asy_norm} for the detailed proof.
\end{proof}

Theorem~\ref{thm:asy_norm} establishes that $\mathbf{c}^{\top} (\theta^{(T+1)} - \theta^{\ast})$ is asymptotically normal provided that
\begin{align}\label{condition1}
    \log T \left[ \frac{\log n}{\sqrt{nL}p} + \sqrt{\frac{\log n}{np}} \left( 1 + \frac{\log n}{\sqrt{np}} \right) \right] \frac{\|\mathbf{c}\|_{1}}{\|\mathbf{c}\|_{2}} \to 0.
\end{align}
This condition is readily satisfied whenever $\mathbf{c}$ is sparse, i.e., supported on a subset of coordinates whose size remains bounded. A direct corollary of Theorem~\ref{thm:asy_norm} is the limiting joint distribution of $\theta^{(T+1)}_{S_{k}} - \theta^{\ast}_{S_{k}}$, obtained via the Cramér--Wold device \citep{cramer1936some}, where $S_{k} \subset [n]$ is any index subset of fixed size $k < \infty$. This result is formalized in Corollary~\ref{cor:asy_norm} below.

\begin{corollary} \label{cor:asy_norm}
    Assume the assumptions of Theorem~\ref{thm:iterative_loss} hold. Then for any fixed $k \in [n]$ and for any fixed $T \lesssim \frac{n^{3}}{p\log n}$, as long as
    \begin{align*}
         \log T \left[ \frac{\log n}{\sqrt{nL}p} + \sqrt{\frac{\log n}{np}} \left( 1 + \frac{\log n}{\sqrt{np}} \right) \right] = o(1),
    \end{align*}
    we have
    \begin{align*}
        \sqrt{L} \left( \left\{ \sum_{t=0}^{T}\left[ \nabla^{2} \mathcal{L}^{(\leq t)}(\theta^{(t)}) \right]^{\dagger} \nabla^{2} \ell_{t}(\theta^{(t)}) \left[ \nabla^{2} \mathcal{L}^{(\leq t)}(\theta^{(t)}) \right]^{\dagger}  \right\}_{S_{k},S_{k}} \right)^{-1/2}  \left( \theta^{(T+1)}_{S_{k}} - \theta^{\ast}_{S_{k}}  \right) \rightsquigarrow  \mathcal{N}(\bm{0}, \boldsymbol{I}_{k}),
    \end{align*}
    where $S_{k}$ is any subset over $[n]$ with size $k$.
\end{corollary}

To the best of our knowledge, no existing literature has investigated the non-asymptotic expansion of the estimators $\theta^{(T)}$ across iterations $T \ge 0$ under our iterative augmentation workflow; our work is the first to address this problem.

We note that \eqref{condition1} imposes a stronger connectivity condition, $\sqrt{nL}p \gtrsim \log n \log T$, compared to the minimal requirement $np \gtrsim \log n$ for connectivity in Erdős–Rényi graphs. This stronger requirement stems from technical challenges in controlling both the higher-order residual terms and the union bound over all $T$ iterations. Relaxing this assumption is a non-trivial task that we leave for future work.
\section{Numerical Results} 
\label{sec:Numerical_Results}

This section presents simulations that investigate the statistical properties of iterative MLE under the synthetic-data augmentation workflow and corroborate our theoretical results.

Specifically, we compare two variants: (i) the unregularized iterative MLE proposed in the main text and (ii) a regularized iterative MLE with a weak regularization parameter $\lambda_T = \frac{T+1}{nL},$ 
which serves as an intermediate step in our theoretical analysis of the unregularized estimator. In particular, we use the regularized estimator to establish that the iterative estimates remain bounded in $\ell_{\infty}$ norm and further show that the unregularized estimator is sufficiently close to its regularized counterpart. We refer to \S\ref{sec:regularized} for further details.

Unless otherwise stated, all reported results are averaged over $200$ independent Monte Carlo replications, and shaded regions in the trajectory plots denote pointwise $95\%$ confidence intervals.

\subsection{Baseline Setting}
We set $n=50$, $p=0.5$, and $L=20$ and keep these parameters fixed throughout all iterations.

Figure~\ref{fig:simulation} summarizes the behavior of the regularized and unregularized estimators over $T\in[50]$ iterations. Figures~\ref{fig:simulation(a)} and~\ref{fig:simulation(b)} show that $\log(\|\hat\theta^{(T)}-\theta^*\|_2)$ and $\log(\|\hat\theta^{(T)}-\theta^*\|_{\infty})$ increase during the first few iterations and then stabilize within a bounded range, indicating that model collapse does not occur in this setting.
 In particular, the statistical errors remain bounded throughout the iterative procedure, and the regularized and unregularized iterative MLEs exhibit closely aligned trajectories.
 
Figure~\ref{fig:simulation(c)} focuses on the top five ranked items, whose estimated scores remain close to their underlying true scores while preserving the correct ordering. Finally, Figure~\ref{fig:simulation(d)} compares the final averaged estimates over 200 Monte Carlo replications with the ground truth across all ranked items, showing that both the estimated ordering and score magnitudes remain well aligned with $\theta^{\ast}$ even as the comparison data become progressively contaminated by synthetic observations.
\begin{figure}[htbp]
    \centering
    \begin{subfigure}[htbp]{0.45\textwidth}
        \centering
        \includegraphics[width=\textwidth]{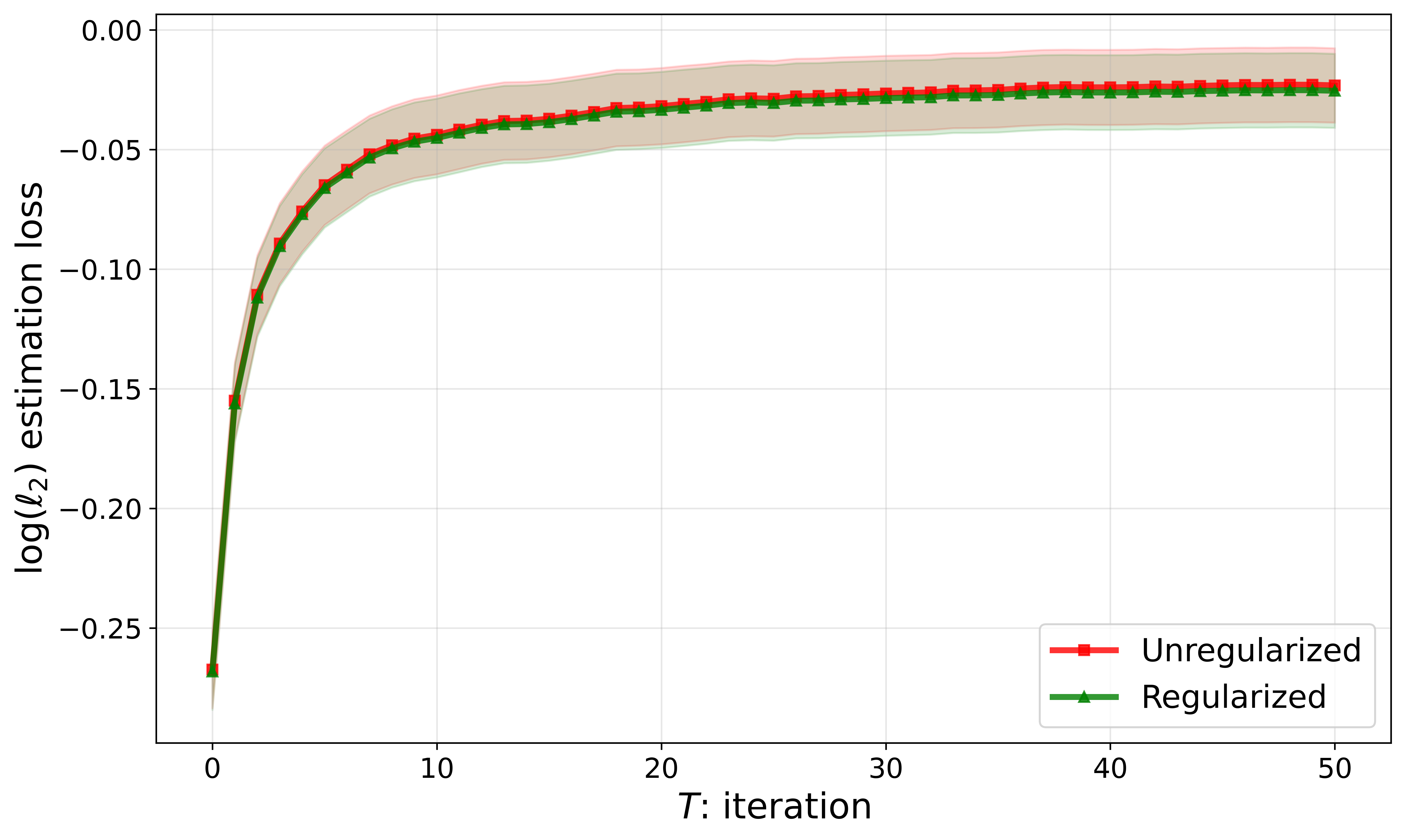}
        \caption{$\log(\|\hat\theta^{(T)}-\theta^*\|_2)$ statistical error vs. iteration $T$}
        \label{fig:simulation(a)}
    \end{subfigure}
    \hfill 
    \begin{subfigure}[htbp]{0.45\textwidth}
        \centering
        \includegraphics[width=\textwidth]{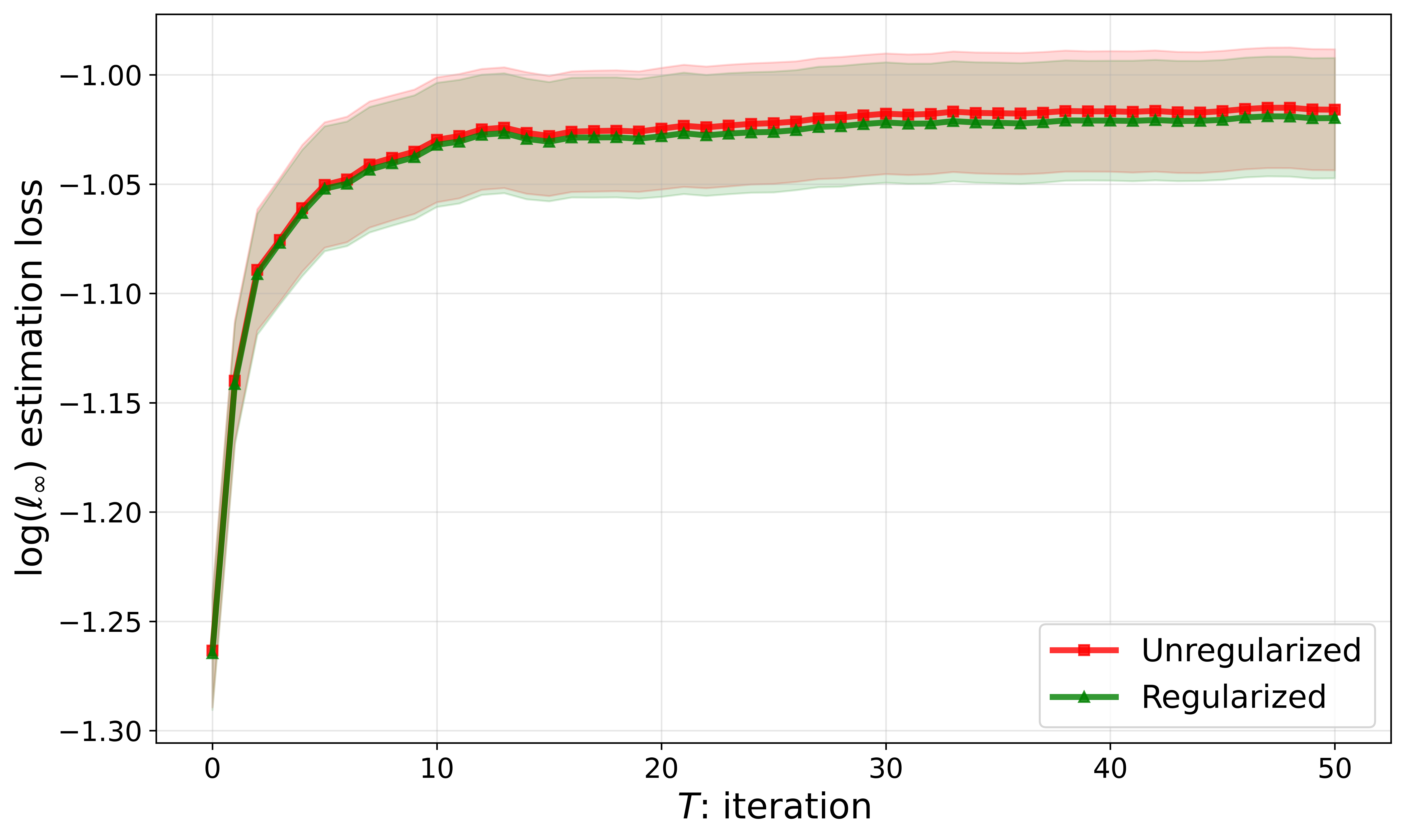}
        \caption{$\log(\|\hat\theta^{(T)}-\theta^*\|_{\infty})$ statistical error vs. iteration $T$}
        \label{fig:simulation(b)}
    \end{subfigure}
    \vspace{1em}
    \begin{subfigure}[htbp]{0.45\textwidth}
        \centering
        \includegraphics[width=\textwidth]{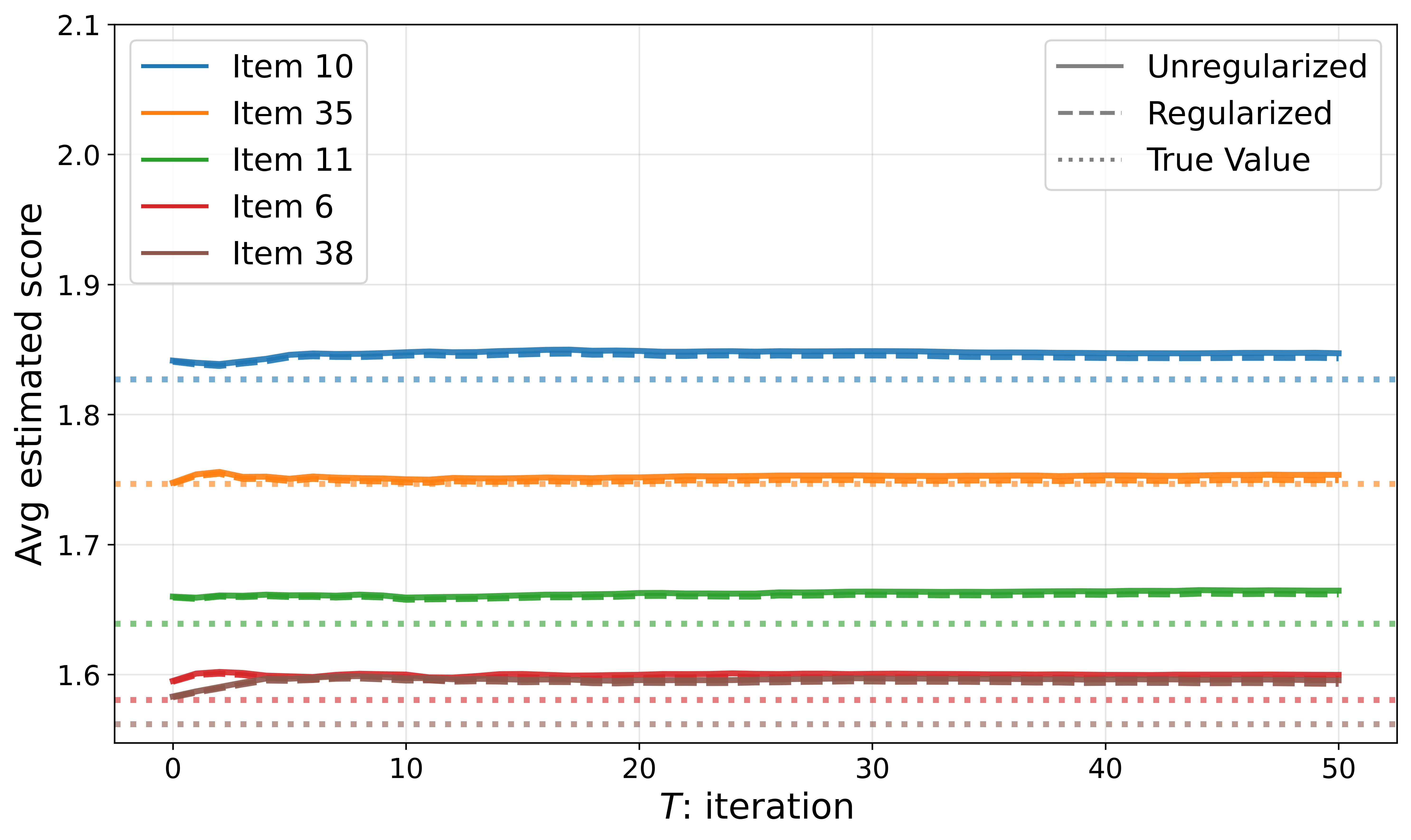}
        \caption{Evolution of the averaged estimated scores for the top five ranked items}
        \label{fig:simulation(c)}
    \end{subfigure}
    \hfill
    \begin{subfigure}[htbp]{0.45\textwidth}
        \centering
        \includegraphics[width=\textwidth]{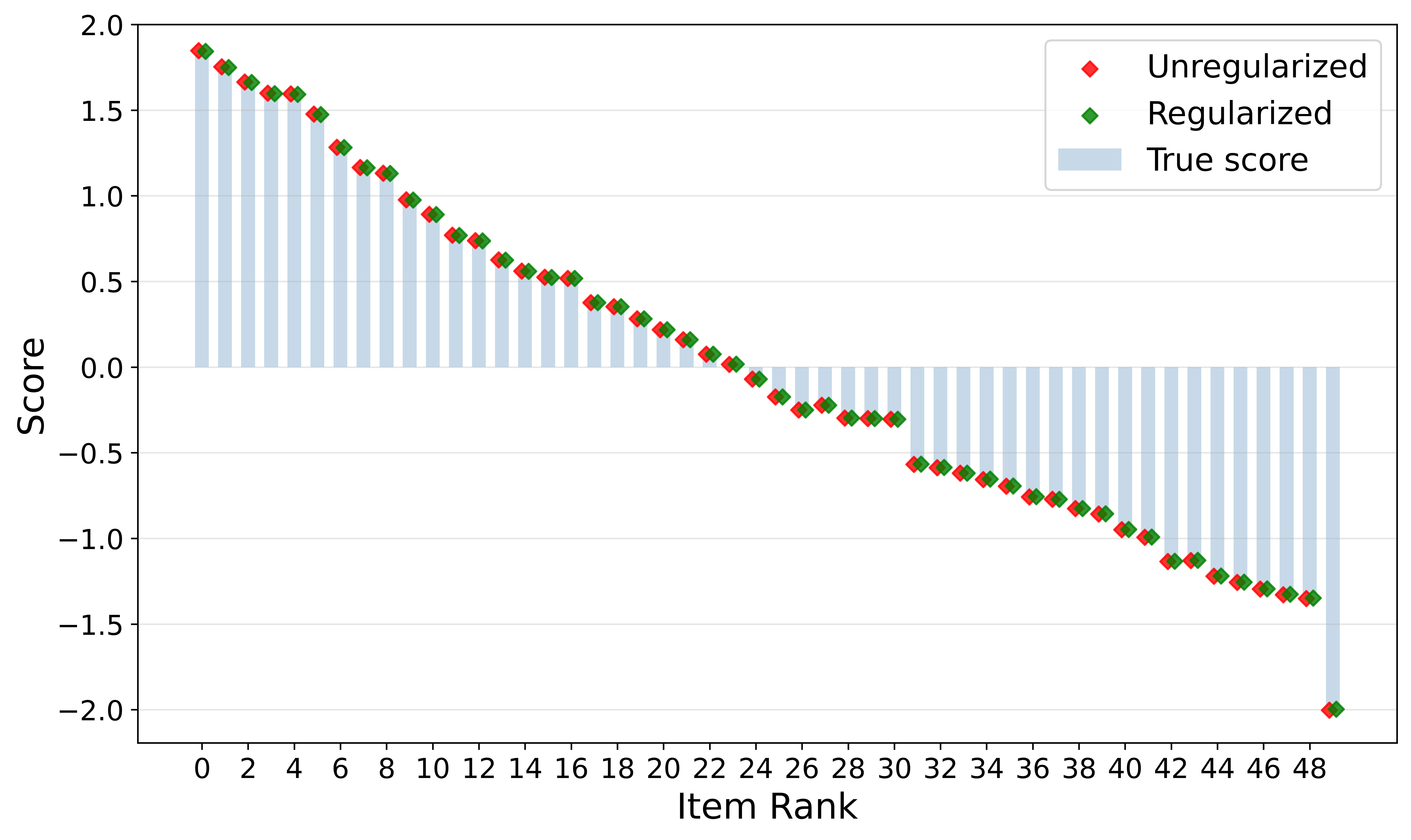}
        \caption{Final estimates $\theta^{(T)}$ for all 50 ranked items}
        \label{fig:simulation(d)}
    \end{subfigure}
    \caption{Empirical performance of the unregularized and regularized iterative MLEs under the synthetic data augmentation workflow, with $n=50$, $T\in[50]$, edge probability $p=0.5$, and $L=20$. Results are averaged over $200$ Monte Carlo simulations.}
    \label{fig:simulation}
\end{figure}
\subsection{Rate of Convergence}
To corroborate the statistical  rates established in Theorem~\ref{thm:iterative_loss}, we empirically examine how the estimation error scales with $1/\sqrt{pL}$. We run the synthetic augment workflow for $T = 50$ rounds and consider $n \in \left\{ 50, 100 \right\}$. To cover a range of $1/\sqrt{pL}$ within $(0,1)$, we examine 6 different $(p,L)$ pairs, which are listed below in Table~\ref{tab:pL_pairs}. For each $(p,L)$ pair, we record the final $\ell_{2}$ and $\ell_{\infty}$ statistical errors at iteration $T = 50$.
\begin{table}[htbp]
    \centering
    \begin{tabular}{c c c c c c c}
        \toprule
        $p$ & 0.8 & 0.5 & 0.6 & 0.3 & 0.4 & 0.3 \\
        \midrule
        $L$ & 50 & 25 & 10 & 10 & 5 & 5 \\
        \bottomrule
    \end{tabular}
    \caption{ $(p, L)$ pairs considered for validating the statistical rates in Theorem~\ref{thm:iterative_loss}.}
    \label{tab:pL_pairs}
\end{table}

The results are shown in Figure~\ref{fig:rate_of_convergence}. For each fixed $n\in \{50,100\}, $ both $\ell_{2}$ and $\ell_{\infty}$ statistical errors exhibit an approximately linear relationship with $1/\sqrt{pL}$ , providing  further support of the statistical rates established in Theorem~\ref{thm:iterative_loss}. As expected, the regularized and unregularized iterative MLEs exhibit nearly identical empirical performance. 
\begin{figure}[htbp]
    \centering
    \begin{subfigure}[b]{0.45\textwidth}
        \centering
        \includegraphics[width=\textwidth]{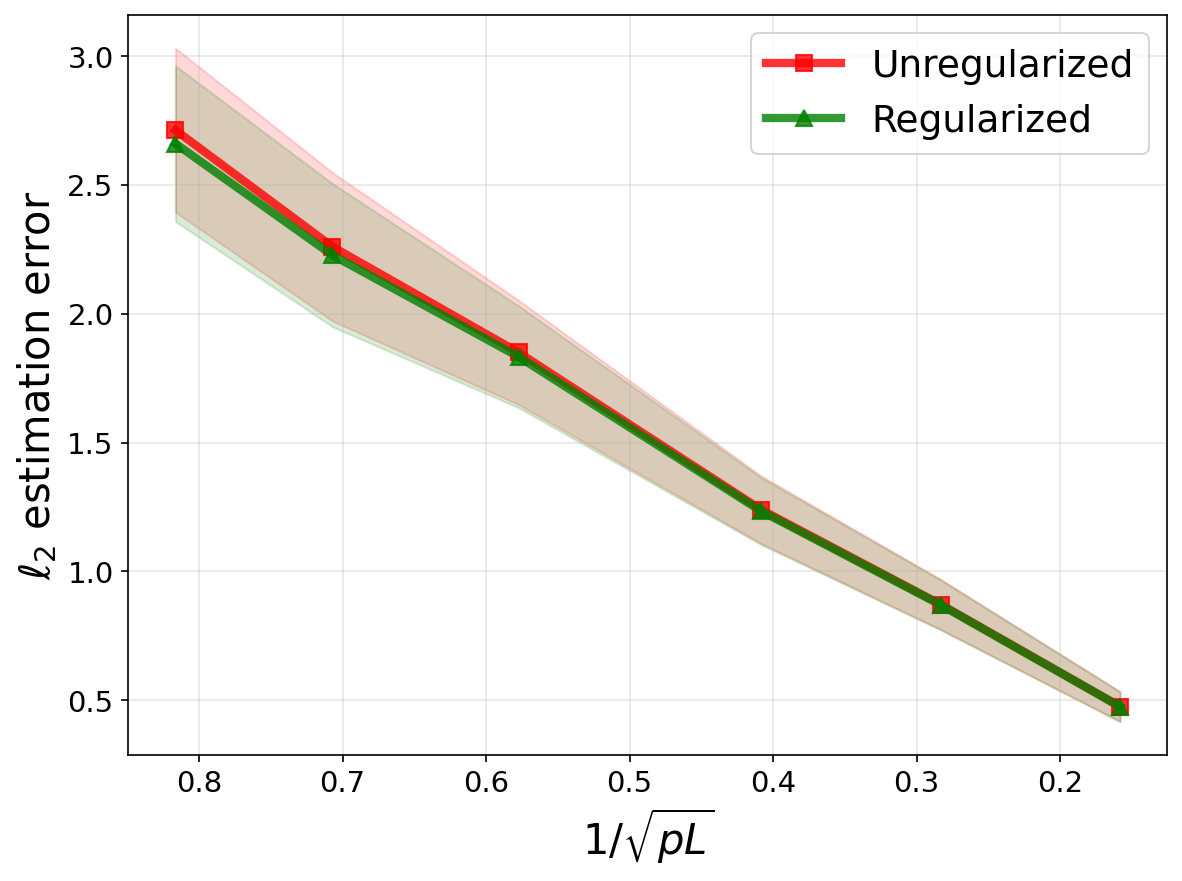}
        \caption{$\ell_2$ estimation error vs.\ $1/\sqrt{pL}$ ($n=50$)}
        \label{fig:rate(a)}
    \end{subfigure}
    \hfill 
    \begin{subfigure}[b]{0.45\textwidth}
        \centering
        \includegraphics[width=\textwidth]{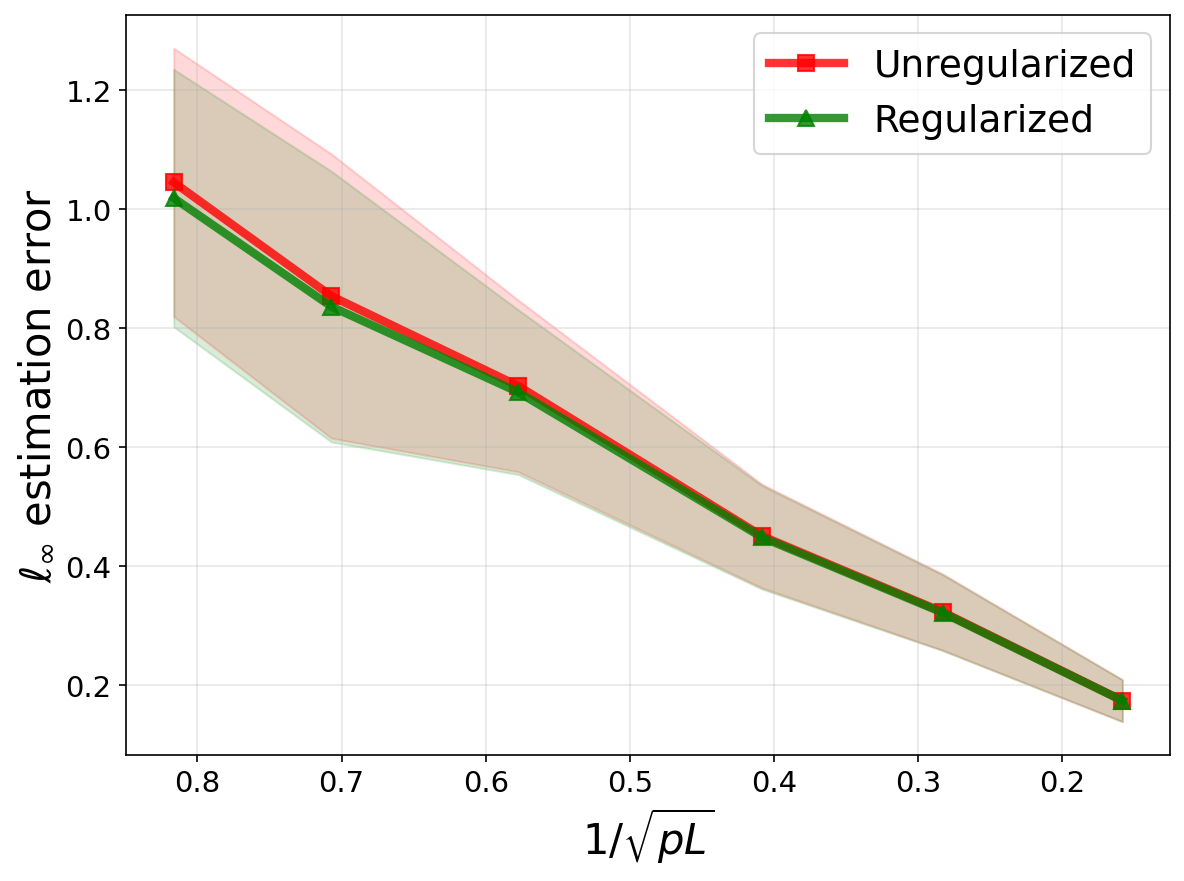}
        \caption{$\ell_\infty$ estimation error vs.\ $1/\sqrt{pL}$ ($n=50$)}
        \label{fig:rate(b)}
    \end{subfigure}
    \vspace{1em}
    \begin{subfigure}[b]{0.45\textwidth}
        \centering
        \includegraphics[width=\textwidth]{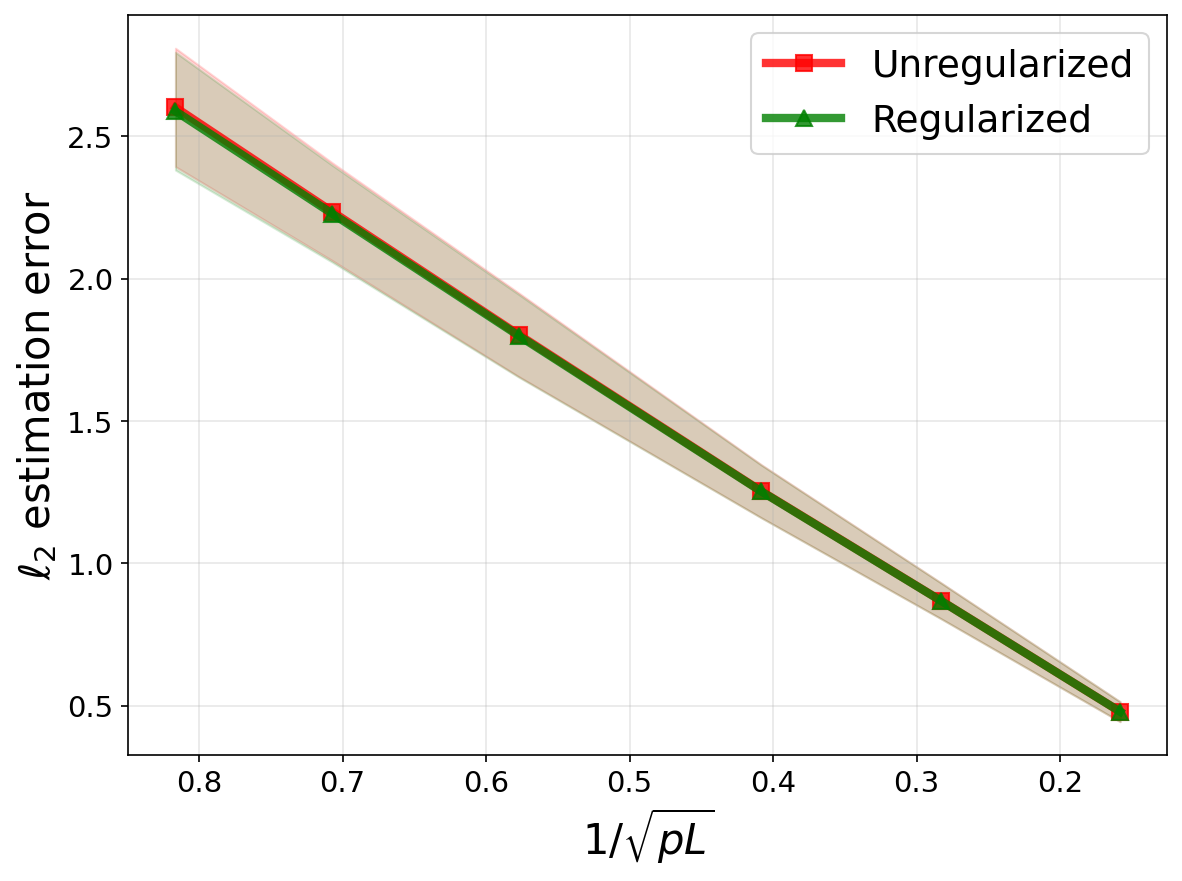}
        \caption{$\ell_2$ estimation error vs.\ $1/\sqrt{pL}$ ($n=100$)}
        \label{fig:rate(c)}
    \end{subfigure}
    \hfill 
    \begin{subfigure}[b]{0.45\textwidth}
        \centering
        \includegraphics[width=\textwidth]{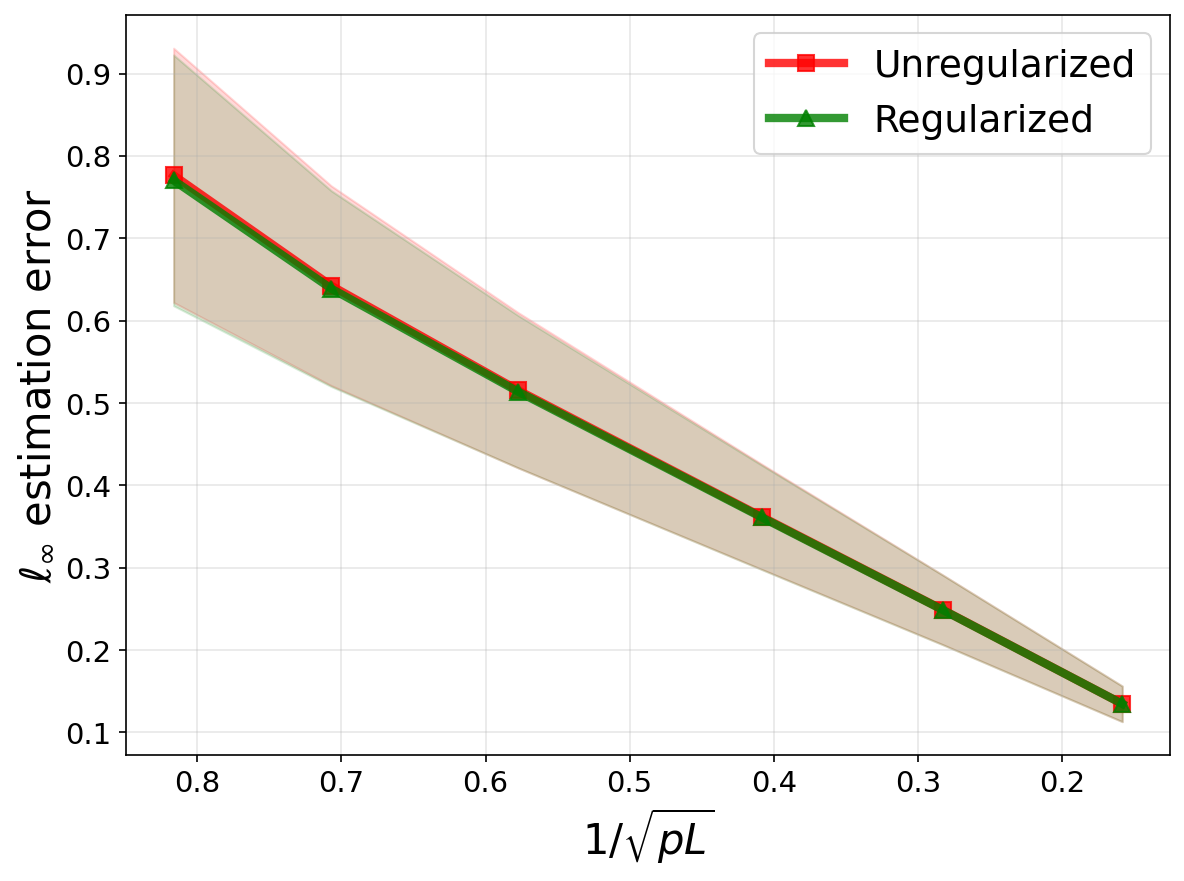}
        \caption{$\ell_\infty$ estimation error vs.\ $1/\sqrt{pL}$ ($n=100$)}
        \label{fig:rate(d)}
    \end{subfigure}
    \caption{Statistical rates of $\norm{\theta^{(T)} - \theta^{\ast}}_2$ and $\norm{\theta^{(T)} - \theta^{\ast}}_\infty$ for the $(p, L)$ pairs listed in Table~\ref{tab:pL_pairs} with $T=50$. The solid red and green lines represent the averaged errors for the unregularized and regularized iterative MLEs, respectively, and the shaded areas indicate their associated standard deviations based on 200 Monte Carlo simulations at iteration $T = 50$. The top row reports results for $n=50$ and the bottom row for $n=100$.}
    \label{fig:rate_of_convergence}
\end{figure}

\subsection{Asymptotic Normality}
Next, we examine the asymptotic normality of the estimators. We define the effective sample size as $n_{a} := n / \log n$,   fix $n = 100$, and consider six $(p,L)$ pairs with $p \in \{6/n_{a}, 12/n_{a}\}$ and $L \in \{2,6,20\}$. For each $(p,L)$, we evaluate the unregularized iterative MLE  at iteration $T = 100$. Figure~\ref{fig:asymptotic_normality} presents the Q-Q plots for checking the normality of $\theta^{(T)}_{1}$, the first coordinate of $\theta^{(T)}$. Across all settings, the empirical quantiles align closely with the reference normal quantiles, providing empirical evidence for the asymptotic normality of $\theta^{(T)}_1$.

\begin{figure}[htbp]
    \centering
    \includegraphics[width=1.0\linewidth]{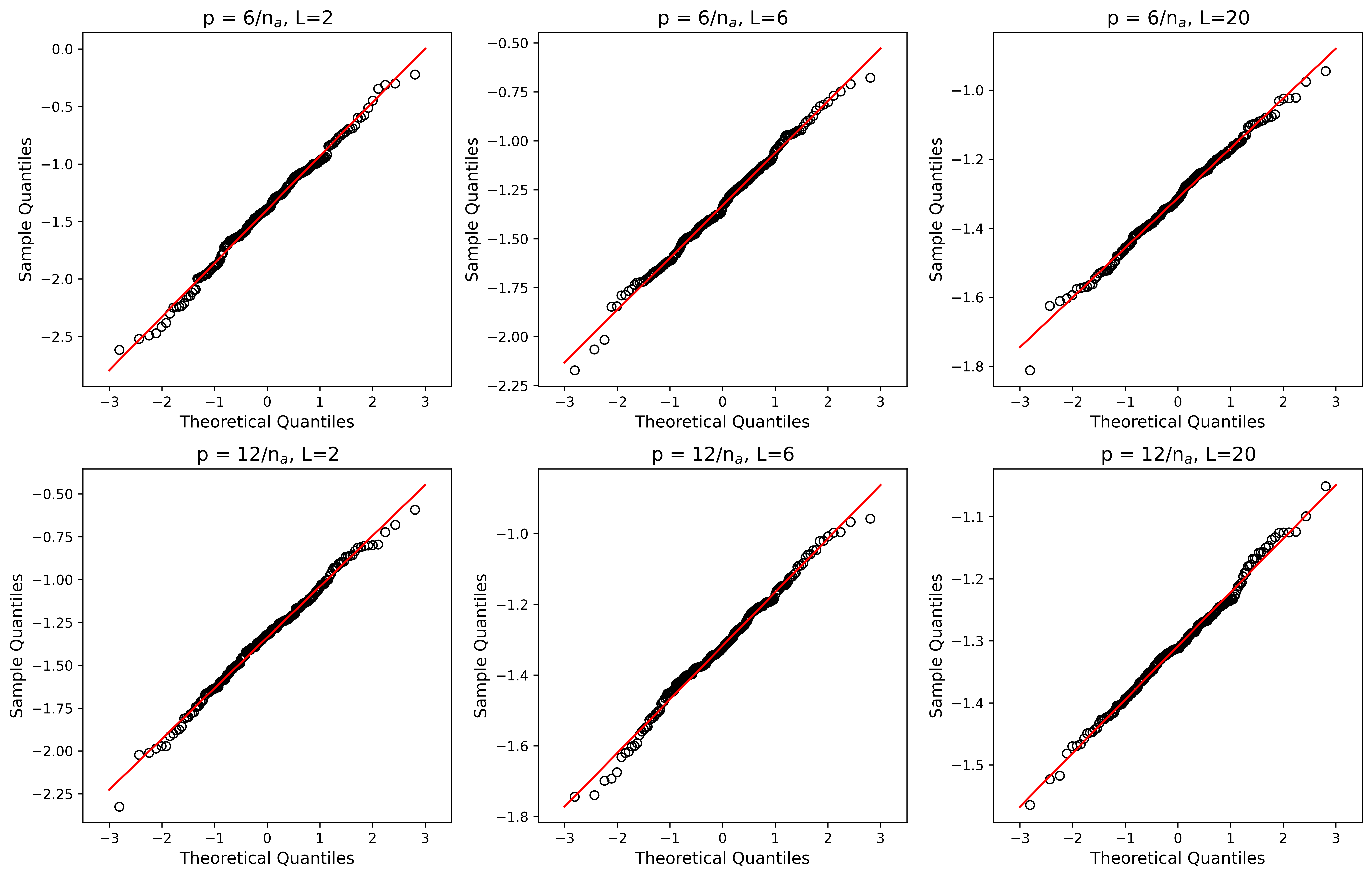}
    \caption{Q-Q Plots for checking the normality of $\left(\theta^{(T)}\right)_{1}$ under varying edge probability $p \in \{6/n_a, 12/n_a\}$ and number of repeated comparisons $L \in \{2,6,20\}$, with fixed $n=100$, effective sample size $n_a := n/\log n$, and iteration $T = 100$. Results are based on 200 Monte Carlo simulations.}
    \label{fig:asymptotic_normality}
\end{figure}


\noindent\textbf{Asymptotic Variance}: Next, we evaluate the asymptotic variance of the estimator. We set $n=200$ and consider two settings, $(p,L)=(0.2,2)$ and $(p,L)=(0.4,10)$. We focus on the linear combination $\bm{c}^{\top}\theta^{(T+1)}$, where $\bm{c}=\boldsymbol{e}_1+\boldsymbol{e}_{200}$ and $\boldsymbol{e}_i$ denotes the $i$th standard basis vector in $\mathbb{R}^{200}$. Based on $300$ Monte Carlo simulations, we plot the following two standardized statistics at iteration $T=50$:

\begin{align*}
    A = \frac{\sqrt{L} \left( \bm{c}^{\top} \theta^{(T+1)} - \bm{c}^{\top} \theta^{\ast} \right)}{\sqrt{V_{n}^{2}}}
    \qquad \text{and} \qquad
    B = \frac{\sqrt{L} \left( \bm{c}^{\top} \theta^{(T+1)} - \bm{c}^{\top} \theta^{\ast} \right)}{\sqrt{\hat{V}_{n}^{2}}},
\end{align*}
where 
\begin{align*}
    V_{n}^{2}
    = &
    \sum_{t=0}^{T} \bm{c}^{\top} \left[ \nabla^{2} \mathcal{L}^{(\leq t)}(\theta^{(t)}) \right]^{\dagger} \nabla^{2} \ell_{t}(\theta^{(t)}) \left[ \nabla^{2} \mathcal{L}^{(\leq t)}(\theta^{(t)}) \right]^{\dagger} \bm{c},
    \\
    \hat{V}_{n}^{2} 
    = & 
    \bm{c}^{\top} \left[ \nabla^{2} \ell_{0}(\theta^{(1)}) \right]^{\dagger} \nabla^{2} \ell_{0}(\theta^{(1)}) \left[ \nabla^{2} \ell_{0}(\theta^{(1)}) \right]^{\dagger} \bm{c} 
    + 
    \sum_{t=1}^{T} \bm{c}^{\top} \left[ \nabla^{2} \mathcal{L}^{(\leq t)}(\theta^{(t)}) \right]^{\dagger} \nabla^{2} \ell_{t}(\theta^{(t)}) \left[ \nabla^{2} \mathcal{L}^{(\leq t)}(\theta^{(t)}) \right]^{\dagger} \bm{c}.
\end{align*}
Here $V_n^2$ denotes the theoretical variance, which requires knowledge of $\theta^{(0)} = \theta^*$ when $t = 0$, while the plug-in estimator $\hat{V}_n^2$ replaces $\theta^{(0)}$ with $\theta^{(1)}$. Figure~\ref{fig:asymptotic_variance} shows that the histograms of both $A$ and $B$ follow closely the standard Gaussian density. The first row corresponds to the sparse regime with $(p, L) = (0.2, 2)$, where, despite the relatively limited amount of data, both standardized statistics remain well approximately by $\mathcal{N}(0,1)$. The second row presents results for $(p, L) = (0.4, 10),$ where we can draw similar conclusions. These results support our theoretical findings in Theorem~\ref{thm:asy_norm}.

\begin{figure}[htbp]
    \centering
    \includegraphics[width=0.8\linewidth]{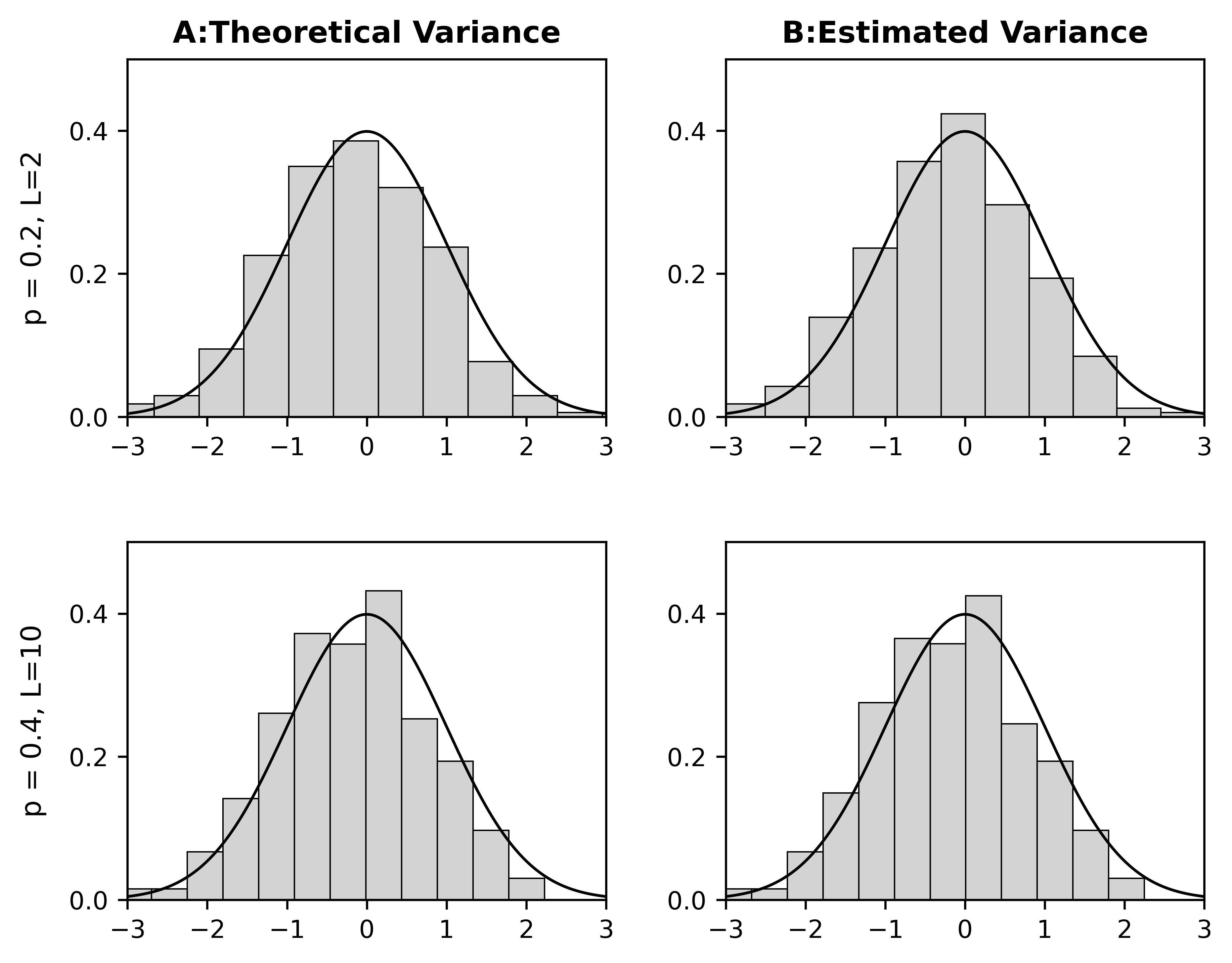}
    \caption{Histograms for assessing the asymptotic variance of $\bm{c}^{\top}\theta^{(T+1)}-\bm{c}^{\top}\theta^{\ast}$. The first row corresponds to $(p,L)=(0.2,2)$, and the second row corresponds to $(p,L)=(0.4,10)$. Results are based on $300$ Monte Carlo simulations with $n=200$ and $T=50$.}
    \label{fig:asymptotic_variance}
\end{figure}

\subsection{Application to Arena Human Preference 140k Dataset}
We apply our proposed method to the Arena Human Preference 140k dataset, available at
\url{https://huggingface.co/datasets/lmarena-ai/arena-human-preference-140k}.
This dataset records pairwise comparison results among $52$ large language models (LLMs) collected through the Arena platform \citep{chiang2024chatbot,searcharena2025}. In total, the dataset comprises $135,634$ human preference votes, each comparing a pair of models with a single declared winner or a draw.

For our analysis, we first remove self-comparisons and all observations whose outcomes are labeled as \texttt{tie} or \texttt{both bad}. This preprocessing step leaves  $98,341$ total number of comparisons among the full set of $52$ models. We then select a subset of $25$ models, including \texttt{Amazon Nova, ChatGPT, Claude, DeepSeek, Gemini, Grok, Hunyuan, Kimi, LLaMA, MiniMax, Mistral, and Qwen}, among others, to exclude models with limited comparisons while ensuring that the leading models are included.
 Restricting to this subset yields $21,685$ comparisons and yields an empirical comparison density $p \approx 0.9333$ and $L$ ranging from $1$ to $274$, see Figure~\ref{fig:real_win_rate_filter_140k}. In the first iteration, we fit the model to this filtered dataset to obtain $\theta^{(1)}$, allowing $L$ to vary across pairs. In subsequent iterations, we study the behavior of the iterative MLE under four fixed $(p,L)$ regimes, $(p,L) \in \{(0.2,2),\ (0.4,5),\ (0.6,10),\ (0.8,20)\}$.

\begin{figure}[htbp]
    \centering
    \includegraphics[width=0.8\linewidth]{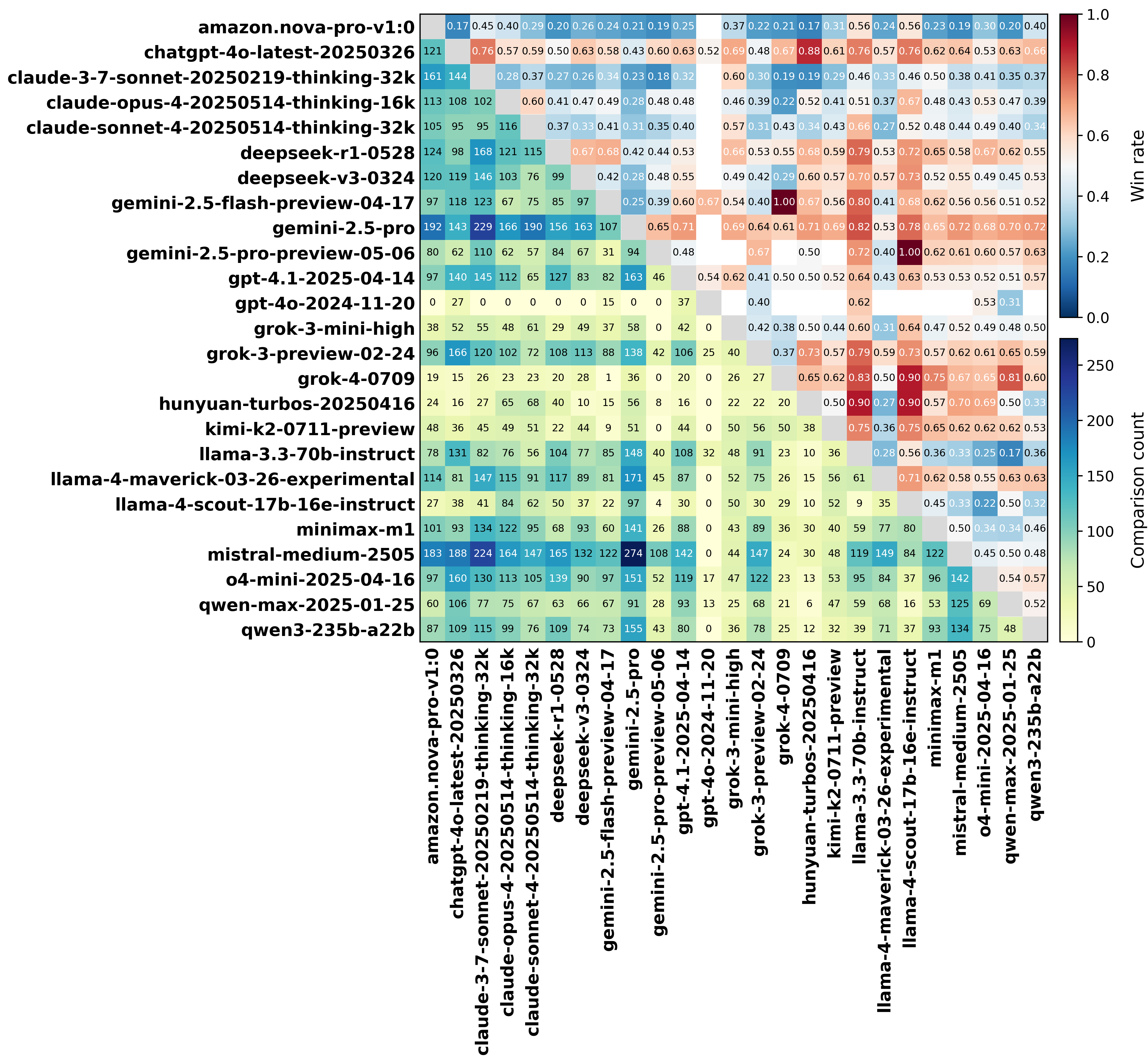}
    \caption{ Win rate (above the diagonal) and comparison count (below the diagonal) between the 25 models in the Arena human preference 140k dataset.}
    \label{fig:real_win_rate_filter_140k}
\end{figure}


Figure~\ref{fig:theta_comparison_both} illustrates the long-run behavior of the two workflows. Under the iterative synthetic data augmentation workflow, the final estimates $\theta^{(T+1)}$ remain close to the initial real-data estimates $\theta^{(1)}$ across all 25 models. This pattern is consistent under all four $(p,L)$ configurations, from the sparse setting $(0.2,2)$ to the dense setting $(0.8,20)$. Specifically, the ranking is essentially preserved after 50 iterations: top models such as \texttt{gemini-2.5-pro} and \texttt{grok-4-0709} remain at the top, while lower-ranked models such as \texttt{amazon.nova-pro-v1:0} remain near the bottom. This stability contrasts with prior concerns about model collapse and suggests that retaining the accumulated real and synthetic data can mitigate progressive degradation over repeated rounds of synthetic augmentation. 

By contrast, the Discard workflow, which retains only the newly generated synthetic data and discards the real data, exhibits substantially greater instability. Once the real data are removed, the estimates $\theta^{(T+1)}$ deviate noticeably from $\theta^{(1)}$, with the discrepancy becoming more pronounced as the synthetic comparison graph becomes sparser. In the most sparse setting, $(p,L)=(0.2,2)$, the final estimates become highly dispersed and no longer preserve the original ranking in a meaningful way.


\begin{figure}[htbp]
    \centering
    \includegraphics[width=\linewidth]{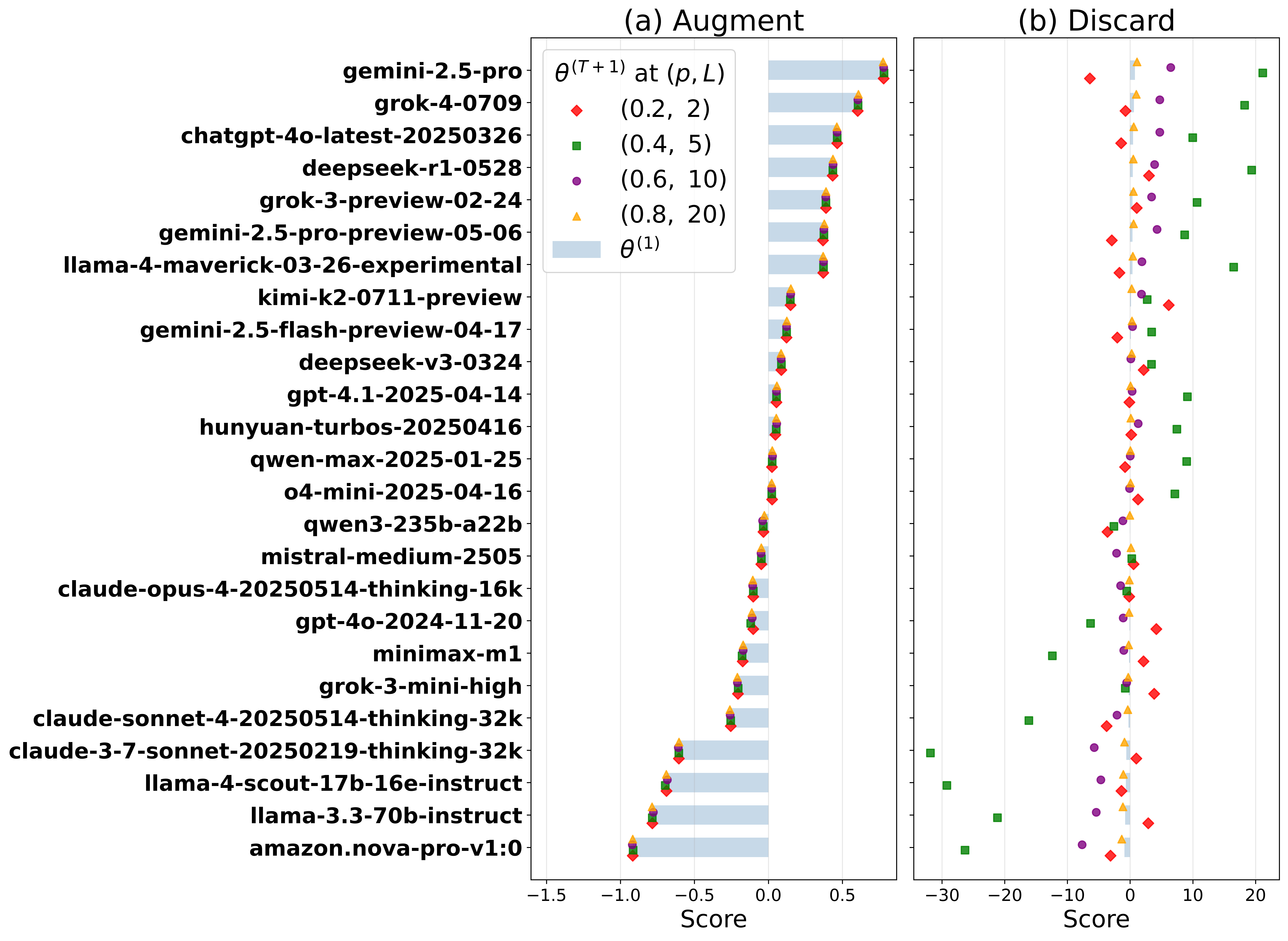}
    \caption{Comparison of BTL score estimates $\theta^{(1)}$ and $\theta^{(T+1)}$
    across 25 selected models from the Arena Human Preference 140k dataset under the augmentation and Discard workflows. Four configurations are shown: $(p, L) = (0.2, 2)$ red diamonds, $(p, L) = (0.4, 5)$ green squares, $(p, L) = (0.6, 10)$ purple circles, and $(p, L) = (0.8, 20)$ orange triangles. Models are ranked by $\theta^{(1)}$ in descending order. Results are averaged over 200 Monte Carlo simulations with fixed $n = 25$ and iteration $T = 50$.}
\label{fig:theta_comparison_both}
\end{figure}


Figure~\ref{fig:score_evolution_both} plots the average estimated scores of the top five models over $50$ iterations. Under the iterative synthetic data augmentation workflow (Figure~\ref{fig:score_evolution_augment}), the trajectories are stable and well separated throughout. After the first round, the scores only changed slightly, and the four $(p,L)$ settings yield  nearly identical trajectories  for each model. In contrast, the discard workflow (Figure~\ref{fig:score_evolution_discard}), exhibits markedly different behavior. The estimated scores increase sharply during the initial iterations and subsequently display increasingly large fluctuations. Even under the denser settings, $(0.6,10)$ and $(0.8,20)$, the trajectories continue to drift rather than stabilize. 

Overall, these results indicate that the iterative synthetic data augmentation workflow maintains stable score trajectories and largely preserves the ranking structure, whereas the Discard workflow exhibits substantial instability over repeated iterations.

\begin{figure}[htbp]
    \centering
    \begin{subfigure}{0.49\textwidth}
        \includegraphics[width=\textwidth]{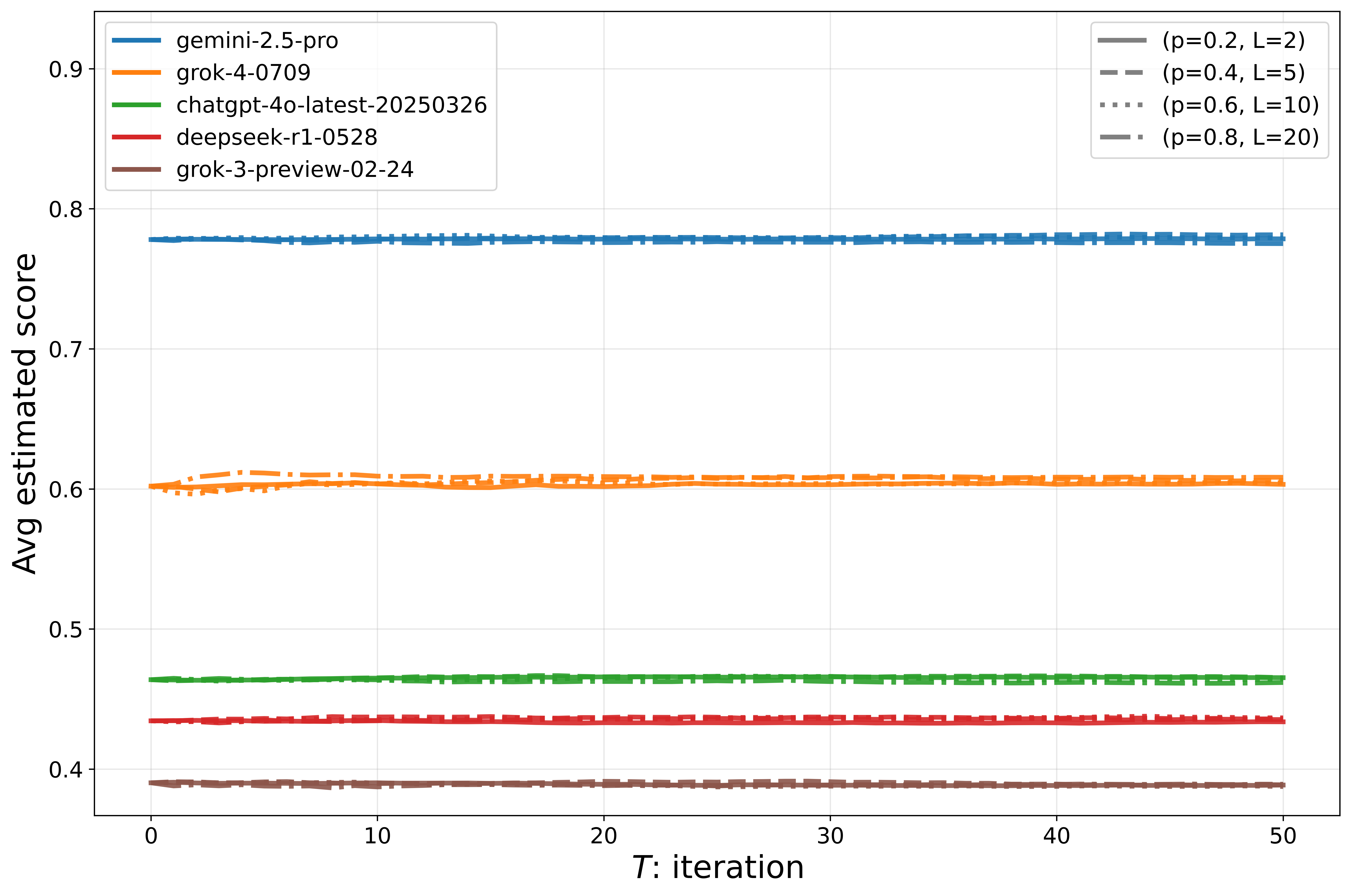}
        \caption{Augmentation workflow}
        \label{fig:score_evolution_augment}
    \end{subfigure}
    \hfill
    \begin{subfigure}{0.49\textwidth}
        \includegraphics[width=\textwidth]{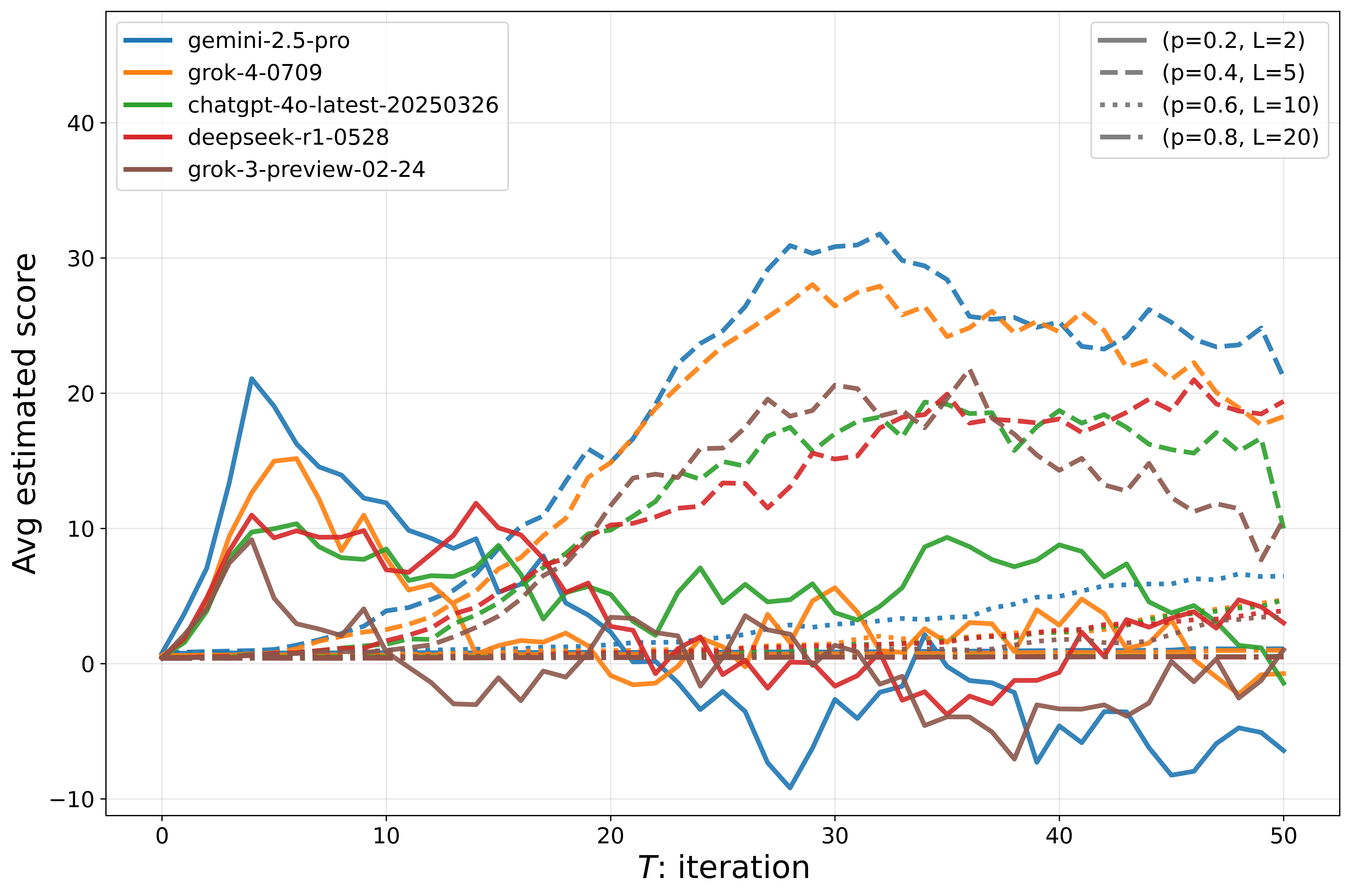}
        \caption{Discard workflow}
        \label{fig:score_evolution_discard}
    \end{subfigure}
    \caption{Score evolution of the top-5 models across 50 iterations under the augmentation and Discard workflows. Colors denote models and line styles denote $(p,L)$ configurations: $(p,L) = (0.2,2)$ solid, $(p,L) = (0.4,5)$ dashed, $(p,L) = (0.6,10)$ dotted, and $(p,L) = (0.8,20)$ dash-dotted. Models are ranked by $\theta^{(1)}$ in descending order. Results are averaged over 200 Monte Carlo simulations with fixed $n = 25$ and
    iteration $T = 50$.}
    \label{fig:score_evolution_both}
\end{figure}

\section{Proof Sketch of Theorem~\ref{thm:iterative_loss}} \label{sec:Proof_Sketch_of_Theorem_1}
In this section, we provide the proof sketch for our main Theorem~\ref{thm:iterative_loss}. The proof proceeds in three steps; see Figure~\ref{fig:proof_structure} and Appendix~\ref{appendix:proof_outline} for a detailed overview. Each step is carried out under the induction hypothesis that the bounds in Theorem~\ref{thm:iterative_loss} hold through iteration $t=T$. We then establish the corresponding bounds at $t=T+1$, thereby closing the induction.
 
\paragraph{Step I: Constant $\ell_\infty$ bound.}

As illustrated in the first panel of Figure~\ref{fig:proof_structure}, we begin with the induction hypothesis that the four bounds hold through iteration $T$, and aim to show that they continue to hold with high probability at iteration $T+1$. We divide the argument into three steps. In \textit{Step~I}, we first show $\norm{\theta^{(T+1)}-\theta^{\ast}}_{\infty} \leq 5$ (see Appendix~\ref{appendix:Step1}). In other words, we show that the iterate $\theta^{(T+1)}$ remains within a constant-sized neighborhood of $\theta^{\ast}$. 

Following a strategy similar to that in \cite{chen2019bspectral,chen2022partial}, we introduce a regularized MLE with a sufficiently small regularization parameter as an intermediate proxy for the unregularized estimator. We analyze the gradient descent trajectory toward the minimizer of this regularized objective starting from $\theta^{\ast}$ and show that the resulting regularized estimator remains bounded with high probability. Since the regularization parameter is chosen to be sufficiently small, we further show that the unregularized estimator stays close to its regularized counterpart, which yields the desired constant-order bound.

A key ingredient in controlling the gradient descent trajectory is a leave-one-out argument, following \cite{chen2019bspectral,chen2022partial,fan2025ranking}. This construction allows us to control the iterates coordinatewise and ensures that no individual coordinate deviates too far from the initial point $\theta^\ast$. Together, these arguments establish the desired $\ell_\infty$ bound and complete \textit{Step~I}.

 
\paragraph{Step II: Optimal rates of $\|\theta^{(T+1)}-\theta^{(T)}\|_2,$ $\|\theta^{(T+1)}-\theta^{(T)}\|_{\infty}$ and $\|\theta^{(T+1)}-\theta^{\ast}\|_{2}$}

In \textit{Step~II}, we control the increment $\theta^{(T+1)}-\theta^{(T)}$ in both the $\ell_2$ and $\ell_\infty$ norms. The argument follows a strategy similar to that in \textit{Step~I}. In particular, we again construct and analyze a gradient descent trajectory. The key difference is that, instead of initializing the trajectory at $\theta^{\ast}$ as in \textit{Step~I}, we initialize it at the previous iterate $\theta^{(T)}$. By the strong convexity established within the bounded region from \textit{Step~I}, the gradient descent iterates converge geometrically to $\theta^{(T+1)}$. At the same time, we show that the entire trajectory remains sufficiently close to its initialization $\theta^{(T)}$. Combining the geometric convergence with this control of the trajectory yields the desired $\ell_2$ and $\ell_\infty$ bounds on the increment $\theta^{(T+1)}-\theta^{(T)}$. We refer interested readers to Section \ref{appendix:Step2} for more details. 

After controlling the increment $\theta^{(T+1)}-\theta^{(T)}$, the bound on $\|\theta^{(T+1)}-\theta^{\ast}\|_{2}$ follows directly from the tower property, together with the induction hypothesis and the results established in the preceding steps. We refer to Section~\ref{appendix:equation_G} for further details.

\paragraph{Step III:~Control $\|\theta^{(T+1)}-\theta^{\ast}\|_{\infty}.$}

Recall that we have the decomposition
\begin{align*}
     \theta^{(T+1)} - \theta^{\ast} = \underbrace{\sum_{t=0}^{T} (\theta^{(t+1)} - \overline{\theta}^{(t+1)})}_{\text{approximation error}} + \underbrace{\sum_{t=0}^{T} ( \overline{\theta}^{(t+1)} - \theta^{(t)})}_{\text{leading term}},
    \end{align*}
Here, $\overline{\theta}^{(t+1)}$, defined in~\eqref{eq:argminLosst_quadratic}, is the minimizer of the quadratic surrogate loss $\overline{\mathcal{L}}^{(\leq t)}$, obtained by locally approximating $\mathcal{L}^{(\leq t)}$ around $\theta^{(t)}$:
\begin{align*}
\overline{\mathcal{L}}^{(\leq t)}(\theta)
=
\mathcal{L}^{(\leq t)}(\theta^{(t)})
+ \nabla \mathcal{L}^{(\leq t)}(\theta^{(t)})^{\top}
\left(\theta-\theta^{(t)}\right)
+ \frac{1}{2}
\left(\theta-\theta^{(t)}\right)^{\top}
\nabla^{2}\mathcal{L}^{(\leq t)}(\theta^{(t)})
\left(\theta-\theta^{(t)}\right).
\end{align*}
Consequently, the increment $\overline{\theta}^{(t+1)}-\theta^{(t)}$ admits the explicit representation $\overline{\theta}^{(t+1)}-\theta^{(t)}
=
-\left[
\nabla^{2}\mathcal{L}^{(\leq t)}(\theta^{(t)})
\right]^{\dagger}
\nabla \ell_t(\theta^{(t)}).$

This explicit representation allows us to directly control the $\ell_{\infty}$ norm of the cumulative leading term,
$\sum_{t=0}^{T}
\left(\overline{\theta}^{(t+1)}-\theta^{(t)}\right),$ 
at the desired order $\mathcal{O}\big(\sqrt{\log n/(npL)}\big)$, by applying the martingale concentration inequality. It therefore remains to show that the cumulative approximation error
$\|\sum_{t=0}^{T}
\left(\theta^{(t+1)}-\overline{\theta}^{(t+1)}\right)\|_{\infty}$
is of a smaller order.

To establish this result, we leverage the technique developed in Theorem~7 of~\cite{fan2024uncertainty}, together with the preceding induction hypotheses. In particular, we show that the approximation error is of smaller order than $\mathcal{O}\big(\sqrt{\log n/(npL)}\big)$ provided that
$\log t \left[
\frac{\sqrt{\log n}}{\sqrt{nL}p}
+
\frac{1}{\sqrt{np}}
\left(
1+\frac{\log n}{\sqrt{np}}
\right)
\right]
=o(1).$\\ This then closes the induction argument.

Altogether, this completes the proof of Theorem~\ref{thm:iterative_loss}. The proof of Theorem~\ref{thm:asy_norm} then follows along similar lines once \textit{Step~III} is established. In particular, the leading term determines the asymptotic variance, while the approximation error is of a smaller order and is therefore asymptotically negligible.

\begin{figure}[htbp]
\centering
\begingroup
\setlength{\abovedisplayskip}{0pt}%
\setlength{\belowdisplayskip}{0pt}%
\setlength{\abovedisplayshortskip}{0pt}%
\setlength{\belowdisplayshortskip}{0pt}%
\scalebox{0.85}{
\begin{tikzpicture}[node distance=0.65cm and 0.35cm,
  aleft/.style={font=\small\itshape, text=red!80!black,
                anchor=east, xshift=-4pt, inner sep=2pt},
  aright/.style={font=\small\itshape, text=red!80!black,
                 anchor=west, xshift=4pt, inner sep=2pt},
]

\node[tbox] (SI){%
  \textcolor{cteal}{\textbf{Step~I (see Appendix~\ref{appendix:Step1})}}
  \begin{align*}
    \norm{\theta^{(T+1)}-\theta^{*}}_{\infty} \leq 5.
  \end{align*}
};

\node[iobox, above=0.65cm of SI, anchor=south] (IN){%
  \parbox{12cm}{%
    \centering
    \textbf{Statistical error bounds at iteration $\bm{t \in [T]}$}
    \begin{alignat*}{2}
      \norm{\theta^{(t)}-\theta^{(t-1)}}_{2}
        &\lesssim \frac{1}{t}\sqrt{\frac{1}{pL}},
      &\qquad
      \norm{\theta^{(t)}-\theta^{(t-1)}}_{\infty}
        &\lesssim \frac{1}{t}\sqrt{\frac{\log n}{npL}},  \\[1pt]
      \norm{\theta^{(t)}-\theta^{*}}_{2}
        &\lesssim \sqrt{\frac{1}{pL}},
      &\qquad
      \norm{\theta^{(t)}-\theta^{*}}_{\infty}
        &\lesssim \sqrt{\frac{\log n}{npL}}.
    \end{alignat*}
  }%
};

\node[wbox, below=0.65cm of SI.south, anchor=north] (SII){%
  \textcolor{cpurple}{\textbf{Step~II (see Appendix~\ref{appendix:Step2})}}
  \begin{alignat*}{2}
    \norm{\theta^{(T+1)}-\theta^{(T)}}_{2}
      \lesssim \frac{1}{T+1}\sqrt{\frac{1}{pL}},
    \qquad
    \norm{\theta^{(T+1)}-\theta^{(T)}}_{\infty}
    \lesssim \frac{1}{T+1}\sqrt{\frac{\log n}{npL}},
    \qquad
    \norm{\theta^{(T+1)}-\theta^{*}}_{2}
      \lesssim \sqrt{\frac{1}{pL}}.
  \end{alignat*}
};

\node[pbox, below=0.65cm of SII.south, anchor=north] (SIII){%
  \textcolor{cpurple}{\textbf{Step~III (see Appendix~\ref{appendix:Step3})}}
  \begin{align*}
    \norm{\theta^{(T+1)}-\theta^{*}}_{\infty}
      \lesssim \sqrt{\frac{\log n}{npL}}.
  \end{align*}
};

\node[iobox, below=0.65cm of SIII, anchor=north] (OUT){%
  \parbox{12cm}{%
    \centering
    \textbf{Statistical error bounds at iteration $\bm{t = T+1}$}%
  }
};

\draw[arr] (IN.south -| SI.north)
    -- node[aleft]  {}
       node[aright] {1st leave-one-out argument}
    (SI.north);

\draw[arr] (SI.south)
    -- node[aright] {2nd leave-one-out argument}
    (SI.south |- SII.north);

\draw[arr] (SII.south)
    -- node[aleft]  {quadratic proxy $\overline{\theta}^{(T+1)}$ (see \S\ref{subsec:asymptotic})}
       node[aright] {3rd leave-one-out argument}
    (SII.south |- SIII.north);

\draw[arr] (SIII.south) 
    -- node[aright] {}
    (SIII.south |- OUT.north);

\coordinate (Lbend) at ($(SII.west)+(-1.2, 0)$);
\draw[arr]
    (SII.west)
    -- (Lbend)
    -- node[aright, pos=0.5] {}
       (Lbend |- OUT.west)
    -- (OUT.west);

\node[below=0.55cm of OUT.south west, anchor=north west,
      font=\small] (LEG){%
  \begin{tikzpicture}
    \draw[draw=cteal, line width=0.5pt, fill=fteal, rounded corners=2pt]
      (0,0) rectangle (0.40,0.22);
    \node[font=\small, anchor=west] at (0.50, 0.11)
      {Stage~1: constant $\ell_\infty$ bound (Step~I)};
    \draw[draw=cpurple, line width=0.5pt, fill=fpurple, rounded corners=2pt]
      (0,-0.38) rectangle (0.40,-0.16);
    \node[font=\small, anchor=west] at (0.50, -0.27)
      {Stage~2: optimal $\ell_\infty$ rate (Steps~II--III)};
  \end{tikzpicture}%
};

\end{tikzpicture}
}
\endgroup

\caption{Proof outline for Theorem~\ref{thm:iterative_loss}.}
\label{fig:proof_structure}
\end{figure}
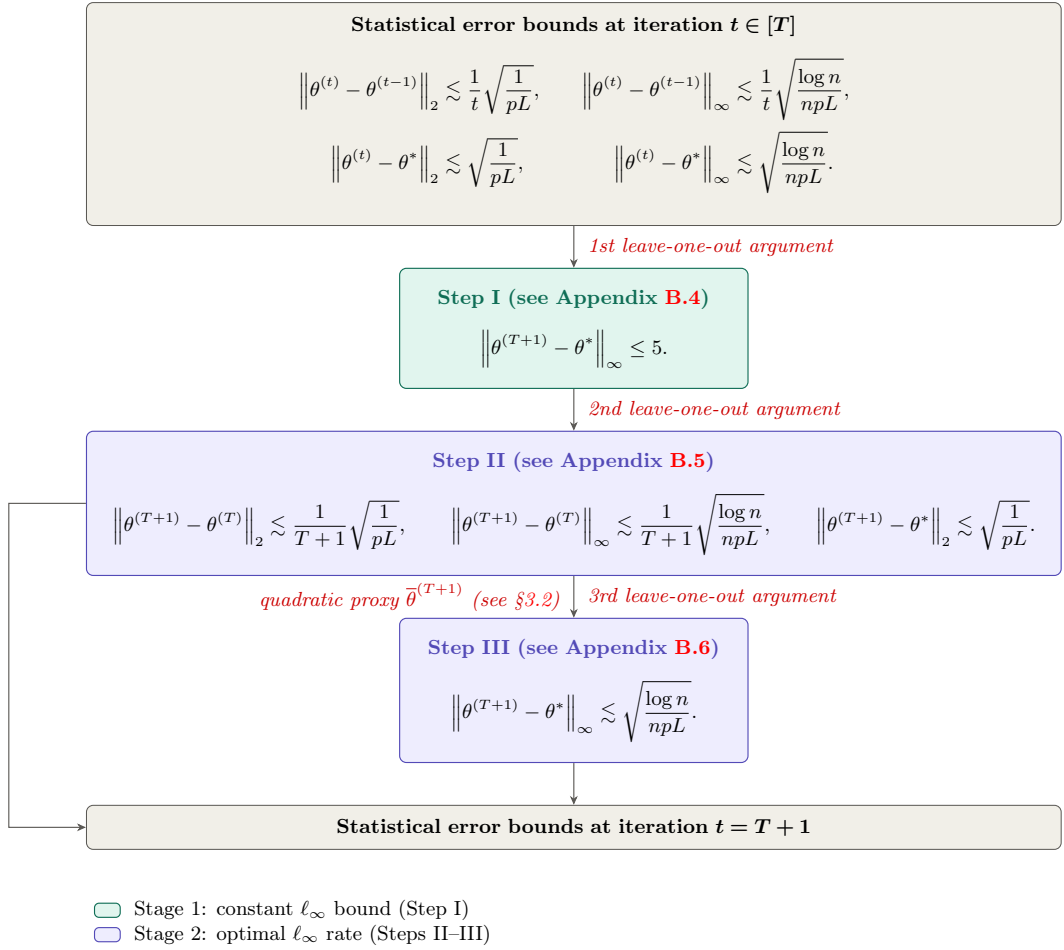

\section{Discussion}
This paper studies entity ranking under iterative synthetic data augmentation and shows that such a workflow can avoid model collapse under the BTL model. Specifically, we establish that the iterative MLE achieves optimal $\ell_2$ and $\ell_\infty$ statistical rates under a sparse comparison graph, even when the fraction of real data vanishes over iterations. The main technical challenge arises from the temporal dependence induced by iterative retraining. To address this challenge, we develop a coupled induction framework incorporating several leave-one-out constructions, which allows us to control error propagation across iterations and further establish the asymptotic normality of the iterative MLE. Numerical experiments, together with an empirical study on the Arena Human Preference 140K dataset, further support our theoretical findings. To the best of our knowledge, this is the first work to rigorously establish that model collapse can be avoided in ranking problems from a high-dimensional and non-asymptotic perspective.

There are several promising directions for future research. First, it would be valuable to incorporate covariates into the ranking framework \cite{fan2024covariate,li2025efficient,dong2026statistical}. Second, it remains an open question whether model collapse can be avoided under the sparsest graph regime. In this paper, our main theoretical results require $\sqrt{nL}p \gtrsim (\log n)^{3/2}$, which excludes the sparsest regime with $p=\cO(\log n/n)$. Extending our theoretical analysis to this regime remains technically challenging and warrants further investigation. Third, recent work \cite{wu2026does} studies model collapse under a graphical interactive training framework. It would be interesting to extend our analysis to this setting and investigate how AI training and evaluation can be integrated with human preferences collected from heterogeneous and interactive sources.

\clearpage
\bibliographystyle{apalike}
\bibliography{main}
\clearpage
\appendix
\section{Preliminaries} 
\label{sec:Preliminaries}
In this Appendix, we articulate some necessary facts of the gradient and the Hessian of the loss function defined in \eqref{eq:argminLossT}. To begin with, they can be expressed as
\begin{align*}
    \nabla \mathcal{L}^{(\leq T)}(\theta) 
    & := 
    \sum_{t=0}^{T} \nabla \ell_{t}(\theta) 
    =
    \sum_{t=0}^{T} \sum_{(i,j) \in \mathcal{E}^{(t)},\, i>j} \left\{ - y_{j,i}^{(t)}+ \frac{e^{\theta_i}}{e^{\theta_i} + e^{\theta_j}} \right\} (\boldsymbol{e}_{i} - \boldsymbol{e}_{j}),
    \\
    \nabla^{2} \mathcal{L}^{(\leq T)}(\theta) 
    &:= 
    \sum_{t=0}^{T} \nabla^{2} \ell_{t}(\theta) = \sum_{t=0}^{T} \sum_{(i,j) \in \mathcal{E}^{(t)},\, i>j} \frac{e^{\theta_i}e^{\theta_j}}{(e^{\theta_i} + e^{\theta_j})^{2}}(\boldsymbol{e}_{i} - \boldsymbol{e}_{j})(\boldsymbol{e}_{i} - \boldsymbol{e}_{j})^{\top},
\end{align*}
Here $\boldsymbol{e}_{1}, \cdots , \boldsymbol{e}_{n}$ stand for the canonical basis vectors in $\mathbb{R}^{n}$. The gradient vector $\nabla \mathcal{L}^{(\leq T)}(\theta)$ is controlled at the truth $\theta^{\ast}$ as follows.

\begin{lemma} \label{lemma:gradient_T_H}
    For any iteration $T \in \mathbb{N}$, the following event
    \begin{align}
        \mathcal{A}_{0}^{(T)} := \Bigg\lbrace \norm{\nabla \mathcal{L}^{(\le T)}(\theta^{(T)})}_2 = \norm{\nabla \ell_{T}(\theta^{(T)}) }_{2} \lesssim \sqrt{\frac{n^2 p}{L}} \Bigg\rbrace
    \end{align}
    occurs with probability exceeding $1 - \mathcal{O}(n^{-10})$.
\end{lemma}
\begin{proof}
    See \S\ref{appendix:gradient_T_H} for a detailed proof.
\end{proof}
Before proceeding to analyzing the Hessian matrix $\nabla^{2} \mathcal{L}^{(\leq T)}(\theta)$, we first illustrate the concentration of the vertex degrees in an Erdo\H{o}s-R\'enyi random graph. 
\begin{lemma}[Degree Concentration] \label{lemma:degree_concentration}
    Suppose that $\mathcal{G}^{(T)} \sim \mathcal{G}_{n,p}$. Let $d^{(T)}_{i}$ be the degree of node $i$, $d^{(T)}_{\operatorname{min}} = \min_{1 \leq i \leq n} d^{(T)}_{i}$ and $d^{(T)}_{\operatorname{max}} = \max_{1 \leq i \leq n} d^{(T)}_{i}$. If $p \geq \frac{c_{0} \log n}{n}$ for some sufficiently large constant $c_{0} > 0$, then the following event
    \begin{align*}
        \mathcal{A}_{1}^{(T)} := \Big\lbrace \frac{1}{2}np \leq d^{(T)}_{\operatorname{min}} \leq d^{(T)}_{\operatorname{max}} \leq \frac{3}{2}np \Big\rbrace
    \end{align*}
    happens with probability exceeding $1 - \mathcal{O}(n^{-10})$ when $n$ is large enough.
\end{lemma}

\begin{proof}
    The proof is a direct consequence of Bernstein’s inequality (or multiplicative Chernoff bounds) and a union bound argument and is hence omitted.
\end{proof}

Throughout the paper, we assume the conditions stated in Lemma~\ref{lemma:degree_concentration} hold. 
In addition, we denote by $\boldsymbol{L}_{\mathcal{G}}^{(T)} = \sum_{(i,j) \in \mathcal{E}^{(T)}, i > j} (\boldsymbol{e}_{i} - \boldsymbol{e}_{j}) (\boldsymbol{e}_{i} - \boldsymbol{e}_{j})^{\top}$ the (unnormalized) Laplacian matrix associated with $\mathcal{G}^{(T)}$. For any matrix $A$, we let
\begin{align}
    \lambda_{\operatorname{min}, \perp} \left( A \right) := \min \lbrace \mu \mid z^{\top}A z \geq \mu \norm{z}_{2}^{2} \text{ for all $z$ with } \boldsymbol{1}^{\top}z = 0 \rbrace,
\end{align}
namely, the smallest eigenvalue when restricted to vectors orthogonal to $\boldsymbol{1}$. We next summarize the smoothness and the strong convexity of the function $\mathcal{L}_{\lambda_{T}}^{\leq (T)}(\theta)$ in Lemma~\ref{lemma:largesteigen} and Lemma~\ref{lemma:smallesteigen}, respectively.

\begin{lemma} \label{lemma:largesteigen}
    Suppose event $\mathcal{A}_{1}^{(T)}$ holds, for any iteration $T \in \mathbb{N}$, one has
    \begin{align}
        \lambda_{\operatorname{max}} \left( \nabla^{2} \mathcal{L}^{\leq (T)}(\theta) \right) \leq (T+1)np, \quad \forall \theta \in \mathbb{R}^{n}.
    \end{align}
\end{lemma}

\begin{proof}
    Note that $\frac{e^{\theta_i}e^{\theta_j}}{(e^{\theta_i} + e^{\theta_j})^{2}} \leq \frac{1}{4}$. It follows immediately from the Hessian $\nabla^{2} \mathcal{L}^{\leq (T)}(\theta)$ that
    \begin{align*}
        \lambda_{\text{max}} \left( \nabla^{2} \mathcal{L}^{\leq (T)}(\theta) \right) \leq \frac{1}{4} \sum_{t=0}^{T} \norm{\boldsymbol{L}_{\mathcal{G}}^{(t)}} \leq \frac{1}{2} \sum_{t=0}^{T} d^{(t)}_{\operatorname{max}},
    \end{align*}
    where $d^{(t)}_{\operatorname{max}}$ is the maximum vertex degree in the graph $\mathcal{G}^{(t)}$. On the event $\mathcal{A}_{1}^{(t)}$ we have $d_{\operatorname{max}} \leq 2np$, which completes the proof.
\end{proof}

\begin{lemma} \label{lemma:smallesteigen}
    Suppose event $\mathcal{A}_{1}^{(T)}$ holds, for any iteration $T \in \mathbb{N}$, with probability exceeding $1 - \mathcal{O}\left(n^{-10} \right)$ one has
    \begin{align}
        \lambda_{\operatorname{min}, \perp} \left( \nabla^{2} \mathcal{L}^{\leq (T)}(\theta) \right) \geq \frac{(T+1)np}{8 \kappa e^{2C}}
    \end{align}
    simultaneously for all $\theta$ obeying $\norm{\theta - \theta^{\ast}}_{\infty} \leq C$ for some $C \geq 0$.
\end{lemma}
\begin{proof}
    We refer to \S\ref{appendix:smallesteigen} for the detailed proof.
\end{proof}
\section{Proof Outline of Theorem~\ref{thm:iterative_loss}} 
\label{appendix:proof_outline}
\subsection{The Regularized MLE} \label{sec:regularized}
To understand the statistical error of $\theta^{(T)}$, we begin with analyzing the regularized MLE $\theta^{(T)}_{\lambda}$, defined as the solution of the following optimization problem
\begin{align} \label{eq:argminLossT_regularzied}
    \theta^{(T+1)}_{\lambda} = \arg\min_{\theta \in \Theta} \ \mathcal{L}^{(\leq T)}_{\lambda_{T}}(\theta),
\end{align} 
where $\mathcal{L}_{\lambda_{T}}^{(\leq T)}(\theta) := \mathcal{L}^{(\leq T)}(\theta) + \frac{1}{2} \lambda_{T} \norm{\theta}_{2}^{2}$ for $\lambda_{T} := \frac{T+1}{n}$. 

Next, we consider the standard gradient descent algorithm that is expected to converge to the minimizer $\theta^{(T+1)}_{\lambda}$ at each fixed iteration $T$, and analyze the trajectory of this iterative algorithm instead. The algorithm is stated in Algorithm~\ref{algorithm:regularized_IMLE}. 
Notably, this gradient descent algorithm is not practical since the initial point is set to be $\theta^{\ast}$.  We shall adopt a time-invariant step size rule at fixed iteration $T$:
\begin{align*}
    \eta_{\tau} \equiv \eta := \frac{1}{\lambda_{T} + (T+1)np}, \quad \tau = 0, 1, 2, \cdots.
\end{align*}

\begin{algorithm}[H] 
\caption{Gradient descent for computing the regularized MLE} \label{algorithm:regularized_IMLE}
\begin{algorithmic}[1] 
\Require Parameter space $\Theta \subseteq \mathbb{R}^n$; family of distributions $\{ p_{\theta} \}_{\theta \in \Theta}$ over input space $\mathcal{Y}$; target parameter $\theta^\ast \in \Theta$.
\State Set $\theta^{(0)} := \theta^\ast$. 
\For{$\tau = 0, 1, 2, \dots$}
    \begin{align*}
        \theta^{(\tau+1)} = \theta^{(\tau)} - \eta_{\tau} \nabla \mathcal{L}_{\lambda_{T}}^{(\leq T)} (\theta^{(\tau)}).
    \end{align*}
\EndFor
\end{algorithmic}
\end{algorithm}
\subsection{\texorpdfstring{Leave-One-Out Arguments and Induction Hypotheses for Iteration $t \in [T]$}{Leave-One-Out Arguments and Induction Hypotheses for iteration t in [T]}}
\label{appendix:loo_induction_1T}
We leverage the leave-one-out technique \citep{ma2018implicit, chen2019bspectral,chen2020noisy} and bound the statistical error by induction. More specifically, for any $m \in [n]$, we first consider the leave-one-out sequence $\lbrace \theta^{(\tau, m)} \rbrace_{\tau = 0,1,\cdots}$ such that
\begin{align*}
    \theta^{(\tau+1, m)} 
    = 
    \theta^{(\tau, m)} - \eta \nabla \mathcal{L}_{\lambda_{T}}^{(m,\leq T)} ( \theta^{(\tau, m)}) 
\end{align*}
where $\theta^{(0, m)} = \theta^{(0)} = \theta^{\ast}$ and
\begin{align*}
    \nabla \mathcal{L}_{\lambda_{T}}^{(m,\leq T)} ( \theta^{(\tau, m)})
    & =
    \nabla \mathcal{L}^{(m,\leq T)} ( \theta^{(\tau, m)}) + \lambda_{T} \theta^{(\tau, m)}
    \\ & =
    \sum_{t=0}^{T} \nabla \ell_{t}^{(m)}(\theta^{(\tau, m)}) + \lambda_{T} \theta^{(\tau, m)},
    \\
    \ell_{t}^{(m)} (\theta) 
    & =
    \sum_{(i,j) \in \mathcal{E}_{t},i > j, i \neq m, j \neq m} \left\{ - y_{j,i}^{(t)} (\theta_i - \theta_j) + \log\left(1 + e^{\theta_i - \theta_j} \right) \right\}
    \\ &\quad +
    \sum_{i \in [n] \backslash \lbrace m \rbrace} p \left\{ - \frac{e^{\theta^{(t)}_{i}}}{e^{\theta_{i}^{(t)}} + e^{\theta_{m}^{(t)}}} (\theta_{i} - \theta_{m}) + \log\left(1 + e^{\theta_{i} - \theta_{m}} \right) \right\}.
\end{align*}
Note that $\boldsymbol{1}^{\top}\theta^{\ast} = 0$ implies $\boldsymbol{1}^{\top} \theta^{(\tau)} = \boldsymbol{1}^{\top} \theta^{(\tau, m)} = 0$ for all $\tau = 0, \cdots, T$.

For the sake of clarity,  we first list two main induction hypotheses needed in our analysis:
\begin{itemize}
    \item[1.] The auxiliary regularized iterated updates at iteration $T$:
    \begin{align}
        \max_{m \in [n]} \norm{ \theta^{(\tau, m)} - \theta^{(\tau)}}_{2} 
        & \leq 
        1, 
        \tag{A}\label{eq:A} 
        \\
        \norm{\theta^{(\tau)} - \theta^{\ast}}_{2} 
        & \leq 
        \sqrt{\frac{n}{\log n}}
        \tag{B}\label{eq:B},
        \\
        \max_{m \in [n]} \left | \theta_{m}^{(\tau, m)} - \theta_{m}^{\ast} \right | 
        & \leq 
        1,  
        \tag{C}\label{eq:C}
    \end{align}
    \item[2.] The primary unregularized iterated updates up to iteration $t = 0, 1, \cdots, T$:
    \begin{align}
        \norm{\theta^{(t)} - \theta^{(t)}_{\lambda}}_{2} 
        & 
        \leq D_{1} \sqrt{\frac{1}{n}},
        \tag{D}\label{eq:D}
        \\
        \norm{\theta^{(t)}-\theta^{(t-1)}}_{\infty} 
        & \leq 
        D_{2} \frac{1}{t} \sqrt{\frac{\log n}{npL}}, 
        \tag{E}\label{eq:E}
        \\
        \norm{\theta^{(t)}-\theta^{(t-1)}}_{2} 
        & \leq 
        D_{3} \frac{1}{t} \sqrt{\frac{1}{pL}}, 
        \tag{F}\label{eq:F}
        \\
        \norm{\theta^{(t)} - \theta^{\ast}}_{2} 
        & \leq 
        D_{4} \sqrt{\frac{1}{pL}},
        \tag{G}\label{eq:G}
        \\
        \norm{\theta^{(t)} - \theta^{\ast}}_{\infty} 
        &\leq 
        D_{5} \sqrt{\frac{\log n}{npL}}, \tag{H}\label{eq:H}
    \end{align}
    where $D_{1}, \cdots, D_{5} > 0$ are some sufficiently large constants and we assume $\theta^{(-1)} = \theta^{(0)} = \theta^{(0,1)} = \cdots = \theta^{(0,n)} = \theta^{\ast}$. 
\end{itemize}
It is obvious that \eqref{eq:A}-\eqref{eq:C} hold for $\tau = 0$ and \eqref{eq:D}-\eqref{eq:H} hold for $t = 0$ automatically. For the auxiliary gradient descent, suppose \eqref{eq:A}-\eqref{eq:C} are true for $1, \cdots, \tau$, then we need to show the same conclusions continue to hold for $\tau + 1$ by induction. Then we use the same technique for the primary iteration of \eqref{eq:D}-\eqref{eq:H} for the case of $T+1$.

There are a few direct consequences of \eqref{eq:A}-\eqref{eq:H} that are worth listing, we summarize them as follows
\begin{align}
    \norm{\theta^{(\tau)} - \theta^{\ast}}_{\infty} 
    & \leq  
    2,
    \tag{I}\label{eq:I}
    \\
    \max_{m \in [n]} \norm{\theta^{(\tau, m)}-\theta^{\ast}}_{\infty} 
    &\leq 
    3,  
    \tag{J}\label{eq:J}
    \\
    \max_{m \in [n]} \norm{\theta^{(\tau, m)}-\theta^{\ast}}_{2} 
    &\leq 
    1 + \sqrt{\frac{n}{\log n}}, 
    \tag{K}\label{eq:K}
    \\
    \max_{m \in [n]} \norm{\theta^{(\tau, m)} - \theta^{(t)}}_{\infty} 
    & \leq  
    3 + D_{5} \sqrt{\frac{\log n}{npL}},
    \tag{L}\label{eq:L}
    \\
    \max_{m \in [n]} \norm{\theta^{(\tau, m)} - \theta^{(t)}}_{2}  
    & \leq  
    1 + \sqrt{\frac{n}{\log n}} + D_{4} \sqrt{\frac{1}{pL}}.
    \tag{M}\label{eq:M}
\end{align}
\subsection{\texorpdfstring{Leave-One-Out Arguments and Induction Hypotheses for Iteration $t = T+1$}{Leave-One-Out Arguments and Induction Hypotheses for iteration t = T+1}}
\label{appendix:loo_induction_T+1}
It is worth noting that, to establish the primary induction hypotheses \eqref{eq:E} and \eqref{eq:F} at iteration $T+1$, it is necessary to introduce an iteration-restricted leave-one-out construction, together with a tertiary induction hypothesis that controls the difference between $\theta^{(T+1)}$ and $\theta^{(T)}$, as follows.

At any iteration $T \geq 1$,  we denote $\theta_{0,\lambda}^{(T+1)}$ as the solution of the following minimization problem
\begin{align} \label{eq:Constraint_MLE_3}
    \theta^{(T+1)}_{0, \lambda} := \underset{\boldsymbol{1}^{\top}\theta = 0} {\arg \min} 
    \; \mathcal{L}_{\lambda_{0}}^{(\leq T)}(\theta),
\end{align}
where 
\begin{align*}
    \mathcal{L}_{\lambda_{0}}^{(\leq T)} ( \theta_{0}^{(\upsilon)}) :=\mathcal{L}^{(\leq T)} ( \theta_{0}^{(\upsilon)}) + \frac{\lambda_{0}}{2} \norm{\theta_{0}^{(\upsilon)}}_{2}^{2} = \sum_{t=0}^{T} \nabla \ell_{t}(\theta_{0}^{(\upsilon)})+ \frac{\lambda_{0}}{2} \norm{\theta_{0}^{(\upsilon)}}_{2}^{2}.
\end{align*}
Similarly, we define an auxiliary gradient descent sequence $\lbrace \theta_{0}^{(\upsilon)} \rbrace$ such that 
\begin{align*}
    \theta_{0}^{(\upsilon+1)} 
    &= 
    \theta_{0}^{(\upsilon)} - \eta_{\upsilon} \nabla \mathcal{L}_{\lambda_{0}}^{(\leq T)} ( \theta_{0}^{(\upsilon)}),
\end{align*}
where we shall adopt a time invariant step size rule:
\begin{align*}
    \eta_{\upsilon} \equiv \eta := \frac{1}{\lambda_{T} + (T+1)np} = \frac{1}{T+1}\frac{1}{\lambda_{0}+np}, \quad \upsilon = 0, 1, 2, \cdots.
\end{align*}
For any $m \in [n]$, we then introduce an iteration-restricted leave-one-out  sequence $\lbrace \theta_{0}^{(\upsilon, m)} \rbrace_{\upsilon = 0, 1, \cdots}$ constructed via
\begin{align*}
    \theta_{0}^{(\upsilon+1, m)} 
    = 
    \theta_{0}^{(\upsilon, m)} - \eta \nabla \mathcal{L}_{\lambda_{0}}^{(T \backslash m,\leq T)} ( \theta_{0}^{(\upsilon, m)}),
\end{align*}
where the iteration-restricted leave-one-out loss function $\mathcal{L}_{\lambda_{0}}^{(T \backslash m,\leq T)} ( \theta)$ replaces all likelihood components involving the $m$-th  item with their expected values at iteration $T$ that 
\begin{align*}
    \mathcal{L}_{\lambda_{0}}^{(T \backslash m,\leq T)} ( \theta)
    & := \mathcal{L}^{(\leq T-1)} (\theta) + \ell_{T}^{(m)}(\theta) + \frac{1}{2} \lambda_{0} \norm{\theta}_{2}^{2}
\end{align*}
and $\ell_{T}^{(m)}$ is $\ell_{T}$ with all terms including $m$ deleted that
\begin{align*}
    \ell_{T}^{(m)} (\theta) = &\sum_{(i,j) \in \mathcal{E}_{T},i > j, i \neq m, j \neq m} \left\{ - y_{j,i}^{(T)} (\theta_i - \theta_j) + \log\left(1 + e^{\theta_i - \theta_j} \right) \right\}
    \\ &\quad +
    \sum_{i \in [n] \backslash \lbrace m \rbrace} p \left\{ - \frac{e^{\theta^{(T)}_{i}}}{e^{\theta_{i}^{(T)}} + e^{\theta_{m}^{(T)}}} (\theta_{i} - \theta_{m}) + \log\left(1 + e^{\theta_{i} - \theta_{m}} \right) \right\}.
\end{align*}
Note that $\theta_{0}^{(0, m)} = \theta_{0}^{(0)} = \theta^{(T)}$, and $\boldsymbol{1}^{\top}\theta^{(T)} = 0$ implies $\boldsymbol{1}^{\top} \theta_{0}^{(\upsilon)} = \boldsymbol{1}^{\top} \theta_{0}^{(\upsilon, m)} = 0$ for all $\tau$.
For clarity, we state the tertiary induction hypotheses.

\begin{align}
    \norm{\theta_{0}^{(\upsilon)} - \theta^{(T)}}_{2} 
    & \leq D_{5} \frac{1}{T+1}
    \sqrt{\frac{\log n}{pL}},
    \tag{A1}\label{eq:A1}
    \\
    \max_{m \in [n]} \left| \theta_{0,m}^{(\upsilon,m)} - \theta_{m}^{(T)} \right| 
    & \leq 
    D_{6} \frac{1}{T+1} \sqrt{\frac{\log n}{npL}},
    \tag{B1}\label{eq:B1}
    \\
    \max_{m \in [n]} \norm{\theta_{0}^{(\upsilon)} - \theta_{0}^{(\upsilon,m)}}_{2} 
    & \leq 
    D_{7} \frac{1}{T+1} \sqrt{\frac{\log n}{npL}},
    \tag{C1}\label{eq:C1}
    \\
    \norm{\theta_{0}^{(\upsilon)}-\theta^{(T)}}_{\infty}
    & \leq 
    D_{8} \frac{1}{T+1} \sqrt{\frac{\log n}{npL}},
    \tag{D1}\label{eq:D1}
\end{align}
thus the following two bounds hold
\begin{align}
    \max_{m \in [n]} \norm{\theta_{0}^{(\upsilon,m)}-\theta^{(T)}}_{\infty} 
    & \leq 
    D_{9} \frac{1}{T+1} \sqrt{\frac{\log n}{npL}},
    \tag{E1}\label{eq:E1}
    \\
    \max_{m \in [n]} \norm{\theta_{0}^{(\upsilon,m)}-\theta^{(T)}}_{2}
    & \leq
    D_{10} \frac{1}{T+1} \sqrt{\frac{\log n}{pL}}.
    \tag{F1}\label{eq:F1}
\end{align}
For $\upsilon = 0$, since $\theta_{0}^{0} = \theta_{0}^{(0,1)} = \cdots = \theta_{0}^{(0,n)}$, the hypotheses \eqref{eq:A1}-\eqref{eq:D1} are satisfied trivially. In the following lemmas, we show that, provided \eqref{eq:A1}-\eqref{eq:D1} hold up to step $\upsilon$ and the primary induction hypotheses \eqref{eq:D}-\eqref{eq:H} hold at iteration $T$, the conclusions of \eqref{eq:A1}-\eqref{eq:D1} also hold at step $\upsilon+1$.
\subsection{Step I}
\label{appendix:Step1}
The objective of this step is to show that $\norm{\theta^{(T+1)} - \theta^{\ast}}_{\infty}$ remains bounded by a finite constant with high probability, and then to sharpen this bound by establishing the auxiliary induction hypotheses at $(\tau+1)$-th step. The key observation is that $\theta^{(t+1)}$ should remain close to its regularized counterpart $\theta_{\lambda}^{(t+1)}$ when the regularization parameter is sufficiently small.

\begin{lemma} \label{lemma:equation_D_plus}
    Suppose the auxiliary induction hypotheses \eqref{eq:A}-\eqref{eq:C} hold for $0,\cdots, \tau^{\ast}$-th step and the primary induction hypotheses \eqref{eq:D}-\eqref{eq:H} hold up to iteration $T$, then with probability exceeding $1 - \mathcal{O}(n^{-7})$, one has
    \begin{align*}
        \norm{\theta^{(T+1)} - \theta^{\ast}}_{\infty} \leq 5.
    \end{align*}
\end{lemma}
\begin{proof}
    See \S\ref{appendix:equation_D_plus} for a detailed proof outline.
\end{proof}

Lemma~\ref{lemma:equation_D_plus} is arguably the key technical step in the analysis of the MLE, since it provides direct control of the spectrum of $\nabla^{2} \mathcal{L}_{\lambda_{T}}^{(\leq T)}(\theta)$. As an intermediate consequence of Lemma~\ref{lemma:equation_D_plus}, the primary induction hypothesis \eqref{eq:D} for iteration $T+1$ follows immediately, and we omit the proof.
\subsection{Step II}
\label{appendix:Step2}
The goal of this step is to establish the primary induction hypotheses \eqref{eq:E}-\eqref{eq:G}. The main idea is to introduce an iteration-restricted leave-one-out construction initialized at $\theta^{(T)}$, which yields refined control of the increment $\theta^{(T+1)}-\theta^{(T)}$. By expressing this increment as the sum of a deterministic bias term and a centered stochastic fluctuation term, and controlling the two components separately, we show that $\theta^{(T+1)}$ does not deviate substantially from $\theta^{(T)}$ in either the $\ell_{\infty}$- and the $\ell_{2}$-norm. We then use the $\ell_{\infty}$ control to further bound $\norm{\theta^{(T+1)} - \theta^{\ast}}_{2}$. 

\begin{lemma} \label{lemma:equation_EFG}
    Suppose the tertiary induction hypotheses \eqref{eq:A1}-\eqref{eq:D1} hold for $0,\cdots, \upsilon^{\ast}$-th step, the auxiliary induction hypotheses \eqref{eq:A}-\eqref{eq:C} hold for $0,\cdots, \tau^{\ast}$-th step and the primary induction hypotheses \eqref{eq:D}-\eqref{eq:H} hold up to iteration $T$, then with probability exceeding $1 - \mathcal{O}(n^{-5})$, one has
    \begin{align*}
        \norm{\theta^{(T+1)}-\theta^{(T)}}_{\infty} 
        & \leq 
        D_{2} \frac{1}{T+1} \sqrt{\frac{\log n}{npL}},
        \\
        \norm{\theta^{(T+1)}-\theta^{(T)}}_{2} 
        & \leq 
        D_{3} \frac{1}{T+1} \sqrt{\frac{1}{pL}},
        \\
        \norm{\theta^{(T+1)}-\theta^{\ast}}_{2} 
        & \leq 
        D_{4} \sqrt{\frac{1}{pL}}.
    \end{align*}
\end{lemma}
\begin{proof}
    See \S\ref{appendix:equation_EFG} for a detailed proof.
\end{proof}
\subsection{Step III}
\label{appendix:Step3}
The main technical challenge here lies in showing that, under suitable regularity conditions, the difference $\overline{\theta}^{(t+1)} - \theta^{(t+1)}$ is negligible compared to $\overline{\theta}^{(t+1)} - \theta^{(t)}$. To address this, we leverage two parallel leave-one-out \citep{chen2019bspectral,chen2021spectral,fan2024uncertainty} constructions on both $\mathcal{L}^{(\le t)}(\cdot)$ and $\overline{\mathcal{L}}^{(\le t)}(\cdot)$ to control the approximation error $\overline{\theta}^{(t+1)} - \theta^{(t+1)}$ in $\ell_{\infty}$-norm. We summarize the upper bound of $\norm{\theta^{(T+1)} - \theta^{\ast}}_{\infty}$ in the following Lemma~\ref{lemma:equation_H}.

\begin{lemma} \label{lemma:equation_H}
    Suppose the auxiliary induction hypotheses \eqref{eq:A}-\eqref{eq:C} hold for $0,\cdots, \tau^{\ast}$-th step and the primary induction hypotheses \eqref{eq:D}-\eqref{eq:H} hold up to iteration $T$, as long as $\frac{\log n}{np} + \frac{1}{t+1} \frac{1}{\sqrt{np}} = \mathcal{O}(1)$ and $\log T \left[ \frac{\sqrt{\log n}}{\sqrt{nL}p} + \frac{1}{\sqrt{np}} \left( 1 + \frac{\log n}{\sqrt{np}} \right) \right] = \mathcal{O}(1)$,  then with probability exceeding $1 - \mathcal{O}(n^{-5})$, one has
    \begin{align*}
        \norm{\theta^{(T+1)} - \theta^{\ast}}_{\infty} 
        &\leq 
        D_{5} \sqrt{\frac{\log n}{npL}}.
    \end{align*}
\end{lemma}
\begin{proof}
    See \S\ref{appendix:equation_H} for a detailed proof.
\end{proof}
\section{Proof of Auxiliary Lemmas in Appendix~\ref{sec:Preliminaries}}
\label{appendix:first}
In this Appendix, we prove detailed proof of aforementioned building blocks.
\subsection{Proof of Lemma~\ref{lemma:gradient_T_H} and Its Corollaries}
\label{appendix:gradient_T_H}
We provide an upper bound for the gradient vector in $\ell_{2}$-norm as stated in Lemma~\ref{lemma:gradient_T_H}.
\begin{proof}
    Observe that
    \begin{align*}
         \left[ \nabla \mathcal{L}^{(\le T)}(\theta^{(T)})\right]_{i} = \left[ \nabla \ell_{T}(\theta^{(T)}) \right]_{i} = \sum_{j \in [n] \backslash \{ i \}} G_{j,i}^{(T)} \left\{ - y_{j,i}^{(T)}+ \frac{e^{\theta_{i}^T}}{e^{\theta_{i}^T} + e^{\theta_{j}^T}} \right\},
    \end{align*}
    one can derive
    \begin{align*}
        \norm{\nabla \ell_{T}(\theta^{(T)})}_{2} = \sqrt{\sum_{i=1}^{n} \left( \sum_{j \in [n] \backslash \{ i \}} G_{j,i}^{(T)} \left\{ - y_{j,i}^{(T)}+ \frac{e^{\theta_{i}^T}}{e^{\theta_{i}^T} + e^{\theta_{j}^T}} \right\} \right)^{2}}.
    \end{align*}
    Let $\mathcal{U} := \left\lbrace u  \in \mathbb{R}^{n}: \sum_{i \in [n]} u_{i}^{2}\leq 1 \right\rbrace$ be the unit ball in $\mathbb{R}^{n}$. then there exists a symmetric $1/2$-net of $\mathcal{V} \subset \mathcal{U}$ such that for any $u \in \mathcal{U}$, there is a $v \in \mathcal{V}$ satisfying $\norm{u-v}_{2} \leq \frac{1}{2}$. moreover, we also have $\log |\mathcal{V}| \leq C^{\prime} n$ for some harmless constant $C^{\prime} > 0$. See Lemma 5.2 in \cite{vershynin2010introduction}.
    
    By Cauchy-Schwarz inequality and dual form of $L_{2}$ norm, for any $u \in \mathcal{U}$, with the corresponding $v \in \mathcal{V}$, we obtain
    \begin{align*}
        u^{\top} \nabla \ell_{T}(\theta^{(T)}) \leq \norm{u}_{2}\norm{\nabla \ell_{T}(\theta^{(T)})}_{2} \leq \norm{\nabla \ell_{T}(\theta^{(T)})}_{2},
    \end{align*}
    and
    \begin{align*}
        u^{\top} \nabla \ell_{T}(\theta^{(T)}) 
        &= 
        v^{\top} \nabla \ell_{T}(\theta^{(T)}) + (u-v)^{\top} \nabla \ell_{T}(\theta^{(T)}) 
        \\ &\leq 
        v^{\top} \nabla \ell_{T}(\theta^{(T)}) + \norm{u-v}_{2}\norm{\nabla \ell_{T}(\theta^{(T)})}_{2} 
        \\ &\leq 
        v^{\top} \nabla \ell_{T}(\theta^{(T)}) + \frac{1}{2} \norm{\nabla \ell_{T}(\theta^{(T)})}_{2},
    \end{align*}
    which further indicate that
    \begin{align*}
        \norm{\nabla \ell_{T}(\theta^{(T)})}_{2} 
        & = 
        \max_{u \in \mathcal{U}} u^{\top} \nabla \ell_{T} (\theta^{(T)}) 
        \\ & \leq
        2 \max_{v \in \mathcal{V}} v^{\top} \nabla \ell_{T} (\theta^{(T)}) 
        \\ & =
        2 \max_{v \in \mathcal{V}} \sum_{i=1}^{n} v_{i} \left( \sum_{j \in [n] \backslash \{ i \}} G_{j,i}^{(T)} \left\{ - y_{j,i}^{(T)}+ \frac{e^{\theta_{i}^T}}{e^{\theta_{i}^T} + e^{\theta_{}^T}} \right\} \right)
        \\ &=
        2 \max_{v \in \mathcal{V}} \sum_{1 \leq i < j \leq n}^{n} G_{j,i}^{(T)} (v_{i} - v_{j}) \left\{ - y_{j,i}^{(T)}+ \frac{e^{\theta_{i}^T}}{e^{\theta_{i}^T} + e^{\theta_{j}^T}} \right\}.
    \end{align*}
    Fix $v \in \mathcal{V}$ and condition on the graph $\mathcal{G}^{(T)}$, we find $G_{j,i}^{(T)} (v_{i} - v_{j}) \left\{ - y_{j,i}^{(T)}+ \frac{e^{\theta_{i}^T}}{e^{\theta_{i}^T} + e^{\theta_{j}^T}} \right\}$ are independent, mean-zero, and bounded. Applying Hoeffding's inequality and union bound on the last line, for any $\epsilon > 0$, we have
    \begin{align*}
        \mathbb{P} \left( \left| v^{\top} \nabla \ell_{T} (\theta^{(T)}) \right| \geq \epsilon \mid \mathcal{G}^{(T)} \right) \leq 2 \exp \left(-\frac{L\,\epsilon^2}{2\sum_{i<j}G_{j,i}^{(T)}(v_i-v_j)^2}\right),
    \end{align*}
    and thus
    \begin{align*}
        \mathbb{P} \left( \max_{v \in \mathcal{V}}\left| v^{\top} \nabla \ell_{T} (\theta^{(T)}) \right| \geq \epsilon \mid \mathcal{G}^{(T)} \right) 
        & \leq 2 
        |\mathcal{V}| \exp \left(-\frac{L\,\epsilon^2}{2  \max_{v \in \mathcal{V}} \sum_{i<j}G_{j,i}^{(T)}(v_i-v_j)^2}\right)
        \\ & \leq 2
        |\mathcal{V}| \exp \left(-\frac{L\,\epsilon^2}{2  \max_{\norm{v}_{2} \leq 1} \sum_{i<j}G_{j,i}^{(T)}(v_i-v_j)^2}\right)
        \\ & = 2
        |\mathcal{V}| \exp \left(-\frac{L\,\epsilon^2}{2  \lambda_{\max} \left( \boldsymbol{L}_{\mathcal{G}}^{(T)} \right)} \right).
    \end{align*}
    With appropriate choice of $\epsilon = C \sqrt{\frac{\left( \log |\mathcal{V}| + a \right) \lambda_{\max} \left( \boldsymbol{L}_{\mathcal{G}}^{(T)} \right)}{L}}$, and rearranging the above inequality, one can achieve
    \begin{align*}
        \mathbb{P} \left( \max_{v \in \mathcal{V}}\left| v^{\top} \nabla \ell_{T} (\theta^{(T)}) \right| \leq C \sqrt{\frac{\left( \log |\mathcal{V}| + a \right) \lambda_{\max} \left( \boldsymbol{L}_{\mathcal{G}}^{(T)} \right)}{L}} \mid \mathcal{G}^{(T)} \right) \geq 1 - 2e^{-a}.
    \end{align*}
    Using $\log|\mathcal{V}|\le C'n$ and setting $a=11\log n$ we can obtain the desired bound for the conclusion by the law of total probability and Lemma \ref{lemma:L_eigenvalue_min2} that
    \begin{align*}
        \norm{\nabla \ell_{T}(\theta^{(T)})}_{2} 
        \lesssim
        \sqrt{\frac{n^{2}p}{L}}
    \end{align*}
    with probability at least $1 - O(n^{-10})$. 
\end{proof}
Once Lemma~\ref{lemma:gradient_T_H} has been established, we have the following two direct corollaries of Lemma~\ref{lemma:gradient_T_H}.
\begin{corollary} \label{cor:sum_gradient_T}
    Suppose event $\mathcal{A}_{1}^{(T)}$ holds, for any iteration $T \in \mathbb{N}$, the following event
    \begin{align*}
        \mathcal{A}_{2}^{(T)}
        := \left\{
        \norm{\sum_{t=0}^{T} \nabla \ell_t(\theta^{(t)})}_{2}
        \lesssim \sqrt{\frac{(T+1)n^{2}p}{L}}
        \right\}
    \end{align*}
    holds with probability at least $1-\mathcal O(n^{-10})$.
\end{corollary}

\begin{proof}
    Define the $\sigma$-field
    \begin{align*}
        \mathcal{H}_{s}:=\sigma \left\langle \{G^{(t)}\}_{t=0}^{s},\{Y^{(t)}\}_{t=0}^{s} \right\rangle, \quad \forall s \in [T].
    \end{align*}
    For $v\in\mathbb{R}^n$, let
    \begin{align*}
        S_v
        &:= v^\top S_{T} = v^\top \sum_{t=0}^{T} \nabla \ell_t(\theta^{(t)})
        \\
        &= \sum_{t=0}^{T}\sum_{1\le j<i\le n} G^{(t)}_{i,j}\,(v_i{-}v_j)
        \left\{\frac{e^{\theta^{(t)}_{i}}}{e^{\theta^{(t)}_{i}}+e^{\theta^{(t)}_{j}}} - y^{(t)}_{j,i}\right\},
        \quad y^{(t)}_{j,i}:=\frac{1}{L}\sum_{l=1}^{L}y^{(t,l)}_{j,i}.
    \end{align*}
    Then $\left\lbrace \left( v^{\top}S_{s}, \mathcal{H}_{s} \right): s \in [T] \right\rbrace$ is a martingale difference array such that $v^{\top}S_{s}$ is bounded for all $s \in [T]$. Azuma-Hoeffding's inequality therefore yields, for any $\epsilon>0$,
    \begin{align*}
        \mathbb{P}\big(|S_v|\ge \epsilon\big)
        \ \le\ 2 \exp\!\left(-\frac{L\,\epsilon^2}{2\sum_{t=0}^{T}\sum_{i>j} G^{(t)}_{i,j}(v_i{-}v_j)^2}\right).
    \end{align*}
    Let $\mathcal{U}:=\{u\in\mathbb R^n:\sum_{i=1}^{n}u_i^2\le 1\}$ and fix a symmetric $1/2$-net $\mathcal{V}\subset\mathcal{U}$ such that $\log|\mathcal{V}|\le C' n$. By the dual characterization of the $\ell_2$-norm,
    \begin{align*}
        \norm{\sum_{t=0}^{T} \nabla \ell_t(\theta^{(t)})}_2 \leq 2 \max_{v\in\mathcal{V}} S_v
        = 2 \max_{v\in\mathcal{V}} |S_v|.
    \end{align*}
    Applying the union bound, we obtain for any $a>0$,
    \begin{align*}
        \mathbb{P}\!\left(\max_{v\in\mathcal{V}}|S_v|\ge \epsilon \right)
        & \leq 
        2 |\mathcal{V}|\exp\!\left(-\frac{L\,\epsilon^2}{2 \max_{v\in\mathcal{V}}\sum_{t=0}^{T}\sum_{i>j}G^{(t)}_{i,j}(v_i{-}v_j)^2}\right)
        \\ & \leq
        2 |\mathcal{V}|\exp\!\left(-\frac{L\,\epsilon^2}{2 \max_{\norm{v}_{2} \leq 1}\sum_{t=0}^{T}\sum_{i>j}G^{(t)}_{i,j}(v_i{-}v_j)^2}\right)
        \\ &=
        2 |\mathcal{V}|\exp\!\left(-\frac{L\,\epsilon^2}{2\lambda_{\max}\!\left(\sum_{t=0}^{T} L_{\mathcal{G}}^{(t)}\right)}\right).
    \end{align*}
    Choose
    \begin{align*}
        \epsilon \;=\; C\,
        \sqrt{\frac{\big(\log|\mathcal{V}|+a\big)\ \lambda_{\max}\!\left(\sum_{t=0}^{T} L_{\mathcal{G}}^{(t)}\right)}{L}},
    \end{align*}
    with a sufficiently large absolute constant $C>0$, which implicates
    \begin{align*}
        \mathbb{P} \left( \norm{\sum_{t=0}^{T}\nabla\ell_t(\theta^{(t)})}_2
        \le C^{\prime\prime}\sqrt{\frac{(T{+}1)\,n^{2}p}{L}}
        \right)
        \geq 1-\mathcal O(n^{-10}),
    \end{align*}
    after absorbing numerical constants into $C''>0$, which proves the claim.

\end{proof}

\begin{corollary} \label{cor:gradient_T_true_1}
    Suppose event $\mathcal{A}_{1}^{(T)}$ holds, for any iteration $T \in \mathbb{N}$, the following event
    \begin{align}
        \mathcal{A}_{3}^{(T)} := \Bigg\lbrace \norm{\nabla \mathcal{L}^{(\leq T)}(\theta^{\ast})}_{2} \lesssim  \sqrt{\frac{ (T+1)n^{2}p}{L}} + np \sum_{t=0}^{T} \norm{\theta^{(t)}-\theta^\ast}_2   \Bigg\rbrace
    \end{align}
    occurs with probability exceeding $1 - \mathcal{O}(n^{-10})$.
\end{corollary}

\begin{proof}
    Observe that
    \begin{align*}
        \lambda_{T} = \frac{T+1}{n} = (T+1) \lambda_{0}
    \end{align*}
    and
    \begin{align} \label{eq:nabla_decomp}
        \nabla \mathcal{L}^{(\leq T)}(\theta^{\ast}) 
        &= 
        \sum_{t=0}^{T} \nabla l_t(\theta^{\ast}) 
        \\ &=
        \sum_{t=0}^{T} \sum_{(i,j) \in \mathcal{E}_t,\, i>j} \left\{ - y_{j,i}^{(t)}+ \pi_{ij}(\theta^{\ast})\right\} (\boldsymbol{e}_{i} - \boldsymbol{e}_{j})
        \\ &= 
        N_{T} + B_{T},
    \end{align}
    where
    \begin{align*}
        N_{T} &:= \sum_{t=0}^{T} \sum_{(i,j) \in \mathcal{E}_t,\, i>j} \left\{ - y_{j,i}^{(t)}+ \pi_{ij}(\theta^{(t)})\right\} (\boldsymbol{e}_{i} - \boldsymbol{e}_{j}) = \sum_{t=0}^{T} \nabla \ell_{t}(\theta^{(t)}) , 
        \\
        B_{T} &:= \sum_{t=1}^{T} \sum_{(i,j) \in \mathcal{E}_t,\, i>j} \left\{ \pi_{ij}(\theta^{\ast}) - \pi_{ij}(\theta^{(t)})\right\} (\boldsymbol{e}_{i} - \boldsymbol{e}_{j}), \quad \pi_{ij}(\theta^{\ast}) = \pi_{ij}(\theta^{(0)}).
    \end{align*}
    Lemma \ref{cor:sum_gradient_T} indicates that 
    \begin{align*}
        \norm{N_{T}}_{2} = \norm{\sum_{t=0}^{T} \nabla \ell_{t}(\theta^{(t)})}_{2} \lesssim \sqrt{\frac{(T+1) n^{2}p}{L}}.
    \end{align*}
    For each fixed $t$, by the mean-value theorem,
    \begin{align*}
        \pi_{ij}(\theta^\ast)-\pi_{ij}(\theta^{(t)}) = \sigma'(\xi_{ij}^{(t)})\big[(\theta^\ast_i-\theta^\ast_j)-(\theta^{(t)}_{i}-\theta^{(t)}_{j})\big] = \sigma'(\xi_{ij}^{(t)})\,( \theta^\ast-\theta^{(t)})^{\top} \left( \mathbf e_i-\mathbf e_j \right),
    \end{align*}
    for some $\xi_{ij}^{(t)}$ on the segment between $( \theta^{(t)}_{i}-\theta^{(t)}_{j} )$ and $( \theta^\ast_i-\theta^\ast_j )$. Then we have
    \begin{align*}
    B_T=\sum_{t=0}^T \sum_{(i,j)\in E_t,\ i>j}\sigma'(\xi_{ij}^{(t)})(\mathbf e_i-\mathbf e_j)(\mathbf e_i-\mathbf e_j)^\top \, (\theta^\ast-\theta^{(t)}).
    \end{align*}
    It  follows immediately from $0\le\sigma'(\cdot)\le \frac{1}{4}$ and $\norm{L_{\mathcal G}^{(t)}} \leq 2 d^{(t)}_{\operatorname{max}} \leq \frac{3np}{2}$ on the event $\mathcal{A}_{1}^{(T)}$ that
    \begin{align*}
        \norm{B_T}_2 \le \frac14\sum_{t=0}^T \norm{L_{\mathcal G}^{(t)}} \,\norm{\theta^{(t)}-\theta^\ast}_2 \leq \frac{3}{8} np \sum_{t=0}^{T} \norm{\theta^{(t)}-\theta^\ast}_2.
    \end{align*}
    This completes the proof.
\end{proof}
\subsection{Proof of Lemma~\ref{lemma:smallesteigen}}
\label{appendix:smallesteigen}
We provide the proof of Lemma~\ref{lemma:smallesteigen} by providing the lower bounds for $\lambda_{\operatorname{min}, \perp} \left( \nabla^{2} \mathcal{L}^{\leq (T)}(\theta) \right)$ and $\lambda_{\operatorname{min}, \perp} \left( L_{\mathcal{G}}^{(T)} \right)$, respectively.
\begin{lemma} \label{lemma:Hessian_eigenvalue_min2}
    For all $\theta \in \mathbb{R}^{n}$ such that $\norm{\theta - \theta^{\ast}}_{\infty} \leq C$ for some $C \geq 0$, we have
    \begin{align}
        \lambda_{\operatorname{min}, \perp} \left( \nabla^{2} \mathcal{L}^{\leq (T)}(\theta) \right) \geq \frac{1}{4 \kappa e^{2C}} \sum_{t=0}^{T} \lambda_{\operatorname{min}, \perp} \left( \boldsymbol{L}_{\mathcal{G}}^{(t)} \right)
    \end{align}
    for any iteration $T \in \mathbb{N}$.
\end{lemma}

\begin{proof}
    By Weyl's inequality, it suffices to prove that
    \begin{align*}
        \underset{1 \leq i, j \leq n}{\min} \frac{e^{\theta_i}e^{\theta_j}}{(e^{\theta_i} + e^{\theta_j})^{2}} \geq \frac{1}{4\kappa e^{2C}}
    \end{align*}
    for all $\theta \in \mathbb{R}^{n}$ obeying $\norm{\theta - \theta^{\ast}}_{\infty} \leq C$.

    For pair $(i,j)$, without loss of generality we assume $\theta_{i} < \theta_{j}$, then obtain
    \begin{align*}
        \frac{e^{\theta_i}e^{\theta_j}}{(e^{\theta_i} + e^{\theta_j})^{2}} = \frac{e^{\theta_i-\theta_j}}{(1 + e^{\theta_i-\theta_j})^{2}} = \frac{e^{-|\theta_i-\theta_j|}}{(1 + e^{-|\theta_i-\theta_j|})^{2}} \geq \frac{1}{4} e^{-|\theta_i-\theta_j|}.
    \end{align*}
    On the other hand, it holds that
    \begin{align*}
        |\theta_i-\theta_j| 
        &\leq |\theta_{i}-\theta^{\ast}_{i}| + |\theta^{\ast}_{i}-\theta^{\ast}_{j}| + |\theta^{\ast}_{j}-\theta_{j}|
        \\ &\leq \theta^{\ast}_{\operatorname{max}} - \theta^{\ast}_{\operatorname{min}} + 2 \norm{\theta-\theta^{\ast}}_{\infty}
        \\ &\leq 
        \log \kappa + 2C.
    \end{align*}
    There, we obtain
    \begin{align*}
        e^{-|\theta_i-\theta_j|} \geq \frac{1}{4\kappa e^{2C}},
    \end{align*}
    which completes the proof.
\end{proof}

\begin{lemma} \label{lemma:L_eigenvalue_min2}
    Let $\mathcal{G}^{(T)} \overset{i.i.d.}{\sim} \mathcal{G}_{n,p}$ for any iteration $T \in \mathbb{N}$, and suppose that $p > \frac{c_{0} \log n}{n}$ for some sufficiently large constant $c_{0} > 0$. Then one has
    \begin{align}
         \mathbb{P} \left( \frac{3}{2}np \geq \lambda_{\operatorname{max}}\left( L_{\mathcal{G}}^{(T)} \right) \geq \lambda_{\operatorname{min}, \perp} \left( L_{\mathcal{G}}^{(T)} \right) \geq \frac{1}{2}np \right) \geq 1 - \mathcal{O} \left( n^{-10} \right).
    \end{align}
\end{lemma}

\begin{proof}
    Let $R^{(T)}$ be any $(n-1) \times n$ partial isometry with orthonormal rows such that the row space is $\boldsymbol{1}^{\perp} := \lbrace v \in \mathbb{R}^{n} \mid \boldsymbol{1}^{\top}v = 0 \rbrace$ at any iteration $T \in \mathbb{N}$. Then, it holds that
    \begin{enumerate}
        \item $R^{(T)} {R^{(T)}}^{\top} = \boldsymbol{I}_{n-1}, \quad R^{(T)} \, \boldsymbol{1} = 0;$
        \item $\norm{R^{(T)} v}_{2} = \norm{v}_{2}, \quad {R^{(T)}}^{\top}{R^{(T)}}v = v, \quad \forall v \in \boldsymbol{1}^{\perp};$
        \item $\norm{L^{(T)}_{\mathcal{G}} } = \norm{ R^{(T)}L^{(T)}_{\mathcal{G}}{R^{(T)}}^{\top} }, \quad \lambda_{\operatorname{min}, \perp} \left( L_{\mathcal{G}}^{(T)} \right) = \lambda_{\operatorname{min}} \left( R^{(T)}L^{(T)}_{\mathcal{G}}{R^{(T)}}^{\top} \right);$
        \item $\mathbb{E} \Big\lbrack L^{(T)}_{\mathcal{G}} \Big\rbrack = p \sum_{1 \leq i < j \leq n} (\boldsymbol{e}_{i} - \boldsymbol{e}_{j}) (\boldsymbol{e}_{i} - \boldsymbol{e}_{j})^{\top} = p \left( n\boldsymbol{I}_{n} - \boldsymbol{1} \boldsymbol{1}^{\top} \right).$
    \end{enumerate}

    Let $E^{(T)}_{i,j} := R^{(T)}(\boldsymbol{e}_{i} - \boldsymbol{e}_{j}) (\boldsymbol{e}_{i} - \boldsymbol{e}_{j})^{\top} {R^{(T)}}^{\top}$ for $(i,j) \in \mathcal{E}_{T}$ and $i > j$. Then we have 
    \begin{enumerate}
        \item $R^{(T)}L^{(T)}_{\mathcal{G}}{R^{(T)}}^{\top} = \sum_{(i,j) \in \mathcal{E}_{T},i > j} E^{(T)} _{i,j};$
        \item $E^{(T)} _{i,j} \succeq 0, \quad \norm{ E^{(T)} _{i,j}} = \norm{(\boldsymbol{e}_{i} - \boldsymbol{e}_{j}) (\boldsymbol{e}_{i} - \boldsymbol{e}_{j})^{\top}} = 2;$
        \item $\lambda_{\operatorname{min}} \left( \mathbb{E} \Big\lbrack R^{(T)}L^{(T)}_{\mathcal{G}}{R^{(T)}}^{\top} \Big\rbrack \right) = \lambda_{\operatorname{min}} \left( p R^{T} \left( n\boldsymbol{I}_{n} - \boldsymbol{1} \boldsymbol{1}^{\top} \right) {R^{(T)}}^{\top} \right) = \lambda_{\operatorname{min}} \left( pn \boldsymbol{I}_{n-1} \right) = pn.$
    \end{enumerate}
    By the matrix Chernoff inequality (Tropp, 2015), we obtain, for $t \in (0,1)$,
    \begin{align*}
        \mathbb{P} \left( \lambda_{\operatorname{min}} \left( R^{(T)}L^{(T)}_{\mathcal{G}}{R^{(T)}}^{\top} \right) \leq t \cdot \lambda_{\operatorname{min}} \left( \mathbb{E} \Big\lbrack R^{(T)}L^{(T)}_{\mathcal{G}}{R^{(T)}}^{\top} \Big\rbrack \right) \right) \leq (n-1) \Big\lbrack \frac{e^{t-1}}{t^{t}} \Big\rbrack^{\lambda_{\operatorname{min}} \left( \mathbb{E} \Big\lbrack R^{(T)}L^{(T)}_{\mathcal{G}}{R^{(T)}}^{\top} \Big\rbrack \right)/2}.
    \end{align*}
    Let $t = \frac{1}{2}$, we have
    \begin{align*}
        \mathbb{P} \left( \lambda_{\operatorname{min}, \perp} \left( L_{\mathcal{G}}^{(T)} \right) \leq \frac{1}{2}pn \right) \leq (n-1) \Big\lbrack \frac{e^{-\frac{1}{2}}}{\frac{1}{2}^{\frac{1}{2}}} \Big\rbrack^{\frac{pn}{2}}.
    \end{align*}
    Similarly,
    \begin{align*}
        \mathbb{P} \left( \lambda_{\operatorname{max}} \left( L_{\mathcal{G}}^{(T)} \right) \geq \frac{3}{2}pn \right) \leq (n-1) \Bigl[ \frac{e^{\frac{1}{2}}}{\frac{3}{2}^{\frac{3}{2}}} \Bigr]^{\frac{pn}{2}}.
    \end{align*}
    As a result, if $p > \frac{c_{0} \log n}{n}$ for some $c_{0} > 0$, we have 
    \begin{align*}
        \mathbb{P} \left( \frac{3}{2}pn \geq \lambda_{\operatorname{max}}\left( L_{\mathcal{G}}^{(T)} \right) \geq \lambda_{\operatorname{min}, \perp} \left( L_{\mathcal{G}}^{(T)} \right) \geq \frac{1}{2}pn \right) \geq  1 - \mathcal{O} \left( n^{-10} \right).
    \end{align*}
    This completes the proof.
\end{proof}

By combining Lemma~\ref{lemma:Hessian_eigenvalue_min2} with Lemma~\ref{lemma:L_eigenvalue_min2}, we reach the conclusion of Lemma~\ref{lemma:smallesteigen}.
\section{Proof of Inference Results in Appendix~\ref{appendix:proof_outline}}
\subsection{Proof Outline of Lemma~\ref{lemma:equation_D_plus} in Step~I}
\label{appendix:equation_D_plus}
Our first step is to establish the following lemma that $\norm{\theta^{(T+1)} - \theta^{\ast}}_{\infty}$ is bounded by a constant with high probability even though the MLE has no constraint or regularization. The lemma also shows that the unregularized estimator $\theta^{(T+1)}$ is sufficiently close to its regularized counterpart $\theta^{(T+1)}_{\lambda}$.

Leveraging the leave-one-out sequences, our want to establish that, for any fixed $T$, the iterates $\{\theta^{(\tau)}\}_{\tau=0,1,\cdots}$ remain close to $\theta^{\ast}$ in the $\ell_{\infty}$-norm. Suppose the auxiliary induction hypotheses \eqref{eq:A}--\eqref{eq:C} hold for $0,\cdots, \tau$-th step and the primary induction hypotheses \eqref{eq:D}--\eqref{eq:H} hold up to iteration $T$, then with probability exceeding $1 - \mathcal{O}(n^{-10})$, we verify the induction hypotheses \eqref{eq:A}--\eqref{eq:C} at $(\tau+1)$-step one by one.
\subsubsection[Proof of Auxiliary Induction Hypothesis (A)]{Proof of Auxiliary Induction Hypothesis \eqref{eq:A}}
\label{appendix:equation_A}
We first establish that the gradient of the loss function $\ell_{t}(\cdot)$ does not deviate substantially, in the $\ell_2$-norm, from the gradient of its leave-one-out counterpart.
\begin{lemma} \label{lemma:gradient_t_diff_LOO}
    The following event holds
    \begin{align*}
        \mathcal A_{4}^{(T)}
        := \left\{ 
        \norm{\nabla \ell_{t}(\theta^{(\tau, m)}) - \nabla \ell_{t}^{(m)}(\theta^{(\tau, m)})}_{2} 
        \lesssim 
        \sqrt{\frac{np \log n}{L}} + \sqrt{np \log n} \norm{\theta^{(\tau, m)} - \theta^{(t)}}_{\infty} \right \}
    \end{align*}
    with probability at least $1-\mathcal O(n^{-10})$.
\end{lemma}
\begin{proof}
See \S\ref{appendix:gradient_t_diff_LOO}.
\end{proof}
Once Lemma~\ref{lemma:gradient_t_diff_LOO} is established, we prove the auxiliary induction hypothesis \eqref{eq:A} as follows.

\begin{proof}
For any $m \in [n]$, by definition we know that
\begin{align*}
    \theta^{(\tau+1)} - \theta^{(\tau + 1, m)} 
    &= 
    \theta^{(\tau)} - \eta \nabla \mathcal{L}_{\lambda_{T}}^{(\le T)}(\theta^{(\tau)}) - \Big\lbrack \theta^{(\tau,m)} - \eta \nabla \mathcal{L}_{\lambda_{T}}^{(m,\le T)}(\theta^{(\tau,m)}) \Big\rbrack 
    \\ &=
    \theta^{(\tau)} - \eta \nabla \mathcal{L}_{\lambda_{T}}^{(\le T)}(\theta^{(\tau)}) - \Big\lbrack \theta^{(\tau, m)} - \eta \nabla \mathcal{L}_{\lambda_{T}}^{(\le T)}(\theta^{(\tau, m)}) \Big\rbrack 
    \\ & \quad- 
    \eta \left( \nabla \mathcal{L}_{\lambda_{T}}^{(\le T)}(\theta^{(\tau, m)}) - \nabla \mathcal{L}_{\lambda_{T}}^{m,\le T)}(\theta^{(\tau,m)})  \right)
    \\ &=
    \Big\lbrack \boldsymbol{I}_{n} - \eta \nabla^{2} \mathcal{L}_{\lambda_{T}}^{(\le T)}(\theta^{(\tau)}(\xi)) \Big\rbrack \left( \theta^{(\tau)} - \theta^{(\tau,m)} \right) 
    \\ & \quad - 
    \eta \left( \nabla \mathcal{L}_{\lambda_{T}}^{(\le T)}(\theta^{(\tau, m)}) - \nabla \mathcal{L}_{\lambda_{T}}^{(m,\le T)}(\theta^{(\tau,m)})  \right)
    \\ &=
    \Big\lbrack (1-\eta\lambda_{T})\boldsymbol{I}_{n} - \eta \sum_{t=0}^{T} \nabla^{2} \ell_{t} ( \theta^{(\tau)}(\xi) ) \Big\rbrack \left( \theta^{(\tau)} - \theta^{(\tau, m)} \right) 
    \\ & \quad - 
    \eta \sum_{t=0}^{T} \left( \nabla \ell_{t}(\theta^{(\tau, m)}) - \nabla \ell_{t}^{(m)}(\theta^{(\tau, m)}) \right),
\end{align*}
where $\theta^{(\tau)}(\xi)$ is a convex combination of $\theta^{(\tau)}$ and $\theta^{(\tau, m)}$. By \eqref{eq:I} and \eqref{eq:J}, we thus have $\norm{\theta^{(\tau)}(\xi) - \theta^{\ast}}_{\infty} \leq 3$, and we can combine this with Lemma \ref{lemma:smallesteigen} to obtain the bound
\begin{align*}
    \norm{\Big\lbrack \boldsymbol{I}_{n} - \eta \nabla^{2} \mathcal{L}_{\lambda_{T}}^{(\le T)}(\theta^{(\tau)}(\xi))  \Big\rbrack \left( \theta^{(\tau)} - \theta^{(\tau, m)} \right) }_{2}  
    & \leq \norm{\Big\lbrack \boldsymbol{I}_{n} - c_{1}(T+1) \eta np \Big\rbrack \left( \theta^{(\tau)} - \theta^{(\tau, m)} \right)}_{2}
    \\ &\leq 
    \left( 1 - c_{1}(T+1) \eta np \right) \cdot \norm{\theta^{(\tau)} - \theta^{(\tau, m)}}_{2}
\end{align*}
for some constant $c_{1}>0$. 

Note that by Lemma \ref{lemma:gradient_t_diff_LOO} and \eqref{eq:L}, for $t>0$, we have 
\begin{align*}
   \norm{\nabla \ell_{t}(\theta^{(\tau, m)}) - \nabla \ell_{t}^{(m)}(\theta^{(\tau, m)})}_{2} 
   & \leq 
   C_{1}\sqrt{\frac{np \log n}{L}} + C_{1}\sqrt{np \log n} \norm{\theta^{(\tau, m)} - \theta^{(t)}}_{\infty}
   \\ & \leq
   C_{1}\sqrt{\frac{np \log n}{L}} + C_{1}\sqrt{np \log n} \left( 3 + D_{5} \sqrt{\frac{\log n}{npL}} \right)
   \\ & = 
   C_{1} \sqrt{np\log n} \left( \sqrt{\frac{1}{L}} + 3 + D_{5} \sqrt{\frac{\log n}{npL}} \right), \quad \forall t = 1, \cdots, T
\end{align*}
and
\begin{align*}
    \norm{\nabla \ell_{0}(\theta^{(\tau, m)}) - \nabla \ell_{0}^{(m)}(\theta^{(\tau, m)})}_{2} 
    & \leq 
    C_{1}\sqrt{\frac{np \log n}{L}} + C_{1}\sqrt{np \log n} \norm{\theta^{(\tau, m)} - \theta^{\ast}}_{\infty}
    \\ & \leq
     C_{1} \sqrt{np\log n} \left( \sqrt{\frac{1}{L}} + 3 \right).
\end{align*}
Telescopic sum and Cauchy-Schwarz inequality further yield
\begin{align*}
    \norm{\nabla \mathcal{L}_{\lambda_{T}}^{(\le T)}(\theta^{(\tau, m)}) - \nabla \mathcal{L}_{\lambda_{T}}^{(m,\le T)}(\theta^{(\tau,m)})}_{2}
    & = 
    \norm{\sum_{t=0}^{T} \left( \nabla \ell_{t}(\theta^{(\tau, m)}) - \nabla \ell_{t}^{(m)}(\theta^{(\tau, m)}) \right)}_{2}
    \\ & \leq
    \sum_{t=0}^{T} \norm{\nabla \ell_{t}(\theta^{(\tau, m)}) - \nabla \ell_{t}^{(m)}(\theta^{(\tau, m)})}_{2} 
    \\ & \leq 
    C_{1} (T+1) \sqrt{np\log n} \left( \sqrt{\frac{1}{L}} + 3 + D_{5} \sqrt{\frac{\log n}{npL}} \right).
\end{align*}
Combine the above results with \eqref{eq:A} together yields 
\begin{align*}
   \norm{\theta^{(\tau+1)} - \theta^{(\tau + 1, m)}}_{2} 
   & \leq
   \left( 1 - c_{1}(T+1) \eta np \right) \cdot \norm{\theta^{(\tau)} - \tau^{(\tau, m)}}_{2} 
   \\ & \quad + 
   \eta C_{1} (T+1) \sqrt{np\log n} \left( \sqrt{\frac{1}{L}} + 3 + D_{5} \sqrt{\frac{\log n}{npL}} \right)
   \\ & \leq \left( 1 - c_{1} \eta_{\ast} np \right) + \eta_{\ast} C_{1} \sqrt{np\log n} \left( \sqrt{\frac{1}{L}} + 3 + D_{5} \sqrt{\frac{\log n}{npL}} \right)
   \\ & \leq 1,
\end{align*}
where the last inequality requires that $\sqrt{np\log n} \left( \sqrt{\frac{1}{L}} + 3 + D_{5} \sqrt{\frac{\log n}{npL}} \right) \leq c_{1} np$, which is implied by the condition that $p \geq \frac{c_{0} \log n}{n}$ for some sufficiently large $c_{0} > 0$. We thus have proved \eqref{eq:A} for $\tau + 1$.
\end{proof}
\subsubsection[Proof of Auxiliary Induction Hypothesis (B)]{Proof of Auxiliary Induction Hypothesis \eqref{eq:B}}
\label{appendix:equation_B}
\begin{proof}
By definition we have
\begin{align*}
    \theta^{(\tau+1)} - \theta^{\ast} 
    &= 
    \theta^{(\tau)} - \eta \nabla \mathcal{L}_{\lambda_{T}}^{(\le T)}(\theta^{(\tau)}) - \theta^{\ast}
    \\ &=
    \theta^{(\tau)} - \eta \nabla \mathcal{L}_{\lambda_{T}}^{(\le T)}(\theta^{(\tau)}) - \Big\lbrack \theta^{\ast} - \eta \nabla \mathcal{L}_{\lambda_{T}}^{(\le T)}(\theta^{\ast}) \Big\rbrack - \eta \nabla \mathcal{L}_{\lambda_{T}}^{(\le T)}(\theta^{\ast})
    \\ &=
    \Big\lbrack \boldsymbol{I}_{n} - \eta \nabla^{2} \mathcal{L}_{\lambda_{T}}^{(\le T)}(\theta^{(\tau)}(\zeta)) \Big\rbrack \left( \theta^{(\tau)} - \theta^{\ast} \right) - \eta \nabla \mathcal{L}_{\lambda_{T}}^{(\le T)}(\theta^{\ast}),
\end{align*}
where $\theta^{(\tau)}(\zeta)$ is a convex combination of $\theta^{(\tau)}$ and $\theta^{\ast}$. Similarly we use Lemma \ref{lemma:smallesteigen} to obtain the bound by \eqref{eq:I} that 
\begin{align*}
    \norm{\Big\lbrack \boldsymbol{I}_{n} - \eta \nabla^{2} \mathcal{L}_{\lambda_{T}}^{(\le T)}(\theta^{(\tau)}(\zeta))  \Big\rbrack \left( \theta^{(\tau)} - \theta^{\ast} \right) }_{2}  
    \leq
    \left( 1 - c_{2}(T+1) \eta np \right) \cdot \norm{\theta^{(\tau)} - \theta^{\ast}}_{2}
\end{align*}
for some constant $c_{2}>0$.

It follows immediately from Corollary \ref{cor:gradient_T_true_1} that there exists some constants $c_{3},c_{4}>0$ such that
\begin{align*}
    \norm{\nabla \mathcal{L}^{(\leq T)}_{\lambda_{T}}(\theta^{\ast})}_{2} 
    & \leq
    \norm{\nabla \mathcal{L}^{(\leq T)}(\theta^{\ast})}_{2} + \lambda_{T} \norm{\theta^{\ast}}_{2}
    \\ & \leq 
    c_{3} \sqrt{\frac{ (T+1)n^{2}p}{L}} + c_{4} np \sum_{t=0}^{T} \norm{\theta^{(t)}-\theta^\ast}_{2} + \lambda_{T} \norm{\theta^{\ast}}_{2}
\end{align*}
Combine the above together with \eqref{eq:B} and \eqref{eq:G} indicates that 
\begin{align*}
    \norm{\theta^{(\tau+1)}-\theta^{\ast}}_{2} 
    & \leq 
    \left( 1 - c_{2}(T+1) \eta np \right) \cdot \norm{\theta^{(\tau)} - \theta^{\ast}}_{2}
    + 
    c_{3} \eta \sqrt{\frac{ (T+1)n^{2}p}{L}}
    \\ & \quad + 
    c_{4} \eta np \sum_{t=0}^{T} \norm{\theta^{(t)}-\theta^\ast}_{2} + \eta \lambda_{T} \norm{\theta^{\ast}}_{2}
    \\ & \leq 
    \left( 1 - c_{2} \eta_{\ast} np \right) \norm{\theta^{(\tau)} - \theta^{\ast}}_{2} + c_{3} \eta_{\ast} \sqrt{\frac{n^{2}p}{ (T+1)L}}
    \\ & \quad + 
    c_{4} \eta_{\ast} np \norm{\theta^{(t)}-\theta^\ast}_{2} + \eta_{\ast} \lambda_{0} \norm{\theta^{\ast}}_{2}
    \\ & \leq
    \left( 1 - c_{2} \eta_{\ast} np \right) \sqrt{\frac{n}{\log n}} + c_{3} \eta_{\ast} \sqrt{\frac{n^{2}p}{ (T+1)L}}
    \\ & \quad +
    c_{4} \eta_{\ast} np D_{4} \sqrt{\frac{1}{pL}} + \eta_{\ast} \frac{\log \kappa}{\sqrt{n}}
    \\ & \leq
    \sqrt{\frac{n}{\log n}},
\end{align*}
where the last inequality is due to $\sqrt{\frac{n^{2}p}{(T+1)L}} + np \sqrt{\frac{1}{pL}} + \frac{\log \kappa}{\sqrt{n}} \leq c_{2}np\sqrt{\frac{n}{\log n}}$ by the choice of $\eta$ and $\lambda_{T}$. Hence, \eqref{eq:B} holds for $\tau + 1$.
\end{proof}
\subsubsection[Proof of Auxiliary Induction Hypothesis (C)]{Proof of Auxiliary Induction Hypothesis \eqref{eq:C}}
\label{appendix:equation_C}
\begin{proof}
For any $m \in [n]$, we observe that
\begin{align*}
    \theta_{m}^{(\tau+1,m)} - \theta_{m}^{\ast} 
    &= 
    \theta_{m}^{(\tau, m)} - \eta \left[ \nabla \mathcal{L}_{\lambda_{T}}^{(m, \le T)}(\theta^{(\tau,m)}) \right]_{m} - \theta_{m}^{\ast}
    \\ &=
    \theta_{m}^{(\tau, m)} - \theta_{m}^{\ast} - \eta \left[ \sum_{t=0}^{T} \sum_{i \in [n] \backslash \lbrace m \rbrace} p \left\{ \frac{e^{\theta^{(t)}_{i}}}{e^{\theta_{i}^{(t)}} + e^{\theta_{m}^{(t)}}} - \frac{e^{\theta_i^{(\tau, m)}}}{e^{\theta_i^{(\tau, m)}} + e^{\theta_{m}^{(\tau, m)}}} \right\}  \right] - \eta \lambda_{T} \theta^{(\tau, m)}_{m},
\end{align*}
where the last line follows by the construction of $\mathcal{L}_{\lambda_{T}}^{(m, \leq T)}$. Apply the mean value theorem to obtain
\begin{align*}
    \frac{e^{\theta^{(t)}}_{i}}{e^{\theta_{i}^{(t)}} + e^{\theta_{m}^{(t)}}} - \frac{e^{\theta_{i}^{(\tau, m)}}}{e^{\theta_i^{(\tau, m)}} + e^{\theta_{m}^{(\tau, m)}}} 
    &= 
    \frac{1}{1 + e^{\theta^{(t)}_{m}- \theta^{(t)}_{i}}} - \frac{1}{1 + e^{\theta_{m}^{(\tau, m)}-\theta_{i}^{(\tau, m)}}}
    \\ &=
    -\frac{e^{\xi_{i}^{(t)}}}{(1+e^{\xi_{i}^{(t)}})^{2}} \left[ \theta^{(t)}_{m} - \theta^{(t)}_{i} - \left( \theta_{m}^{(\tau, m)}-\theta_{i}^{(\tau, m)} \right) \right],
\end{align*}
where $\xi_{i}^{(t)}$ is some real number lying between $\theta^{(t)}_{m}- \theta^{(t)}_{i}$ and $\theta_{m}^{(\tau, m)}-\theta_{i}^{(\tau, m)}$. 

Rearranging the above results yields
\begin{align*}
    \theta_{m}^{(\tau+1,m)} - \theta_{m}^{\ast} 
    &=
    \theta_{m}^{(\tau, m)} - \theta_{m}^{\ast} + \eta p \left[ \sum_{t=0}^{T} \sum_{i \in [n] \backslash \lbrace m \rbrace} \frac{e^{\xi_{i}^{(t)}}}{(1+e^{\xi_{i}^{(t)}})^{2}} \left[ \theta^{(t)}_{m} - \theta^{(t)}_{i} - \left( \theta_{m}^{(\tau, m)}-\theta_{i}^{(\tau, m)} \right) \right]  \right] - \eta \lambda_{T} \theta^{(\tau, m)}_{m}
    \\ &=
    \left( 1 - \eta \lambda_{T} - \eta p \sum_{i \in [n] \backslash \{ m \}} \frac{e^{\xi_{i}^{(0)}}}{(1+e^{\xi_{i}^{(0)}})^{2}}  \right) \left( \theta_{m}^{(\tau, m)} - \theta_{m}^{\ast} \right) - \eta \lambda_{T} \theta^{\ast}_{m}
    \\ & \quad + 
    \eta p \sum_{i \in [n] \backslash \{ m \}} \frac{e^{\xi_{i}^{(0)}}}{(1+e^{\xi_{i}^{(0)}})^{2}} \left( \theta_{i}^{(\tau, m)} - \theta_{i}^{\ast} \right) 
    \\ & \quad 
    - 
    \eta p  \sum_{t=1}^{T} \sum_{i \in [n] \backslash \lbrace m \rbrace} \frac{e^{\xi_{i}^{(t)}}}{(1+e^{\xi_{i}^{(t)}})^{2}} \left( \theta_{m}^{(\tau, m)}-\theta^{(t)}_{m} \right) 
    \\ & \quad + 
    \eta p \sum_{t=1}^{T} \sum_{i \in [n] \backslash \lbrace m \rbrace} \frac{e^{\xi_{i}^{(t)}}}{(1+e^{\xi_{i}^{(t)}})^{2}} \left( \theta_{i}^{(\tau, m)} - \theta^{(t)}_{i}  \right) 
    \\ &=
    \left( 1 - \eta \lambda_{T} - \eta p \sum_{i \in [n] \backslash \{ m \}} \frac{e^{\xi_{i}^{(0)}}}{(1+e^{\xi_{i}^{(0)}})^{2}}  \right) \left( \theta_{m}^{(\tau, m)} - \theta_{m}^{\ast} \right) - \eta \lambda_{T} \theta^{\ast}_{m}
    \\ & \quad + 
    \eta p \sum_{i \in [n] \backslash \{ m \}} \frac{e^{\xi_{i}^{(0)}}}{(1+e^{\xi_{i}^{(0)}})^{2}} \left( \theta_{i}^{(\tau, m)} - \theta_{i}^{\ast} \right) 
    \\ & \quad - 
    \eta p \sum_{t=1}^{T} \sum_{i \in [n] \backslash \lbrace m \rbrace} \frac{e^{\xi_{i}^{(t)}}}{(1+e^{\xi_{i}^{(t)}})^{2}} \left( \theta_{m}^{(\tau, m)} - \theta^{\ast}_{m} \right) 
    \\ & \quad + 
    \eta p \sum_{t=1}^{T} \sum_{i \in [n] \backslash \lbrace m \rbrace} \frac{e^{\xi_{i}^{(t)}}}{(1+e^{\xi_{i}^{(t)}})^{2}}   \left( \theta^{(t)}_{m} - \theta^{\ast}_{m} \right)
    \\ & \quad + 
    \eta p \sum_{t=1}^{T} \sum_{i \in [n] \backslash \lbrace m \rbrace} \frac{e^{\xi_{i}^{(t)}}}{(1+e^{\xi_{i}^{(t)}})^{2}} \left( \theta_{i}^{(\tau, m)} - \theta^{(t)}_{i} \right)
    \\ &=
    \left( 1 - \eta \lambda_{T} - \eta p \sum_{t=0}^{T} \sum_{i \in [n] \backslash \{ m \}} \frac{e^{\xi_{i}^{(t)}}}{(1+e^{\xi_{i}^{(t)}})^{2}}  \right) \left( \theta_{m}^{(\tau, m)} - \theta_{m}^{\ast} \right) - \eta \lambda_{T} \theta^{\ast}_{m}
    \\ & \quad +
    \eta p \sum_{t=0}^{T} \sum_{i \in [n] \backslash \lbrace m \rbrace} \frac{e^{\xi_{i}^{(t)}}}{(1+e^{\xi_{i}^{(t)}})^{2}}   \left( \theta^{(t)}_{m} -\theta^{\ast}_{m} \right)
    \\ & \quad + 
    \eta p \sum_{t=0}^{T} \sum_{i \in [n] \backslash \lbrace m \rbrace} \frac{e^{\xi_{i}^{(t)}}}{(1+e^{\xi_{i}^{(t)}})^{2}} \left( \theta_{i}^{(\tau, m)} - \theta^{(t)}_{i} \right),
\end{align*}
and hence
\begin{align*}
    \left| \theta_{m}^{(\tau+1,m)} - \theta_{m}^{\ast} \right| 
    & \leq 
    \left| 1 - \eta \lambda_{T} - \eta p \sum_{t=0}^{T} \sum_{i \in [n] \backslash \{ m \}} \frac{e^{\xi_{i}^{(t)}}}{(1+e^{\xi_{i}^{(t)}})^{2}}  \right| \cdot \left| \theta_{m}^{(\tau, m)} - \theta_{m}^{\ast} \right| + \eta \lambda_{T} \left| \theta_{m}^{\ast} \right|
    \\ \quad & +
    \frac{\eta p}{4} \sum_{t=0}^{T} \sum_{i \in [n] \backslash \lbrace m \rbrace} \left| \theta_{m}^{\ast} - \theta_{m}^{(t)} \right|
    +
    \frac{\eta p}{4} \sum_{t=0}^{T} \sum_{i \in [n] \backslash \lbrace m \rbrace} \left| \theta_{i}^{(\tau, m)} - \theta_{i}^{(t)} \right|
    \\ & \leq
    \left| 1 - \eta \lambda_{T} - \eta p \sum_{t=0}^{T} \sum_{i \in [n] \backslash \{ m \}} \frac{e^{\xi_{i}^{(t)}}}{(1+e^{\xi_{i}^{(t)}})^{2}}  \right| \cdot \left| \theta_{m}^{(\tau, m)} - \theta_{m}^{\ast} \right| + \eta \lambda_{T} \norm{\theta^{\ast}}_{\infty}
    \\ \quad & +
    \frac{\eta p}{4} n \sum_{t=0}^{T} \norm{\theta^{\ast} - \theta^{(t)}}_{\infty} + \frac{\eta p}{4} \sqrt{n} \sum_{t=0}^{T} \norm{\theta^{(\tau, m)} - \theta^{(t)}}_{2}.
\end{align*} 
Here, the first inequality comes from the triangle inequality and elementary inequality $\frac{e^{\xi_{i}^{(t)}}}{(1+\xi_{i}^{(t)})^{2}} \leq \frac{1}{4}$ for any $\xi_{i}^{(t)} \in \mathbb{R}$, whereas the second relation holds owing to the Cauchy-Schwarz inequality, namely 
\begin{align*}
    \sum_{t=0}^{T} \sum_{i \in [n] \backslash \lbrace m \rbrace} \left| \theta_{m}^{\ast} - \theta_{m}^{(t)} \right|
    & \leq
    n \sum_{t=0}^{T} \norm{\theta^{\ast} - \theta^{t}}_{\infty},
    \\
    \sum_{t=0}^{T} \sum_{i \in [n] \backslash \lbrace m \rbrace} \left| \theta_{i}^{(\tau, m)} - \theta_{i}^{(t)} \right|
    & \leq
    \sqrt{n} \sum_{t=0}^{T} \norm{\theta^{(\tau, m)} - \theta^{(t)}}_{2},
\end{align*}
and we also obtain that
\begin{align*}
    1 - \eta \lambda_{T} - \eta p \sum_{t=0}^{T} \sum_{i \in [n] \backslash \{ m \}} \frac{e^{\xi_{i}^{(t)}}}{(1+e^{\xi_{i}^{(t)}})^{2}}  \geq 1 - \eta \lambda_{T} - \eta p (T+1)\frac{(n-1)}{4} \geq 1 - \eta \left\lbrack \lambda_{T} + (T+1)np \right\rbrack = 0.
\end{align*}
To further upper bound $\left| \theta_{m}^{\tau+1, m} - \theta_{m}^{\ast} \right|$, it suffices to obtain a lower bound on $\frac{e^{\xi_{i}^{(t)}}}{(1+\xi_{i}^{(t)})^{2}}$. Toward this, we can guarantee $\left|\xi_{i}^{(t)} - \theta^{(t)}_{m} + \theta^{(t)}_{i} \right| \leq\left| \theta^{(t)}_{m} - \theta_{m}^{(\tau, m)} \right| + \left| \theta^{(t)}_{i} - \theta_{i}^{(\tau, m)} \right| \leq 6 + 2 D_{5} \sqrt{\frac{\log n}{npL}}$ by \eqref{eq:L}, which implies $\norm{\xi_{i}^{t}}_{\infty}$ is bounded. We then have $\min_{0 \leq t \leq T} \sum_{i \in [n] \backslash \{ m \}} \frac{e^{\xi_{i}^{(t)}}}{(1+e^{\xi_{i}^{(t)}})^{2}} \geq c_{5} n$ for some constant $c_{5} > 0$. Combine the above bounds with \eqref{eq:H} and \eqref{eq:M} thus yields
\begin{align*}
    \left| \theta_{m}^{(\tau+1,m)} - \theta_{m}^{\ast} \right| 
    & \leq
    (1 - \eta \lambda_{T} - c_{5} \eta (T+1) np) \left| \theta_{m}^{(\tau, m)} - \theta_{m}^{\ast} \right| + \eta \lambda_{T} \norm{\theta^{\ast}}_{\infty}
    \\ & \quad +
    \frac{\eta p}{4} n \sum_{t=0}^{T} \norm{\theta^{\ast} - \theta^{(t)}}_{\infty} + \frac{\eta p}{4} \sqrt{n} \sum_{t=0}^{T} \norm{\theta^{(\tau, m)} - \theta^{(t)}}_{2}
    \\ & \leq 
    (1 - c_{5} \eta (T+1) np) \left| \theta_{m}^{(\tau, m)} - \theta_{m}^{\ast} \right| + \eta \lambda_{T} \norm{\theta^{\ast}}_{\infty}
    \\ & \quad +
    \eta p n \sum_{t=0}^{T} \norm{\theta^{\ast} - \theta^{(t)}}_{\infty} + \eta p \sqrt{n} \sum_{t=0}^{T} \norm{\theta^{(\tau, m)} - \theta^{(t)}}_{2}
    \\ & \leq 
    (1 - c_{5} \eta_{\ast} np) + \eta_{\ast} \frac{\log \kappa}{n}
    \\ & \quad +
    \eta_{\ast} p n \left( D_{5} \sqrt{\frac{\log n}{npL}} \right)
    +
    \eta_{\ast} p \sqrt{n} \left( 1 + \sqrt{\frac{n}{\log n}} + D_{4} \sqrt{\frac{1}{pL}} \right)
    \\ & \leq
    1,
\end{align*}
where the last inequality is because of $\eta_{\ast} p n \left( D_{5} \sqrt{\frac{\log n}{npL}} \right) + \eta_{\ast} p \sqrt{n} \left( 1 + \sqrt{\frac{n}{\log n}} + D_{4} \sqrt{\frac{1}{pL}} \right) + \eta\lambda_{T}\norm{\theta^{\ast}}_{\infty} \leq c_{5}\eta_{\ast}np$. We thus have proved \eqref{eq:C} for $\tau+1$.
\end{proof}
\subsubsection{Proof of Lemma~\ref{lemma:equation_D_plus}} \label{appendix:equation_D_plus_plus}
To summarize, we have shown that \eqref{eq:A}-\eqref{eq:C} hold for all $\tau \leq \tau^{\ast}$ with probability at least $1 - \mathcal{O}(\tau^{\ast} n^{-10})$. Moreover, \eqref{eq:I} holds for all $\tau \leq \tau^{\ast}$, and hence $\norm{\theta^{(\tau^{\ast})} - \theta^{\ast}}_{\infty} \leq 2$. Invoking the convergence properties of gradient descent for strongly convex objectives \citep{bubeck2015convex}, we first obtain the following lemma, which provides a finite $\ell_{\infty}$ statistical error bound on $\theta^{(t+1)}_{\lambda}$ around $\theta^{\ast}$.
\begin{lemma} \label{lemma:inner_reg_bound}
    Suppose the auxiliary induction hypotheses \eqref{eq:A}-\eqref{eq:C} hold for $0,\cdots, \tau^{\ast}$-th step and the primary induction hypotheses \eqref{eq:D}-\eqref{eq:H} hold up to iteration $T$, then for any $t = 0, \cdots, T$, we have
    \begin{align*}
        \norm{\theta^{(t+1)}_{\lambda} - \theta^{\ast}}_{\infty} \leq 4
    \end{align*}
    with probability at least $1 - \mathcal{O}(n^{-7})$.
\end{lemma}
\begin{proof}
See \S\ref{appendix:inner_reg_bound}.
\end{proof}

Next, we establish a finite $\ell_{\infty}$-norm bound for $\norm{\theta^{(t+1)} - \theta^{\ast}}_{\infty}$, which is a key step for controlling the spectrum of the Hessian matrix $\nabla^{2} \mathcal{L}^{(\leq T)}(\cdot)$ when exploring the statistical properties of the MLE. In addition, we show that the two estimators $\theta^{(t+1)}_{\lambda}$ and $\theta^{(t+1)}$ are sufficiently close.
\begin{lemma} \label{lemma:inner_unreg_bound}
    Suppose the auxiliary induction hypotheses \eqref{eq:A}-\eqref{eq:C} hold for $0,\cdots, \tau^{\ast}$-th step and the primary induction hypotheses \eqref{eq:D}-\eqref{eq:H} hold up to iteration $T$, then for any $t = 0, \cdots, T$, we have
    \begin{align*}
        \norm{\theta^{(t+1)} - \theta^{\ast}}_{\infty} \leq 5
    \end{align*}
    with probability at least $1 - \mathcal{O}(n^{-7})$.
\end{lemma}
\begin{proof}
See \S\ref{appendix:inner_unreg_bound}.
\end{proof}
\subsection{Proof Outline of Lemma~\ref{lemma:equation_EFG} in Step~II}
\label{appendix:equation_EFG}
Our second step is to show that $\theta^{(T+1)}$ does not deviate substantially from $\theta^{(T)}$ in either the $\ell_2$- or the $\ell_\infty$-norm, which then allows us to derive the optimal rate for $\norm{\theta^{(T+1)}-\theta^\ast}_2$.

Analogous to Step~I, we introduce an iteration-restricted leave-one-out construction and use it to show that, for any fixed $T$, the iterates $\{\theta^{(\upsilon)}\}_{\upsilon=0,1,\ldots}$ remain close to $\theta^{(T)}$ in the $\ell_\infty$-norm. Assuming that the tertiary induction hypotheses \eqref{eq:A1}--\eqref{eq:D1} hold up to $\upsilon$-step, and that the primary induction hypotheses \eqref{eq:D}--\eqref{eq:H} hold at iteration $T$, we then verify the tertiary induction hypotheses at$(\upsilon+1)$-step, one by one, with high probability.

\subsubsection[Proof of Tertiary Induction Hypothesis (A1)]{Proof of Tertiary Induction Hypothesis \eqref{eq:A1}}
\label{appendix:equation_A1}
\begin{proof}
In view of the gradient update rule, by the fundamental theorem of calculus, we have
\begin{align*}
    \theta_{0}^{(\upsilon+1)} - \theta^{(T)}
    & =
    \theta_{0}^{(\upsilon)} - \eta \nabla \mathcal{L}_{\lambda_{0}}^{(\leq T)} ( \theta_{0}^{(\upsilon)}) - \theta^{(T)}
    \\ &=
    \theta_{0}^{(\upsilon)} - \eta \nabla \mathcal{L}_{\lambda_{0}}^{(\leq T)} ( \theta_{0}^{(\upsilon)}) - \left[ \theta^{(T)}  - \eta \nabla \mathcal{L}_{\lambda_{0}}^{(\leq T)} ( \theta^{(T)}) \right] - \eta \nabla \mathcal{L}_{\lambda_{0}}^{(\leq T)} ( \theta^{(T)})
    \\ & =
    \left\{ \boldsymbol{I_{n}} - \eta \int_{0}^{1} \nabla^{2} \mathcal{L}_{\lambda_{0}}^{(\leq T)}(\theta_{0}^{(\upsilon)}(\iota)) d \iota \right\} (\theta_{0}^{(\upsilon)}-\theta^{(T)}) - \eta \nabla \mathcal{L}_{\lambda_{0}}^{(\leq T)} ( \theta^{(T)}),
\end{align*}
where we denote $\theta_{0}^{(\upsilon)}(\iota) := \theta^{(T)} + \iota (\theta_{0}^{(\upsilon)}-\theta^{(T)})$ for $\iota \in [0,1]$. Combining the induction hypothesis \eqref{eq:H} and \eqref{eq:D1}, one can see that for all $0 \leq \iota \leq 1$,
\begin{align*}
    \norm{\theta_{0}^{(\upsilon)}(\iota) - \theta^{\ast}}_{\infty}
    & \leq
    \norm{\theta_{0}^{(\upsilon)}(\iota) - \theta^{(T)}}_{\infty} + \norm{\theta^{(T)}-\theta^{\ast}}_{\infty}
    \\ & \leq
    \norm{\theta_{0}^{(\upsilon)}- \theta^{(T)}}_{\infty} + \norm{\theta^{(T)}-\theta^{\ast}}_{\infty}
    \\ & \leq
    D_{8} \frac{1}{T+1} \sqrt{\frac{\log n}{npL}} + D_{5} \sqrt{\frac{\log n}{npL}} \leq \epsilon
\end{align*}
for any sufficiently small $\epsilon > 0$.

This together with Lemma \ref{lemma:largesteigen} and Lemma \ref{lemma:smallesteigen} reveals that for any $0 \leq \iota \leq 1$,
\begin{align*}
    \frac{T+1}{10 \kappa}np + \lambda_{0} \leq \frac{(T+1)np}{8 \kappa e^{2\epsilon}} + \lambda_{0} \leq \lambda_{\min,\perp} \left( \nabla^{2} \mathcal{L}_{\lambda_{0}}^{(\leq T)}(\theta_{0}^{(\upsilon)}(\iota)) \right) \leq \lambda_{\max} \left( \nabla^{2} \mathcal{L}_{\lambda_{0}}^{(\leq T)}(\theta_{0}^{(\upsilon)}(\iota)) \right) \leq \lambda_{0} + (T+1)np,
\end{align*}
where the first inequality holds as long as $\epsilon > 0$ is small enough. Denoting $\boldsymbol{A}^{(\upsilon)} = \int_{0}^{1} \nabla^{2} \mathcal{L}_{\lambda_{0}}^{(\leq T)}(\theta_{0}^{(\upsilon)}(\iota)) d \iota$ and using the triangle inequality with $\boldsymbol{1}_{n}^{\top} (\theta_{0}^{(\upsilon)}-\theta^{(T)})=0$, we can derive 
\begin{align*}
    \norm{\theta_{0}^{(\upsilon+1)}-\theta^{(T)}}_{2}
    & \leq
    \norm{( \boldsymbol{I}_{n} - \eta \boldsymbol{A}^{(\upsilon)} ) (\theta_{0}^{(\upsilon)}-\theta^{(T)})}_{2} + \eta \norm{\nabla \mathcal{L}_{\lambda_{0}}^{(\leq T)} ( \theta^{(T)})}_{2} 
    \\ & \leq
    \max \left\{ \left| 1 - \eta \lambda_{\min,\perp}(\boldsymbol{A}) \right|, \left| 1 - \eta \lambda_{\max,\perp}(\boldsymbol{A}) \right| \right\} \norm{\theta_{0}^{(\upsilon)}-\theta^{(T)}}_{2} + \eta \norm{\nabla \mathcal{L}_{\lambda_{0}}^{(\leq T)} ( \theta^{(T)})}_{2} 
    \\ & \leq
    \left( 1 - \frac{T+1}{10 \kappa} \eta np \right) \norm{\theta_{0}^{(\upsilon)}-\theta^{(T)}}_{2} + \eta \norm{\nabla \mathcal{L}^{(\leq T)} ( \theta^{(T)})}_{2} + \eta \lambda_{0} \norm{\theta^{(T)} - \theta^{\ast}}_{2} + \eta \lambda_{0} \norm{\theta^{\ast}}_{2}
    \\ & \leq
    \left( 1 - \frac{T+1}{10 \kappa} \eta np \right) D_{5} \frac{1}{T+1} \sqrt{\frac{\log n}{pL}} + \eta c_{8} \sqrt{\frac{n^{2}p}{L}} + \eta \lambda_{0} D_{4} \sqrt{\frac{1}{pL}} + \eta \lambda_{0} \sqrt{n} \log \kappa
    \\ & \leq
    D_{5} \frac{1}{T+1} \sqrt{\frac{\log n}{pL}}
\end{align*}
for some constant $D_{5}, c_{8} > 0$. Here, the fourth inequality use of the facts of Lemma \ref{lemma:gradient_T_H} with hypotheses \eqref{eq:A1} and \eqref{eq:G}. The last line holds due to $\sqrt{\frac{n^{2}p}{L}} + \lambda_{0} \sqrt{\frac{1}{pL}} + \lambda_{0} \sqrt{n} \log \kappa = o( \sqrt{\frac{n^{2}p \log n}{L}})$. Here, we succeed to prove \eqref{eq:A1} for $\upsilon+1$.
\end{proof}
\subsubsection[Proof of Tertiary Induction Hypothesis (B1)]{Proof of Tertiary Induction Hypothesis \eqref{eq:B1}}
\label{appendix:equation_B1}
\begin{proof}
Note that
\begin{align*}
    \theta_{0,m}^{(\upsilon+1,m)} - \theta_{m}^{(T)} 
    &= 
    \theta_{0, m}^{(\upsilon, m)} - \eta \left[ \nabla \mathcal{L}_{\lambda_{0}}^{(T \backslash  m, \le T)}(\theta_{0}^{(\upsilon,m)}) \right]_{m} - \theta_{m}^{(T)}
    \\ &=
    \theta_{0,m}^{(\upsilon, m)} -\theta_{m}^{(T)} - \eta \lambda_{0} \theta^{(\upsilon, m)}_{m}
    \\ & -
    \eta \sum_{i \in [n] \backslash \lbrace m \rbrace}  p \left\{ \frac{e^{\theta_{i}^{(T)}}}{e^{\theta_{i}^{(T)}} + e^{\theta_{m}^{(T)}}} - \frac{e^{\theta_{0,i}^{(\upsilon, m)}}}{e^{\theta_{0,i}^{(\upsilon, m)}} + e^{\theta_{0,m}^{(\upsilon, m)}}} \right\} 
    \\ & -
    \eta \sum_{t=0}^{T-1} \sum_{i \in [n] \backslash \{ m \}} G_{m,i}^{(t)} \left\{y_{m,i}^{(t)} - \frac{e^{\theta_{0,i}^{(\upsilon, m)}}}{e^{\theta_{0,i}^{(\upsilon, m)}} + e^{\theta_{0,m}^{(\upsilon, m)}}} \right\} 
    \\ &=
    \left( 1 - \eta \lambda_{0} \right) \left( \theta_{0,m}^{(\upsilon, m)} - \theta_{m}^{(T)} \right) - \eta \lambda_{0} \theta^{(T)}_{m}
    \\ & -
    \eta \sum_{i \in [n] \backslash \lbrace m \rbrace}  p \left\{ \frac{e^{\theta_{i}^{(T)}}}{e^{\theta_{i}^{(T)}} + e^{\theta_{m}^{(T)}}} - \frac{e^{\theta_{0,i}^{(\upsilon, m)}}}{e^{\theta_{0,i}^{(\upsilon, m)}} + e^{\theta_{0,m}^{(\upsilon, m)}}} \right\} 
    \\ & -
    \eta \sum_{t=0}^{T-1} \sum_{i \in [n] \backslash \{ m \}} G_{m,i}^{(t)} \left\{y_{m,i}^{(t)} - \frac{e^{\theta_{i}^{(T)}}}{e^{\theta_{i}^{(T)}} + e^{\theta_{m}^{(T)}}} \right\} 
    \\ & -
    \eta \sum_{t=0}^{T-1} \sum_{i \in [n] \backslash \{ m \}} G_{m,i}^{(t)} \left\{ \frac{e^{\theta_{i}^{(T)}}}{e^{\theta_{i}^{(T)}} + e^{\theta_{m}^{(T)}}} - \frac{e^{\theta_{0,i}^{(\upsilon, m)}}}{e^{\theta_{0,i}^{(\upsilon, m)}} + e^{\theta_{0,m}^{(\upsilon, m)}}} \right\}
    \\ &=
    \left( 1 - \eta \lambda_{0} \right) \left( \theta_{0,m}^{(\upsilon, m)} - \theta_{m}^{(T)} \right) - \eta \lambda_{0} \theta^{(T)}_{m}
    \\ & -
    \eta \sum_{i \in [n] \backslash \lbrace m \rbrace}  p \left\{ \frac{e^{\theta_{i}^{(T)}}}{e^{\theta_{i}^{(T)}} + e^{\theta_{m}^{(T)}}} - \frac{e^{\theta_{0,i}^{(\upsilon, m)}}}{e^{\theta_{0,i}^{(\upsilon, m)}} + e^{\theta_{0,m}^{(\upsilon, m)}}} \right\} 
    \\ & -
    \eta \sum_{t=0}^{T-1} \sum_{i \in [n] \backslash \{ m \}} G_{m,i}^{(t)} \left\{ \frac{e^{\theta_{i}^{(T)}}}{e^{\theta_{i}^{(T)}} + e^{\theta_{m}^{(T)}}} - \frac{e^{\theta_{0,i}^{(\upsilon, m)}}}{e^{\theta_{0,i}^{(\upsilon, m)}} + e^{\theta_{0,m}^{(\upsilon, m)}}} \right\}
\end{align*}
where the last equality follows by the construction of $\left[ \nabla\mathcal{L}^{(\leq T-1)}(\theta^{(T)}) \right]_{m} = 0$. Apply the mean value theorem to obtain
\begin{align*}
    \frac{e^{\theta_{i}^{(T)}}}{e^{\theta_{i}^{(T)}} + e^{\theta_{m}^{(T)}}} - \frac{e^{\theta_{0,i}^{(\upsilon, m)}}}{e^{\theta_{0,i}^{(\upsilon, m)}} + e^{\theta_{0,m}^{(\upsilon, m)}}}
    &= 
    \frac{1}{1 + e^{\theta^{(T)}_{m}- \theta^{(T)}_{i}}} - \frac{1}{1 + e^{\theta_{0,m}^{(\upsilon, m)}-\theta_{0,i}^{(\upsilon, m)}}}
    \\ &=
    -\frac{e^{\xi_{0,i}^{(T)}}}{(1+e^{\xi_{0,i}^{(T)}})^{2}} \left[ \theta^{(T)}_{m} - \theta^{(T)}_{i} - \left( \theta_{0,m}^{(\upsilon, m)}-\theta_{0,i}^{(\upsilon, m)} \right) \right],
\end{align*}
where $\xi_{0,i}^{(T)}$ is some real number lying between $\theta^{(T)}_{m}- \theta^{(T)}_{i}$ and $\theta_{0,m}^{(\upsilon, m)}-\theta_{0,i}^{(\upsilon, m)}$. 

Rearranging the above results yields
\begin{align*}
    \theta_{0,m}^{(\upsilon+1,m)} -\theta_{m}^{(T)}  
    &=
    \left( 1 - \eta \lambda_{0} \right) \left( \theta_{0,m}^{(\upsilon, m)} - \theta_{m}^{(T)} \right) - \eta \lambda_{0} \theta^{(T)}_{m}
    \\ & - 
    \eta p \sum_{i\in[n] \backslash \{ m\}} \frac{e^{\xi_{0,i}^{(T)}}}{(1+e^{\xi_{0,i}^{(T)}})^{2}} \left[ \left( \theta_{0,m}^{(\upsilon, m)}-\theta_{0,i}^{(\upsilon, m)} \right) - \left( \theta^{(T)}_{m} - \theta^{(T)}_{i} \right) \right]
    \\ & -
    \eta \sum_{t=0}^{T-1} \sum_{i\in[n] \backslash \{ m\}} G^{(t)}_{m,i} \frac{e^{\xi_{0,i}^{(T)}}}{(1+e^{\xi_{0,i}^{(T)}})^{2}} \left[ \left( \theta_{0,m}^{(\upsilon, m)}-\theta_{0,i}^{(\upsilon, m)} \right) - \left( \theta^{(T)}_{m} - \theta^{(T)}_{i} \right) \right]
    \\ & = 
    \left[ 1 - \eta \lambda_{0} - 
    \eta p \sum_{i\in[n] \backslash \{ m\}} \frac{e^{\xi_{0,i}^{(T)}}}{(1+e^{\xi_{0,i}^{(T)}})^{2}} - \eta \sum_{t=0}^{T-1} \sum_{i\in[n] \backslash \{ m\}} G^{(t)}_{m,i}  \frac{e^{\xi_{0,i}^{(T)}}}{(1+e^{\xi_{0,i}^{(T)}})^{2}} \right] \left( \theta_{0,m}^{(\upsilon, m)} - \theta_{m}^{(T)} \right) 
    \\ & + 
    \eta p \sum_{i\in[n] \backslash \{ m\}} \frac{e^{\xi_{0,i}^{(T)}}}{(1+e^{\xi_{0,i}^{(T)}})^{2}} \left( \theta_{0,i}^{(\upsilon, m)} - \theta_{i}^{(T)} \right)
    \\ & + 
    \eta \sum_{t=0}^{T-1} \sum_{i\in[n] \backslash \{ m\}} G^{(t)}_{m,i}  \frac{e^{\xi_{0,i}^{(T)}}}{(1+e^{\xi_{0,i}^{(T)}})^{2}} \left( \theta_{0,i}^{(\upsilon, m)} - \theta_{i}^{(T)} \right)
    \\ & - 
    \eta \lambda_{0} \theta^{(T)}_{m},
\end{align*}
and hence
\begin{align*}
    \left| \theta_{0,m}^{(\upsilon+1,m)} - \theta_{m}^{(T)}   \right| 
    & \leq
    \left| 1 - \eta \lambda_{0} - 
    \eta p \sum_{i\in[n] \backslash \{ m\}} \frac{e^{\xi_{0,i}^{(T)}}}{(1+e^{\xi_{0,i}^{(T)}})^{2}} - \eta \sum_{t=0}^{T-1} \sum_{i\in[n] \backslash \{ m\}} G^{(t)}_{m,i}  \frac{e^{\xi_{0,i}^{(T)}}}{(1+e^{\xi_{0,i}^{(T)}})^{2}} \right| \cdot \left| \theta_{0,m}^{(\upsilon, m)} - \theta_{m}^{(T)} \right|
    \\ & + 
    \frac{\eta p}{4} \sum_{i\in[n] \backslash \{ m\}}  \left| \theta_{0,i}^{(\upsilon, m)} - \theta_{i}^{(T)} \right|
    + 
    \frac{\eta}{4} \sum_{t=0}^{T-1} \sum_{i\in[n] \backslash \{ m\}} G^{(t)}_{m,i}  \left| \theta_{0,i}^{(\upsilon, m)} - \theta_{i}^{(T)} \right|
    + 
    \eta \lambda_{0} \left| \theta^{(T)}_{m} \right|
    \\ & \leq
    \left|1 - \eta \lambda_{0} - 
    \eta p \sum_{i\in[n] \backslash \{ m\}} \frac{e^{\xi_{0,i}^{(T)}}}{(1+e^{\xi_{0,i}^{(T)}})^{2}} - \eta \sum_{t=0}^{T-1} \sum_{i\in[n] \backslash \{ m\}} G^{(t)}_{m,i}  \frac{e^{\xi_{0,i}^{(T)}}}{(1+e^{\xi_{0,i}^{(T)}})^{2}} \right| \cdot \left| \theta_{0,m}^{(\upsilon, m)} - \theta_{m}^{(T)} \right|
    \\ & + 
    \frac{\eta p\sqrt{n}}{4} \norm{\theta_{0}^{(\upsilon, m)} - \theta^{(T)}}_{2}
    + 
    \frac{3\eta p\sqrt{n}}{8} \sum_{t=0}^{T-1}  \norm{\theta_{0}^{(\upsilon, m)} - \theta^{(T)}}_{2}
    + 
    \eta \lambda_{0} \norm{\theta^{(T)}}_{\infty}
\end{align*}

Here, the first inequality comes from the triangle inequality and elementary inequality $\frac{e^{\xi_{0,i}^{(T)}}}{(1+\xi_{0,i}^{(T)})^{2}} \leq \frac{1}{4}$ for any $\xi_{0,i}^{(T)} \in \mathbb{R}$, whereas the second relation holds owing to $\sum_{i \in [n] \backslash \lbrace m \rbrace} \left| \theta^{(\upsilon, m)}_{0,i} - \theta_{i}^{(T)}  \right| \leq \sqrt{n} \norm{ \theta^{(\upsilon, m)}_{0} - \theta^{(T)}}_{2}$. We also obtain that
\begin{align*}
    & 1 - \eta \lambda_{0} - 
    \eta p \sum_{i\in[n] \backslash \{ m\}} \frac{e^{\xi_{0,i}^{(T)}}}{(1+e^{\xi_{0,i}^{(T)}})^{2}} - \eta \sum_{t=0}^{T-1} \sum_{i\in[n] \backslash \{ m\}} G^{(t)}_{m,i}  \frac{e^{\xi_{0,i}^{(T)}}}{(1+e^{\xi_{0,i}^{(T)}})^{2}}
    \\ & \geq
    1 - \eta \lambda_{0} - 
    \eta p \cdot \frac{1}{4} (n-1) - \frac{1}{4} \eta T \cdot \frac{3}{2} np
    \\ & \geq 
    1 - \eta \left[ \lambda_{0} + (T+1) np \right]
    \\ & > 
    0.
\end{align*}
To further upper bound $\left| \theta_{0,m}^{(\upsilon+1, m)} - \theta_{m}^{(T)} \right|$, it suffices to obtain a lower bound on $\frac{e^{\xi_{0,i}^{(T)}}}{(1+\xi_{0,i}^{(T)})^{2}}$. Toward this, we can guarantee $\left|\xi_{0,i}^{(T)} - \theta^{(T)}_{0,m} + \theta^{(T)}_{0,i} \right| \leq\left| \theta^{(T)}_{m} - \theta_{0,m}^{(\upsilon, m)} \right| + \left| \theta^{(T)}_{i} - \theta_{0,i}^{(\upsilon, m)} \right| \leq 2D_{9} \frac{1}{T+1} \sqrt{\frac{\log n}{npL}}$ by \eqref{eq:E1}, which implies $\norm{\xi_{0,i}^{T}}_{\infty}$ is bounded. We then have $\frac{e^{\xi_{0,i}^{(T)}}}{(1+e^{\xi_{0,i}^{(T)}})^{2}} \geq c_{9}$ for some constant $c_{9} > 0$. Combine the above bounds with \eqref{eq:B1}, \eqref{eq:E1}, \eqref{eq:H} and Lemma \ref{lemma:degree_concentration} thus yields
\begin{align*}
    \left| \theta_{0,m}^{(\upsilon+1,m)} - \theta_{m}^{(T)}   \right| 
    & \leq
    \left( 1 - \eta \lambda_{0} - 
    \eta p \sum_{i\in[n] \backslash \{ m\}} \frac{e^{\xi_{0,i}^{(T)}}}{(1+e^{\xi_{0,i}^{(T)}})^{2}} - \eta \sum_{t=0}^{T-1} \sum_{i\in[n] \backslash \{ m\}} G^{(t)}_{m,i}  \frac{e^{\xi_{0,i}^{(T)}}}{(1+e^{\xi_{0,i}^{(T)}})^{2}} \right) \cdot \left| \theta_{0,m}^{(\upsilon, m)} - \theta_{m}^{(T)} \right|
    \\ & + 
    \frac{\eta p\sqrt{n}}{4} \norm{\theta_{0}^{(\upsilon, m)} - \theta^{(T)}}_{2}
    + 
    \frac{3\eta p\sqrt{n}}{8} \sum_{t=0}^{T-1}  \norm{\theta_{0}^{(\upsilon, m)} - \theta^{(T)}}_{2}
    + 
    \eta \lambda_{0} \norm{\theta^{(T)}}_{\infty}
    \\ & \leq 
    \left( 1 - c_{9} \eta np - \frac{1}{2}c_{9}T \eta np \right) \left| \theta_{0,m}^{(\upsilon, m)} - \theta_{m}^{(T)} \right|
    + 
    \eta p \sqrt{n} \sum_{t=0}^{T}  \norm{\theta_{0}^{(\upsilon, m)} - \theta^{(T)}}_{2}
    \\ & + 
    \eta \lambda_{0} \norm{\theta^{(\ast)}}_{\infty}
    + 
    \eta \lambda_{0} \norm{\theta^{(T)} - \theta^{\ast}}_{\infty}
    \\ & \leq 
    \left( 1 - \frac{1}{2}c_{9} (T+1) \eta np \right) D_{6} \frac{1}{T+1} \sqrt{\frac{\log n}{npL}}
    + 
    \eta p \sqrt{n} D_{10} \sqrt{\frac{\log n}{pL}} 
    \\ & + 
    \eta \lambda_{0} \log \kappa + \eta \lambda_{0} D_{5} \sqrt{\frac{\log n}{npL}}
    \\ & \leq
    D_{6} \frac{1}{T+1} \sqrt{\frac{\log n}{npL}}
\end{align*}
as long as $\frac{1}{2}c_{9}D_{6} \gg D_{10}$. We thus have proved \eqref{eq:B1} for $\upsilon+1$.
\end{proof}
\subsubsection[Proof of Tertiary Induction Hypothesis (C1)]{Proof of Tertiary Induction Hypothesis \eqref{eq:C1}}
\label{appendix:equation_C1}
\begin{proof}
We observe that 
\begin{align*}
    \theta_{0}^{(\upsilon+1)} - \theta_{0}^{(\upsilon + 1, m)} 
    &= 
    \theta_{0}^{(\upsilon)} - \eta \nabla \mathcal{L}_{\lambda_{0}}^{(\le T)}(\theta_{0}^{(\upsilon)}) - \Big\lbrack \theta_{0}^{(\upsilon,m)} - \eta \nabla \mathcal{L}_{\lambda_{0}}^{(T \backslash m,\le T)}(\theta_{0}^{(\upsilon,m)}) \Big\rbrack 
    \\ &=
    \theta_{0}^{(\upsilon)} - \eta \nabla \mathcal{L}_{\lambda_{0}}^{(\le T)}(\theta_{0}^{(\upsilon)}) - \Big\lbrack \theta_{0}^{(\upsilon, m)} - \eta \nabla \mathcal{L}_{\lambda_{0}}^{(\le T)}(\theta_{0}^{(\upsilon, m)}) \Big\rbrack 
    \\ & \quad- 
    \eta \left( \nabla \mathcal{L}_{\lambda_{0}}^{(\le T)}(\theta_{0}^{(\upsilon, m)}) - \nabla \mathcal{L}_{\lambda_{0}}^{(T \backslash m,\le T)}(\theta_{0}^{(\upsilon,m)})  \right)
    \\ &=
    \Big\lbrack \boldsymbol{I}_{n} - \eta \nabla^{2} \mathcal{L}_{\lambda_{0}}^{(\le T)}(\theta_{0}^{(\upsilon)}(\xi)) \Big\rbrack \left( \theta_{0}^{(\upsilon)} - \theta_{0}^{(\upsilon,m)} \right)
    \\ & \quad - 
    \eta \left( \nabla \ell_{T}(\theta_{0}^{(\upsilon, m)}) - \nabla \ell_{T}^{(m)}(\theta_{0}^{(\upsilon, m)}) \right),
\end{align*}
where $\theta_{0}^{(\upsilon)}(\xi)$ is a convex combination of $\theta_{0}^{(\upsilon)}$ and $\theta_{0}^{(\upsilon, m)}$. By \eqref{eq:D1}, \eqref{eq:E1} and \eqref{eq:H}, we thus have $\norm{\theta_{0}^{(\upsilon)}(\xi) - \theta^{\ast}}_{\infty} \leq \norm{\theta_{0}^{(\upsilon)} -\theta^{(T)}}_{\infty} + \norm{\theta_{0}^{(\upsilon,m)} -\theta^{(T)}}_{\infty} + \norm{\theta^{(T)}-\theta^{\ast}}_{\infty} \leq D_{8} \frac{1}{T+1} \sqrt{\frac{\log n}{npL}} + D_{9} \frac{1}{T+1} \sqrt{\frac{\log n}{npL}} + D_{5} \sqrt{\frac{\log n}{npL}}$, and we can combine this with Lemma \ref{lemma:smallesteigen} to obtain the bound
\begin{align*}
    \norm{\Big\lbrack \boldsymbol{I}_{n} - \eta \nabla^{2} \mathcal{L}_{\lambda_{0}}^{(\le T)}(\theta_{0}^{(\upsilon)}(\xi))  \Big\rbrack \left( \theta_{0}^{(\upsilon)} - \theta_{0}^{(\upsilon, m)} \right) }_{2}  
    & \leq \norm{\Big\lbrack \boldsymbol{I}_{n} - c_{10}(T+1) \eta np} \Big\rbrack \left( \theta_{0}^{(\upsilon)} - \theta_{0}^{(\upsilon, m)} \right)_{2}
    \\ &\leq 
    \left( 1 - c_{10}(T+1) \eta np \right) \cdot \norm{\theta_{0}^{(\upsilon)} - \theta_{0}^{(\upsilon, m)}}_{2}
\end{align*}
for some constant $c_{10}>0$. 

Note that by Lemma \ref{lemma:gradient_t_diff_LOO} and \eqref{eq:E1}, for $t>0$, we have 
\begin{align*}
   \norm{\nabla \ell_{T}(\theta_{0}^{(\upsilon, m)}) - \nabla \ell_{T}^{(m)}(\theta_{0}^{(\upsilon, m)})}_{2} 
   & \leq 
   C_{1}\sqrt{\frac{np \log n}{L}} + C_{1}\sqrt{np \log n} \norm{\theta_{0}^{(\upsilon, m)} - \theta^{(T)}}_{\infty} 
   \\ & = 
   C_{1}\sqrt{\frac{np \log n}{L}} + C_{1} D_{9} \sqrt{np \log n} \sqrt{\frac{\log n}{npL}}
   \\ & \leq 
   C_{1} \sqrt{\frac{np\log n}{L}} \left( 1 + D_{9} \sqrt{\frac{\log n}{np}} \right).
\end{align*}
Combine the above results with \eqref{eq:C1} together yields 
\begin{align*}
   \norm{\theta_{0}^{(\upsilon+1)} - \theta_{0}^{(\upsilon + 1, m)}}_{2} 
   & \leq
   \left( 1 - c_{10}(T+1) \eta np \right) \cdot \norm{\theta_{0}^{(\upsilon)} - \theta_{0}^{(\upsilon, m)}}_{2} 
   \\ & \quad + 
   \eta C_{1} \sqrt{\frac{np\log n}{L}} \left( 1 + D_{9} \sqrt{\frac{\log n}{np}} \right)
   \\ & \leq 
   \left( 1 - c_{10}(T+1) \eta np \right) D_{7} \frac{1}{T+1} \sqrt{\frac{\log n}{npL}}
   \\ & \quad + 
   \eta C_{1} \sqrt{\frac{np\log n}{L}} \left( 1 + D_{9} \sqrt{\frac{\log n}{np}} \right)
   \\ & \leq 
   D_{7} \frac{1}{T+1} \sqrt{\frac{\log n}{npL}},
\end{align*}
where the last inequality requires that $c_{10}D_{7} \gg C_{1}$, which is implied by the condition that $p \geq \frac{c_{0} \log n}{n}$ for some sufficiently large $c_{0} > 0$. We thus have proved \eqref{eq:C1} for $\upsilon + 1$.    
\end{proof}
\subsubsection[Proof of Tertiary Induction Hypothesis (D1)]{Proof of Tertiary Induction Hypothesis \eqref{eq:D1}}
\label{appendix:equation_D1}
\begin{proof}
Consider any $m \in [n]$, it is easily seen from the triangle inequality that
\begin{align*}
    \left| \theta_{0,m}^{(\upsilon+1)}-\theta_{m}^{(T)} \right|
    & \leq
    \left| \theta_{0,m}^{(\upsilon+1)}-\theta_{0,m}^{(\upsilon+1,m)} \right| + \left| \theta_{m}^{(T)}-\theta_{0,m}^{(\upsilon+1,m)} \right|
    \\ & \leq
    \norm{\theta^{(\upsilon+1)}_{0}-\theta^{(\upsilon+1,m)}_{0}}_{2} + \max_{m \in [n]} \left| \theta_{0,m}^{(\upsilon+1,m)} - \theta_{m}^{(T)} \right| 
    \\ & \leq
    D_{7} \frac{1}{T+1} \sqrt{\frac{\log n}{npL}} + D_{6} \frac{1}{T+1} \sqrt{\frac{\log n}{npL}}
    \\ & \leq
    D_{8} \frac{1}{T+1} \sqrt{\frac{\log n}{npL}},
\end{align*}
with the proviso that $D_{8} \geq D_{7} + D_{6}$. Therefore, we proved \eqref{eq:D1} for $\upsilon+1$.
\end{proof}
\subsubsection{Auxiliary Lemmas}

To summarize, we have shown that \eqref{eq:A1}-\eqref{eq:D1} hold for all $\upsilon \leq \upsilon^{\ast}$ with probability at least $1-\mathcal{O}(\upsilon^{\ast} n^{-10})$. Note that \eqref{eq:E1} holds for all $\upsilon \leq \upsilon^{\ast}$ as well and we thus have $\norm{\theta_{0}^{(\upsilon^{\ast})}-\theta^{(T)}}_{\infty} \leq D_{8} \frac{1}{T+1} \sqrt{\frac{\log n}{npL}}$. We then obtain the following lemmas based on the convergence properties of gradient descent to derive the $\ell_{\infty}$ and $\ell_{2}$ bounds on the regularized counterpart $\theta_{0,\lambda}^{(T+1)}$.

\begin{lemma} \label{lemma:auxil_reg_bound}
    Suppose the tertiary induction hypotheses \eqref{eq:A1}-\eqref{eq:D1} hold for $0,\cdots, \upsilon^{\ast}$-th step and the primary induction hypotheses \eqref{eq:D}-\eqref{eq:H} hold at iteration $T$,  then we have
    \begin{align*}
        \norm{\theta^{(T+1)}_{0,\lambda} - \theta^{(T)}}_{\infty} \leq 
        D_{8} \frac{2}{T+1} \sqrt{\frac{\log n}{npL}}
    \end{align*}
    with probability at least $1 - \mathcal{O}(n^{-5})$.
\end{lemma}
\begin{proof}
See \S\ref{appendix:auxil_reg_bound}.
\end{proof}

\begin{lemma} \label{lemma:auxil_reg_bound_2}
    Suppose the tertiary induction hypotheses \eqref{eq:A1}-\eqref{eq:D1} hold for $0,\cdots, \upsilon^{\ast}$-th step and the primary induction hypotheses \eqref{eq:D}-\eqref{eq:H} hold at iteration $T$,  then we have
    \begin{align*}
        \norm{\theta^{(T+1)}_{0,\lambda} - \theta^{(T)}}_{2} \leq 
        D_{11} \frac{1}{T+1} \sqrt{\frac{1}{pL}}
    \end{align*}
    with probability at least $1 - \mathcal{O}(n^{-5})$.
\end{lemma}
\begin{proof}
See \S\ref{appendix:auxil_reg_bound_2}.
\end{proof}

\begin{lemma} \label{lemma:auxil_reg_true_bound_2}
    Suppose the tertiary induction hypotheses \eqref{eq:A1}-\eqref{eq:D1} hold for $0,\cdots, \upsilon^{\ast}$-th step and the primary induction hypotheses \eqref{eq:D}-\eqref{eq:H} hold at iteration $T$,  then we have
    \begin{align*}
        \norm{\theta^{(T+1)} - \theta^{(T+1)}_{0,\lambda}}_{2} \leq 
        D_{12} \frac{1}{T+1} \sqrt{\frac{\log n}{npL}}
    \end{align*}
    with probability at least $1 - \mathcal{O}(n^{-5})$.
\end{lemma}
\begin{proof}
See \S\ref{appendix:auxil_reg_true_bound_2}.
\end{proof}

Next, We introduce the following useful fact that bounds the difference between two pseudo-inverse matrices.
\begin{lemma} \label{lemma:inverse_diff_L2}
    Let $A, B \in \mathbb{R}^{d \times d}$ be positive definite matrices, then
    \begin{align*}
        \norm{A^{-1} - B^{-1}}_{2} \leq \frac{\norm{A-B}_{2}}{\lambda_{\operatorname{min}}(A)\lambda_{\operatorname{min}}(B)}
    \end{align*}
\end{lemma}

\begin{proof}
    \begin{align*}
        \norm{A^{-1} - B^{-1}}_{2} = \norm{A^{-1} (B-A) B^{-1}}_{2} \leq \norm{A^{-1}}_{2} \norm{A-B}_{2} \norm{B^{-1}}_{2} = \frac{\norm{A-B}_{2}}{\lambda_{\operatorname{min}}(A)\lambda_{\operatorname{min}}(B)}.
    \end{align*}
\end{proof}

\begin{lemma} \label{lemma:pseudoinverse_diff_L2}
    Let $A,B\in\mathbb{R}^{n\times n}$ be symmetric positive semi-definite matrices with
    $\mathrm{Null}(A)=\mathrm{Null}(B)=\mathrm{span}\{\boldsymbol{1}\}$.
    Let $\Pi := I - \frac{1}{n}\boldsymbol{1}\boldsymbol{1}^{\top}$ be the orthogonal projector onto
    $1^\perp := \{x\in\mathbb{R}^n:\boldsymbol{1}^{\top}x=0\}$, then
    \begin{align*}
        \big \|A^{\dagger}- B^{\dagger}\big\|_{2}
        \leq
        \frac{\big\|\Pi(A-B)\Pi\big\|_{2}}
        {\lambda_{\min,\perp}(A)\lambda_{\min,\perp}(B)}
        \leq
        \frac{\|A-B\|_{2}}
        {\lambda_{\min,\perp}(A) \lambda_{\min,\perp}(B)} .
    \end{align*}
\end{lemma}
\begin{proof}
Lemma~\ref{lemma:pseudoinverse_diff_L2} then follows as a direct consequence of Lemma~\ref{lemma:inverse_diff_L2}, and its proof is therefore omitted.
\end{proof}
\subsubsection[Proof of Primary Induction Hypothesis (E) and (F)]{Proof of Primary Induction Hypothesis \eqref{eq:E} and \eqref{eq:F}}
\label{appendix:equation_E}
We can now verify the primary induction hypotheses \eqref{eq:E} and \eqref{eq:F} as immediate consequences of the three lemmas stated above.
\begin{proof}
since
\begin{align*}
    \norm{\theta^{(T+1)}- \theta^{(T)}}_{\infty} 
    & \leq
    \norm{\theta^{(T+1)}_{0, \lambda}-\theta^{(T)} }_{\infty} + \norm{\theta^{(T+1)} -\theta^{(T+1)}_{0, \lambda}}_{2} 
    \\ & \leq 
    2D_{8} \frac{1}{T+1} \sqrt{\frac{\log n}{npL}} +D_{12} \frac{1}{T+1} \sqrt{\frac{\log n}{npL}}
    \\ & \leq
    D_{2} \frac{1}{T+1} \sqrt{\frac{\log n}{npL}}
\end{align*}
and
\begin{align*}
    \norm{\theta^{(T+1)} - \theta^{(T)}}_{2} 
    & \leq
    \norm{\theta^{(T+1)}_{0, \lambda}-\theta^{(T)} }_{2} + \norm{\theta^{(T+1)} -\theta^{(T+1)}_{0, \lambda}}_{2} 
    \\ & \leq 
    D_{11} \frac{1}{T+1} \sqrt{\frac{1}{pL}} + D_{12} \frac{1}{T+1} \sqrt{\frac{\log n}{npL}}
    \\ & \leq
    D_{3} \frac{1}{T+1} \sqrt{\frac{1}{pL}}
\end{align*}
as long as $D_{2} \geq D_{12} + 2D_{8}$ and $D_{3} \geq D_{11} + D_{12}$. Here, we showed \eqref{eq:E} and \eqref{eq:F} for $T+1$.
\end{proof}
\subsubsection[Proof of Primary Induction Hypothesis (G)]{Proof of Primary Induction Hypothesis \eqref{eq:G}}
\label{appendix:equation_G}
To bound the increment $\theta^{(T+1)} - \theta^{(T)}$, we start from its telescopic decomposition and split it into two contributions: a deterministic bias term induced by the perturbation of the pseudo-inverse Hessian, and a centered stochastic fluctuation term that can be represented as a martingale. The bias term is controlled using operator-norm bounds for $\big[\nabla^{2}\mathcal{L}^{(\le t)}(\theta^{(t)})\big]^{\dagger}$ together with the $\ell_{\infty}$-control on $\|\theta^{(t+1)} - \theta^{(t)}\|_{\infty}$, whereas the martingale fluctuation term is handled via Azuma--Hoeffding’s inequality. 
\begin{proof}
Observe that
\begin{align*}
    0
    = 
    \nabla \mathcal{L}^{(\leq t)}(\theta^{(t+1)})
    &=
    \nabla \mathcal{L}^{(\leq t)}(\theta^{(t)}) + \nabla^{2} \mathcal{L}^{(\leq t)}(\theta^{(t)}) \left( \theta^{(t+1)}-\theta^{(t)} \right) + \mathbf{U}^{(t)} \left( \theta^{(t+1)}-\theta^{(t)} \right),
    \end{align*}
where 
\begin{align*}
    \norm{\mathbf{U}^{(t)}}_{2} 
    & =
    \sup_{\norm{\boldsymbol{v}_{1}}_{2} = \norm{\boldsymbol{v}_{2}}_{2} = 1, \boldsymbol{v}_{1},\boldsymbol{v}_{2} \neq 0} \boldsymbol{v}_{1}^{\top} \mathbf{U}^{(t)} \boldsymbol{v}_{2} 
    \\ &=
    \sup_{\norm{\boldsymbol{v}_{1}}_{2} = \norm{\boldsymbol{v}_{2}}_{2} = 1, \boldsymbol{v}_{1},\boldsymbol{v}_{2} \neq 0} \left| \frac{1}{2} \sum_{k=0}^{t} \nabla^{3}\ell_{k}(\theta^{(k)}(\iota)) \left( \boldsymbol{v}_{1}, \boldsymbol{v}_{2}, \theta^{(t+1)}-\theta^{(t)} \right) \right|
    \\ & =
    \sup_{\norm{\boldsymbol{v}_{1}}_{2} = \norm{\boldsymbol{v}_{2}}_{2} = 1, \boldsymbol{v}_{1},\boldsymbol{v}_{2} \neq 0} \left| \frac{1}{2} \sum_{k=0}^{t}\sum_{(i,j) \in \mathcal{E}_{k}}  \sigma^{\prime\prime}(\theta^{(k)}_{i}(\iota)-\theta^{(k)}_{j}(\iota))(\boldsymbol{e}_{i} - \boldsymbol{e}_{j})^{\otimes 3} \left( \boldsymbol{v}_{1}, \boldsymbol{v}_{2}, \theta^{(t+1)}-\theta^{(t)} \right) \right|
    \\ & \lesssim
    \left| \sum_{k=0}^{t}\sum_{(i,j) \in \mathcal{E}_{k}} (\boldsymbol{v}_{1,i}-\boldsymbol{v}_{1,j})(\boldsymbol{v}_{2,i}-\boldsymbol{v}_{2,j})(\theta_{i}^{(t+1)}-\theta_{j}^{(t+1)}) \right|
    \\ & \lesssim 
    \sum_{k=0}^{t} \norm{\boldsymbol{v}_{1}}_{\boldsymbol{L}_{\mathcal{G}}^{(k)}} \norm{\boldsymbol{v}_{2}}_{\boldsymbol{L}_{\mathcal{G}}^{(k)}} \cdot \max_{i,j \in [n]} \left| \theta^{(t+1)}_{i}-\theta^{(t)}_{j} \right|
    \\ & \leq
    \sum_{k=0}^{t} \lambda_{\operatorname{max}} \left( \boldsymbol{L}_{\mathcal{G}}^{(k)} \right) \cdot \norm{\theta^{(t+1)}-\theta^{(t)}}_{\infty}
    \\ & \lesssim
    (t+1)np \norm{\theta^{(t+1)}-\theta^{(t)}}_{\infty}.
\end{align*}

Recall that
\begin{align*}
    \norm{\theta^{(T+1)}-\theta^{\ast}}_{2} 
    &=
    \norm{\sum_{t=0}^{T} \theta^{(t+1)}-\theta^{(t)}}_{2}
    \\ & =
    \norm{\sum_{t=0}^{T} \Big\lbrack \nabla^{2} \mathcal{L}^{(\le t)}(\theta^{(t)}) + \mathbf{U}^{(t)} \Big\rbrack^{\dagger} \nabla \ell_{t}(\theta^{(t)})}_{2}
    \\ & =
    \norm{\sum_{t=0}^{T} \left\{ \Big\lbrack \nabla^{2} \mathcal{L}^{(\le t)}(\theta^{(t)}) + \mathbf{U}^{(t)}\Big\rbrack^{\dagger} - \Big\lbrack \nabla^{2} \mathcal{L}^{(\le t)}(\theta^{(t)})\Big\rbrack^{\dagger} \right\} \nabla \ell_{t}(\theta^{(t)})}_{2}
    \\ & \quad + 
    \norm{\sum_{t=0}^{T} \Big\lbrack \nabla^{2} \mathcal{L}^{(\le t)}(\theta^{(t)}) \Big\rbrack^{\dagger} \nabla \ell_{t}(\theta^{(t)})}_{2}.
\end{align*}
For the first term, by Lemma \ref{lemma:smallesteigen} and Lemma \ref{lemma:pseudoinverse_diff_L2}, we derive
\begin{align*}
    \norm{\Big\lbrack \nabla^{2} \mathcal{L}^{(\le t)}(\theta^{(t)}) + \mathbf{U}^{(t)}\Big\rbrack^{\dagger} - \Big\lbrack \nabla^{2} \mathcal{L}^{(\le t)}(\theta^{(t)})\Big\rbrack^{\dagger}}_{2} 
    &=
    \norm{\Big\lbrack \nabla^{2} \mathcal{L}^{(\le t)}(\theta^{(t)}(\iota))\Big\rbrack^{\dagger} - \Big\lbrack \nabla^{2} \mathcal{L}^{(\le t)}(\theta^{(t)})\Big\rbrack^{\dagger}}_{2} 
    \\ & \leq \frac{\norm{\mathbf{U}^{(t)}}_{2}}{\lambda_{\min, \perp} \left( \nabla^{2} \mathcal{L}^{(\le t)}(\theta^{(t)}(\iota)) \right) \lambda_{\min, \perp} \left( \nabla^{2} \mathcal{L}^{(\le t)}(\theta^{(t)}) \right)}
    \\ & \leq
    \frac{(t+1)np \norm{\theta^{(t+1)}-\theta^{(t)}}_{\infty}}{c_{14}(t+1)^{2}n^{2}p^{2}}
    \\ & =
    \frac{\norm{\theta^{(t+1)}-\theta^{(t)}}_{\infty}}{c_{14}(t+1)np}
\end{align*}
for some constant $c_{14} > 0$.

Combining the above bounds with Lemma \ref{lemma:gradient_T_H} and \eqref{eq:E}, there exists a constant $c_{15} > 0$ such that 
\begin{align*}
    & \norm{\sum_{t=0}^{T} \left\{ \Big\lbrack \nabla^{2} \mathcal{L}^{(\le t)}(\theta^{(t)}) + \mathbf{U}^{(t)}\Big\rbrack^{\dagger} - \Big\lbrack \nabla^{2} \mathcal{L}^{(\le t)}(\theta^{(t)})\Big\rbrack^{\dagger} \right\} \nabla \ell_{t}(\theta^{(t)})}_{2}
    \\ & \leq  
    \sum_{t=0}^{T} \norm{\Big\lbrack \nabla^{2} \mathcal{L}^{(\le t)}(\theta^{(t)}) + \mathbf{U}^{(t)}\Big\rbrack^{\dagger} - \Big\lbrack \nabla^{2} \mathcal{L}^{(\le t)}(\theta^{(t)})\Big\rbrack^{\dagger}}_{2} \norm{\nabla \ell_{t}(\theta^{(t)})}_{2}
    \\ & \leq 
    \sum_{t=0}^{T} \frac{\norm{\theta^{(t+1)}-\theta^{(t)}}_{\infty}}{c_{14}(t+1)np} \norm{\nabla \ell_{t}(\theta^{(t)})}_{2}
    \\ & \leq 
    c_{15} \sum_{t=0}^{T} \frac{1}{(t+1)^{2}} \frac{1}{np} \sqrt{\frac{\log n}{npL}} \sqrt{\frac{n^{2}p}{L}}
    \\ & \leq 
    c_{15} \sqrt{\frac{1}{pL}} \sqrt{\frac{\log n}{npL}}
\end{align*}
due to $\sum_{t=0}^{T}\frac{1}{(t+1)^{2}} \leq \frac{\pi^{2}}{6}$.

For the second term, define the $\sigma$-field
 \begin{align*}
    \mathcal{H}_{s}:=\sigma \left\langle \{G^{(t)}\}_{t=0}^{s},\{\theta^{(t)}\}_{t=0}^{s} \right\rangle, \quad \forall s \in [T].
\end{align*}
For $v \in \mathbb{R}^{n}$ with $\norm{v}_{2} \leq 1$, let
\begin{align*}
    S_{v} 
    & := 
    v^{\top} S_{T} 
    = 
    v^{\top} \sum_{t=0}^{T} \Big\lbrack \nabla^{2} \mathcal{L}^{(\le t)}(\theta^{(t)})\Big\rbrack^{\dagger} \nabla \ell_{t}(\theta^{(t)})
    \\ & =
    \sum_{t=0}^{T} \sum_{1\le j<i\le n} v^{\top}\Big\lbrack \nabla^{2} \mathcal{L}^{(\le t)}(\theta^{(t)})\Big\rbrack^{\dagger}(\boldsymbol{e}_i - \boldsymbol{e}_j) \cdot G^{(t)}_{i,j}
    \left\{\frac{e^{\theta^{(t)}_{i}}}{e^{\theta^{(t)}_{i}}+e^{\theta^{(t)}_{j}}} - y^{(t)}_{j,i}\right\},
    \quad y^{(t)}_{j,i}:=\frac{1}{L}\sum_{l=1}^{L}y^{(t,l)}_{j,i}.
\end{align*}
Then $\left\lbrace \left( v^{\top}S_{s}, \mathcal{H}_{s} \right): s \in [T] \right\rbrace$ is a martingale such that $v^{\top}S_{s}$ is bounded for all $s \in [T]$. Azuma-Hoeffding's inequality therefore yields, for any $\epsilon > 0$,
\begin{align*}
    \mathbb{P}\big(|S_v|\ge \epsilon \big)
    \ \le\ 2 \exp\!\left(-\frac{L\,\epsilon^2}{2 \sum_{t=0}^{T} \sum_{i>j} G^{(t)}_{i,j}\left( v^{\top}\Big\lbrack \nabla^{2} \mathcal{L}^{(\le t)}(\theta^{(t)})\Big\rbrack^{\dagger} (\boldsymbol{e}_i - \boldsymbol{e}_j) \right)^2}\right).
\end{align*}
With Lemma \ref{lemma:largesteigen} and Lemma \ref{lemma:smallesteigen}, the quadratic form in the denominator can be written as
\begin{align*}
    & \sum_{t=0}^{T}\sum_{i>j} G^{(t)}_{i,j} \left( v^{\top}\Big\lbrack \nabla^{2} \mathcal{L}^{(\le t)}(\theta^{(t)})\Big\rbrack^{\dagger} (\boldsymbol{e}_i - \boldsymbol{e}_j) \right)^2
    \\ & =
    v^{\top} \Big\lbrack \nabla^{2} \mathcal{L}^{(\le t)}(\theta^{(t)})\Big\rbrack^{\dagger} \sum_{t=0}^{T} \sum_{i>j}  G^{(t)}_{i,j} (\boldsymbol{e}_i - \boldsymbol{e}_j) (\boldsymbol{e}_i - \boldsymbol{e}_j)^{\top} \Big\lbrack \nabla^{2} \mathcal{L}^{(\le t)}(\theta^{(t)})\Big\rbrack^{\dagger} v
    \\ & =
    \sum_{t=0}^{T} v^{\top} \Big\lbrack \nabla^{2} \mathcal{L}^{(\le t)}(\theta^{(t)})\Big\rbrack^{\dagger} \boldsymbol{L}_{\mathcal{G}}^{(t)} \Big\lbrack \nabla^{2} \mathcal{L}^{(\le t)}(\theta^{(t)})\Big\rbrack^{\dagger} v 
    \\ & \leq
    \sum_{t=0}^{T} \lambda_{\text{max}} \left( \boldsymbol{L}_{\mathcal{G}}^{(t)} \right) \norm{ \Big\lbrack \nabla^{2} \mathcal{L}^{(\le t)}(\theta^{(t)})\Big\rbrack^{\dagger} }_{2}^{2} \norm{v}_{2}^{2}
    \\ & \lesssim
    \sum_{t=0}^{T} \frac{np}{(t+1)^{2}n^{2}p^{2}} 
    \\ & \lesssim
    \frac{1}{np}.
\end{align*}
Let $\mathcal{U}:=\{u\in\mathbb R^n:\sum_{i=1}^{n}u_i^2\le 1\}$ and fix a symmetric $1/2$-net $\mathcal{V}\subset\mathcal{U}$ such that $\log|\mathcal{V}|\le C' n$. As usual,
\begin{align*}
    \norm{ \sum_{t=0}^{T} \Big\lbrack \nabla^{2} \mathcal{L}^{(\le t)}(\theta^{(t)})\Big\rbrack^{\dagger} \nabla \ell_{t}(\theta^{(t)}) }_{2} = \sup_{\norm{v}_{2} \leq 1} \left| S_{v} \right| \leq 2 \max_{v \in \mathcal{V}} \left| S_{v} \right|.
\end{align*}
Apply the union bound, for any $a > 0$ and any $\epsilon > 0$, we have
\begin{align*}
    \mathbb{P} \left(\max_{v\in\mathcal{V}}|S_v|\ge \epsilon \right)
    & \leq 
    2 |\mathcal{V}|\exp\!\left(-\frac{L\,\epsilon^2}{2 \max_{v\in\mathcal{V}} \sum_{t=0}^{T} \sum_{i>j} G^{(t)}_{i,j}\left( v^{\top}\Big\lbrack \nabla^{2} \mathcal{L}^{(\le t)}(\theta^{(t)})\Big\rbrack^{\dagger} (\boldsymbol{e}_i - \boldsymbol{e}_j) \right)^2}\right)
    \\ & \leq
    2 |\mathcal{V}|\exp\!\left(-\frac{L\,\epsilon^2}{2 \max_{\norm{v}_{2} \leq 1} \sum_{t=0}^{T} \sum_{i>j} G^{(t)}_{i,j}\left( v^{\top}\Big\lbrack \nabla^{2} \mathcal{L}^{(\le t)}(\theta^{(t)})\Big\rbrack^{\dagger} (\boldsymbol{e}_i - \boldsymbol{e}_j) \right)^2}\right).
\end{align*}
Choose 
\begin{align*}
    \epsilon = c_{16} \sqrt{\frac{\left( \log |\mathcal{V}| + 11\log n \right)\sum_{t=0}^{T} \lambda_{\text{max}} \left( \boldsymbol{L}_{\mathcal{G}}^{(t)} \right) \norm{ \Big\lbrack \nabla^{2} \mathcal{L}^{(\le t)}(\theta^{(t)})\Big\rbrack^{\dagger} }_{2}^{2}}{L}}
\end{align*}
with a sufficiently large absolute constant $c_{16}>0$, which indicates that
\begin{align*}
    \norm{ \sum_{t=0}^{T} \Big\lbrack \nabla^{2} \mathcal{L}^{(\le t)}(\theta^{(t)})\Big\rbrack^{\dagger} \nabla \ell_{t}(\theta^{(t)}) }_{2} \leq c_{17} \sqrt{\frac{1}{pL}}
\end{align*}
with probability at least $1 - \mathcal{O}(n^{-10})$ and a constant $c_{17} > 0$.
Putting the bounds of above two terms together,
\begin{align*}
    \norm{\theta^{(T+1)}-\theta^{\ast}}_{2}  \leq c_{17} \sqrt{\frac{1}{pL}} + c_{15} \sqrt{\frac{1}{pL}} \sqrt{\frac{\log n}{npL}} \leq D_{4} \sqrt{\frac{1}{pL}},
\end{align*}
where $D_{4} \gg c_{15} + c_{17}$ is sufficiently large and $npL \gtrsim \log n$.

Therefore, \eqref{eq:G} holds for $T + 1$.
\end{proof}
\subsection{Proof Outline of Lemma~\ref{lemma:equation_H} in Step~III}
\label{appendix:equation_H}
In this Appendix, we outline the proof of Lemma~\ref{lemma:equation_H} by combining the arguments developed in the previous three steps. 
\subsubsection{Useful Auxiliary Facts}
We begin by collecting several building blocks and useful auxiliary facts.
\begin{lemma} \label{lemma:hessian_dif_L2}
    The following event holds
    \begin{align*}
        \mathcal{A}_{4}^{(T)}
        := \left\{ \norm{\nabla^{2} \mathcal{L}^{(\leq T)}(\theta^{\prime})-\nabla^{2} \mathcal{L}^{(\leq T)}(\theta^{\prime\prime})}_{2} \lesssim (T+1)np \norm{\theta^{\prime}-\theta^{\prime\prime}}_{\infty}
         \right\}
    \end{align*}
    with probability at least $1-\mathcal{O}(n^{-10})$.
\end{lemma}
\begin{proof}
See \S\ref{appendix:hessian_dif_L2}.
\end{proof}

\begin{lemma} \label{lemma:diff_qua_T+1_true}
    With probability exceeding $1 - \mathcal{O}(n^{-10})$, one has
    \begin{align*}
        \norm{\sum_{t=0}^{T}\overline{\theta}^{(t+1)} - \theta^{(t)} }_{\infty} \leq D_{13} \sqrt{\frac{\log n}{npL}}.
    \end{align*}
\end{lemma}
\begin{proof}
See \S\ref{appendix:diff_qua_T+1_true}.
\end{proof}

\begin{lemma} \label{lemma:diff_ori_qudra_t+1_l2}
    Suppose the auxiliary induction hypotheses \eqref{eq:A}-\eqref{eq:C} hold for $0,\cdots, \tau^{\ast}$-th step and the primary induction hypotheses \eqref{eq:D}-\eqref{eq:H} hold up to iteration $T$. Then for any $t = 0, \cdots, T$, one has
    \begin{align*}
        \norm{\overline{\theta}^{(t+1)} - \theta^{(t+1)}}_{2} \leq D_{14} \frac{1}{(t+1)^{2}} \sqrt{\frac{\log n}{n}} \frac{1}{pL}
    \end{align*}
    with probability exceeding $1 - \mathcal{O}(n^{-10})$.
\end{lemma}
\begin{proof}
See \S\ref{appendix:diff_ori_qudra_t+1_l2}.
\end{proof}
\subsubsection{Preliminary Results on the Leave-One-Out Quadratic Loss}
We denote by $\Psi = \left\{ \boldsymbol{e}_{i} - \boldsymbol{e}_{j}: i,j \in [n] \right\}$ and $\Psi^{\perp} = \text{span} \{ \boldsymbol{1}_{n} \}$. Intuitively, we can always move any $\theta \in \mathbb{R}^{n}$ along a direction $\bm{v} \in \Psi^{\perp}$ that does not change both $\mathcal{L}^{(\leq t)}(\cdot)$ and $\overline{\mathcal{L}}^{(\leq t)}(\cdot)$ until it satisfies the identifiability constraint $\Theta = \left\{ \theta \in \mathbb{R}^{n}: \boldsymbol{1}_{n}^{\top} \theta = 0 \right\}$. By the definition of $\mathcal{L}^{(\leq t)}(\cdot)$ and $\overline{\mathcal{L}}^{(\leq t)}(\cdot)$ we know that for any $\theta \in \mathbb   {R}^{n}$ and $\bm{v} = c \boldsymbol{1} \in \Psi^{\perp}$ for some $c \in \mathbb{R}$, we have
\begin{align*}
    \mathcal{L}^{(\leq t)}(\theta) = \mathcal{L}^{(\leq t)}(\theta + \bm{v}) \quad \text{and} \quad \overline{\mathcal{L}}^{(\leq t)}(\theta) = \overline{\mathcal{L}}^{(\leq t)}(\theta + \bm{v}). 
\end{align*}

Formally, we have the following lemma.

\begin{lemma} \label{lemma:theta_shift}
    For any $\theta \in \mathbb{R}^{n}$, there exists a $\bm{v} \in \Psi^{\perp}$ such that $\theta + \bm{v} \in \Theta$.
\end{lemma}
\begin{proof}
See \S\ref{appendix:theta_shift}.
\end{proof}

We adopt the following notations for a given vector $\rm{x} \in \mathbb{R}^{n}$,
\begin{align*}
    \mathcal{L}^{(\leq t)}|_{\rm{x}_{-i}}(x_{i}) = \mathcal{L}^{(\leq t)} (\breve{\theta})|_{\breve{\theta}_{i} = \rm{x}_{i}, \breve{\theta}_{-i} = \rm{x}_{-i}} \quad \text{and} \quad \overline{\mathcal{L}}^{(\leq t)}|_{\rm{x}_{-i}}(x_{i}) = \overline{\mathcal{L}}^{(\leq t)} (\breve{\theta})|_{\breve{\theta}_{i} = \rm{x}_{i}, \breve{\theta}_{-i} = \rm{x}_{-i}} 
\end{align*}
which act as the marginal loss of $\mathcal{L}^{(\leq t)}(\cdot)$ and $\overline{\mathcal{L}}^{(\leq t)}(\cdot)$, respectively, in the $i$-th coordinate given other coordinates fixed. With Lemma \ref{lemma:theta_shift}, we have the following proposition.

\begin{proposition} \label{prop:minimizer_L_i}
    For $i \in [n]$, $\theta_{i}^{(t+1)}$ is the minimizer of the univariate function $\mathcal{L}^{(\leq t)}|_{\theta_{-i}^{(t+1)}}(x_{i})$.
\end{proposition}
\begin{proof}
See \S\ref{appendix:minimizer_L_i}.
\end{proof}

Fix the $n-1$ coordinates by $\theta^{(t+1)}_{-i}$, we define $\overline{\theta}_{i}^{(t+1,\prime)}$ as the minimizer of $\overline{\mathcal{L}}^{(\leq t)}|_{\theta_{-i}^{(t+1)}}(x_{i})$. Plugging $\theta_{-i}^{(t+1)}(z) = \theta^{(t)} + \left( z-\theta^{(t)}_{i} \right) \boldsymbol{e}_{i} + \sum_{j \in [n] \backslash \{i\}} \left(\theta^{(t+1)}_{j} - \theta^{(t)}_{j} \right) \boldsymbol{e}_{j}$ into $\overline{\mathcal{L}}^{(\leq t)}(\cdot)$ yields
\begin{align*}
    \overline{\mathcal{L}}^{(\leq t)}(\theta_{-i}^{(t+1)}(z))
    & =
    \mathcal{L}^{(\leq t)}(\theta^{(t)}) + \left(\theta_{-i}^{(t+1)}(z)-\theta^{(t)}\right)^{\top} \nabla \mathcal{L}^{(\leq t)}(\theta^{(t)}) 
    \\ & \quad + 
    \frac{1}{2} \left(\theta_{-i}^{(t+1)}(z)-\theta^{(t)}\right)^{\top} \nabla^{2} \mathcal{L}^{(\leq t)}(\theta^{(t)})\left(\theta_{-i}^{(t+1)}(z)-\theta^{(t)}\right)
    \\ & =
    D_{13} + \left(z-\theta^{(t)}_{i}\right) \left[ \nabla \mathcal{L}^{(\leq t)}(\theta^{(t)}) \right]_{i} + \frac{1}{2} \left(z-\theta^{(t)}_{i}\right)^{2} \left[\nabla^{2} \mathcal{L}^{(\leq t)}(\theta^{(t)})\right]_{i,i} 
    \\ & \quad + 
    \left(z-\theta^{(t)}_{i}\right) \sum_{j \in [n] \backslash \{i\}}  \left(\theta^{(t+1)}_{j} - \theta^{(t)}_{j}\right) \left[\nabla^{2} \mathcal{L}^{(\leq t)}(\theta^{(t)})\right]_{i,j}
\end{align*}
for some constant $D_{13}$. First-order condition $\frac{\partial}{\partial z}\overline{\mathcal{L}}^{(\leq t)}(\theta_{-i}^{(t+1)}(z)) = 0$ further gives
\begin{align} \label{eq:diff_quadratic_prime_true}
    \overline{\theta}_{i}^{(t+1,\prime)} = \theta_{i}^{(t)} - \frac{\left[ \nabla \mathcal{L}^{(\leq t)}(\theta^{(t)}) \right]_{i} + \sum_{j \in [n] \backslash \{i\}} \left(\theta^{(t+1)}_{j} - \theta^{(t)}_{j}\right) \left[\nabla^{2} \mathcal{L}^{(\leq t)} (\theta^{(t)})\right]_{i,j}  }{\left[\nabla^{2} \mathcal{L}^{(\leq t)}(\theta^{(t)})\right]_{i,i}}.
\end{align}
Similar to Proposition \ref{prop:minimizer_L_i}, we have the following proposition for $\overline{\mathcal{L}}^{(\leq t)}|_{\overline{\theta}_{-i}^{(t+1)}}(x_{i})$. 

\begin{proposition} \label{prop:minimizer_quadratic_L_i}
    For $i \in [n]$, $\overline{\theta}_{i}^{(t+1)}$ is the minimizer of the univariate function $\overline{\mathcal{L}}^{(\leq t)}|_{\overline{\theta}_{-i}^{(t+1)}}$.
\end{proposition}
\begin{proof}
See \S\ref{appendix:minimizer_quadratic_L_i}.
\end{proof}

Similarly, we obtain
\begin{align} \label{eq:diff_quadratic_true}
    \overline{\theta}_{i}^{(t+1)} = \theta_{i}^{(t)} - \frac{\left[ \nabla \mathcal{L}^{(\leq t)}(\theta^{(t)}) \right]_{i} + \sum_{j \in [n] \backslash \{i\}} \left(\overline{\theta}^{(t+1)}_{j} - \theta^{(t)}_{j}\right) \left[\nabla^{2} \mathcal{L}^{(\leq t)} (\theta^{(t)})\right]_{i,j}  }{\left[\nabla^{2} \mathcal{L}^{(\leq t)}(\theta^{(t)})\right]_{i,i}}.
\end{align}

In order to bound $\left| \theta^{(t+1)}_{i} - \overline{\theta}^{(t+1)}_{i} \right|$, we are about to bound $\left| \overline{\theta}^{(t+1, \prime)}_{i} - \overline{\theta}^{(t+1)}_{i} \right|$ and $\left| \theta^{(t+1)}_{i} - \overline{\theta}^{(t+1, \prime)}_{i} \right|$ separately. However, before providing an upper bound for $\left| \overline{\theta}^{(t+1, \prime)}_{i} - \overline{\theta}^{(t+1)}_{i} \right|$, we first include several auxiliary results which are helpful to the proof of Lemma \ref{lemma:diff_quadratic}. Recall that we have constructed the iteration-restricted leave-one-out loss function $\mathcal{L}^{(t \backslash m,\leq t)} ( \theta)$ that replaces all likelihood components involving the $m$-th  item with their expected values at iteration $t$ that 
\begin{align*}
    \mathcal{L}^{(t \backslash m,\leq t)} ( \theta)
    & := \mathcal{L}^{(\leq t-1)} (\theta) + \ell_{t}^{(m)}(\theta)
\end{align*}
and $\ell_{t}^{(m)}$ is $\ell_{t}$ with all terms including $m$ deleted that
\begin{align*}
    \ell_{t}^{(m)} (\theta) = &\sum_{(i,j) \in \mathcal{E}_{t},i > j, i \neq m, j \neq m} \left\{ - y_{j,i}^{(t)} \left(\theta_i - \theta_j\right) + \log\left(1 + e^{\theta_i - \theta_j} \right) \right\}
    \\ &\quad +
    \sum_{i \in [n] \backslash \lbrace m \rbrace} p \left\{ - \frac{e^{\theta^{(t)}_{i}}}{e^{\theta_{i}^{(t)}} + e^{\theta_{m}^{(t)}}} \left(\theta_{i} - \theta_{m}\right) + \log\left(1 + e^{\theta_{i} - \theta_{m}} \right) \right\}.
\end{align*}
Here we define $\theta^{(t+1)}_{(i)}$ as
\begin{align*} 
    \theta^{(t+1)}_{(i)} := \underset{\theta \in \Theta}{\arg\min} \;\mathcal{L}^{(t \backslash i, \leq t)} ( \theta)
\end{align*}
and let $\overline{\theta}^{(t+1)}_{(i)}$ be the solution of the following linear equations
\begin{align*}
    \begin{cases}
        \mathcal{P}\nabla \mathcal{L}^{(t \backslash i, \leq t)}(\theta^{(t)}) + \mathcal{P}\nabla^{2} \mathcal{L}^{(t \backslash i, \leq t)}(\theta^{(t)}) \left( \overline{\theta}^{(t+1)}_{(i)} - \theta^{(t)}\right) = \bm{0};
        \\
        \mathcal{P} \overline{\theta}^{(t+1)}_{(i)} = \overline{\theta}^{(t+1)}_{(i)},
    \end{cases}
\end{align*}
which reduces to
\begin{align*}
    \overline{\theta}^{(t+1)}_{(i)} - \theta^{(t)} 
    = 
    -\left[ \nabla^{2} \mathcal{L}^{(t \backslash i, \leq t)}(\theta^{(t)}) \right]^{\dagger} \nabla \mathcal{L}^{(t \backslash i, \leq t)}(\theta^{(t)}).
\end{align*} 

\begin{lemma} \label{lemma:iterative_t_LOO}
    Under the assumptions of Theorem \ref{thm:iterative_loss}, for $i \in [n]$, with probability at least $1 - \mathcal{O}(n^{-5})$ we have
    \begin{align*}
        \norm{\theta^{(t+1)}_{(i)} - \theta^{(t)}}_{\infty}
        & \leq 
        D_{15} \frac{1}{t+1} \sqrt{\frac{\log n}{npL}},
        \\
        \norm{\theta^{(t+1)}_{(i)} - \theta^{(t)}}_{2} 
        & \leq 
        D_{16} \frac{1}{t+1} \sqrt{\frac{1}{pL}},
        \\
        \norm{\theta^{(t+1)}_{(i)} - \theta^{(t+1)}}_{2}
        & \leq
        D_{17} \frac{1}{t+1} \sqrt{\frac{\log n}{npL}},
        \\
        \norm{\theta^{(t+1)}_{(i)} - \overline{\theta}^{(t+1)}_{(i)}}_{2} 
        & \leq
        D_{18} \frac{1}{t+1} \frac{\log n}{\sqrt{n}pL}, 
    \end{align*}
    for some constants $D_{15}, D_{16}, D_{17}, D_{18} > 0$.
\end{lemma}

The proof of Lemma~\ref{lemma:iterative_t_LOO} closely parallels the previous argument and is therefore omitted. In particular, Lemma~\ref{lemma:iterative_t_LOO} can be established by mimicking the proof of the bounds on $\norm{\overline{\theta}^{(t+1)} - \theta^{(t)}}_{\infty}$, $\norm{\overline{\theta}^{(t+1)} - \theta^{(t)}}_{2}$, and $\norm{\theta^{(t+1)} - \overline{\theta}^{(t+1)}}_{2}$.

\begin{lemma} \label{lemma:supp_loss_t}
    For $i \in [n]$ and $t \in \mathbb{N}$, with probability at least $1 - \mathcal{O}(n^{-10})$ we have
    \begin{align*}
        \left| \left[\nabla \mathcal{L}^{(\leq t)} (\theta^{(t)})\right]_{i} \right| 
        & \leq 
        D_{19} \sqrt{\frac{np \log n}{L}},
        \\
        \sum_{j \in [n] \backslash \{i\}}\left[\nabla^{2} \mathcal{L}^{(\leq t)} (\theta^{(t)})\right]_{i,j}^{2}
        & \leq
        D_{20} (t+1)^{2}np,
        \\
        \sum_{j \in [n] \backslash \{i\}} \left| \left[\nabla^{2} \mathcal{L}^{(\leq t)} (\theta^{(t)})\right]_{i,j} \right|
        & \leq
        D_{21} (t+1) np, 
        \\
        \left| y_{j,i}^{(t)} - \mathbb{E} \left[ y_{j,i}^{(t)} \right] \right| 
        & \leq 
        D_{22} \sqrt{\frac{\log n}{L}}, \text{ for any } i, j \in [n], i \neq j,
    \end{align*}
    for some constants $D_{19}, D_{20}, D_{21}, D_{22} > 0$.
\end{lemma}
\begin{proof}
See \S\ref{appendix:supp_loss_t}.
\end{proof}

\begin{lemma} \label{lemma:diff_quadratic}
    Under the assumptions of Theorem \ref{thm:iterative_loss}, for $i \in [n]$, with probability at least $1 - \mathcal{O}(n^{-5})$ we have
    \begin{align*}
        \left| \overline{\theta}^{(t+1, \prime)}_{i} - \overline{\theta}^{(t+1)}_{i} \right| 
        & \leq 
        D_{23} \frac{1}{t+1} \frac{1}{\sqrt{npL}} \frac{\log n}{\sqrt{nL}p} + D_{23} \frac{1}{t+1} \frac{1}{np} \sqrt{\frac{\log n}{L}} \left( 1 + \frac{\log n}{\sqrt{np}} \right) 
        \\ & + 
        D_{23} \frac{1}{t+1} \left( \frac{\log n}{np} + \frac{1}{\sqrt{np}} \right) \norm{\overline{\theta}^{(t+1)} - \theta^{(t+1)}}_{\infty}
    \end{align*}
    for some constant $D_{23} > 0$.
\end{lemma}
\begin{proof}
See \S\ref{appendix:diff_quadratic}.
\end{proof}

We now describe the intuition of bounding $\left| \theta^{(t+1)}_{i} - \overline{\theta}^{(t+1, \prime)}_{i} \right|$. Since $\theta^{(t+1)}_{i}$ is the minimizer of $\mathcal{L}^{(\leq t)}|_{\theta_{-i}^{(t+1)}}(x_{i})$ and $\overline{\theta}^{(t+1, \prime)}_{i}$ is the minimizer of $\overline{\mathcal{L}}^{(\leq t)}|_{\theta_{-i}^{(t+1)}}(x_{i})$. Therefore, as long as $\mathcal{L}^{(\leq t)}|_{\theta_{-i}^{(t+1)}}(\cdot)$ and $\overline{\mathcal{L}}^{(\leq t)}|_{\theta_{-i}^{(t+1)}}(\cdot)$ are close enough, the difference between their minimizers $\theta^{(t+1)}_{i}$ and $\overline{\theta}^{(t+1, \prime)}_{i}$ is small. We summarize this finding in the following lemma \ref{lemma:diff_quadratic_MLE}.

\begin{lemma} \label{lemma:diff_quadratic_MLE}
    Under the assumptions of Theorem \ref{thm:iterative_loss}, for $i \in [n]$, with probability at least $1 - \mathcal{O}(n^{-5})$ we have
    \begin{align*}
        \left| \theta^{(t+1)}_{i} - \overline{\theta}^{(t+1, \prime)}_{i} \right| \leq D_{24} \frac{1}{(t+1)^{2}} \frac{\log n}{npL}
    \end{align*}
    for some constant $D_{24} > 0$.
\end{lemma}

\begin{proof}
See \S\ref{appendix:diff_quadratic_MLE}.
\end{proof}

Finally, combining the conclusions of Lemma \ref{lemma:diff_quadratic_MLE} and \ref{lemma:diff_quadratic} we get the following theorem for $\norm{\theta^{(t+1)} - \overline{\theta}^{(t+1)}}_{\infty}$.

\begin{lemma} \label{lemma:diff_ori_qudra_t+1}
    Under the assumptions of Theorem \ref{thm:iterative_loss}, with probability at least $1 - \mathcal{O}(n^{-5})$, we have
    \begin{align*}
        \norm{\theta^{(t+1)} - \overline{\theta}^{(t+1)}}_{\infty} \leq D_{25}\frac{1}{t+1} \sqrt{\frac{\log n}{npL}} \left[ \frac{\sqrt{\log n}}{\sqrt{nL}p}  + \frac{1}{\sqrt{np}} \left( 1 + \frac{\log n}{\sqrt{np}} \right) \right]
    \end{align*}
    for some constant $D_{25} > 0$.
\end{lemma}
\begin{proof}
See \S\ref{appendix:diff_ori_qudra_t+1}.
\end{proof}
\subsubsection[{Proof of Primary Induction Hypothesis (H)}]{Proof of Primary Induction Hypothesis \eqref{eq:H}}
Combining Lemma~\ref{lemma:diff_qua_T+1_true} with Lemma~\ref{lemma:diff_ori_qudra_t+1}, we can now establish Lemma~\ref{lemma:equation_H} as follows.
\begin{proof}
As long as $\log T \left[ \frac{\sqrt{\log n}}{\sqrt{nL}p} + \frac{1}{\sqrt{np}} \left( 1 + \frac{\log n}{\sqrt{np}} \right) \right] = \mathcal{O}(1)$, we can derive
\begin{align*}
    \norm{\sum_{t=0}^{T}\theta^{(t+1)} - \overline{\theta}^{(t+1)}}_{\infty}
    & \leq
    \sum_{t=0}^{T} \norm{\theta^{(t+1)} - \overline{\theta}^{(t+1)}}_{\infty}
    \\ & \leq
    D_{25} \sum_{t=0}^{T} \frac{1}{t+1} \sqrt{\frac{\log n}{npL}} \left[ \frac{\sqrt{\log n}}{\sqrt{nL}p} + \frac{1}{\sqrt{np}} \left( 1 + \frac{\log n}{\sqrt{np}} \right) \right]
    \\ & \leq
    D_{26} \sqrt{\frac{\log n}{npL}} 
\end{align*}
for some constant $D_{26} > 0$ with probability exceeding $1 - \mathcal{O}(n^{-5})$.

In this way, we can further guarantee
\begin{align*}
    \norm{\theta^{(T+1)}-\theta^{\ast}}_{\infty}
    & =
    \norm{\sum_{t=0}^{T} \theta^{(t+1)}-\theta^{(t)}}_{\infty}
    \\ & \leq
    \norm{\sum_{t=0}^{T} \theta^{(t+1)} - \overline{\theta}^{(t+1)}}_{\infty} + 
    \norm{\sum_{t=0}^{T} \overline{\theta}^{(t+1)} - \theta^{(t)}}_{\infty} 
    \\ & \leq
    D_{5} \sqrt{\frac{\log n}{npL}} 
\end{align*}
with probability at least $1 - \mathcal{O}(n^{-5})$, where $D_{5} \geq D_{13} + D_{26}$. Thus \eqref{eq:H} holds for $T + 1$.
\end{proof}
\section{Proof of Theorem~\ref{thm:asy_norm}}
\label{appendix:asy_norm}

\begin{proof}
    The following content is conditioned on the event $\left\{ \mathcal{A}_{t} \right\}_{t=0}^{T+1}$. Let $\bar{\bm{c}} = \mathcal{P} \bm{c}$ be the projection of $\bm{c}$ onto linear space $\Theta$. Since $\bar{\bm{c}}$ is the projection of $\bm{c}$ onto linear space $\Theta$, we obtain
    \begin{align*}
        \bm{c}^{\top} \overline{\theta}^{(t+1)} - \bm{c}^{\top} \theta^{(t)} = \bar{\bm{c}}^{\top} \overline{\theta}^{(t+1)} - \bar{\bm{c}}^{\top} \theta^{(t)}
    \end{align*}
    for any iteration $t = 0, \cdots, T$ with $\overline{\theta}^{(t+1)}, \theta^{(t)} \in \Theta$. Since $\nabla \ell_{t}(\theta^{(t)}) + \nabla^{2} \mathcal{L}^{(\leq t)}(\theta^{(t)})  (\overline{\theta}^{(t+1)} - \theta^{(t)}) 
    = \bm{0}$, we define
    \begin{align*}
        X_t{} 
        & :=
        \bar{\bm{c}}^{\top} \overline{\theta}^{(t+1)} - \bar{\bm{c}}^{\top} \theta^{(t)}
        \\ & =
        \bm{c}^{\top} \overline{\theta}^{(t+1)} - \bm{c}^{\top} \theta^{(t)}
        \\ & = 
        - \bm{c}^{\top} \left[ \nabla^{2} \mathcal{L}^{(\leq t)}(\theta^{(t)})\right]^{\dagger}\nabla \ell_{t}
        (\theta^{(t)}) 
        \\ & =
        \sum_{1\le j<i\le n} \underbrace{G^{(t)}_{i,j} \bm{c}^{\top}\Big\lbrack \nabla^{2} \mathcal{L}^{(\le t)}(\theta^{(t)})\Big\rbrack^{\dagger}(\boldsymbol{e}_i - \boldsymbol{e}_j)}_{\text{fixed given $\mathcal{F}_{t-1}$}} \cdot 
        \underbrace{ \frac{1}{L}\sum_{l=1}^{L} \left\{y^{(t,l)}_{j,i} - \frac{e^{\theta^{(t)}_{i}}}{e^{\theta^{(t)}_{i}}+e^{\theta^{(t)}_{j}}}\right\}}_{\text{mean $0$ given $\mathcal{F}_{t-1}$}},
    \end{align*}
    where $\mathcal{F}_{t-1}:= \sigma \left\langle \{\mathcal{G}^{(s)}\}_{s=0}^{t},\{Y^{(s)}\}_{s=0}^{t-1} \right\rangle$ and $\mathcal{F}_{-1} = \sigma \left\langle \mathcal{G}^{(0)} \right\rangle$.

    We first show that $\lbrace X_{t}\rbrace_{t=0}^{T+1}$ is a martingale difference sequence w.r.t. $\lbrace \mathcal{F}_{t}\rbrace_{t=0}^{T+1}$ due to 
    \begin{align*}
        \mathbb{E} \left[ \nabla \ell_{t}(\theta^{(t)}) \mid \mathcal{F}_{t-1} \right] = \mathbb{E} \left[ \sum_{1 \leq j < i \leq n} G_{j,i}^{(t)} \left\{ - y_{j,i}^{(t)}+ \frac{e^{\theta^{(t)}_i}}{e^{\theta^{(t)}_i} + e^{\theta^{(t)}_j}} \right\} (\boldsymbol{e}_{i} - \boldsymbol{e}_{j}) \mid \mathcal{F}_{t-1} \right] = \bm{0},
    \end{align*}
    which implies
    \begin{align*}
        \mathbb{E} \left[ X_{t} \mid \mathcal{F}_{t-1} \right]
        = 
        - \bm{c}^{\top} \left[ \nabla^{2} \mathcal{L}^{(\leq t)}(\theta^{(t)}) \right]^{\dagger} \mathbb{E} \left[  \nabla \ell_{t}(\theta^{(t)}) \mid \mathcal{F}_{t-1} \right] = 0.
    \end{align*}
    By Lemma~\ref{lemma:largesteigen} and Lemma~\ref{lemma:smallesteigen}, we further have
    \begin{align*}
        & \mathbb{E}\left[ X_{t}^{2} \mid \mathcal{F}_{t-1} \right] 
        = 
        \text{Var} \left[ X_{t} \mid \mathcal{F}_{t-1} \right]
        \\ = & 
        \text{Var} \left[ \sum_{1 \leq j < i \leq n} G_{i,j}^{(t)} \left\{ y_{j,i}^{(t)} - \frac{e^{\theta^{(t)}_{i}}}{e^{\theta^{(t)}_{i}}+e^{\theta^{(t)}_{j}}} \right\} \bm{c}^{\top} \left[ \nabla^{2} \mathcal{L}^{(\leq t)}(\theta^{(t)}) \right]^{\dagger} (\boldsymbol{e}_i - \boldsymbol{e}_j) \mid \mathcal{F}_{t-1} \right]
        \\ = & 
        \sum_{1 \leq j < i \leq n} G_{i,j}^{(t)} \left\{ \bm{c}^{\top} \left[ \nabla^{2} \mathcal{L}^{(\leq t)}(\theta^{(t)}) \right]^{\dagger} (\boldsymbol{e}_i - \boldsymbol{e}_j) \right\}^{2} \text{Var} \left[ y_{j,i}^{(t)} - \frac{e^{\theta^{(t)}_{i}}}{e^{\theta^{(t)}_{i}}+e^{\theta^{(t)}_{j}}} \mid \mathcal{F}_{t-1} \right]
        \\ = &
        \frac{1}{L} \sum_{1 \leq j < i \leq n} G_{i,j}^{(t)} \left\{ \bm{c}^{\top} \left[ \nabla^{2} \mathcal{L}^{(\leq t)}(\theta^{(t)}) \right]^{\dagger} (\boldsymbol{e}_i - \boldsymbol{e}_j) \right\}^{2}\frac{e^{\theta_i^{(t)}}e^{\theta_j^{(t)}}}{(e^{\theta_i^{(t)}} + e^{\theta_j^{(t)}})^{2}}
        \\ = & 
        \frac{1}{L} \bm{c}^{\top} \left[ \nabla^{2} \mathcal{L}^{(\leq t)}(\theta^{(t)}) \right]^{\dagger} \left[\sum_{(i,j) \in \mathcal{E}^{(t)},\, i>j} \frac{e^{\theta_i^{(t)}}e^{\theta_j^{(t)}}}{(e^{\theta_i^{(t)}} + e^{\theta_j^{(t)}})^{2}}(\boldsymbol{e}_{i} - \boldsymbol{e}_{j})(\boldsymbol{e}_{i} - \boldsymbol{e}_{j})^{\top} \right] \left[ \nabla^{2} \mathcal{L}^{(\leq t)}(\theta^{(t)}) \right]^{\dagger} \bm{c}
        \\ = & 
        \frac{1}{L} \bm{c}^{\top} \left[ \nabla^{2} \mathcal{L}^{(\leq t)}(\theta^{(t)}) \right]^{\dagger}   \nabla^{2} \ell_{t}(\theta^{(t)})  \left[ \nabla^{2} \mathcal{L}^{(\leq t)}(\theta^{(t)}) \right]^{\dagger} \bm{c}
    \end{align*}
    and thus
    \begin{align*}
        \mathbb{E}[X_{t}^{2}] = \mathbb{E} \left\{  \mathbb{E} \left[ X_{t}^{2} \mid \mathcal{F}_{t-1} \right]\right\} = \mathbb{E} \left\{ \frac{1}{L} \bm{c}^{\top} \left[ \nabla^{2} \mathcal{L}^{(\leq t)}(\theta^{(t)}) \right]^{\dagger}   \nabla^{2} \ell_{t}(\theta^{(t)})  \left[ \nabla^{2} \mathcal{L}^{(\leq t)}(\theta^{(t)}) \right]^{\dagger} \bm{c} \right\}.
    \end{align*}

   Motivated by the martingale central limit theorem \citep{brown1971martingale}, we introduce a fine-grained martingale difference array $\left\{ X_{t,(i,j),\ell} : t = 0,\dots,T;\ (i,j)\in \mathcal{E}^{(t)},\, j < i;\ \ell \in [L] \right\}$ equipped with a deterministic ordering $\left\{ (i,j,\ell) : (i,j)\in \mathcal{E}^{(t)},\, j < i;\ \ell \in [L] \right\}$ for each fixed $t$. This ordering guarantees that the associated martingale difference sequence $\{ X_{t,(i,j),\ell} \}$ has length $k_{n,t} \lesssim (t+1)n^{2}pL$ for every $t = 0,\dots,T$, where
    \begin{align*}
        X_{t, (i,j), l} := \underbrace{G^{(t)}_{i,j} \bm{c}^{\top}\Big\lbrack \nabla^{2} \mathcal{L}^{(\le t)}(\theta^{(t)})\Big\rbrack^{\dagger}(\boldsymbol{e}_i - \boldsymbol{e}_j)}_{\text{fixed given $\mathcal{D}_{t-1, (i,j), l}$}} \cdot 
        \underbrace{ \frac{1}{L}\left\{y^{(t,l)}_{j,i} - \frac{e^{\theta^{(t)}_{i}}}{e^{\theta^{(t)}_{i}}+e^{\theta^{(t)}_{j}}}\right\}}_{\text{mean $0$ given $\mathcal{D}_{t-1, (i,j), l}$}}
    \end{align*}
    and
    \begin{align*}
        D_{t-1, (i,j), l} := \mathcal{F}_{t-1} \vee \sigma \left\langle y_{j^{\prime},i^{\prime}}^{(t, l^{\prime})}: (j^{\prime},i^{\prime},l^{\prime}) \leq (i,j,l) \right\rangle.
    \end{align*}
    The martingale CLT theorem then reduces our task to verifying the following two conditions:
    \begin{align*}
        \frac{1}{s_{n}^{2}}
        \sum_{t=0}^{T} \sum_{(i,j)\in\mathcal{E}^{(t)}, j<i} \sum_{l=1}^{L} \mathbb{E}\Bigl[ X_{t, (i,j), l}^{2}\,
        \boldsymbol{1}_{|X_{t, (i,j), l}| \geq \varepsilon s_{n}} \Bigr]
        \;\to\; 0, \forall \varepsilon>0
        \quad\text{and}\quad
        \frac{V_{n}^{2}}{s_{n}^{2}} \;\xrightarrow{p}\; 1,
    \end{align*}
    where 
    \begin{align*}
        S_{n} & := \sum_{t=0}^{T} \sum_{(i,j)\in\mathcal{E}^{(t)}, j<i} \sum_{l=1}^{L} X_{t, (i,j), l} = \sum_{t=0}^{T} X_{t} = \sum_{t=0}^{T} \left( \bar{\bm{c}}^{\top} \overline{\theta}^{(t+1)} - \bar{\bm{c}}^{\top} \theta^{(t)} \right), 
        \\ 
        V_{n}^{2} & := \sum_{t=0}^{T} \sum_{(i,j)\in\mathcal{E}^{(t)}, j<i} \sum_{l=1}^{L} \mathbb{E} \left[ X_{t, (i,j), l}^{2} \mid D_{t-1, (i,j), l}\right] = \sum_{t=0}^{T} \mathbb{E}\left[ X_{t}^{2} \mid \mathcal{F}_{t-1} \right],
        \\
        s_{n}^{2} & := \sum_{t=0}^{T} \mathbb{E} \left[ X_{t}^{2} \right].
    \end{align*}

    For the Lindeberg condition, we in fact establish a stronger result
    \begin{align*}
        \max_{t, (i,j), l} \frac{ \left| X_{t, (i,j), l} \right|}{s_{n}} \xrightarrow{p} 0
    \end{align*}
    due to
    \begin{align*}
        & \frac{1}{s_{n}^{2}}
        \sum_{t=0}^{T} \sum_{(i,j)\in\mathcal{E}^{(t)}, j<i} \sum_{l=1}^{L} \mathbb{E}\Bigl[ X_{t, (i,j), l}^{2}\,
        \boldsymbol{1}_{|X_{t, (i,j), l}| \geq \varepsilon s_{n}} \Bigr]
        \\ \leq &
        \frac{1}{s_{n}^{2}} \sum_{t=0}^{T} \sum_{(i,j)\in\mathcal{E}^{(t)}, j<i} \sum_{l=1}^{L} \mathbb{E} \left[ X_{t, (i,j), l}^{2} \boldsymbol{1}_{\max_{t, (i,j), l} \left| X_{t, (i,j), l} \right| \geq \varepsilon s_{n}} \right]
        \\ \leq &
        \frac{1}{s_{n}^{2}} \max_{t, (i,j), l} X_{t, (i,j), l}^{2}  \sum_{t=0}^{T} \sum_{(i,j)\in\mathcal{E}^{(t)}, j<i} \sum_{l=1}^{L} \mathbb{P} \left(\max_{t, (i,j), l} \left| X_{t, (i,j), l} \right| \geq \varepsilon s_{n}\right) 
    \end{align*}
    Combining with Lemma~\ref{lemma:smallesteigen}, we know
    \begin{align*}
        \left| X_{t, (i,j), l} \right| 
        & = 
        \left| G^{(t)}_{i,j} \bm{c}^{\top}\Big\lbrack \nabla^{2} \mathcal{L}^{(\le t)}(\theta^{(t)})\Big\rbrack^{\dagger}(\boldsymbol{e}_i - \boldsymbol{e}_j) \cdot \frac{1}{L}\left\{y^{(t,l)}_{j,i} - \frac{e^{\theta^{(t)}_{i}}}{e^{\theta^{(t)}_{i}}+e^{\theta^{(t)}_{j}}}\right\} \right|
        \\ & \leq
        \frac{1}{L} \norm{\bm{c}}_{2} \norm{\Big\lbrack \nabla^{2} \mathcal{L}^{(\le t)}(\theta^{(t)})\Big\rbrack^{\dagger}}_{2} \norm{\boldsymbol{e}_i - \boldsymbol{e}_j}_{2}
        \\ & \lesssim
        \frac{1}{t+1} \frac{\norm{\bm{c}}_{2}}{npL}.
    \end{align*}
    Consequently,
    \begin{align*}
        \max_{t, (i,j), l} \frac{ \left| X_{t, (i,j), l} \right|}{s_{n}} \lesssim \frac{\norm{\bm{c}}_{2}/npL}{\norm{\bm{c}}_{2}/\sqrt{npL}} = \frac{1}{\sqrt{npL}} \to 0.
    \end{align*}

    For the normalized quadratic variation of $X_{t, (i,j), l}$, it is equivalent to prove 
    \begin{align*}
        \frac{\left| V_{n}^{2} - s_{n}^{2} \right|}{s_{n}^{2}} \leq \underbrace{\frac{\left| V_{n}^{2} - (V_{n}^{\ast})^{2} \right|}{s_{n}^{2}}}_{u_{1}} + \underbrace{\frac{\left|(V_{n}^{\ast})^{2} - \mathbb{E} \left[(V_{n}^{\ast})^{2}\right] \right|}{s_{n}^{2}}}_{u_{2}} + \underbrace{\frac{\left|s_{n}^{2} - \mathbb{E} \left[(V_{n}^{\ast})^{2}\right] \right|}{s_{n}^{2}}}_{u_{3}} \xrightarrow{p} 0,
    \end{align*}
    where
    \begin{align*}
        (V_{n}^{\ast})^{2} := \frac{1}{L} \sum_{t=0}^{T} \bm{c}^{\top} \left[ \nabla^{2} \mathcal{L}^{(\leq t)}(\theta^{\ast}) \right]^{\dagger}   \nabla^{2} \ell_{t}(\theta^{\ast})  \left[ \nabla^{2} \mathcal{L}^{(\leq t)}(\theta^{\ast}) \right]^{\dagger} \bm{c}.
    \end{align*}
    For term $u_{1}$, we first notice the following decomposition
    \begin{align*}
        & \left[ \nabla^{2} \mathcal{L}^{(\leq t)}(\theta^{(t)}) \right]^{\dagger}   \nabla^{2} \ell_{t}(\theta^{(t)}) \left[ \nabla^{2} \mathcal{L}^{(\leq t)}(\theta^{(t)}) \right]^{\dagger} - \left[ \nabla^{2} \mathcal{L}^{(\leq t)}(\theta^{\ast}) \right]^{\dagger}   \nabla^{2} \ell_{t}(\theta^{\ast})  \left[ \nabla^{2} \mathcal{L}^{(\leq t)}(\theta^{\ast}) \right]^{\dagger}
        \\ = & 
        \left\{ \left[ \nabla^{2} \mathcal{L}^{(\leq t)}(\theta^{(t)}) \right]^{\dagger} - \left[ \nabla^{2} \mathcal{L}^{(\leq t)}(\theta^{\ast}) \right]^{\dagger} \right\} \nabla^{2} \ell_{t}(\theta^{(t)}) \left[ \nabla^{2} \mathcal{L}^{(\leq t)}(\theta^{(t)}) \right]^{\dagger}
        \\ & +
        \left[ \nabla^{2} \mathcal{L}^{(\leq t)}(\theta^{\ast}) \right]^{\dagger} \left\{ \nabla^{2} \ell_{t}(\theta^{(t)}) - \nabla^{2} \ell_{t}(\theta^{\ast}) \right\} \left[ \nabla^{2} \mathcal{L}^{(\leq t)}(\theta^{(t)}) \right]^{\dagger}
        \\ & + 
        \left[ \nabla^{2} \mathcal{L}^{(\leq t)}(\theta^{\ast}) \right]^{\dagger} \nabla^{2} \ell_{t}(\theta^{\ast}) \left\{ \left[ \nabla^{2} \mathcal{L}^{(\leq t)}(\theta^{(t)}) \right]^{\dagger} - \left[ \nabla^{2} \mathcal{L}^{(\leq t)}(\theta^{\ast}) \right]^{\dagger} \right\}.
    \end{align*}
    Combining Lemma~\ref{lemma:pseudoinverse_diff_L2} with Lemma~\ref{lemma:hessian_dif_L2}, we obtian
    \begin{align*}
        \norm{\nabla^{2} \ell_{t}(\theta^{(t)}) - \nabla^{2} \ell_{t}(\theta^{\ast})}_{2} 
        & \lesssim 
        np \norm{\theta^{(t)} - \theta^{\ast}}_{\infty},
        \\
        \norm{\left[ \nabla^{2} \mathcal{L}^{(\leq t)}(\theta^{(t)}) \right]^{\dagger} - \left[ \nabla^{2} \mathcal{L}^{(\leq t)}(\theta^{\ast}) \right]^{\dagger}}_{2} & \lesssim \frac{\norm{\theta^{(t)} - \theta^{\ast}}_{\infty}}{(t+1) np},
    \end{align*}
    which implies 
    \begin{align*}
        & \left| V_{n}^{2} - (V_{n}^{\ast})^{2} \right|
        \\ = & 
        \left| \sum_{t=0}^{T} \frac{1}{L} \bm{c}^{\top} \left\{ \left[ \nabla^{2} \mathcal{L}^{(\leq t)}(\theta^{(t)}) \right]^{\dagger}   \nabla^{2} \ell_{t}(\theta^{(t)})  \left[ \nabla^{2} \mathcal{L}^{(\leq t)}(\theta^{(t)}) \right]^{\dagger} -\left[ \nabla^{2} \mathcal{L}^{(\leq t)}(\theta^{\ast}) \right]^{\dagger}   \nabla^{2} \ell_{t}(\theta^{\ast})  \left[ \nabla^{2} \mathcal{L}^{(\leq t)}(\theta^{\ast}) \right]^{\dagger} \right\} \bm{c} \right|
        \\ \lesssim &
        \sum_{t=0}^{T} \frac{1}{L} \norm{\bm{c}}_{2}^{2} \norm{\left[ \nabla^{2} \mathcal{L}^{(\leq t)}(\theta^{(t)}) \right]^{\dagger}   \nabla^{2} \ell_{t}(\theta^{(t)})  \left[ \nabla^{2} \mathcal{L}^{(\leq t)}(\theta^{(t)}) \right]^{\dagger} -\left[ \nabla^{2} \mathcal{L}^{(\leq t)}(\theta^{\ast}) \right]^{\dagger}   \nabla^{2} \ell_{t}(\theta^{\ast})  \left[ \nabla^{2} \mathcal{L}^{(\leq t)}(\theta^{\ast}) \right]^{\dagger}}_{2}
        \\ \lesssim & 
        \sum_{t=0}^{T} \frac{1}{L} \norm{\bm{c}}_{2}^{2}
        \norm{\left[ \nabla^{2} \mathcal{L}^{(\leq t)}(\theta^{(t)}) \right]^{\dagger} - \left[ \nabla^{2} \mathcal{L}^{(\leq t)}(\theta^{\ast}) \right]^{\dagger}}_{2} \norm{\nabla^{2} \ell_{t}(\theta^{(t)})}_{2} \norm{\left[ \nabla^{2} \mathcal{L}^{(\leq t)}(\theta^{(t)}) \right]^{\dagger}}_{2}
        \\ & + 
        \sum_{t=0}^{T} \frac{1}{L} \norm{\bm{c}}_{2}^{2} \norm{\left[ \nabla^{2} \mathcal{L}^{(\leq t)}(\theta^{\ast}) \right]^{\dagger}}_{2} \norm{\nabla^{2} \ell_{t}(\theta^{(t)}) - \nabla^{2} \ell_{t}(\theta^{\ast})}_{2} \norm{\left[ \nabla^{2} \mathcal{L}^{(\leq t)}(\theta^{(t)}) \right]^{\dagger}}_{2}
        \\ & + 
        \sum_{t=0}^{T} \frac{1}{L} \norm{\bm{c}}_{2}^{2} \norm{\left[ \nabla^{2} \mathcal{L}^{(\leq t)}(\theta^{\ast}) \right]^{\dagger}}_{2} \norm{\nabla^{2} \ell_{t}(\theta^{\ast})}_{2} \norm{\left[ \nabla^{2} \mathcal{L}^{(\leq t)}(\theta^{(t)}) \right]^{\dagger} - \left[ \nabla^{2} \mathcal{L}^{(\leq t)}(\theta^{\ast}) \right]^{\dagger}}_{2}
        \\ \lesssim & 
        \sum_{t=0}^{T} \norm{\bm{c}}_{2}^{2} \frac{1}{(t+1)^{2}} \frac{1}{npL} \norm{\theta^{(t)}-\theta^{\ast}}_{2}
        \\ \lesssim & 
        \frac{\norm{\bm{c}}_{2}^{2}}{npL} \sqrt{\frac{\log n}{npL}}.
    \end{align*}

    As a result, with the help of theorem~\ref{thm:iterative_loss}, we have
    \begin{align*}
        u_{1} := \frac{\left| V_{n}^{2} - (V_{n}^{\ast})^{2} \right|}{s_{n}^{2}} \lesssim \sqrt{\frac{\log n}{npL}} \to 0.
    \end{align*}

    For the second term $u_{2}$, we define
    \begin{align*}
        (\overline{V}_{n}^{\ast})^{2} := \frac{1}{L} \sum_{t=0}^{T} \bm{c}^{\top} \mathbb{E} \left\{ \left[ \nabla^{2} \mathcal{L}^{(\leq t)}(\theta^{\ast}) \right]^{\dagger} \right\}  \mathbb{E} \left[ \nabla^{2} \ell_{t}(\theta^{\ast}) \right] \mathbb{E} \left\{ \left[ \nabla^{2} \mathcal{L}^{(\leq t)}(\theta^{\ast}) \right]^{\dagger} \right\} \bm{c}.
    \end{align*}
    Then it suffices to show
    \begin{align*}
        \frac{\left| (V_{n}^{\ast})^{2} -(\overline{V}_{n}^{\ast})^{2} \right|}{s_{n}^{2}} \xrightarrow{p} 0
    \end{align*}
    due to 
    \begin{align*}
        \left| (V_{n}^{\ast})^{2} -\mathbb{E} \left[ (V_{n}^{\ast})^{2} \right] \right| 
        & \leq 
        \left| (V_{n}^{\ast})^{2} -(\overline{V}_{n}^{\ast})^{2} \right| + \left| \mathbb{E} \left[ (V_{n}^{\ast})^{2} \right] - (\overline{V}_{n}^{\ast})^{2} \right|
        \\ & \leq
        \left| (V_{n}^{\ast})^{2} -(\overline{V}_{n}^{\ast})^{2} \right| + \mathbb{E}  \left| (V_{n}^{\ast})^{2} - (\overline{V}_{n}^{\ast})^{2} \right|.
    \end{align*}
    Let $\nabla^{2}\ell_{t}(\theta^{\ast}) = \sum_{i > j} Z_{j,i}$, where $Z_{j,i} := (G_{j,i}^{(t)} - p) \frac{e^{\theta_i^{\ast}}e^{\theta_j^{\ast}}}{(e^{\theta_i^{\ast}} + e^{\theta_j^{\ast}})^{2}}(\boldsymbol{e}_{i} - \boldsymbol{e}_{j})(\boldsymbol{e}_{i} - \boldsymbol{e}_{j})^{\top}$, then $\left\{ Z_{j,i} \right\}$ are in dependent with mean zero. Since $\norm{Z_{j,i}}_{2} \leq \frac{1}{2}$ and $\norm{\sum_{i > j} \mathbb{E} \left[ Z_{j,i}^{2} \right]}_{2} \lesssim p(1-p) \norm{n\boldsymbol{I}_{n} - \boldsymbol{1}\boldsymbol{1}^{\top}}_{2} \lesssim np.$
    By matrix Bernstein inequality, we have
    \begin{align*}
        \norm{\nabla^{2} \ell_{t}(\theta^{\ast}) - \mathbb{E} \left[ \nabla^{2} \ell_{t}(\theta^{\ast}) \right]}_{2} \lesssim \sqrt{np \log n} + \log n \lesssim \sqrt{np \log n}
    \end{align*}
    with probability at least $1 - \mathcal{O}(n^{-10})$.

    Similarly in the case of $u_1$, we notice that
    \begin{align*}
        & \left[ \nabla^{2} \mathcal{L}^{(\leq t)}(\theta^{\ast}) \right]^{\dagger}   \nabla^{2} \ell_{t}(\theta^{\ast})  \left[ \nabla^{2} \mathcal{L}^{(\leq t)}(\theta^{\ast}) \right]^{\dagger} - \mathbb{E} \left\{ \left[ \nabla^{2} \mathcal{L}^{(\leq t)}(\theta^{\ast}) \right]^{\dagger} \right\}  \mathbb{E} \left[ \nabla^{2} \ell_{t}(\theta^{\ast}) \right] \mathbb{E} \left\{ \left[ \nabla^{2} \mathcal{L}^{(\leq t)}(\theta^{\ast}) \right]^{\dagger} \right\}
        \\ = & 
        \left( \left[ \nabla^{2} \mathcal{L}^{(\leq t)}(\theta^{\ast}) \right]^{\dagger} - \mathbb{E} \left\{ \left[ \nabla^{2} \mathcal{L}^{(\leq t)}(\theta^{\ast}) \right]^{\dagger} \right\} \right) \nabla^{2} \ell_{t}(\theta^{\ast}) \left[ \nabla^{2} \mathcal{L}^{(\leq t)}(\theta^{\ast}) \right]^{\dagger}
        \\ & + 
        \mathbb{E} \left\{ \left[ \nabla^{2} \mathcal{L}^{(\leq t)}(\theta^{\ast}) \right]^{\dagger} \right\} \left( \nabla^{2} \ell_{t}(\theta^{\ast}) - \mathbb{E} \left[ \nabla^{2} \ell_{t}(\theta^{\ast}) \right] \right) \left[ \nabla^{2} \mathcal{L}^{(\leq t)}(\theta^{\ast}) \right]^{\dagger}
        \\ & + 
        \mathbb{E} \left\{ \left[ \nabla^{2} \mathcal{L}^{(\leq t)}(\theta^{\ast}) \right]^{\dagger} \right\}  \mathbb{E} \left[ \nabla^{2} \ell_{t}(\theta^{\ast}) \right] \left( \left[ \nabla^{2} \mathcal{L}^{(\leq t)}(\theta^{\ast}) \right]^{\dagger} - \mathbb{E} \left\{ \left[ \nabla^{2} \mathcal{L}^{(\leq t)}(\theta^{\ast}) \right]^{\dagger} \right\} \right),
    \end{align*}
    which implies
    \begin{align*}
        & \left| (V_{n}^{\ast})^{2} -(\overline{V}_{n}^{\ast})^{2} \right|
        \\ \lesssim & 
        \sum_{t=0}^{T} \frac{1}{L} \norm{\bm{c}}_{2}^{2} \norm{\left[ \nabla^{2} \mathcal{L}^{(\leq t)}(\theta^{\ast}) \right]^{\dagger} - \mathbb{E} \left\{ \left[ \nabla^{2} \mathcal{L}^{(\leq t)}(\theta^{\ast}) \right]^{\dagger} \right\}}_{2} \norm{\nabla^{2} \ell_{t}(\theta^{\ast})}_{2}  \norm{\left[ \nabla^{2} \mathcal{L}^{(\leq t)}(\theta^{\ast}) \right]^{\dagger}}_{2}
        \\ & + 
        \sum_{t=0}^{T} \frac{1}{L} \norm{\bm{c}}_{2}^{2} \norm{\mathbb{E} \left\{ \left[ \nabla^{2} \mathcal{L}^{(\leq t)}(\theta^{\ast}) \right]^{\dagger} \right\}}_{2} \norm{\left( \nabla^{2} \ell_{t}(\theta^{\ast}) - \mathbb{E} \left[ \nabla^{2} \ell_{t}(\theta^{\ast}) \right] \right)}_{2} \norm{\left[ \nabla^{2} \mathcal{L}^{(\leq t)}(\theta^{\ast}) \right]^{\dagger}}_{2}
        \\ & + 
        \sum_{t=0}^{T} \frac{1}{L} \norm{\bm{c}}_{2}^{2} \norm{\mathbb{E} \left\{ \left[ \nabla^{2} \mathcal{L}^{(\leq t)}(\theta^{\ast}) \right]^{\dagger} \right\}}_{2} \norm{\mathbb{E} \left[ \nabla^{2} \ell_{t}(\theta^{\ast}) \right]}_{2} \norm{\left( \left[ \nabla^{2} \mathcal{L}^{(\leq t)}(\theta^{\ast}) \right]^{\dagger} - \mathbb{E} \left\{ \left[ \nabla^{2} \mathcal{L}^{(\leq t)}(\theta^{\ast}) \right]^{\dagger} \right\} \right)}_{2} 
        \\ \lesssim &
        \frac{\norm{\bm{c}}_{2}^{2}}{npL} \sqrt{\frac{\log n}{np}}.
    \end{align*}
    Therefore,
    \begin{align*}
        \frac{\left| (V_{n}^{\ast})^{2} -(\overline{V}_{n}^{\ast})^{2} \right|}{s_{n}^{2}} \lesssim \sqrt{\frac{\log n}{np}} \to 0.
    \end{align*}
    As a result, we derive
    \begin{align*}
        u_{2} := \frac{\left| (V_{n}^{\ast})^{2} -(\overline{V}_{n}^{\ast})^{2} \right|}{s_{n}^{2}} \xrightarrow{p} 0.
    \end{align*}
    
    We then turn to the third term $u_{3}$. This is immediate from the same bound as $u_{1}$ by taking expectations
    \begin{align*}
        u_{3} := \frac{\left|s_{n}^{2} - \mathbb{E} \left[(V_{n}^{\ast})^{2}\right] \right|}{s_{n}^{2}} \leq \frac{\mathbb{E} \left| V_{n}^{2} - (V_{n}^{\ast})^{2} \right|}{s_{n}^{2}} \leq \sqrt{\frac{\log n}{npL}} \to 0.
    \end{align*}
    
    To summarize, as long as $np \gg \log n$ we have
    \begin{align*}
        \sup_{x \in \mathbb{R}} \left| \mathbb{P} \left( \frac{\sqrt{L} \sum_{t=0}^{T} \left( \bm{c}^{\top} \overline{\theta}^{(t+1)} - \bm{c}^{\top} \theta^{(t)} \right)}{\sqrt{\sum_{t=0}^{T}\bm{c}^{\top} \left[ \nabla^{2} \mathcal{L}^{(\leq t)}(\theta^{(t)}) \right]^{\dagger}   \nabla^{2} \ell_{t}(\theta^{(t)})  \left[ \nabla^{2} \mathcal{L}^{(\leq t)}(\theta^{(t)}) \right]^{\dagger} \bm{c}}} \leq x \right) - \mathbb{P} \left( \mathcal{N}(0,1) \leq x \right) \right| = r_{n}
    \end{align*}
    with $r_{n} := o(1)$.
    
    On the other hand, by Lemma \ref{lemma:diff_ori_qudra_t+1} we have
    \begin{align*}
        & \left| \frac{\sqrt{L} \sum_{t=0}^{T}\left(\bm{c}^{\top} \theta^{(t+1)} - \bm{c}^{\top} \overline{\theta}^{(t+1)}\right)}{\sqrt{\sum_{t=0}^{T}\bm{c}^{\top} \left[ \nabla^{2} \mathcal{L}^{(\leq t)}(\theta^{(t)}) \right]^{\dagger} \nabla^{2} \ell_{t}(\theta^{(t)}) \left[ \nabla^{2} \mathcal{L}^{(\leq t)}(\theta^{(t)}) \right]^{\dagger} \bm{c}}} \right| 
        \\ \lesssim & 
        \norm{\bm{c}}_{1} \sum_{t=0}^{T} \frac{1}{t+1} \sqrt{\frac{\log n}{npL}} \left[ \frac{\sqrt{\log n}}{\sqrt{nL}p}  + \frac{1}{\sqrt{np}} \left( 1 + \frac{\log n}{\sqrt{np}} \right) \right] \bigg/
        \sqrt{\sum_{t=0}^{T} \norm{\bm{c}}_{2}^{2} \frac{1}{(t+1)^{2}} \frac{1}{npL}} 
        \\ \lesssim & 
        \log T \left[ \frac{\log n}{\sqrt{nL}p} + \sqrt{\frac{\log n}{np}} \left( 1 + \frac{\log n}{\sqrt{np}} \right) \right] \frac{\norm{\bm{c}}_{1}}{\norm{\bm{c}}_{2}}
    \end{align*}
    with probability exceeding $1 - \mathcal{O}(n^{-5})$.

    For simplicity we denote by $\Gamma = \log T \left[ \frac{\log n}{\sqrt{nL}p} + \sqrt{\frac{\log n}{np}} \left( 1 + \frac{\log n}{\sqrt{np}} \right) \right] \frac{\norm{\bm{c}}_{1}}{\norm{\bm{c}}_{2}}$. Consider event $A = \left\{ \left| \frac{\sqrt{L} \sum_{t=0}^{T}\left(\bm{c}^{\top} \theta^{(t+1)} - \bm{c}^{\top} \overline{\theta}^{(t+1)}\right)}{\sqrt{\sum_{t=0}^{T}\bm{c}^{\top} \left[ \nabla^{2} \mathcal{L}^{(\leq t)}(\theta^{(t)}) \right]^{\dagger} \nabla^{2} \ell_{t}(\theta^{(t)}) \left[ \nabla^{2} \mathcal{L}^{(\leq t)}(\theta^{(t)}) \right]^{\dagger} \bm{c}}} \right|  \leq \Lambda \Gamma \right\}$, where $\Lambda > 0$ is some constant such that $\mathbb{P}(A^{c}) = \mathcal{O}(n^{-5})$. Then we consider the following three events
    \begin{align*}
        B_{1} & = \left\{ \left| \frac{\sqrt{L} \left(\bm{c}^{\top} \theta^{(T+1)} - \bm{c}^{\top} \theta^{\ast}\right)}{\sqrt{\sum_{t=0}^{T} \bm{c}^{\top} \left[ \nabla^{2} \mathcal{L}^{(\leq t)}(\theta^{(t)}) \right]^{\dagger} \nabla^{2} \ell_{t}(\theta^{(t)}) \left[ \nabla^{2} \mathcal{L}^{(\leq t)}(\theta^{(t)}) \right]^{\dagger} \bm{c}}} \right| \leq x \right\},
        \\ 
        B_{2} & = \left\{ \left| \frac{\sqrt{L} \sum_{t=0}^{T} \left(\bm{c}^{\top} \overline{\theta}^{(t+1)} - \bm{c}^{\top} \theta^{(t)}\right)}{\sqrt{\sum_{t=0}^{T} \bm{c}^{\top} \left[ \nabla^{2} \mathcal{L}^{(\leq t)}(\theta^{(t)}) \right]^{\dagger} \nabla^{2} \ell_{t}(\theta^{(t)}) \left[ \nabla^{2} \mathcal{L}^{(\leq t)}(\theta^{(t)}) \right]^{\dagger} \bm{c}}} \right| \leq x - \Lambda \Gamma \right\},
        \\ 
        B_{3} & = \left\{ \left| \frac{\sqrt{L} \sum_{t=0}^{T} \left(\bm{c}^{\top} \overline{\theta}^{(t+1)} - \bm{c}^{\top} \theta^{(t)}\right)}{\sqrt{\sum_{t=0}^{T} \bm{c}^{\top} \left[ \nabla^{2} \mathcal{L}^{(\leq t)}(\theta^{(t)}) \right]^{\dagger} \nabla^{2} \ell_{t}(\theta^{(t)}) \left[ \nabla^{2} \mathcal{L}^{(\leq t)}(\theta^{(t)}) \right]^{\dagger} \bm{c}}} \right| \leq x + \Lambda \Gamma \right\}.
    \end{align*}
    Then we have 
    \begin{align*}
        & \left| \mathbb{P} \left( \frac{\sqrt{L} \left( \bm{c}^{\top} \theta^{(T+1)} - \bm{c}^{\top} \theta^{\ast} \right)}{\sqrt{\sum_{t=0}^{T} \bm{c}^{\top} \left[ \nabla^{2} \mathcal{L}^{(\leq t)}(\theta^{(t)}) \right]^{\dagger} \nabla^{2} \ell_{t}(\theta^{(t)}) \left[ \nabla^{2} \mathcal{L}^{(\leq t)}(\theta^{(t)}) \right]^{\dagger} \bm{c}}} \leq x \right) - \mathbb{P} \left( \mathcal{N}(0,1) \leq x \right) \right|
        \\ = &
        \left| \mathbb{P}(B_{1} \cap A) - \mathbb{P}(B_{1} \cap A^{c}) + \mathbb{P} \left( \mathcal{N}(0,1) \leq x \right)\right|
        \\ \lesssim & 
        \left| \mathbb{P}(B_{1} \cap A) + \mathbb{P} \left( \mathcal{N}(0,1) \leq x \right)\right| + \mathbb{P}(A^{c}).
    \end{align*}
    Observe that $B_{2} \cap A \subset B_{1} \cap A \subset B_{3} \cap A$, we derive
    \begin{align*}
        & \left| \mathbb{P} \left( \frac{\sqrt{L} \left( \bm{c}^{\top} \theta^{(T+1)} - \bm{c}^{\top} \theta^{\ast} \right)}{\sqrt{\sum_{t=0}^{T} \bm{c}^{\top} \left[ \nabla^{2} \mathcal{L}^{(\leq t)}(\theta^{(t)}) \right]^{\dagger} \nabla^{2} \ell_{t}(\theta^{(t)}) \left[ \nabla^{2} \mathcal{L}^{(\leq t)}(\theta^{(t)}) \right]^{\dagger} \bm{c}}} \leq x \right) - \mathbb{P} \left( \mathcal{N}(0,1) \leq x \right) \right|
        \\ = &
        \max \left\{ \left| \mathbb{P}(B_{3}) - \mathbb{P}(\mathcal{N}(0,1) \leq x) \right|, \left| \mathbb{P}(B_{2} \cap A) - \mathbb{P}(\mathcal{N}(0,1) \leq x) \right| \right\} + \mathbb{P}(A^{c})
        \\ \lesssim &
        \max \left\{ r_{n} + \Gamma, r_{n} + \Gamma + n^{-5} \right\} + n^{-5}
        \\ \lesssim &
        r_{n} + \Gamma + n^{-5}.
    \end{align*}
    Since the above inequality holds for every $x \in \mathbb{R}$, we prove the desired result.

    Thus, we finally conclude our proof of Theorem~\ref{thm:asy_norm}.
\end{proof}

\section{Proof of Auxiliary Lemmas in Appendix~\ref{appendix:proof_outline}}
\label{appendix:last}
In this Appendix, we prove detailed proof of aforementioned building blocks.
\subsection{Proof of Auxiliary Lemmas in Appendix~\ref{appendix:equation_A}}
\subsubsection{Proof of Lemma~\ref{lemma:gradient_t_diff_LOO}}
\label{appendix:gradient_t_diff_LOO}
\begin{proof}
    Observe that
    \begin{align*}
        & \nabla \ell_{t}(\theta^{(\tau, m)}) - \nabla \ell_{t}^{(m)}(\theta^{(\tau, m)})
        \\ = & 
        \sum_{(i,j) \in \mathcal{E}_t,\, i>j} \left\{ - y_{j,i}^{(t)}+ \frac{e^{\theta_i^{(\tau, m)}}}{e^{\theta_i^{(\tau, m)}} + e^{\theta_j^{(\tau, m)}}} \right\} (\boldsymbol{e}_{i} - \boldsymbol{e}_{j}) 
        \\ & - 
        \sum_{(i,j) \in \mathcal{E}_{t},i > j, i \neq m, j \neq m} \left\{ - y_{j,i}^{(t)}+ \frac{e^{\theta_i^{(\tau, m)}}}{e^{\theta_i^{(\tau, m)}} + e^{\theta_j^{(\tau, m)}}} \right\} (\boldsymbol{e}_{i} - \boldsymbol{e}_{j}) 
        \\ & -
        \sum_{i \in [n] \backslash \lbrace m \rbrace} p \left\{ - \frac{e^{\theta_{i}^{(t)}}}{e^{\theta_{i}^{(t)}}+e^{\theta_{m}^{(t)}}} + \frac{e^{\theta_{i}^{(\tau, m)}}}{e^{\theta_{i}^{(\tau, m)}}+e^{\theta_{m}^{(\tau, m)}}} \right\} (\boldsymbol{e}_{i} - \boldsymbol{e}_{m}) 
        \\ = &
        \sum_{i \in [n] \backslash \{m\}} G_{m,i}^{(t)} \left\{ - y_{m,i}^{(t)} + \frac{e^{\theta_{i}^{(\tau, m)}}}{e^{\theta_{i}^{(\tau, m)}}+e^{\theta_{m}^{(\tau, m)}}} \right\} (\boldsymbol{e}_{i} - \boldsymbol{e}_{m})
        \\ & -
        \sum_{i \in [n] \backslash \lbrace m \rbrace} p \left\{ - \frac{e^{\theta_{i}^{(t)}}}{e^{\theta_{i}^{(t)}}+e^{\theta_{m}^{(t)}}} + \frac{e^{\theta_{i}^{(\tau, m)}}}{e^{\theta_{i}^{(\tau, m)}}+e^{\theta_{m}^{(\tau, m)}}} \right\} (\boldsymbol{e}_{i} - \boldsymbol{e}_{m})
        \\ = &
        \underbrace{\frac{1}{L} \sum_{i:(i,m)\in\mathcal{E}_{t}} \sum_{l=1}^{L} \left\{ -y_{m,i}^{(t,l)} + \frac{e^{\theta_{i}^{(t)}}}{e^{\theta_{i}^{(t)}}+e^{\theta_{m}^{(t)}}} \right\} (\boldsymbol{e}_{i} - \boldsymbol{e}_{m})}_{:= \boldsymbol{u}^{m}}
        \\ & +
        \underbrace{\sum_{i \in [n] \backslash \{ m \}} \left( G_{m,i}^{(t)} - p \right) \left\{ -\frac{e^{\theta_{i}^{(t)}}}{e^{\theta_{i}^{(t)}}+e^{\theta_{m}^{(t)}}} + \frac{e^{\theta_{i}^{(\tau, m)}}}{e^{\theta_{i}^{(\tau, m)}}+e^{\theta_{m}^{(\tau, m)}}} \right\}  (\boldsymbol{e}_{i} - \boldsymbol{e}_{m})}_{:= \boldsymbol{v}^{m}},
    \end{align*}
    we are going to control $\boldsymbol{u}^{m}$ and $\boldsymbol{v}^{m}$ seperately.

    Keep $i$ for the edge neighbor of $m$ and use $k \in \{ 1, \cdots, n \}$ only as the coordinate index. Since $\boldsymbol{e}_{k}^{\top}(\boldsymbol{e}_{i} - \boldsymbol{e}_{m}) = \boldsymbol{1}_{k=i}-\boldsymbol{1}_{k=m}$, we have
    \begin{align*}
        u_{k}^{m}
        &=
        \boldsymbol{e}_{k}^{\top} \boldsymbol{u}^{m}
        = 
        \sum_{i: (i,m)\in\mathcal{E}_{t}} \frac{1}{L} \sum_{l=1}^{L} \left\{ -y_{m,i}^{(t,l)} + \frac{e^{\theta_{i}^{(t)}}}{e^{\theta_{i}^{(t)}}+e^{\theta_{m}^{(t)}}} \right\} \left( \boldsymbol{1}_{k=i}-\boldsymbol{1}_{k=m} \right),
        \\
        v_{k}^{m}
        &=
        \boldsymbol{e}_{k}^{\top} \boldsymbol{v}^{m}
        =
        \sum_{i \in [n] \backslash \{ m \}} \left( G_{m,i}^{(t)} - p \right) \left\{ -\frac{e^{\theta_{i}^{(t)}}}{e^{\theta_{i}^{(t)}}+e^{\theta_{m}^{(t)}}} + \frac{e^{\theta_{i}^{(\tau, m)}}}{e^{\theta_{i}^{(\tau, m)}}+e^{\theta_{m}^{(\tau, m)}}} \right\} \left( \boldsymbol{1}_{k=i}-\boldsymbol{1}_{k=m} \right),
    \end{align*}
    For the first term $\boldsymbol{u}^{m}$, we note that
    \begin{align*}
        u_{k}^{m} = 
        \begin{cases}
            \frac{1}{L} \sum_{l=1}^{L} \left\{ -y_{m,i}^{(t,l)} + \frac{e^{\theta_{i}^{(t)}}}{e^{\theta_{i}^{(t)}}+e^{\theta_{m}^{(t)}}} \right\} & (k,m) \in \mathcal{E}_{t} \Rightarrow k = i
            \\
            - \sum_{i: (i,m)\in\mathcal{E}_{t}} \frac{1}{L} \sum_{l=1}^{L} \left\{ -y_{m,i}^{(t,l)} + \frac{e^{\theta_{i}^{(t)}}}{e^{\theta_{i}^{(t)}}+e^{\theta_{m}^{(t)}}} \right\} & k=m
            \\
            0 & o.w.
        \end{cases}.
    \end{align*}
    Since $-y_{m,i}^{(t,l)} + \frac{e^{\theta_{i}^{(t)}}}{e^{\theta_{i}^{(t)}}+e^{\theta_{m}^{(t)}}}$ is $i.i.d.$ bounded by $[-1,1]$ with  $\mathbb{E} \left[ -y_{m,i}^{(t,l)} + \frac{e^{\theta_{i}^{(t)}}}{e^{\theta_{i}^{(t)}}+e^{\theta_{m}^{(t)}}} \right] = 0$ and $d_{m}^{(t)} \lesssim np$ by \ref{lemma:degree_concentration}, we can apply Hoeffding's inequality and union bounds to derive
    \begin{align*}
        \mathbb{P} \left( |u_{k}^{m}| \geq \delta \mid \mathcal{E}_{t} \right) 
        \leq 
        2e^{-L\delta^{2}/2} \text{ with } (k,m) \in \mathcal{E}_{t}
        \quad \text{and} \quad
        \mathbb{P} \left( |\mu_{m}^{m}| \geq \xi \mid \mathcal{E}_{t} \right) 
        \leq
        2e^{-L\xi^{2} / (2 d_{m}^{(t)})}.
    \end{align*}
    With $\delta = C \sqrt{\frac{\log n}{L}}$ and $\xi = C\sqrt{\frac{np \log n}{L}}$, union bound yields
    \begin{align*}
        |u_{k}^{m}| \lesssim \sqrt{\frac{\log n}{L}}  \text{ with } (k,m) \in \mathcal{E}_{t}
        \quad \text{and} \quad
        |u_{m}^{m}| \lesssim \sqrt{\frac{np\log n}{L}},
    \end{align*}
    which further implies
    \begin{align*}
        \norm{\boldsymbol{u}^{m}}_{2} \leq |u_{m}^{m}| + \sqrt{\sum_{i: (i,m)\in\mathcal{E}_{t}} |u_{i}^{m}|^{2}} \lesssim \sqrt{\frac{np \log n}{L}}, \quad \forall 1 \leq m \leq n.
    \end{align*}
    For the second term $\boldsymbol{v}^{m}$, this is a zero-mean random vector that satisfies
    \begin{align*}
        v_{k}^{m} =
        \begin{cases}
            \left\{ -\frac{e^{\theta_{i}^{(t)}}}{e^{\theta_{i}^{(t)}}+e^{\theta_{m}^{(t)}}} + \frac{e^{\theta_{i}^{(\tau, m)}}}{e^{\theta_{i}^{(\tau, m)}}+e^{\theta_{m}^{(\tau, m)}}} \right\}(1-p) & (k,m) \in \mathcal{E}_{t} \Rightarrow k = i
            \\
            - \sum_{i \in [n] \backslash \{ m \}} \left( G_{m,i}^{(t)} - p \right) \left\{ -\frac{e^{\theta_{i}^{(t)}}}{e^{\theta_{i}^{(t)}}+e^{\theta_{m}^{(t)}}} + \frac{e^{\theta_{i}^{(\tau, m)}}}{e^{\theta_{i}^{(\tau, m)}}+e^{\theta_{m}^{(\tau, m)}}} \right\} & k = m
            \\
            - \left\{ -\frac{e^{\theta_{i}^{(t)}}}{e^{\theta_{i}^{(t)}}+e^{\theta_{m}^{(t)}}} + \frac{e^{\theta_{i}^{(\tau, m)}}}{e^{\theta_{i}^{(\tau, m)}}+e^{\theta_{m}^{(\tau, m)}}} \right\}p & o.w.
        \end{cases}.
    \end{align*}

    Define $g(x) = \frac{1}{1+e^{x}}$ for $x \in \mathbb{R}$, we have $|g^{\prime}(x)| \leq 1$ and thus
    \begin{align*}
        \left| -\frac{e^{\theta_{i}^{(t)}}}{e^{\theta_{i}^{(t)}}+e^{\theta_{m}^{(t)}}} + \frac{e^{\theta_{i}^{(\tau, m)}}}{e^{\theta_{i}^{(\tau, m)}}+e^{\theta_{m}^{(\tau, m)}}} \right|
        &=
        \left| -\frac{1}{1+e^{\theta_{m}^{(t)}-\theta_{i}^{(t)}}} + \frac{1}{1+e^{\theta_{m}^{(\tau, m)}-\theta_{i}^{(\tau, m)}}} \right|
        \\ &=
        \left| g \left( \theta_{m}^{(\tau, m)}-\theta_{i}^{(\tau, m)} \right) - g \left( \theta_{m}^{(t)}-\theta_{i}^{(t)} \right) \right|
        \\ &\leq
        \left| \left( \theta_{m}^{(\tau, m)}-\theta_{i}^{(\tau, m)} \right) - \left( \theta_{m}^{(t)}-\theta_{i}^{(t)} \right) \right|
        \\ &\leq
        \left| \theta_{i}^{(t)} - \theta_{i}^{(\tau, m)} \right| + \left| \theta_{m}^{(t)} - \theta_{m}^{(\tau, m)} \right|,
    \end{align*}
    which indicates that
    \begin{align*}
        \left| -\frac{e^{\theta_{i}^{(t)}}}{e^{\theta_{i}^{(t)}}+e^{\theta_{m}^{(t)}}} + \frac{e^{\theta_{i}^{(\tau, m)}}}{e^{\theta_{i}^{(\tau, m)}}+e^{\theta_{m}^{(\tau, m)}}} \right|
        &\leq 2 
        \norm{\theta^{(\tau, m)}-\theta^{(t)}}_{\infty},
        \\
        \sum_{i=1}^{n}\left| -\frac{e^{\theta_{i}^{(t)}}}{e^{\theta_{i}^{(t)}}+e^{\theta_{m}^{(t)}}} + \frac{e^{\theta_{i}^{(\tau, m)}}}{e^{\theta_{i}^{(\tau, m)}}+e^{\theta_{m}^{(\tau, m)}}} \right|^{2}
        &\leq 4n 
        \norm{\theta^{(\tau, m)}-\theta^{(t)}}_{\infty}^{2}.
    \end{align*}
    Applying Bernstein inequality we obtain
    \begin{align*}
        |v_{m}^{m}| 
        &\lesssim 
        \sqrt{\left( p \sum_{i=1}^{n} \left| -\frac{e^{\theta_{i}^{(t)}}}{e^{\theta_{i}^{(t)}}+e^{\theta_{m}^{(t)}}} + \frac{e^{\theta_{i}^{(\tau, m)}}}{e^{\theta_{i}^{(\tau, m)}}+e^{\theta_{m}^{(\tau, m)}}} \right|^{2}\right) \log n} 
        \\ & \quad +  
        \max_{1 \leq i \leq n} \left| -\frac{e^{\theta_{i}^{(t)}}}{e^{\theta_{i}^{(t)}}+e^{\theta_{m}^{(t)}}} + \frac{e^{\theta_{i}^{(\tau, m)}}}{e^{\theta_{i}^{(\tau, m)}}+e^{\theta_{m}^{(\tau, m)}}} \right| \log n
        \\ &\lesssim
        \left( \sqrt{np \log n} + \log n \right)  \norm{\theta^{(\tau, m)}-\theta^{(t)}}_{\infty}.
    \end{align*}
    As a consequence, 
    \begin{align*}
        \norm{\boldsymbol{v}^{m}}_{2}
        & \leq 
        |v_{m}^{m}| + \sqrt{\sum_{i: (i,m)\in\mathcal{E}_{t}} |v_{i}^{m}|^{2}}  + \sqrt{\sum_{i: (i,m)\notin\mathcal{E}_{t} \text{ and } i \neq m} |v_{i}^{m}|^{2}}
        \\ & \leq
        \left( \sqrt{np \log n} + \log n \right) \norm{\theta^{(\tau, m)}-\theta^{(t)}}_{\infty} + 
        \\ & \quad + \norm{\theta^{(\tau, m)}-\theta^{(t)}}_{\infty} + 
        p \sqrt{n} \norm{\theta^{(\tau, m)}-\theta^{(t)}}_{\infty}
        \\ & \lesssim
        \sqrt{np \log n} \norm{\theta^{(\tau, m)}-\theta^{(t)}}_{\infty},
    \end{align*}
    as long as $np \gtrsim \log n$.

    Putting the above results together, we see that
    \begin{align*}
        \norm{\nabla \ell_{t}(\theta^{(\tau, m)}) - \nabla \ell_{t}^{(m)}(\theta^{(\tau, m)})}_{2} 
        & \leq 
        \norm{\boldsymbol{u}^{m}}_{2} + \norm{\boldsymbol{v}^{m}}_{2}
        \\ & \lesssim
        \sqrt{\frac{np \log n}{L}} + \sqrt{np \log n} \norm{\theta^{(\tau, m)}-\theta^{(t)}}_{\infty}.
    \end{align*}
\end{proof}

\subsection{Proof of Auxiliary Lemmas in Appendix~\ref{appendix:equation_D_plus_plus}}
\subsubsection{Proof of Lemma~\ref{lemma:inner_reg_bound}}
\label{appendix:inner_reg_bound}
\begin{proof}
With a standard optimization result for a strongly convex objective function based on the auxiliary regularized gradient descent, for any $t = 0, 1, \cdots, T$, we have
\begin{align*}
    \norm{\theta^{(\tau^{\ast})} - \theta^{(t+1)}_{\lambda}}_{2} \leq \left( 1 - \frac{\lambda_{t}}{\lambda_{t}+(t+1)np}\right)^{\tau^{\ast}} \norm{\theta^{(t+1)}_{\lambda} - \theta^{\ast}}_{2}.
\end{align*}
By triangle inequality, we have
\begin{align*}
    \norm{\theta^{(t+1)}_{\lambda} - \theta^{\ast}}_{\infty} & \leq \norm{\theta^{(\tau^{\ast})} - \theta^{(t+1)}_{\lambda}}_{2} + \norm{\theta^{(\tau^{\ast})} - \theta^{\ast}}_{\infty} \\ & \leq 
    \left( 1 - \frac{\lambda_{t}}{\lambda_{t}+(t+1)np}\right)^{\tau^{\ast}} \sqrt{n} \norm{\theta^{(t+1)}_{\lambda} - \theta^{\ast}}_{\infty} + 2.
\end{align*}
Since $1 - \frac{\lambda_{t}}{\lambda_{t}+(t+1)np} \leq 1 - \frac{1}{1+n^{2}}$, we can take $\tau^{\ast} = n^{3}$ in order that $\left( 1 - \frac{\lambda_{t}}{\lambda_{t}+(t+1)np}\right)^{\tau^{\ast}} \sqrt{n} \leq \frac{1}{2}$. This implies $\norm{\theta^{(t+1)}_{\lambda} - \theta^{\ast}}_{\infty} \leq 4$ with probability at least $1 - \mathcal{O}(n^{-7})$.
\end{proof}
\subsubsection{Proof of Lemma~\ref{lemma:inner_unreg_bound}}
\label{appendix:inner_unreg_bound}
\begin{proof}
Define a constraint unregularized MLE as
\begin{align} \label{eq:Constraint_MLE_1}
    \theta^{(t+1)}_{\operatorname{con}} := \underset{\boldsymbol{1}^{\top}\theta = 0; \norm{\theta - \theta^{\ast}}_{\infty} \leq 5} {\arg\min} \mathcal{L}^{(\leq t)}(\theta)
\end{align}
Therefore, $\theta^{(t+1)}_{\lambda}$ is feasible for the above optimization problem. We then have
\begin{align*}
    \mathcal{L}^{(\leq t)}(\theta^{(t+1)}_{\lambda}) \geq \mathcal{L}^{(\leq t)}(\theta^{(t+1)}_{\operatorname{con}}).
\end{align*}
We apply Taylor expansion, and obtain
\begin{align*}
    \mathcal{L}^{(\leq t)}(\theta^{(t+1)}_{\operatorname{con}}) &= 
    \mathcal{L}^{(\leq t)}(\theta^{(t+1)}_{\lambda}) + \left( \theta^{(t+1)}_{\operatorname{con}} - \theta^{(t+1)}_{\lambda} \right)^{\top} \nabla \mathcal{L}^{(\leq t)}(\theta^{(t+1)}_{\lambda})
    \\ & \quad + 
    \frac{1}{2} \left( \theta^{(t+1)}_{\operatorname{con}} - \theta^{(t+1)}_{\lambda} \right)^{\top} \nabla^{2} \mathcal{L}^{(\leq t)}(\theta^{(t+1)}_{\lambda}(\varsigma)) \left( \theta^{(t+1)}_{\operatorname{con}} - \theta^{(t+1)}_{\lambda} \right),
\end{align*}
where $\theta^{(t+1)}_{\lambda}(\varsigma)$ is a convex combination of $\theta^{(t+1)}_{\operatorname{con}}$ and $\theta^{(t+1)}_{\lambda}$. We can derive $\norm{\theta^{(t+1)}_{\lambda}(\varsigma) - \theta^{\ast}}_{\infty} \leq 5$ due to $\norm{\theta^{(t+1)}_{\lambda} - \theta^{\ast}}_{\infty} \leq 4$ and $\norm{\theta^{(t+1)}_{\operatorname{con}} - \theta^{\ast}}_{\infty} \leq 5$ by the definition of $\theta^{(t+1)}_{\operatorname{con}}$. By Lemma \ref{lemma:smallesteigen}, we get the lower bound
\begin{align*}
    \mathcal{L}^{(\leq t)}(\theta^{(t+1)}_{\operatorname{con}}) & \geq 
    \mathcal{L}^{(\leq t)}(\theta^{(t+1)}_{\lambda}) + \left( \theta^{(t+1)}_{\operatorname{con}} - \theta^{(t+1)}_{\lambda} \right)^{\top} \nabla \mathcal{L}^{(\leq T)}(\theta^{(t+1)}_{\lambda})
    \\ & \quad + 
    c_{6}(t+1)np \norm{\theta^{(t+1)}_{\operatorname{con}} - \theta^{(t+1)}_{\lambda}}_{2}^{2}
\end{align*}
for some constant $c_{6} > 0$.

Consider the facts that 
\begin{align*}
    \nabla \mathcal{L}^{(\leq t)}(\theta^{(t+1)}_{\lambda}) + \lambda_{t}\theta^{(t+1)}_{\lambda} = \nabla \mathcal{L}^{(\leq t)}_{\lambda}(\theta^{(t+1)}_{\lambda}) = 0
\end{align*} 
and
\begin{align*}
    \norm{\theta^{(t+1)}_{\lambda}}_{2} \leq \sqrt{n} \norm{\theta^{(t+1)}_{\lambda}}_{\infty} \leq \sqrt{n} \norm{\theta^{(t+1)}_{\lambda} - \theta^{\ast}}_{\infty} + \sqrt{n} \norm{\theta^{\ast}}_{\infty} \leq c_{7} \sqrt{n},
\end{align*}
then we derive 
\begin{align*}
    \norm{\theta^{(t+1)}_{\operatorname{con}} - \theta^{(t+1)}_{\lambda}}_{2}
    \leq
    \frac{\norm{\nabla \mathcal{L}^{(\leq t)}(\theta^{(t+1)}_{\lambda})}_{2}}{c_{6}(t+1)np}
    =
    \frac{\lambda_{t} \norm{\theta^{(t+1)}_{\lambda}}_{2}}{c_{6}(t+1)np}
    \leq 
    \frac{\lambda_{t}c_{7} \sqrt{n}}{c_{6}(t+1)np} 
    \leq 
    D_{1} \sqrt{\frac{1}{n}},
\end{align*}
for some constant $c_{7}, D_{1} > 0$.
In this way, since
\begin{align*}
    \norm{\theta^{(t+1)}_{\operatorname{con}}-\theta^{\ast}}_{\infty} \leq \norm{\theta^{(t+1)}_{\lambda}-\theta^{\ast} }_{\infty} + \norm{\theta^{(t+1)}_{\operatorname{con}}-\theta^{(t+1)}_{\lambda}}_{\infty} \leq 4 + D_{1} \sqrt{\frac{1}{n}} \leq \frac{9}{2},
\end{align*}
the minimizer of constrained unregularzied MLE problem~\eqref{eq:Constraint_MLE_1} is in the interior of the constraint. By convexity, we have $\theta^{(t+1)}_{\operatorname{con}} = \theta^{(t+1)}$. Thus with probability at least $1 - \mathcal{O}(n^{-7})$, we obtain
\begin{align*}
    \norm{\theta^{(t+1)}-\theta^{\ast}}_{\infty} &\leq 5
\end{align*}
for any $t = 0, 1, \cdots, T$.
\end{proof}

\subsection{Proof Auxiliary Lemmas in Appendix~\ref{appendix:equation_EFG}}
\subsubsection{Proof of Lemma~\ref{lemma:auxil_reg_bound}}
\label{appendix:auxil_reg_bound}
\begin{proof}
With a standard optimization result for a strongly convex objective function based on the auxiliary regularized gradient descent, we have
\begin{align*}
    \norm{\theta_{0}^{(\upsilon^{\ast})} - \theta^{(T+1)}_{0,\lambda}}_{2} 
    & \leq 
    \left( 1 - \frac{\lambda_{0}}{\lambda_{0}+(T+1)np}\right)^{\upsilon^{\ast}} \norm{\theta^{(T+1)}_{0,\lambda} - \theta^{(T)}}_{2}
\end{align*}
By triangle inequality, we have
\begin{align*}
    \norm{\theta^{(T+1)}_{0,\lambda} - \theta^{(T)}}_{\infty} 
    & \leq 
    \norm{\theta_{0}^{(\upsilon^{\ast})} - \theta^{(T+1)}_{0,\lambda}}_{2} + \norm{\theta_{0}^{(\upsilon^{\ast})} - \theta^{(T)}}_{\infty} 
    \\ & \leq 
    \left( 1 - \frac{\lambda_{0}}{\lambda_{0}+(T+1)np}\right)^{\upsilon^{\ast}} \sqrt{n} \norm{\theta^{(T+1)}_{0,\lambda} - \theta^{(T)}}_{\infty} + D_{8} \frac{1}{T+1} \sqrt{\frac{\log n}{npL}}.
\end{align*}
We can take $\upsilon^{\ast} = \lceil (1+(T+1)n^{2}p) \log (2\sqrt{n}) \rceil$ with $T \leq \frac{1}{n^{2}p} \left( \frac{n^{5}}{\log (2\sqrt{n})} - 1\right) - 1 \lesssim \frac{n^{3}}{p\log n}$ in order that $\left( 1 - \frac{\lambda_{0}}{\lambda_{0}+(T+1)np}\right)^{\upsilon^{\ast}} \sqrt{n} \leq \frac{1}{2}$. This implies $\norm{\theta^{(T+1)}_{0,\lambda} - \theta^{(T)}}_{\infty} \leq D_{8} \frac{2}{T+1} \sqrt{\frac{\log n}{npL}}$ with probability at least $1 - \mathcal{O}(n^{-5})$.
\end{proof}
\subsubsection{Proof of Lemma~\ref{lemma:auxil_reg_bound_2}}
\label{appendix:auxil_reg_bound_2}
\begin{proof}
Note that $\theta^{(T)}$ is feasible for the regularized optimization problem \eqref{eq:Constraint_MLE_3}, then we observe that
\begin{align*}
    \mathcal{L}_{\lambda_{0}}^{(\leq T)}(\theta^{(T)}) 
    & \geq 
    \mathcal{L}_{\lambda_{0}}^{(\leq T)}(\theta_{0, \lambda}^{(T+1)}) 
    \\ & =
    \mathcal{L}_{\lambda_{0}}^{(\leq T)}(\theta^{(T)}) + \left( \theta^{(T+1)}_{0, \lambda} - \theta^{(T)} \right)^{\top} \nabla \mathcal{L}_{\lambda_{0}}^{(\leq T)}(\theta^{(T)})
    \\ & \quad + 
    \frac{1}{2}  \left( \theta^{(T+1)}_{0, \lambda} - \theta^{(T)} \right)^{\top} \nabla^{2} \mathcal{L}_{\lambda_{0}}^{(\leq T)}(\theta^{(T+1)}_{0, \lambda}(\varrho))  \left( \theta^{(T+1)}_{0, \lambda} - \theta^{(T)} \right),
\end{align*}
where $\theta^{(T+1)}_{0, \lambda}(\varrho)$ is a convex combination of $\theta^{(T)}$ and $\theta^{(T+1)}_{0, \lambda}$. With \eqref{eq:G}, We can derive $\norm{\theta^{(T+1)}_{0,\lambda}(\varrho) - \theta^{(T)}}_{\infty} \leq \norm{\theta^{(T+1)}_{0, \lambda} - \theta^{(T)}}_{\infty} + \norm{\theta^{(T)} - \theta^{\ast}}_{\infty} + \norm{\theta^{\ast}}_{\infty} \leq  D_{8} \frac{2}{T+1} \sqrt{\frac{\log n}{npL}} + D_{4} \sqrt{\frac{1}{pL}} + \log \kappa$. Combining with Lemma \ref{lemma:gradient_T_H} and Lemma \ref{lemma:smallesteigen}, we derive
\begin{align*}
    \norm{\theta_{0, \lambda}^{(T+1)} - \theta^{(T)}}_{2} 
    & \leq
    \frac{2 \norm{\nabla \mathcal{L}_{\lambda_{0}}^{(\leq T)}(\theta^{(T)})}_{2}}{\lambda_{\min, \perp} \left( \nabla^{2} \mathcal{L}_{\lambda_{0}}^{(\leq T)}(\theta^{(T+1)}_{0, \lambda}(\varrho)) \right)}
    \\ & \leq 
    2 \frac{\norm{\nabla \ell_{T}(\theta^{(T)})}_{2} + \lambda_{0} \norm{\theta^{(T)} - \theta^{\ast}}_{2} + \lambda_{0} \norm{\theta^{\ast}}_{2}}{\lambda_{\min, \perp} \left( \nabla^{2} \mathcal{L}_{\lambda_{0}}^{(\leq T)}(\theta^{(T+1)}_{0, \lambda}(\varrho)) \right)}
    \\ & \leq
    \frac{2}{c_{12}} \frac{\sqrt{\frac{n^{2}p}{L}} + D_{4} \frac{1}{n}\sqrt{\frac{1}{pL}} + \frac{1}{\sqrt{n}} \log \kappa}{(T+1)np}
    \\ & \leq
    D_{11} \frac{1}{T+1} \sqrt{\frac{1}{pL}}
\end{align*}
for some constant $c_{12}, D_{11} > 0$.
\end{proof}
\subsubsection{Proof of Lemma~\ref{lemma:auxil_reg_true_bound_2}}
\label{appendix:auxil_reg_true_bound_2}
\begin{proof}
since $\theta^{(T+1)}_{0,\lambda}$ is feasible for $\mathcal{L}^{(\leq T)}(\cdot)$, we have
\begin{align*}
    \mathcal{L}^{(\leq T)}(\theta^{(T+1)}_{0,\lambda}) \geq \mathcal{L}^{(\leq T)}(\theta^{(T+1)}).
\end{align*}
Applying Taylor expansion shows
\begin{align*}
    \mathcal{L}^{(\leq T)}(\theta^{(T+1)}) &= 
    \mathcal{L}^{(\leq T)}(\theta^{(T+1)}_{0, \lambda}) + \left( \theta^{(T+1)} - \theta^{(T+1)}_{0, \lambda} \right)^{\top} \nabla \mathcal{L}^{(\leq T)}(\theta^{(T+1)}_{0, \lambda})
    \\ & \quad + 
    \frac{1}{2} \left( \theta^{(T+1)} - \theta^{(T+1)}_{0, \lambda} \right)^{\top} \nabla^{2} \mathcal{L}^{(\leq T)}(\theta^{(T+1)}_{0, \lambda}(\varsigma)) \left( \theta^{(T+1)} - \theta^{(T+1)}_{0, \lambda} \right),
\end{align*}
where $\theta^{(T+1)}_{\lambda}(\varsigma)$ is a convex combination of $\theta^{(T+1)}$ and $\theta^{(T+1)}_{0, \lambda}$. We can derive $\norm{\theta^{(T+1)}_{0,\lambda}(\varsigma) - \theta^{(T)}}_{\infty}$ is bounded due to $\norm{\theta^{(T+1)}_{0,\lambda} - \theta^{(T)}}_{\infty} \leq D_{8}\frac{2}{T+1} \sqrt{\frac{\log n}{npL}}$ and $\norm{\theta^{(T+1)}- \theta^{(T)}}_{\infty} \leq \norm{\theta^{(T+1)} - \theta^{\ast}} + \norm{\theta^{(T)} - \theta^{\ast}} + 2 \norm{\theta^{\ast}}_{\infty} \leq 10 + 2 \log \kappa$. By Lemma \ref{lemma:smallesteigen}, we get the lower bound
\begin{align*}
    \mathcal{L}^{(\leq T)}(\theta^{(T+1)}) & \geq 
    \mathcal{L}^{(\leq T)}(\theta^{(T+1)}_{0, \lambda}) + \left( \theta^{(T+1)} - \theta^{(T+1)}_{0, \lambda} \right)^{\top} \nabla \mathcal{L}^{(\leq T)}(\theta^{(T+1)}_{0,\lambda})
    \\ & \quad + 
    c_{13}(T+1)np \norm{\theta^{(T+1)} - \theta^{(T+1)}_{0,\lambda}}_{2}^{2}
\end{align*}
for some constant $c_{13} > 0$. 

Consider the facts that 
\begin{align*}
    \nabla \mathcal{L}^{(\leq T)}(\theta^{(T+1)}_{0,\lambda}) + \lambda_{0}\theta^{(T+1)}_{0,\lambda} = \nabla \mathcal{L}^{(\leq T)}_{\lambda_{0}}(\theta^{(T+1)}_{0, \lambda}) = 0
\end{align*} 
and
\begin{align*}
    \norm{\theta^{(T+1)}_{0, \lambda}}_{2}  
    & \leq \norm{\theta^{(T+1)}_{0,\lambda} - \theta^{(T)}}_{2} + \norm{\theta^{(T)} - \theta^{\ast}}_{2}  + \norm{\theta^{(T)}}_{2} 
    \\ & \leq 
    D_{11} \frac{1}{T+1} \sqrt{\frac{1}{pL}} + 
    D_{4} \sqrt{\frac{1}{pL}} + \sqrt   {n} \log \kappa
\end{align*}
then we derive 
\begin{align*}
    \norm{\theta^{(T+1)} - \theta^{(T+1)}_{0,\lambda}}_{2}
    & \leq
    \frac{\norm{\nabla \mathcal{L}^{(\leq T)}(\theta^{(T+1)}_{0,\lambda})}_{2}}{c_{13}(T+1)np}
    =
    \frac{\lambda_{0} \norm{\theta^{(T+1)}_{0,\lambda}}_{2}}{c_{13}(T+1)np}
    \\ & \leq 
    \frac{D_{11}}{c_{13}(T+1)^{2}n^{2}p} \sqrt{\frac{1}{pL}} + 
    \frac{D_{4}}{c_{13}(T+1)n^{2}p} \sqrt{\frac{1}{pL}} + \frac{1}{c_{13}(T+1)n^{2}p} \sqrt{n} \log \kappa
    \\ & \leq
    D_{12} \frac{1}{T+1} \sqrt{\frac{\log n}{npL}}
\end{align*}
for some constant $D_{12} > 0$.
\end{proof}

\subsection{Proof Auxiliary Lemmas in Appendix~\ref{appendix:equation_H}}
\subsubsection{Proof of Lemma~\ref{lemma:hessian_dif_L2}}
\label{appendix:hessian_dif_L2}
\begin{proof}
By definition and Lemma \ref{lemma:degree_concentration}, we have
\begin{align*}
    & \norm{\nabla^{2} \mathcal{L}^{(\leq T)}(\theta^{\prime})-\nabla^{2} \mathcal{L}^{(\leq T)}(\theta^{\prime\prime})}_{2}
    \\ & =
    \norm{\sum_{t=0}^{T} \sum_{(i,j) \in \mathcal{E}_t,\, i>j} \left[ \frac{e^{\theta_i^{\prime}}e^{\theta_j^{\prime}}}{(e^{\theta_i^{\prime}} + e^{\theta_j^{\prime}})^{2}} -\frac{e^{\theta_i^{\prime\prime}}e^{\theta_j^{\prime\prime}}}{(e^{\theta_i^{\prime\prime}} + e^{\theta_j^{\prime\prime}})^{2}}\right](\boldsymbol{e}_{i} - \boldsymbol{e}_{j})(\boldsymbol{e}_{i} - \boldsymbol{e}_{j})^{\top}}_{2}
    \\ &=
    \norm{\sum_{t=0}^{T} \sum_{(i,j) \in \mathcal{E}_t,\, i>j} \left[ \sigma^{\prime}(\theta_i^{\prime}-\theta_j^{\prime}) - \sigma^{\prime}(\theta_i^{\prime\prime}-\theta_j^{\prime\prime}) \right](\boldsymbol{e}_{i} - \boldsymbol{e}_{j})(\boldsymbol{e}_{i} - \boldsymbol{e}_{j})^{\top}}_{2}
    \\ &\lesssim
    \norm{\sum_{t=0}^{T} \sum_{(i,j) \in \mathcal{E}_t,\, i>j} \left[ (\theta_i^{\prime}-\theta_i^{\prime\prime}) - (\theta_j^{\prime}-\theta_j^{\prime\prime}) \right](\boldsymbol{e}_{i} - \boldsymbol{e}_{j})(\boldsymbol{e}_{i} - \boldsymbol{e}_{j})^{\top}}_{2}
    \\ & \lesssim
    \norm{\sum_{t=0}^{T} \boldsymbol{L}_{\mathcal{G}}^{(t)}}_{2} \norm{\theta^{\prime}-\theta^{\prime\prime}}_{\infty}
    \\ & \lesssim
    (T+1)np \norm{\theta^{\prime}-\theta^{\prime\prime}}_{\infty}
\end{align*}
with probability exceeding $1-\mathcal{O}(n^{-10})$.
\end{proof}
\subsubsection{Proof of Lemma~\ref{lemma:diff_qua_T+1_true}}
\label{appendix:diff_qua_T+1_true}
\begin{proof}
Fix a coordinate $k$, we define 
\begin{align*}
    X_{k}^{(t)} 
    & := 
    - \boldsymbol{e}_{k}^{\top} \left[ \nabla^{2} \mathcal{L}^{(\leq t)}(\theta^{(t)}) \right]^{\dagger} \nabla \ell_{t}(\theta^{(t)}) 
    \\ & =
    \sum_{1\le j<i\le n} \underbrace{G^{(t)}_{i,j} \boldsymbol{e}_{k}^{\top}\Big\lbrack \nabla^{2} \mathcal{L}^{(\le t)}(\theta^{(t)})\Big\rbrack^{\dagger}(\boldsymbol{e}_i - \boldsymbol{e}_j)}_{\text{fixed given $\mathcal{H}_{t}$}} \cdot 
    \underbrace{ \frac{1}{L}\sum_{l=1}^{L} \left\{\frac{e^{\theta^{(t)}_{i}}}{e^{\theta^{(t)}_{i}}+e^{\theta^{(t)}_{j}}} - y^{(t,l)}_{j,i} \right\}}_{\text{mean $0$ given $\mathcal{H}_{t}$}},
\end{align*}
where $\left\{ \boldsymbol{e}_{k} \right\}_{k=1}^{n}$ is the $k$-th standard basis vector in $\mathbb{R}^{n}$. Set $\norm{S_{T}}_{\infty} := \max_{k \in [n]} \left| S_{k} \right| = \max_{k \in [n]} \left| \sum_{t=0}^{T} X_{k}^{(t)} \right|$, we can apply sub-Gaussian Azuma-Hoeffding inequality to $\left\lbrace \left( X_{k}^{(t)}, \mathcal{H}_{t} \right): t \in [T] \right\rbrace$ and union bound over $k \in [n]$, where $\mathcal{H}_{s}:=\sigma \left\langle \{\mathcal{G}^{(t)}\}_{t=0}^{s},\{Y^{(t)}\}_{t=0}^{s} \right\rangle, \forall s \in [T]$.

We first guarantee that $\left\lbrace \left( X_{k}^{(t)}, \mathcal{H}_{t}  \right): t \in [T] \right\rbrace$ is a martingale difference array due to
\begin{align*}
    \mathbb{E} \left[ \nabla \ell_{t}(\theta^{(t)}) \mid \mathcal{H}_{t-1}, \mathcal{G}^{(t)} \right] = \mathbb{E} \left[ \sum_{1 \leq j < i \leq n} G_{j,i}^{(t)} \left\{ - y_{j,i}^{(t)}+ \frac{e^{\theta^{(t)}_i}}{e^{\theta^{(t)}_i} + e^{\theta^{(t)}_j}} \right\} (\boldsymbol{e}_{i} - \boldsymbol{e}_{j}) \mid \mathcal{H}_{t-1}, \mathcal{G}^{(t)} \right] = 0,
\end{align*}
which implies
\begin{align*}
    \mathbb{E} \left[ X_{k}^{(t)} \mid \mathcal{H}_{t-1}, \mathcal{G}^{(t)} \right]
    = 
    - \boldsymbol{e}_{k}^{\top} \left[ \nabla^{2} \mathcal{L}^{(\leq t)}(\theta^{(t)}) \right]^{\dagger} \mathbb{E} \left[  \nabla \ell_{t}(\theta^{(t)}) \mid \mathcal{H}_{t-1}, \mathcal{G}^{(t)} \right] = 0,
\end{align*}
and thus
\begin{align*}
    \mathbb{E} \left[ X_{k}^{(t)} \mid \mathcal{H}_{t-1} \right] 
    = 
    \mathbb{E} \left\{ \mathbb{E} \left[ X_{k}^{(t)} \mid \mathcal{H}_{t-1}, \mathcal{G}^{(t)} \right] \mid \mathcal{H}_{t-1}\right\} = 0.
\end{align*}

Next we claim that $\left\{ X_{k}^{(t)} \right\}_{t=0}^{T}$ is a sub-Gaussian martingale difference sequence. Note that Given $\mathcal{H}_{t}$, the summands are independent, have mean zero, and are bounded since $\frac{e^{\theta^{(t)}_{i}}}{e^{\theta^{(t)}_{i}}+e^{\theta^{(t)}_{j}}} - y^{(t,l)}_{j,i} \in [-1, 1]$. By Hoeffding's lemma, we can guarantee that
\begin{align*}
    \mathbb{E} \left\{ \exp \left[ \lambda \left( \frac{1}{L}\sum_{l=1}^{L} \left\{\frac{e^{\theta^{(t)}_{i}}}{e^{\theta^{(t)}_{i}}+e^{\theta^{(t)}_{j}}} - y^{(t,l)}_{j,i} \right\} \right) \right] \mid \mathcal{H}_{t} \right\} 
    \leq
    \exp \left( \frac{\lambda^{2}}{2L} \right)
\end{align*}
and 
\begin{align*}
    \mathbb{E} \left[ \exp \left( \lambda X_{k}^{(t)} \right) \mid \mathcal{H}_{t} \right] & \leq 
    \exp \left[ \frac{\lambda^{2}}{2L} \sum_{1\le j<i\le n} G^{(t)}_{i,j} \left( \boldsymbol{e}_{k}^{\top}\Big\lbrack \nabla^{2} \mathcal{L}^{(\le t)}(\theta^{(t)})\Big\rbrack^{\dagger} (\boldsymbol{e}_i - \boldsymbol{e}_j) \right)^2 \right]
    \\ & \leq
    \exp \left( \frac{\lambda^{2}}{2L} \boldsymbol{e}_{k}^{\top} \Big\lbrack \nabla^{2} \mathcal{L}^{(\le t)}(\theta^{(t)})\Big\rbrack^{\dagger} \sum_{i>j}  G^{(t)}_{i,j} (\boldsymbol{e}_i - \boldsymbol{e}_j) (\boldsymbol{e}_i - \boldsymbol{e}_j)^{\top} \Big\lbrack \nabla^{2} \mathcal{L}^{(\le t)}(\theta^{(t)})\Big\rbrack^{\dagger} \boldsymbol{e}_{k} \right)
    \\ & \leq
    \exp \left( \frac{\lambda^{2}}{2L} \lambda_{\text{max}} \left( \boldsymbol{L}_{\mathcal{G}}^{(t)} \right) \norm{ \Big\lbrack \nabla^{2} \mathcal{L}^{(\le t)}(\theta^{(t)})\Big\rbrack^{\dagger}}_{2}^{2} \norm{\boldsymbol{e}_{k}}_{2}^{2} \right)
    \\ & \leq
    \exp \left( \frac{\lambda^{2}}{2} \sigma_{t}^{2} \right),
\end{align*}
where $\sigma_{t}^{2} \asymp \frac{1}{(t+1)^{2}}\frac{1}{npL}$ by Lemma \ref{lemma:largesteigen} and Lemma \ref{lemma:smallesteigen}. 

Based on the Azuma-Hoeffding inequality with $\sigma_{t}^{2}$-sub-gaussian martingale differences $X_{k}^{(t)}$, then for $\epsilon \geq 0$, we have
\begin{align*}
    \mathbb{P} \left( e_{k}^{\top} \left| \sum_{t=0}^{T} \left[ \nabla^{2} \mathcal{L}^{(\leq t)}(\theta^{(t)}) \right]^{\dagger} \nabla \ell_{t}(\theta^{(t)}) \right| \geq \epsilon \right)
    =
    \mathbb{P} \left( \left| \sum_{t=0}^{T} X^{(t)}_{k} \right| \geq \epsilon \right) \leq 
    2 \exp \left( -\frac{\epsilon^{2}}{2 \sum_{t=0}^{T} \sigma_{t}^{2}} \right)
\end{align*}
and thus
\begin{align*}
    \mathbb{P} \left( \norm{\sum_{t=0}^{T} \overline{\theta}^{(t+1)} - \theta^{(t)}}_{\infty} \geq \epsilon\right)
    & =
    \mathbb{P} \left( \norm{\sum_{t=0}^{T} \left[ \nabla^{2} \mathcal{L}^{(\leq t)}(\theta^{(t)}) \right]^{\dagger} \nabla \ell_{t}(\theta^{(t)})}_{\infty} \geq \epsilon\right)
    \\ & =
    \mathbb{P} \left( \max_{k \in [n]} e_{k}^{\top} \left| \sum_{t=0}^{T} \left[ \nabla^{2} \mathcal{L}^{(\leq t)}(\theta^{(t)}) \right]^{\dagger} \nabla \ell_{t}(\theta^{(t)}) \right| \geq \epsilon \right)
    \\ & \leq
    \sum_{k \in [n]}
    \mathbb{P} \left( e_{k}^{\top} \left| \sum_{t=0}^{T} \left[ \nabla^{2} \mathcal{L}^{(\leq t)}(\theta^{(t)}) \right]^{\dagger} \nabla \ell_{t}(\theta^{(t)}) \right| \geq \epsilon \right)
    \\ & \leq
    2n \exp \left( -\frac{\epsilon^{2}}{2 \sum_{t=0}^{T} \sigma_{t}^{2}} \right).
\end{align*}
With appropriate choice of $\epsilon$, we claim that there exists a constant $D_{13} > 0$ such that
\begin{align*}
    \norm{\sum_{t=0}^{T}\overline{\theta}^{(t+1)} - \theta^{(t)} }_{\infty} \leq D_{13} \sqrt{\frac{\log n}{npL}}
\end{align*}
with probability at least $1 - O(n^{-10})$ since $\sum_{t=0}^{T} \sigma_{t}^{2} \lesssim \frac{1}{npL}$.
\end{proof}
\subsubsection{Proof of Lemma~\ref{lemma:diff_ori_qudra_t+1_l2}}
\label{appendix:diff_ori_qudra_t+1_l2}
\begin{proof}
we observe that
\begin{align*}
    0
    = 
    \nabla \mathcal{L}^{(\leq t)}(\theta^{(t+1)})
    &=
    \nabla \mathcal{L}^{(\leq t)}(\theta^{(t)}) + \nabla^{2} \mathcal{L}^{(\leq t)}(\theta^{(t)}) \left( \theta^{(t+1)}-\theta^{(t)} \right) + \mathbf{R}^{(t)}
    \\ &=
    \nabla \mathcal{L}^{(\leq t)}(\theta^{(t)}) + \left\{ \int_{0}^{1} \nabla^{2} \mathcal{L}^{(\leq t)} (\theta^{(t)}+\delta(\theta^{(t+1)}-\theta^{(t)})) d \delta \right\} \left( \theta^{(t+1)}-\theta^{(t)} \right),
\end{align*}
where
\begin{align*}
    \mathbf{R}^{(t)} = \left\{ \int_{0}^{1} \nabla^{2} \mathcal{L}^{(\leq t)} (\theta^{(t)}+\delta(\theta^{(t+1)}-\theta^{(t)})) - \nabla^{2} \mathcal{L}^{(\leq t)}(\theta^{(t)}) d \delta \right\} \left( \theta^{(t+1)}-\theta^{(t)} \right).
\end{align*}
It immediately follows from Lemma \ref{lemma:hessian_dif_L2} that
\begin{align*}
    \norm{\mathbf{R}^{(t)}}_{2}
    & \lesssim 
    \int_{0}^{1} (t+1) np \norm{\theta^{(t+1)}-\theta^{(t)}}_{\infty} d \delta \cdot \norm{\theta^{(t+1)}-\theta^{(t)}}_{2}
    \\ & \lesssim 
    (t+1) np \cdot \norm{\theta^{(t+1)}-\theta^{(t)}}_{\infty} \cdot \norm{\theta^{(t+1)}-\theta^{(t)}}_{2}.
\end{align*}
Combine the above two solutions, we have
\begin{align*}
    \theta^{(t+1)} - \theta^{(t)} 
    & = 
    -\left[ \nabla^{2} \mathcal{L}^{(\leq t)}(\theta^{(t)}) \right]^{\dagger} \left( \nabla \mathcal{L}^{(\leq t)}(\theta^{(t)}) + \mathbf{R}^{(t)} \right)
    \\ & =
    -\left[ \nabla^{2} \mathcal{L}^{(\leq t)}(\theta^{(t)}) \right]^{\dagger} \nabla \mathcal{L}^{(\leq t)}(\theta^{(t)}) - \left[ \nabla^{2} \mathcal{L}^{(\leq t)}(\theta^{(t)}) \right]^{\dagger} \mathbf{R}^{(t)}
    \\ &=
    \overline{\theta}^{(t+1)} - \theta^{(t)} - \left[ \nabla^{2} \mathcal{L}^{(\leq t)}(\theta^{(t)}) \right]^{\dagger} \mathbf{R}^{(t)},
\end{align*}
thus 
\begin{align*}
    \norm{\theta^{(t+1)} - \overline{\theta}^{(t+1)}}_{2} 
    & = 
    \norm{\left[ \nabla^{2} \mathcal{L}^{(\leq t)}(\theta^{(t)}) \right]^{\dagger}}_{2} \norm{\mathbf{R}^{(t)}}_{2}
    \\ & \lesssim
    \frac{1}{(t+1)np} \cdot (t+1)np \cdot D_{2} \frac{1}{t+1} \sqrt{\frac{\log n}{npL}} \cdot D_{3} \frac{1}{t+1} \sqrt{\frac{1}{pL}}
    \\ & \lesssim
    D_{2}D_{3} \frac{1}{(t+1)^{2}} \sqrt{\frac{\log n}{npL}} \sqrt{\frac{1}{pL}}
    \\ & \leq
    D_{14} \frac{1}{(t+1)^{2}} \sqrt{\frac{\log n}{n}} \frac{1}{pL}. 
\end{align*}
for some sufficiently large constant $D_{14} > 0$.
\end{proof}
\subsubsection{Proof of Lemma~\ref{lemma:theta_shift}}
\label{appendix:theta_shift}
\begin{proof}
Take $\boldsymbol{v} = c \boldsymbol{1}_{n} \in \Psi^{\perp}$ with $c = -\frac{1}{n}\boldsymbol{1}_{n}^{\top} \theta$, one can verify that $\theta + \boldsymbol{v} \in \Theta$ since
\begin{align*}
    \boldsymbol{1}_{n}^{\top} \left( \theta + \boldsymbol{v} \right)
    =
    \boldsymbol{1}_{n}^{\top} \theta + c \boldsymbol{1}_{n}^{\top} \boldsymbol{1}_{n}
    =
    \boldsymbol{1}_{n}^{\top} \theta - \frac{1}{n}\boldsymbol{1}_{n}^{\top} \theta \cdot n = 0.
\end{align*}
\end{proof}    
\subsubsection{Proof of Proposition~\ref{prop:minimizer_L_i}}
\label{appendix:minimizer_L_i}
\begin{proof}
Assume that there exists a $z$ such that $\mathcal{L}^{(\leq t)}|_{\theta_{-i}^{(t+1)}}(z) < \mathcal{L}^{(\leq t)}|_{\theta_{-i}^{(t+1)}}(\theta^{(t+1)}_{i})$. Then we let $\boldsymbol{w} \in \mathbb{R}^{n}$ be the vector such that $\boldsymbol{w}_{-i} = \theta_{-i}^{(t+1)}$ and $w_{i} = z$. And, let $\boldsymbol{v}$ be the vector in $\Psi^{\perp}$ such that $\boldsymbol{w} + \boldsymbol{v} \in \Theta$. Then we have
\begin{align*}
    \mathcal{L}^{(\leq t)} (\boldsymbol{w} + \boldsymbol{v}) = \mathcal{L}^{(\leq t)}(\boldsymbol{w}) = \mathcal{L}^{(\leq t)}|_{\theta_{-i}^{(t+1)}}(z) < \mathcal{L}^{(\leq t)}|_{\theta_{-i}^{(t+1)}}(\theta^{(t+1)}_{i}) = \mathcal{L}^{(\leq t)}(\theta^{(t+1)}).
\end{align*}
This contradicts to the definition of $\theta^{(t+1)}$ in \eqref{eq:argminLossT}.
\end{proof}
\subsubsection{Proof of Proposition~\ref{prop:minimizer_quadratic_L_i}}
\label{appendix:minimizer_quadratic_L_i}
\begin{proof}
If we assume that there exists a $z$ such that $\overline{\mathcal{L}}^{(\leq t)}|_{\overline{\theta}_{-i}^{(t+1)}}(z) < \overline{\mathcal{L}}^{(\leq t)}|_{\overline{\theta}_{-i}^{(t+1)}}(\overline{\theta}^{(t+1)}_{i})$. Then we let $\boldsymbol{w} \in \mathbb{R}^{n}$ be the vector such that $\boldsymbol{w}_{-i} = \theta_{-i}^{(t+1)}$ and $w_{i} = z$. And, let $\boldsymbol{v}$ be the vector in $\Psi^{\perp}$ such that $\boldsymbol{w} + \boldsymbol{v} \in \Theta$. Then we have
\begin{align*}
    \overline{\mathcal{L}}^{(\leq t)} (\boldsymbol{w} + \boldsymbol{v}) = \overline{\mathcal{L}}^{(\leq t)}(\boldsymbol{w}) = \overline{\mathcal{L}}^{(\leq t)}|_{\overline{\theta}_{-i}^{(t+1)}}(z) < \overline{\mathcal{L}}^{(\leq t)}|_{\overline{\theta}_{-i}^{(t+1)}}(\overline{\theta}^{(t+1)}_{i}) = \overline{\mathcal{L}}^{(\leq t)}(\overline{\theta}^{(t+1)}).
\end{align*}
This contradicts to the definition of $\overline{\theta}^{(t+1)}$ in \eqref{eq:argminLosst_quadratic}.
\end{proof}
\subsubsection{Proof of Lemma~\ref{lemma:supp_loss_t}}
\label{appendix:supp_loss_t}
\begin{proof}
    (1) By definition for $i \in [n]$ we have
    \begin{align*}
        \left[ \nabla \mathcal{L}^{(\leq t)} (\theta^{(t)}) \right]_{i} 
        & = 
        \left[ \nabla \ell_{t} (\theta^{(t)}) \right]_{i}
        \\ & =
        \sum_{j \in [n] \backslash \{ i \}} G_{j,i}^{(t)} \left\{ - y_{j,i}^{(t)}+ \frac{e^{\theta_i^{(t)}}}{e^{\theta_i^{(t)}} + e^{\theta_j^{(t)}}} \right\}
        \\ & = 
        \frac{1}{L} \sum_{j \in [n] \backslash \{ i \}} \sum_{l=1}^{L} G_{j,i}^{(t)} \left\{ - y_{j,i}^{(t,l)}+ \frac{e^{\theta_i^{(t)}}}{e^{\theta_i^{(t)}} + e^{\theta_j^{(t)}}} \right\}.
    \end{align*}
    Since $\left| - y_{j,i}^{(t,l)}+ \frac{e^{\theta_i^{(t)}}}{e^{\theta_i^{(t)}} + e^{\theta_j^{(t)}}} \right| \leq 1$, by Bernstein inequality we have
    \begin{align*}
        \left| \left[ \nabla \mathcal{L}^{(\leq t)} (\theta^{(t)}) \right]_{i} - \mathbb{E} \left\{ \left[ \nabla \mathcal{L}^{(\leq t)} (\theta^{(t)}) \right]_{i} \mid \mathcal{G}^{(t)} \right\} \right| 
        & \lesssim 
        \frac{1}{L} \left( \sqrt{\log n \cdot L \left( \sum_{j \in [n] \backslash \{ i \}} G_{j,i}^{(t)} \right)} + \log n \right) 
        \\ & \lesssim 
        \sqrt{\frac{np \log n}{L}}
    \end{align*}
    with probability exceeding $1 - \mathcal{O}(n^{-10})$, as long as $npL \gtrsim \log n$. On the other hand, since  $\mathbb{E} \left[  - y_{j,i}^{(t,l)}+ \frac{e^{\theta_i^{(t)}}}{e^{\theta_i^{(t)}} + e^{\theta_j^{(t)}}}\right] = 0$, we know that $\mathbb{E} \left\{ \left[ \nabla \mathcal{L}^{(\leq t)} (\theta^{(t)}) \right]_{i} \mid \mathcal{G}^{(t)} \right\}  = 0$. As a result,  we have
    \begin{align*}
        \left| \left[\nabla \mathcal{L}^{(\leq t)} (\theta^{(t)})\right]_{i} \right| 
        & \leq 
        D_{19} \sqrt{\frac{np \log n}{L}}
    \end{align*}
    with a sufficiently large constant $D_{19} > 0$.

    (2) By definition we have
    \begin{align*}
        \sum_{j \in [n] \backslash \{i\}} \left[\nabla^{2} \mathcal{L}^{(\leq t)} (\theta^{(t)})\right]_{i,j}^{2}
        & = 
        \sum_{j \in [n] \backslash \{i\}} \left( \sum_{l=0}^{t}-G^{(l)}_{j,i} \frac{e^{\theta_i^{(t)}}e^{\theta_j^{(t)}}}{(e^{\theta_i^{(t)}} + e^{\theta_j^{(t)}})^{2}} \right)^{2}
        \\ & \lesssim
        \sum_{j \in [n] \backslash \{i\}} \left( \sum_{l=0}^{t} G^{(l)}_{j,i} \right)^{2}
        \\ & \lesssim
        \sum_{j \in [n] \backslash \{i\}} \left( \sum_{l=0}^{t} G^{(l)}_{j,i} + \sum_{l=0}^{t} G^{(l)}_{j,i} \sum_{0 \leq s \leq t, s \neq l} {G^{(s)}_{j,i}} \right)
        \\ & \lesssim
        (t+1) \sum_{l=0}^{t} \sum_{j \in [n] \backslash \{i\}} G^{(l)}_{j,i}
        \\ & \leq
        D_{20} (t+1)^{2} np
    \end{align*}
    with probability at least $1 - \mathcal{O}(n^{-10})$.

    (3) Observed that
    \begin{align*}
        \sum_{j \in [n] \backslash \{i\}} \left| \left[\nabla^{2} \mathcal{L}^{(\leq t)} (\theta^{(t)})\right]_{i,j} \right|
        & = 
        \sum_{j \in [n] \backslash \{i\}} \left| \sum_{l=0}^{t} -G^{(l)}_{j,i} \frac{e^{\theta_i^{(t)}}e^{\theta_j^{(t)}}}{(e^{\theta_i^{(t)}} + e^{\theta_j^{(t)}})^{2}} \right|
        \\ & \lesssim
        \sum_{l=0}^{t}  \sum_{j \in [n] \backslash \{i\}} G^{(l)}_{j,i}
        \\ & \leq
        D_{21} (t+1) np
    \end{align*}
    with probability at least $1 - \mathcal{O}(n^{-10})$.

    (4) For $i,j \in [n], i \neq j$, by definition we know that $y_{j,i}^{(t)} = \frac{1}{L} \sum_{l=1}^{L} y_{j,i}^{(t,l)}$ is the average of $L$ independent Bernoulli random variables. By Hoeffding's inequality we know that
    \begin{align*}
        \left| y_{j,i}^{(t)} - \mathbb{E} \left[ y_{j,i}^{(t)} \right] \right| \leq
        D_{22} \sqrt{\frac{\log n}{L}}
    \end{align*}
    with probability at least $1 - \mathcal{O}(n^{-12})$. Since there are at most $n(n-1) \leq n^{2}$ ordered pairs $(i,j)$ with $i \neq j$, by union bound we know that
    \begin{align*}
        \left| y_{j,i}^{(t)} - \mathbb{E} \left[ y_{j,i}^{(t)} \right] \right|
        \leq
        D_{22} \sqrt{\frac{\log n}{L}}
    \end{align*}
    holds for all $i, j \in [n], i \neq j$ with probability at least $1 - \mathcal{O}(n^{-10})$.
\end{proof}
\subsubsection{Proof of Lemma~\ref{lemma:diff_quadratic}}
\label{appendix:diff_quadratic}
\begin{proof}
Recall that we have the following two equations
\begin{align*}
    \overline{\theta}_{i}^{(t+1,\prime)} - \theta_{i}^{(t)}
    & = 
    - \frac{\left[ \nabla \mathcal{L}^{(\leq t)}(\theta^{(t)}) \right]_{i} + \sum_{j \in [n] \backslash \{i\}} \left(\theta^{(t+1)}_{j} - \theta^{(t)}_{j}\right) \left[\nabla^{2} \mathcal{L}^{(\leq t)} (\theta^{(t)})\right]_{i,j}}{\left[\nabla^{2} \mathcal{L}^{(\leq t)}(\theta^{(t)})\right]_{i,i}},
    \\
    \overline{\theta}_{i}^{(t+1)} - \theta_{i}^{(t)} 
    & =
    - \frac{\left[ \nabla \mathcal{L}^{(\leq t)}(\theta^{(t)}) \right]_{i} + \sum_{j \in [n] \backslash \{i\}} \left(\overline{\theta}^{(t+1)}_{j} - \theta^{(t)}_{j}\right) \left[\nabla^{2} \mathcal{L}^{(\leq t)} (\theta^{(t)})\right]_{i,j}  }{\left[\nabla^{2} \mathcal{L}^{(\leq t)}(\theta^{(t)})\right]_{i,i}}
\end{align*}
thus we can derive
\begin{align*}
    \overline{\theta}_{i}^{(t+1,\prime)} - \overline{\theta}_{i}^{(t+1)} 
    & = 
    \frac{\sum_{j \in [n] \backslash \{i\}} \left(\theta^{(t+1)}_{j} - \overline{\theta}^{(t+1)}_{j}\right) \left[\nabla^{2} \mathcal{L}^{(\leq t)} (\theta^{(t)})\right]_{i,j}  }{\left[\nabla^{2} \mathcal{L}^{(\leq t)}(\theta^{(t)})\right]_{i,i}}
    \\ & =
    \underbrace{\frac{\sum_{j \in [n] \backslash \{i\}} \left(\theta^{(t+1)}_{j} - \theta^{(t+1)}_{(i),j}\right) \left[\nabla^{2} \mathcal{L}^{(\leq t)} (\theta^{(t)})\right]_{i,j}}{\left[\nabla^{2} \mathcal{L}^{(\leq t)}(\theta^{(t)})\right]_{i,i}}}_{B_{1}} 
    \\ & \quad + 
    \underbrace{\frac{\sum_{j \in [n] \backslash \{i\}} \left(\theta^{(t+1)}_{(i),j} - \overline{\theta}^{(t+1)}_{(i),j} \right) \left[\nabla^{2} \mathcal{L}^{(\leq t)} (\theta^{(t)})\right]_{i,j}  }{\left[\nabla^{2} \mathcal{L}^{(\leq t)}(\theta^{(t)})\right]_{i,i}}}_{B_{2}}
    \\ & \quad +
    \underbrace{\frac{\sum_{j \in [n] \backslash \{i\}} \left(\overline{\theta}^{(t+1)}_{(i),j} - \overline{\theta}_{j}^{(t+1)} \right) \left[\nabla^{2} \mathcal{L}^{(\leq t)} (\theta^{(t)})\right]_{i,j}  }{\left[\nabla^{2} \mathcal{L}^{(\leq t)}(\theta^{(t)})\right]_{i,i}}}_{B_{3}}.
\end{align*}
    Next, we bound $B_{1} - B_{3}$ one by one. Before proceeding, the denominator $\left[\nabla^{2} \mathcal{L}^{(\leq t)}(\theta^{(t)})\right]_{i,i}$ can be bounded as 
\begin{align*}
    \left[\nabla^{2} \mathcal{L}^{(\leq t)}(\theta^{(t)})\right]_{i,i} 
    = 
    \sum_{l=0}^{t} \sum_{j \in [n] \backslash \{i\}} G_{j,i}^{(t)} \frac{e^{\theta_j^{(t)}}e^{\theta_i^{(t)}}}{(e^{\theta_j^{(t)}} + e^{\theta_i^{(t)}})^{2}}
    \geq c_{18} (t+1) np
\end{align*}
with probability at least $1 - \mathcal{O}(n^{-10})$ for some constant $c_{18} > 0$ due to $\norm{\theta^{(t)}}_{\infty} \leq \norm{\theta^{(t)}-\theta^{\ast}}_{\infty} + \norm{\theta^{\ast}}_{\infty} \leq 5 + \log \kappa$ and Lemma \ref{lemma:degree_concentration}.

For $B_{1}$, by Lemma \ref{lemma:iterative_t_LOO}, we have  
\begin{align*}
    \norm{\theta^{(t+1)}_{(i)} - \theta^{(t+1)}}_{2}
    & \leq
    D_{17} \frac{1}{t+1} \sqrt{\frac{\log n}{npL}}.
\end{align*}
With Cauchy-Schwarz inequality and Lemma \ref{lemma:supp_loss_t}, the numerator of $B_{1}$ can bounded as
\begin{align*}
    \left| \sum_{j \in [n] \backslash \{i\}} \left(\theta^{(t+1)}_{j} - \theta^{(t+1)}_{(i),j}\right) \left[\nabla^{2} \mathcal{L}^{(\leq t)} (\theta^{(t)})\right]_{i,j} \right|
    & \leq 
    \norm{\theta^{(t+1)}_{(i)} - \theta^{(t+1)}}_{2} \sqrt{\sum_{j \in [n] \backslash \{i\}}\left[\nabla^{2} \mathcal{L}^{(\leq t)} (\theta^{(t)})\right]_{i,j}^{2}} 
    \\ & \leq
    D_{17} \frac{1}{t+1} \sqrt{\frac{\log n}{npL}} \sqrt{D_{20} (t+1)^{2}np}
    \\ & =
    D_{17}\sqrt{D_{20}} \sqrt{\frac{\log n}{L}}
\end{align*}
with probability at least $1 - \mathcal{O}(n^{-5})$. As a results, $B_{1}$ can be bounded as
\begin{align*}
    \left| B_{1} \right| \leq \frac{D_{17}\sqrt{D_{20}}}{c_{18}} \frac{1}{t+1} \frac{1}{np} \sqrt{\frac{\log n}{L}}.
\end{align*}
with probability at least $1 - \mathcal{O}(n^{-5})$.

When it comes to $B_{2}$, by Bernstein's inequality we know that
\begin{align*}
    & \left| \sum_{j \in [n] \backslash \{i\}} \left(\theta^{(t+1)}_{(i),j} - \overline{\theta}^{(t+1)}_{(i),j} \right) \left[\nabla^{2} \ell_{t}(\theta^{(t)})\right]_{i,j} - \mathbb{E} \left[ \sum_{j \in [n] \backslash \{i\}} \left(\theta^{(t+1)}_{(i),j} - \overline{\theta}^{(t+1)}_{(i),j} \right) \left[\nabla^{2} \ell_{t} (\theta^{(t)})\right]_{i,j} \mid \theta^{(t+1)}_{(i)}, \overline{\theta}^{(t+1)}_{(i)} \right] \right|
    \\ \lesssim & 
    \sqrt{\log n \sum_{j \in [n] \backslash \{i\}} \left(\theta^{(t+1)}_{(i),j} - \overline{\theta}^{(t+1)}_{(i),j} \right)^{2} \mathbb{E} \left\{ \left[\nabla^{2} \ell_{t} (\theta^{(t)})\right]_{i,j}^{2}\right\}} + \log n \max_{j \in [n] \backslash \{i\}} \left| \theta^{(t+1)}_{(i),j} - \overline{\theta}^{(t+1)}_{(i),j}\right| \left| \left[\nabla^{2} \ell_{t} (\theta^{(t)})\right]_{i,j} \right|
    \\ \lesssim & 
    \sqrt{\log n \sum_{j \in [n] \backslash \{i\}} \left(\theta^{(t+1)}_{(i),j} - \overline{\theta}^{(t+1)}_{(i),j} \right)^{2} \mathbb{E} \left[ \left(G_{j,i}^{(t)} \right)^{2}\right]} + \log n \max_{j \in [n] \backslash \{i\}} \left| \theta^{(t+1)}_{(i),j} - \overline{\theta}^{(t+1)}_{(i),j} \right| \left| G_{j,i}^{(t)} \right|
    \\ \lesssim &
    \sqrt{\log n \sum_{j \in [n] \backslash \{i\}} \left(\theta^{(t+1)}_{(i),j} - \overline{\theta}^{(t+1)}_{(i),j} \right)^{2} \mathbb{E} \left[ G_{j,i}^{(t)} \right]} + \log n \norm{\theta^{(t+1)}_{(i)} - \overline{\theta}^{(t+1)}_{(i)}}_{\infty} 
    \\ \lesssim & 
    \sqrt{p \log n} \norm{\theta^{(t+1)}_{(i)} - \overline{\theta}^{(t+1)}_{(i)}}_{2} + \log n \norm{\theta^{(t+1)}_{(i)} -\overline{\theta}^{(t+1)}_{(i)}}_{\infty}
\end{align*}
with probability at least $1 - \mathcal{O}(n^{-10})$. Moreover, by Lemma \ref{lemma:iterative_t_LOO} we have
\begin{align*}
    & \mathbb{E} \left[ \sum_{j \in [n] \backslash \{i\}} \left(\theta^{(t+1)}_{(i),j} - \overline{\theta}^{(t+1)}_{(i),j} \right) \left[\nabla^{2} \ell_{(t)}(\theta^{(t)})\right]_{i,j} \mid \theta^{(t+1)}_{(i)}, \overline{\theta}^{(t+1)}_{(i)} \right]
    \\ \lesssim & 
    \sqrt{\sum_{j \in [n] \backslash \{i\}} \left(\theta^{(t+1)}_{(i),j} - \overline{\theta}^{(t+1)}_{(i),j} \right)^{2} } \sqrt{\sum_{j \in [n] \backslash \{i\}} \left(\mathbb{E} \left\{ \left[\nabla^{2} \ell_{(t)} (\theta^{(t)})\right]_{i,j} \right\}\right)^{2}}
    \\ \lesssim &
    p\sqrt{n} \norm{\theta^{(t+1)}_{(i)} -\overline{\theta}^{(t+1)}_{(i)}}_{2}
\end{align*}
with probability at least $1 - \mathcal{O}(n^{-10})$.
On the other hand, by Cauchy-Schwarz inequality, we have
\begin{align*}
    & \sum_{j \in [n] \backslash \{i\}} \left(\theta^{(t+1)}_{(i),j} - \overline{\theta}^{(t+1)}_{(i),j} \right) \left[\nabla^{2} \mathcal{L}^{(\leq t-1)} (\theta^{(t)})\right]_{i,j}
    \\ \lesssim &
    \sum_{j \in [n] \backslash \{i\}} \sum_{l=0}^{t-1} \left(G_{j,i}^{(l)}\right)^{2} \left(\theta^{(t+1)}_{(i),j} - \overline{\theta}^{(t+1)}_{(i),j} \right) 
    \\ \lesssim &
    \sqrt{\sum_{j \in [n] \backslash \{i\}} \sum_{l=0}^{t-1} \left(\theta^{(t+1)}_{(i),j} - \overline{\theta}^{(t+1)}_{(i),j} \right)^{2} \left(G_{j,i}^{(l)}\right)^{2}} \sqrt{\sum_{j \in [n] \backslash \{i\}} \sum_{l=0}^{t-1}\left(G_{j,i}^{(l)}\right)^{2}}
    \\ \lesssim &
    t\sqrt{np} \norm{\theta^{(t+1)}_{(i)} -\overline{\theta}^{(t+1)}_{(i)}}_{2}
\end{align*}
with probability exceeding $1 - \mathcal{O}(n^{-10})$. 
As a result, there exists a constant $c_{19} > 0$ such that
\begin{align*}
    \left| B_{2} \right| 
    & \lesssim 
    \frac{ (t+1)\sqrt{np} \norm{\theta^{(t+1)}_{(i)} - \overline{\theta}^{(t+1)}_{(i)}}_{2} + \log n \norm{\theta^{(t+1)}_{(i)} -\overline{\theta}^{(t+1)}_{(i)}}_{\infty}}{(t+1)np}
    \\ & \le
    c_{19}D_{18} \frac{1}{t+1} \frac{1}{\sqrt{npL}} \frac{\log n}{\sqrt{nL}p} + c_{19} \frac{1}{t+1} \frac{\log n}{np} \norm{\theta^{(t+1)}_{(i)} -\overline{\theta}^{(t+1)}_{(i)}}_{\infty}
\end{align*} 
with probability at least $1 - \mathcal{O}(n^{-5})$.

We finally consider bounding $B_{3}$. By definition, we know that
\begin{align*}
    \begin{cases}
        \nabla \mathcal{L}^{(\leq t)}(\theta^{(t)}) +  \nabla^{2} \mathcal{L}^{(\leq t)}(\theta^{(t)}) \left( \overline{\theta}^{(t+1)} - \theta^{(t)} \right) = \bm{0};
        \\
        \nabla \mathcal{L}^{(t \backslash i, \leq t)}(\theta^{(t)}) + \nabla^{2} \mathcal{L}^{(t \backslash i, \leq t)}(\theta^{(t)}) \left( \overline{\theta}^{(t+1)}_{(i)} - \theta^{(t)} \right) = \bm{0}.
    \end{cases}
\end{align*}
Combine the two equations we get
\begin{align*}
    \bm{0} 
    & = 
    \underbrace{\left( \nabla \mathcal{L}^{(\leq t)}(\theta^{(t)}) - \nabla \mathcal{L}^{(t \backslash i, \leq t)}(\theta^{(t)}) \right)}_{\boldsymbol{w}_{1}} + \underbrace{\left( \nabla^{2} \mathcal{L}^{(\leq t)}(\theta^{(t)}) - \nabla^{2} \mathcal{L}^{(t \backslash i, \leq t)}(\theta^{(t)}) \right) \left( \overline{\theta}^{(t+1)} - \theta^{(t)} \right)}_{\boldsymbol{w}_{2}}  
    \\ & \quad + 
    \nabla^{2} \mathcal{L}^{(t \backslash i, \leq t)}(\theta^{(t)}) \left( \overline{\theta}^{(t+1)} - \overline{\theta}^{(t+1)}_{(i)} \right)
\end{align*}

For $\bm{w}_{1}$, similar to Lemma \ref{lemma:gradient_t_diff_LOO}, by definition we have
\begin{align*}
    \nabla \mathcal{L}^{(\leq t)}(\theta^{(t)}) - \nabla \mathcal{L}^{(t \backslash i, \leq t)}(\theta^{(t)})
    & = 
    \nabla \ell_{t}(\theta^{(t)}) - \nabla \ell_{t}^{(i)}(\theta^{(t)})
    \\ & = 
    \frac{1}{L} \sum_{j \in [n] \backslash \{ i \}} G_{j,i}^{(t)} \sum_{l=1}^{L} \left\{ -y_{j,i}^{(t,l)} + \frac{e^{\theta_{i}^{(t)}}}{e^{\theta_{i}^{(t)}}+e^{\theta_{j}^{(t)}}} \right\} (\boldsymbol{e}_{i} - \boldsymbol{e}_{j}).
\end{align*}
For $\left[ \nabla \mathcal{L}^{(\leq t)}(\theta^{(t)}) - \nabla \mathcal{L}^{(t \backslash i, \leq t)}(\theta^{(t)}) \right]_{i}$, by Hoeffding inequality and Lemma \ref{lemma:degree_concentration}, we have
\begin{align*}
    \left| \left[ \nabla \mathcal{L}^{(\leq t)}(\theta^{(t)}) - \nabla \mathcal{L}^{(t \backslash i, \leq t)}(\theta^{(t)}) \right]_{i} \right|
    & = 
    \left| \sum_{j \in [n] \backslash \{ i \}} G_{j,i}^{(t)} \sum_{l=1}^{L} \frac{1}{L} \left\{ -y_{j,i}^{(t,l)} + \frac{e^{\theta_{i}^{(t)}}}{e^{\theta_{i}^{(t)}}+e^{\theta_{j}^{(t)}}} \right\} \right|
    \\ & \lesssim
    \sqrt{\sum_{j \in [n] \backslash \{ i \}} G_{j,i}^{(t)} \frac{\log n}{L}}
    \\ & \lesssim
    \sqrt{\frac{np \log n}{L}}
\end{align*}
with probability at least $1 - \mathcal{O}(n^{-10})$. Furthermore, Lemmas \ref{lemma:supp_loss_t} and \ref{lemma:degree_concentration} further guarantee that
\begin{align*}
     \sum_{j \in [n] \backslash \{ i \}} \left[ \nabla \mathcal{L}^{(\leq t)}(\theta^{(t)}) - \nabla \mathcal{L}^{(t \backslash i, \leq t)}(\theta^{(t)}) \right]_{j}^{2}
     & = 
     \sum_{j \in [n] \backslash \{ i \}}  G_{j,i}^{t} \left\{ -y_{j,i}^{(t)} + \frac{e^{\theta_{i}^{(t)}}}{e^{\theta_{i}^{(t)}}+e^{\theta_{j}^{(t)}}} \right\}^{2}
     \\ & =
     \sum_{j \in [n] \backslash \{ i \}}  G_{j,i}^{t} \left\{ -y_{j,i}^{(t)} + \mathbb{E} \left[ y_{j,i}^{(t)}\right] \right\}^{2}
     \\ & \lesssim
     D_{22}^{2} \frac{np \log n}{L} 
\end{align*}
with probability at least $1 - \mathcal{O}(n^{-10})$. Therefore, for $\bm{w}_{1}$ we have
\begin{align*}
    \norm{\bm{w}_{1}}_{2}^{2} 
    & =
    \norm{\nabla \mathcal{L}^{(\leq t)}(\theta^{(t)}) - \nabla \mathcal{L}^{(t \backslash i, \leq t)}(\theta^{(t)})}_{2}^{2}
    \\ & \leq
    \left[ \nabla \mathcal{L}^{(\leq t)}(\theta^{(t)}) - \nabla \mathcal{L}^{(t \backslash i, \leq t)}(\theta^{(t)}) \right]_{i}^{2} + \sum_{j \in [n] \backslash \{ i \}} \left[ \nabla \mathcal{L}^{(\leq t)}(\theta^{(t)}) - \nabla \mathcal{L}^{(t \backslash i, \leq t)}(\theta^{(t)}) \right]_{j}^{2}
    \\ & \leq
    c_{20}^{2} D^{2}_{18} \frac{np \log n}{L}
\end{align*}
for some constant $c_{20}>0$ with probability exceeding $1 - \mathcal{O}(n^{-10})$.

For $\bm{w}_{2}$, since
\begin{align*}
    \nabla^{2} \mathcal{L}^{(\leq t)}(\theta^{(t)}) - \nabla^{2} \mathcal{L}^{(t \backslash i, \leq t)}(\theta^{(t)})
    & = 
    \sum_{j \in [n] \backslash \{ i \}} \left( G_{j,i}^{(t)} - p \right) \frac{e^{\theta_{i}^{(t)}}e^{\theta_{j}^{(t)}}}{(e^{\theta_{i}^{(t)}}+e^{\theta_{j}^{(t)}})^{2}} (\boldsymbol{e}_{i} - \boldsymbol{e}_{j}) (\boldsymbol{e}_{i} - \boldsymbol{e}_{j})^{\top},
\end{align*}
by Lemma \ref{lemma:degree_concentration} it holds that
\begin{align*}
    & \norm{\left( \nabla^{2} \mathcal{L}^{(\leq t)}(\theta^{(t)}) - \nabla^{2} \mathcal{L}^{(t \backslash i, \leq t)}(\theta^{(t)}) \right) \left( \overline{\theta}^{(t+1)} - \theta^{(t)} \right)}_{2}
    \\ = & 
    \norm{\sum_{j \in [n] \backslash \{ i \}} \left( G_{j,i}^{(t)} - p \right) \frac{e^{\theta_{i}^{(t)}}e^{\theta_{j}^{(t)}}}{(e^{\theta_{i}^{(t)}}+e^{\theta_{j}^{(t)}})^{2}} \left[(\boldsymbol{e}_{i} - \boldsymbol{e}_{j})^{\top}\left( \overline{\theta}^{(t+1)} - \theta^{(t)} \right)\right](\boldsymbol{e}_{i} - \boldsymbol{e}_{j}) }_{2}
    \\ \leq & 
    \sum_{j \in [n] \backslash \{ i \}} \left|  G_{j,i}^{(t)} - p \right| \left| (\boldsymbol{e}_{i} - \boldsymbol{e}_{j})^{\top}\left( \overline{\theta}^{(t+1)} - \theta^{(t)} \right) \right| \norm{\boldsymbol{e}_{i} - \boldsymbol{e}_{j}}_{2}
    \\ \leq &
    \sum_{j \in [n] \backslash \{ i \}} \left( G_{j,i}^{(t)} + p \right) \norm{\overline{\theta}^{(t+1)} - \theta^{(t)}}_{\infty} \norm{\boldsymbol{e}_{i} - \boldsymbol{e}_{j}}_{2}
    \\ \lesssim & 
    np \norm{\overline{\theta}^{(t+1)} - \theta^{(t)}}_{\infty}
\end{align*}
with probability at least $1 - \mathcal{O}(n^{-10})$. On the other hand, by \eqref{eq:E} we obtain
\begin{align*}
    \norm{\overline{\theta}^{(t+1)} - \theta^{(t)}}_{\infty}
    & \leq
    \norm{\theta^{(t+1)} - \theta^{(t)}}_{\infty} + \norm{\overline{\theta}^{(t+1)} - \theta^{(t+1)}}_{\infty} 
    \\ & \leq 
    D_{2} \frac{1}{t+1} \sqrt{\frac{\log n}{npL}} + \norm{\overline{\theta}^{(t+1)} - \theta^{(t+1)}}_{\infty}
\end{align*}
with probability at least $1 - \mathcal{O}(n^{-10})$. That is to say, we have
\begin{align*}
    \norm{\bm{w}_{2}}_{2} \leq  c_{21}D_{2} \frac{1}{t+1} \sqrt{\frac{np \log n}{L}} + c_{21}np \norm{\overline{\theta}^{(t+1)} - \theta^{(t+1)}}_{\infty}
\end{align*}
for some constant $c_{21}>0$ with probability exceeding $1 - \mathcal{O}(n^{-10})$.

Combine the upper bounds of $\norm{\bm{w}_{1}}_{2}$ and $\norm{\bm{w}_{2}}_{2}$, we know that
\begin{align*}
    \norm{\nabla^{2} \mathcal{L}^{(t \backslash i, \leq t)}(\theta^{(t)}) \left( \overline{\theta}^{(t+1)} - \overline{\theta}^{(t+1)}_{(i)} \right)}_{2}
    & = 
    \norm{\bm{w}_{1}}_{2} + \norm{\bm{w}_{2}}_{2} 
    \\ & \leq 
    (c_{21}D_{2} + c_{20}D_{22}) \sqrt{\frac{np \log n}{L}} + c_{21}np \norm{\overline{\theta}^{(t+1)} - \theta^{(t+1)}}_{\infty}
\end{align*}
with probability at least $1 - \mathcal{O}(n^{-5})$.

Since $\overline{\theta}^{(t+1)} - \overline{\theta}^{(t+1)}_{(i)} \in \Theta$, we have
\begin{align*}
    & \lambda_{\min, \perp} \left( \nabla^{2} \mathcal{L}^{(t \backslash i, \leq t)}(\theta^{(t)}) \right) \norm{\overline{\theta}^{(t+1)} - \overline{\theta}^{(t+1)}_{(i)}}_{2}^{2}
    \\ \leq & 
    \left( \overline{\theta}^{(t+1)} - \overline{\theta}^{(t+1)}_{(i)} \right)^{\top} \nabla^{2} \mathcal{L}^{(t \backslash i, \leq t)}(\theta^{(t)}) \left( \overline{\theta}^{(t+1)} - \overline{\theta}^{(t+1)}_{(i)} \right)
    \\ \leq & 
    \norm{\overline{\theta}^{(t+1)} - \overline{\theta}^{(t+1)}_{(i)}}_{2} \norm{\nabla^{2} \mathcal{L}^{(t \backslash i, \leq t)}(\theta^{(t)}) \left( \overline{\theta}^{(t+1)} - \overline{\theta}^{(t+1)}_{(i)} \right)}_{2}.
\end{align*}
As a results, there exists a constant $c_{22} > 0$ such that 
\begin{align*}
    \norm{\overline{\theta}^{(t+1)} - \overline{\theta}^{(t+1)}_{(i)}}_{2} 
    & \leq
    \frac{\norm{\nabla^{2} \mathcal{L}^{(t \backslash i, \leq t)}(\theta^{(t)}) \left( \overline{\theta}^{(t+1)} - \overline{\theta}^{(t+1)}_{(i)} \right)}_{2}}{\lambda_{\min, \perp} \left( \nabla^{2} \mathcal{L}^{(t \backslash i, \leq t)}(\theta^{(t)}) \right)}
    \\ & \leq
    \frac{\norm{\nabla^{2} \mathcal{L}^{(t \backslash i, \leq t)}(\theta^{(t)}) \left( \overline{\theta}^{(t+1)} - \overline{\theta}^{(t+1)}_{(i)} \right)}_{2}}{c_{22}(t+1)np}
    \\ & \leq 
    \frac{c_{21}D_{2} + c_{20}D_{22}}{c_{22}} \frac{1}{t+1} \sqrt{\frac{\log n}{npL}} + \frac{c_{21}}{c_{22}} \frac{1}{t+1} \norm{\overline{\theta}^{(t+1)} - \theta^{(t+1)}}_{\infty}
\end{align*}
with probability at least $1 - \mathcal{O}(n^{-10})$. Therefore, for $B_{3}$ we finally achieve
\begin{align*}
    \left| B_{3} \right|
    & =
    \left|\frac{\sum_{j \in [n] \backslash \{i\}} \left(\overline{\theta}^{(t+1)}_{(i),j} - \overline{\theta}_{j}^{(t+1)} \right) \left[\nabla^{2} \mathcal{L}^{(\leq t)} (\theta^{(t)})\right]_{i,j}  }{\left[\nabla^{2} \mathcal{L}^{(\leq t)}(\theta^{(t)})\right]_{i,i}} \right|
    \\ & \leq
    \frac{\norm{\overline{\theta}^{(t+1)} - \overline{\theta}^{(t+1)}_{(i)}}_{2} \sqrt{\sum_{j \in [n] \backslash \{i\}}\left[\nabla^{2} \mathcal{L}^{(\leq t)} (\theta^{(t)})\right]_{i,j}^{2}}}{c_{18}(t+1)np}
    \\ & \leq 
    \frac{\sqrt{D_{20}}}{c_{18}} \frac{1}{\sqrt{np}} \norm{\overline{\theta}^{(t+1)} - \overline{\theta}^{(t+1)}_{(i)}}_{2}
    \\ & \leq
    \frac{\sqrt{D_{20}}}{c_{18}} \frac{c_{21}D_{2} + c_{20}D_{22}}{c_{22}} \frac{1}{t+1} \frac{1}{np} \sqrt{\frac{\log n}{L}} +  \frac{\sqrt{D_{20}}}{c_{18}} \frac{c_{21}}{c_{22}} \frac{1}{t+1} \frac{1}{\sqrt{np}} \norm{\overline{\theta}^{(t+1)} - \theta^{(t+1)}}_{\infty}
\end{align*}
with probability at least $1 - \mathcal{O}(n^{-5})$.

On the other hand, by Lemma \ref{lemma:iterative_t_LOO}, we know that
\begin{align*}
    \norm{\theta^{(t+1)}_{(i)} -\overline{\theta}^{(t+1)}_{(i)}}_{\infty}
    & \leq 
    \norm{\theta^{(t+1)}_{(i)} -\theta^{(t+1)}}_{2} + \norm{\theta^{(t+1)} -\overline{\theta}^{(t+1)}}_{\infty} + \norm{\overline{\theta}^{(t+1)} -\overline{\theta}^{(t+1)}_{(i)}}_{2}
    \\ & \leq
    \left( D_{17} + \frac{c_{21}D_{2} + c_{20}D_{22}}{c_{22}} \right) \frac{1}{t+1} \sqrt{\frac{\log n}{npL}} + \left(1+ \frac{c_{21}}{c_{22}}\right) \norm{\overline{\theta}^{(t+1)} - \theta^{(t+1)}}_{\infty}
\end{align*}
with probability at least $1 - \mathcal{O}(n^{-5})$. Combine with upper bound of $\left| B_{2} \right|$ we have
\begin{align*}
    \left| B_{2} \right| 
    & \leq 
    c_{19}D_{18} \frac{1}{t+1} \frac{1}{\sqrt{npL}} \frac{\log n}{\sqrt{nL}p} + c_{19} \left( D_{17} + \frac{c_{21}D_{2} + c_{20}D_{22}}{c_{22}} \right) \frac{1}{(t+1)^{2}} \frac{\log n}{np}  \sqrt{\frac{\log n}{npL}}
    \\ & + 
    c_{19}\left(1+ \frac{c_{21}}{c_{22}}\right) \frac{1}{t+1} \frac{\log n}{np}\norm{\overline{\theta}^{(t+1)} - \theta^{(t+1)}}_{\infty}
\end{align*}
with probability at least $1 - \mathcal{O}(n^{-5})$. Combine with $\left| B_{1} \right|$ we know that
\begin{align*}
    \left| \overline{\theta}_{i}^{(t+1,\prime)} - \overline{\theta}_{i}^{(t+1)} \right|
    & \leq 
    \frac{D_{17} \sqrt{D_{20}}}{c_{18}} \frac{1}{t+1} \frac{1}{np} \sqrt{\frac{\log n}{L}} 
    \\ & + 
    c_{19}D_{18} \frac{1}{t+1} \frac{1}{\sqrt{npL}} \frac{\log n}{\sqrt{nL}p} + c_{19} \left( D_{17} + \frac{c_{21}D_{2} + c_{20}D_{22}}{c_{22}} \right) \frac{1}{(t+1)^{2}} \frac{\log n}{np}  \sqrt{\frac{\log n}{npL}}
    \\ & + 
    c_{19}\left(1+ \frac{c_{21}}{c_{22}}\right)\frac{1}{t+1} \frac{\log n}{np}\norm{\overline{\theta}^{(t+1)} - \theta^{(t+1)}}_{\infty}
    \\ & +
    \frac{\sqrt{D_{20}}}{c_{18}} \frac{c_{21}D_{2} + c_{20}D_{22}}{c_{22}} \frac{1}{t+1} \frac{1}{np} \sqrt{\frac{\log n}{L}} +  \frac{\sqrt{D_{20}}}{c_{18}} \frac{c_{21}}{c_{22}} \frac{1}{t+1} \frac{1}{\sqrt{np}} \norm{\overline{\theta}^{(t+1)} - \theta^{(t+1)}}_{\infty}
    \\ & \lesssim
    \frac{1}{t+1} \frac{1}{\sqrt{npL}} \frac{\log n}{\sqrt{nL}p} + \frac{1}{t+1} \frac{1}{np} \sqrt{\frac{\log n}{L}} + \frac{1}{(t+1)^{2}} \frac{\log n}{np}  \sqrt{\frac{\log n}{npL}} 
    \\ & + 
    \frac{1}{t+1} \frac{\log n}{np} \norm{\overline{\theta}^{(t+1)} - \theta^{(t+1)}}_{\infty} + \frac{1}{t+1} \frac{1}{\sqrt{np}} \norm{\overline{\theta}^{(t+1)} - \theta^{(t+1)}}_{\infty}
    \\ & \leq
    D_{23} \frac{1}{t+1} \frac{1}{\sqrt{npL}} \frac{\log n}{\sqrt{nL}p} + D_{23} \frac{1}{t+1} \frac{1}{np} \sqrt{\frac{\log n}{L}} \left( 1 + \frac{\log n}{\sqrt{np}} \right) 
    \\ & + 
    D_{23} \frac{1}{t+1} \left( \frac{\log n}{np} + \frac{1}{\sqrt{np}} \right) \norm{\overline{\theta}^{(t+1)} - \theta^{(t+1)}}_{\infty}
\end{align*}
for some sufficiently large constant $D_{23} > 0$ with probability exceeding $1 - \mathcal{O}(n^{-5})$.
\end{proof}
\subsubsection{Proof of Lemma~\ref{lemma:diff_quadratic_MLE}}
\label{appendix:diff_quadratic_MLE}
\begin{proof}
    Since $\theta^{(t+1)}_{i}$ is the minimizer of $\mathcal{L}^{(\leq t)}|_{\theta_{-i}^{(t+1)}}(\cdot)$, we know that the gradient of $\mathcal{L}^{(\leq t)}|_{\theta_{-i}^{(t+1)}}(\cdot)$ with respect to the $i$-th coordinate is $\left( \mathcal{L}^{(\leq t)}|_{\theta_{-i}^{(t+1)}} \right)^{\prime} (\theta^{(t+1)}_{i}) = 0$. By the mean value theorem we have
    \begin{align*}
        \left( \mathcal{L}^{(\leq t)}|_{\theta_{-i}^{(t+1)}} \right)^{\prime} (\theta^{(t+1)}_{i}) = \left( \mathcal{L}^{(\leq t)}|_{\theta_{-i}^{(t+1)}} \right)^{\prime} (\theta^{(t)}_{i}) + \left( \mathcal{L}^{(\leq t)}|_{\theta_{-i}^{(t+1)}} \right)^{\prime\prime} (b_{1}) \left( \theta_{i}^{(t+1)} - \theta^{(t)}_{i} \right),
    \end{align*}
    where $b_{1}$ is a convex combination of $\theta^{(t)}_{i}$ and $\theta^{(t+1)}_{i}$. As a result, we have
    \begin{align*}
        \theta_{i}^{(t+1)} = \theta^{(t)}_{i} - \frac{\left( \mathcal{L}^{(\leq t)}|_{\theta_{-i}^{(t+1)}} \right)^{\prime} (\theta^{(t)}_{i})}{\left( \mathcal{L}^{(\leq t)}|_{\theta_{-i}^{(t+1)}} \right)^{\prime\prime} (b_{1})}.
    \end{align*}
    By definitions of $\mathcal{L}^{(\leq t)}(\cdot)$ and $\mathcal{L}^{(\leq t)}|_{\rm{x}_{-i}}(\cdot)$, we have
    \begin{align*}
        \left( \mathcal{L}^{(\leq t)}|_{\theta_{-i}^{(t+1)}} \right)^{\prime} (x) 
        & = 
        \sum_{l=0}^{t} \sum_{j \in [n] \backslash \{ i \}} G_{j,i}^{(l)} \left\{ - y_{j,i}^{(l)}+ \frac{e^{x}}{e^{x} + e^{\theta^{(t+1)}_{j}}} \right\},
        \\
        \left( \mathcal{L}^{(\leq t)}|_{\theta_{-i}^{(t+1)}} \right)^{\prime\prime} (x) 
        & = 
        \sum_{l=0}^{t} \sum_{j \in [n] \backslash \{i\}} G_{j,i}^{(l)} \frac{e^{x}e^{\theta^{(t+1)}_{j}}}{(e^{x} + e^{\theta^{(t+1)}_{j}})^{2}}.
    \end{align*}
    We first estimate the difference $\left( \mathcal{L}^{(\leq t)}|_{\theta_{-i}^{(t+1)}} \right)^{\prime\prime} (b_{1}) - \left[ \nabla^{2} \mathcal{L}^{(\leq t)}(\theta^{(t)})\right]_{i,i}$. We have
    \begin{align*}
        & \left| \left( \mathcal{L}^{(\leq t)}|_{\theta_{-i}^{(t+1)}} \right)^{\prime\prime} (b_{1}) - \left[ \nabla^{2} \mathcal{L}^{(\leq t)}(\theta^{(t)})\right]_{i,i} \right|
        \\ = &
        \left| \sum_{l=0}^{t} \sum_{j \in [n] \backslash \{i\}} G_{j,i}^{(l)} \frac{e^{b_{1}}e^{\theta^{(t+1)}_{j}}}{(e^{b_{1}} + e^{\theta^{(t+1)}_{j}})^{2}} - \sum_{l=0}^{t} \sum_{j \in [n] \backslash \{i\}} G_{j,i}^{(l)} \frac{e^{\theta^{(t)}_{i}}e^{\theta^{(t)}_{j}}}{(e^{\theta^{(t)}_{i}} + e^{\theta^{(t)}_{j}})^{2}} \right|
        \\ \leq &
        \sum_{l=0}^{t} \sum_{j \in [n] \backslash \{i\}} G_{j,i}^{(l)} \left| \frac{e^{b_{1}}e^{\theta^{(t+1)}_{j}}}{(e^{b_{1}} + e^{\theta^{(t+1)}_{j}})^{2}} - \frac{e^{\theta^{(t)}_{i}}e^{\theta^{(t)}_{j}}}{(e^{\theta^{(t)}_{i}} + e^{\theta^{(t)}_{j}})^{2}} \right|
        \\ \lesssim & 
        \sum_{l=0}^{t} \sum_{j \in [n] \backslash \{i\}} G_{j,i}^{(l)} \left| (b_{1} - \theta_{j}^{(t+1)}) - (\theta_{i}^{(t)} - \theta_{j}^{(t+1)}) \right|
        \\ \lesssim &
        (t+1) np \norm{\theta^{(t+1)}-\theta^{(t)}}_{\infty}
    \end{align*}
    and
    \begin{align*}
        \left| \left( \mathcal{L}^{(\leq t)}|_{\theta_{-i}^{(t+1)}} \right)^{\prime\prime} (b_{1}) \right| 
        & \geq
        \left| \left[ \nabla^{2} \mathcal{L}^{(\leq t)}(\theta^{(t)})\right]_{i,i} \right| -
        \left| \left( \mathcal{L}^{(\leq t)}|_{\theta_{-i}^{(t+1)}} \right)^{\prime\prime} (b_{1}) - \left[ \nabla^{2} \mathcal{L}^{(\leq t)}(\theta^{(t)})\right]_{i,i} \right| 
        \\ & \gtrsim
        (t+1)np - (t+1) np \norm{\theta^{(t+1)}-\theta^{(t)}}_{\infty}
    \end{align*}
    with probability at least $1 - \mathcal{O}(n^{-5})$. On the other hand, we have
    \begin{align*}
        & \left( \mathcal{L}^{(\leq t)}|_{\theta_{-i}^{(t+1)}} \right)^{\prime} (\theta_{i}^{(t)}) - \left\{ \left[ \nabla \mathcal{L}^{(\leq t)}(\theta^{(t)}) \right]_{i} + \sum_{j \in [n] \backslash \{i\}} \left(\theta^{(t+1)}_{j} - \theta^{(t)}_{j}\right) \left[\nabla^{2} \mathcal{L}^{(\leq t)} (\theta^{(t)})\right]_{i,j} \right\}
        \\ = & 
        \sum_{l=0}^{t} \sum_{j \in [n] \backslash \{ i \}} G_{j,i}^{(l)} \left\{ - y_{j,i}^{(l)} + \frac{e^{\theta_{i}^{(t)}}}{e^{\theta_{i}^{(t)}} + e^{\theta^{(t+1)}_{j}}} \right\} - \sum_{l=0}^{t} \sum_{j \in [n] \backslash \{ i \}} G_{j,i}^{(l)} \left\{ - y_{j,i}^{(l)}+ \frac{e^{\theta_{i}^{(t)}}}{e^{\theta_{i}^{(t)}} + e^{\theta_{j}^{(t)}}} \right\}
        \\ & -
        \sum_{j \in [n] \backslash \{i\}} \left( \theta^{(t)}_{j} - \theta^{(t+1)}_{j} \right) \sum_{l=0}^{t} G^{(l)}_{j,i} \frac{e^{\theta_{i}^{(t)}}e^{\theta_{j}^{(t)}}}{(e^{\theta_{i}^{(t)}} + e^{\theta_{j}^{(t)}})^{2}}
        \\ = & 
        \sum_{l=0}^{t} \sum_{j \in [n] \backslash \{ i \}} G^{(l)}_{j,i} \left\{ \frac{e^{\theta_{i}^{(t)}}}{e^{\theta_{i}^{(t)}} + e^{\theta^{(t+1)}_{j}}} - \frac{e^{\theta_{i}^{(t)}}}{e^{\theta_{i}^{(t)}} + e^{\theta_{j}^{(t)}}} - \left( \theta^{(t)}_{j} - \theta^{(t+1)}_{j} \right) \frac{e^{\theta_{i}^{(t)}}e^{\theta_{j}^{(t)}}}{(e^{\theta_{i}^{(t)}} + e^{\theta_{j}^{(t)}})^{2}} \right\}.
    \end{align*}
    By Taylor expansion we know that
    \begin{align*}
        \frac{e^{\theta_{i}^{(t)}}}{e^{\theta_{i}^{(t)}} + e^{\theta^{(t+1)}_{j}}}
        & = \frac{e^{\theta_{i}^{(t)}}}{e^{\theta_{i}^{(t)}} + e^{\theta_{j}^{(t)}}} 
        +
        \frac{e^{\theta_{i}^{(t)}}e^{\theta_{j}^{(t)}}}{(e^{\theta_{i}^{(t)}} + e^{\theta_{j}^{(t)}})^{2}} \left( \theta_{j}^{(t)} - \theta_{j}^{(t+1)} \right)
        + \frac{e^{b_{2}} ( 1 - e^{b_{2}})}{(1 + e^{b_{2}})^{3}} \left( \theta_{j}^{(t)} - \theta_{j}^{(t+1)} \right)^{2},
    \end{align*}
    where $b_{2}$ is the convex combination of $\theta_{i}^{(t)} - \theta_{j}^{(t+1)}$ and $\theta_{i}^{(t)} - \theta_{j}^{(t)}$. As a result, we have
    \begin{align*}
        & \left| \left( \mathcal{L}^{(\leq t)}|_{\theta_{-i}^{(t+1)}} \right)^{\prime} (\theta_{i}^{(t)}) - \left\{ \left[ \nabla \mathcal{L}^{(\leq t)}(\theta^{(t)}) \right]_{i} + \sum_{j \in [n] \backslash \{i\}} \left(\theta^{(t+1)}_{j} - \theta^{(t)}_{j}\right) \left[\nabla^{2} \mathcal{L}^{(\leq t)} (\theta^{(t)})\right]_{i,j} \right\} \right|
        \\ \leq & 
        \sum_{l=0}^{t} \sum_{j \in [n] \backslash \{ i \}} G^{(l)}_{j,i} \left| \frac{e^{\theta_{i}^{(t)}}}{e^{\theta_{i}^{(t)}} + e^{\theta^{(t+1)}_{j}}} - \frac{e^{\theta_{i}^{(t)}}}{e^{\theta_{i}^{(t)}} + e^{\theta_{j}^{(t)}}} - \left( \theta^{(t)}_{j} - \theta^{(t+1)}_{j} \right) \frac{e^{\theta_{i}^{(t)}}e^{\theta_{j}^{(t)}}}{(e^{\theta_{i}^{(t)}} + e^{\theta_{j}^{(t)}})^{2}} \right|
        \\ \leq & 
        \sum_{l=0}^{t} \sum_{j \in [n] \backslash \{ i \}} G^{(l)}_{j,i} \left| \frac{e^{b_{2}} ( 1 - e^{b_{2}})}{(1 + e^{b_{2}})^{3}}  \right| \left( \theta_{j}^{(t)} - \theta_{j}^{(t+1)} \right)^{2}
        \\ \lesssim & 
        (t+1) np \norm{\theta^{(t+1)} - \theta^{(t)}}_{\infty}^{2}
    \end{align*}
    with probability exceeding $1 - \mathcal{O}(n^{-5})$. Recall that
    \begin{align*}
        \overline{\theta}_{i}^{(t+1,\prime)} - \theta_{i}^{(t)}
        & = 
        - \frac{\left[ \nabla \mathcal{L}^{(\leq t)}(\theta^{(t)}) \right]_{i} + \sum_{j \in [n] \backslash \{i\}} \left(\theta^{(t+1)}_{j} - \theta^{(t)}_{j}\right) \left[\nabla^{2} \mathcal{L}^{(\leq t)} (\theta^{(t)})\right]_{i,j}}{\left[\nabla^{2} \mathcal{L}^{(\leq t)}(\theta^{(t)})\right]_{i,i}},
    \end{align*}
    then we can derive
    \begin{align*}
        & \theta^{(t+1)}_{i} - \overline{\theta}^{(t+1, \prime)}_{i} 
        \\ = &
        \left( \theta_{i}^{(t+1)} - \theta^{(t)}_{i} \right) - \left( \overline{\theta}_{i}^{(t+1,\prime)} - \theta_{i}^{(t)} \right)
        \\ =  & 
        - \frac{\left( \mathcal{L}^{(\leq t)}|_{\theta_{-i}^{(t+1)}} \right)^{\prime} (\theta^{(t)}_{i})}{\left( \mathcal{L}^{(\leq t)}|_{\theta_{-i}^{(t+1)}} \right)^{\prime\prime} (b_{1})} 
        + 
        \frac{\left[ \nabla \mathcal{L}^{(\leq t)}(\theta^{(t)}) \right]_{i} + \sum_{j \in [n] \backslash \{i\}} \left(\theta^{(t+1)}_{j} - \theta^{(t)}_{j}\right) \left[\nabla^{2} \mathcal{L}^{(\leq t)} (\theta^{(t)})\right]_{i,j}}{\left[\nabla^{2} \mathcal{L}^{(\leq t)}(\theta^{(t)})\right]_{i,i}}
        \\ = &
        \underbrace{\frac{\left( \mathcal{L}^{(\leq t)}|_{\theta_{-i}^{(t+1)}} \right)^{\prime\prime} (b_{1}) - \left[\nabla^{2} \mathcal{L}^{(\leq t)}(\theta^{(t)})\right]_{i,i}}{\left( \mathcal{L}^{(\leq t)}|_{\theta_{-i}^{(t+1)}} \right)^{\prime\prime} (b_{1}) \cdot \left[\nabla^{2} \mathcal{L}^{(\leq t)}(\theta^{(t)})\right]_{i,i}} \left\{ \left[ \nabla \mathcal{L}^{(\leq t)}(\theta^{(t)}) \right]_{i} + \sum_{j \in [n] \backslash \{i\}} \left(\theta^{(t+1)}_{j} - \theta^{(t)}_{j}\right) \left[\nabla^{2} \mathcal{L}^{(\leq t)} (\theta^{(t)})\right]_{i,j} \right\}}_{B_{1}}
        \\ & + 
        \underbrace{\frac{\left( \mathcal{L}^{(\leq t)}|_{\theta_{-i}^{(t+1)}} \right)^{\prime} (\theta^{(t)}_{i}) - \left\{ \left[ \nabla \mathcal{L}^{(\leq t)}(\theta^{(t)}) \right]_{i} + \sum_{j \in [n] \backslash \{i\}} \left(\theta^{(t+1)}_{j} - \theta^{(t)}_{j}\right) \left[\nabla^{2} \mathcal{L}^{(\leq t)} (\theta^{(t)})\right]_{i,j} \right\}}{\left( \mathcal{L}^{(\leq t)}|_{\theta_{-i}^{(t+1)}} \right)^{\prime\prime} (b_{1})}}_{B_{2}},
    \end{align*}
    where
    \begin{align*}
        \left| B_{1} \right| 
        & \lesssim 
        \frac{(t+1) np \norm{\theta^{(t+1)}-\theta^{(t)}}_{\infty}}{(t+1)np \left[(t+1)np - (t+1) np \norm{\theta^{(t+1)}-\theta^{(t)}}_{\infty}\right]} 
        \\ & \quad \cdot 
        \left| \left[ \nabla \mathcal{L}^{(\leq t)}(\theta^{(t)}) \right]_{i} + \sum_{j \in [n] \backslash \{i\}} \left(\theta^{(t+1)}_{j} - \theta^{(t)}_{j}\right) \left[\nabla^{2} \mathcal{L}^{(\leq t)} (\theta^{(t)})\right]_{i,j} \right|
        \\ & \lesssim
        D_{2} \frac{1}{(t+1)^{2}} \frac{1}{np} \sqrt{\frac{\log n}{npL}} \left| \left[ \nabla \mathcal{L}^{(\leq t)}(\theta^{(t)}) \right]_{i} + \sum_{j \in [n] \backslash \{i\}} \left(\theta^{(t+1)}_{j} - \theta^{(t)}_{j}\right) \left[\nabla^{2} \mathcal{L}^{(\leq t)} (\theta^{(t)})\right]_{i,j} \right|
    \end{align*}
    and
    \begin{align*}
        \left| B_{2} \right|
        \lesssim
        \frac{(t+1)np\norm{\theta^{(t+1)} - \theta^{(t)}}_{\infty}^{2}}{(t+1)np - (t+1) np \norm{\theta^{(t+1)}-\theta^{(t)}}_{\infty}}
        \lesssim
        D_{2}^{2} \frac{1}{(t+1)^{2}} \frac{\log n}{npL}
    \end{align*}
    with probability exceeding $1 - \mathcal{O}(n^{-5})$. On the other hand, we know that
    \begin{align*}
        & \left| \left[ \nabla \mathcal{L}^{(\leq t)}(\theta^{(t)}) \right]_{i} + \sum_{j \in [n] \backslash \{i\}} \left(\theta^{(t+1)}_{j} - \theta^{(t)}_{j}\right) \left[\nabla^{2} \mathcal{L}^{(\leq t)} (\theta^{(t)})\right]_{i,j} \right|
        \\ \leq &
        \left| \left[ \nabla \mathcal{L}^{(\leq t)}(\theta^{(t)}) \right]_{i} \right| + \left| \sum_{j \in [n] \backslash \{i\}} \left(\theta^{(t+1)}_{j} - \theta^{(t)}_{j}\right) \left[\nabla^{2} \mathcal{L}^{(\leq t)} (\theta^{(t)})\right]_{i,j} \right|
        \\ \leq &
        \left| \left[ \nabla \mathcal{L}^{(\leq t)}(\theta^{(t)}) \right]_{i} \right| + \norm{\theta^{(t+1)} - \theta^{(t)}}_{\infty} \sum_{j \in [n] \backslash \{i\}} \left| \left[\nabla^{2} \mathcal{L}^{(\leq t)} (\theta^{(t)})\right]_{i,j} \right|
        \\ \leq &
        D_{19} \sqrt{\frac{np \log n}{L}} + D_{2} \frac{1}{t+1} \sqrt{\frac{\log n}{npL}} D_{21} (t+1)np 
        \\ \leq &
        c_{23} \sqrt{\frac{np \log n}{L}}
    \end{align*}
    for some constant $c_{23} \geq D_{19} + D_{21}$ with probability at least $1 - \mathcal{O}(n^{-5})$. To sum up, we have
    \begin{align*}
        \left| \theta^{(t+1)}_{i} - \overline{\theta}^{(t+1, \prime)}_{i} \right| 
        & \leq 
        \left| B_{1} \right| + \left| B_{2} \right|
        \\ & \leq 
        D_{2} \frac{1}{(t+1)^{2}} \frac{1}{np} \sqrt{\frac{\log n}{npL}} c_{23} \sqrt{\frac{np \log n}{L}} + D_{2}^{2} \frac{1}{(t+1)^{2}} \frac{\log n}{npL}
        \\ & \leq 
        D_{24} \frac{1}{(t+1)^{2}} \frac{\log n}{npL}
    \end{align*}
    for some constant $D_{24} \geq (1 + c_{23})D_{2}^{2}$ with probability at least $1 - \mathcal{O}(n^{-5})$.
\end{proof}
\subsubsection{Proof of Lemma~\ref{lemma:diff_ori_qudra_t+1}}
\label{appendix:diff_ori_qudra_t+1}
\begin{proof}
    Combine Lemma \ref{lemma:diff_quadratic} and  Lemma \ref{lemma:diff_quadratic_MLE}, we know that
    \begin{align*}
        \norm{\theta^{(t+1)} - \overline{\theta}^{(t+1)}}_{\infty} 
        & \leq
        D_{24} \frac{1}{(t+1)^{2}} \frac{\log n}{npL} + D_{23} \frac{1}{t+1} \frac{1}{\sqrt{npL}} \frac{\log n}{\sqrt{nL}p} + D_{23} \frac{1}{t+1} \frac{1}{np} \sqrt{\frac{\log n}{L}} \left( 1 + \frac{\log n}{\sqrt{np}} \right) 
        \\ & + 
        D_{23} \frac{1}{t+1} \left( \frac{\log n}{np} + \frac{1}{\sqrt{np}} \right) \norm{\overline{\theta}^{(t+1)} - \theta^{(t+1)}}_{\infty} 
    \end{align*}
    with probability at least $1 - \mathcal{O}(n^{-5})$. To reveal the constant hidden in the above inequality, we write it as
    \begin{align*}
        \norm{\theta^{(t+1)} - \overline{\theta}^{(t+1)}}_{\infty} 
        & \leq
        \left( D_{23} + D_{24} \right) \left[ \frac{1}{t+1} \frac{1}{\sqrt{npL}} \frac{\log n}{\sqrt{nL}p} + \frac{1}{t+1} \frac{1}{np} \sqrt{\frac{\log n}{L}} \left( 1 + \frac{\log n}{\sqrt{np}} \right) \right]
        \\ & \quad +
        \left( D_{23} + D_{24} \right) \frac{1}{t+1} \left( \frac{\log n}{np} +  \frac{1}{\sqrt{np}} \right) \norm{\overline{\theta}^{(t+1)} - \theta^{(t+1)}}_{\infty} 
    \end{align*}
    with probability at least $1 - \mathcal{O}(n^{-5})$. As a result, as long as $\frac{1}{t+1}  \left( \frac{\log n}{np} + \frac{1}{\sqrt{np}} \right) \leq \frac{1}{2(D_{23}+D_{24})}$,
    we have
    \begin{align*}
        \norm{\theta^{(t+1)} - \overline{\theta}^{(t+1)}}_{\infty} 
        & \leq 
        2(D_{23}+D_{24}) \left[ \frac{1}{t+1} \frac{1}{\sqrt{npL}} \frac{\log n}{\sqrt{nL}p} + \frac{1}{t+1} \frac{1}{np} \sqrt{\frac{\log n}{L}} \left( 1 + \frac{\log n}{\sqrt{np}} \right) \right]
        \\ & \leq
        D_{25}\frac{1}{t+1} \sqrt{\frac{\log n}{npL}} \left[ \frac{\sqrt{\log n}}{\sqrt{nL}p} + \frac{1}{\sqrt{np}} \left( 1 + \frac{\log n}{\sqrt{np}} \right) \right]
    \end{align*}
    for some constant $D_{25} \geq 2(D_{23}+D_{24})$ with probability exceeding $1 - \mathcal{O}(n^{-5})$.
\end{proof}
\end{document}